\documentclass[10pt]{article}

\usepackage{a4wide}
\usepackage{amsmath,amsfonts,amssymb,amsthm,dsfont}

\usepackage{mathrsfs}
\usepackage{amscd}
\newcounter{theorem}
\renewcommand{\thetheorem}{\arabic{section}.\arabic{theorem}}

\newenvironment{thm}[1]{\par
\begin{sloppypar}\refstepcounter{theorem}%
\noindent{\bf #1 \thetheorem.}\it{}}{\end{sloppypar}}

\newenvironment{theorem}{\begin{thm}{Theorem}}{\end{thm}}
\newenvironment{proposition}{\begin{thm}{Proposition}}{\end{thm}}
\newenvironment{corollary}{\begin{thm}{Corollary}}{\end{thm}}
\newenvironment{lemma}{\begin{thm}{Lemma}}{\end{thm}}

\newenvironment{defi}[1]{\par
\begin{sloppypar}\refstepcounter{theorem}%
\noindent{\bf #1 \thetheorem.}\rm{}}{\end{sloppypar}}
\newenvironment{definition}{\begin{defi}{Definition}}{\end{defi}}

\newenvironment{remark}{\begin{defi}{Remark}}{\end{defi}}

\newenvironment{ex}{\begin{defi}{Example}}{\end{defi}}

\def\R{{\rm I\kern-.2em R}}   
   
     \def\cqfd{\hfill $\square$}
\def\X{\mathcal X}

\def\id{{\rm id\hspace*{-1.5pt}l}}
\def\dbl{[\hspace*{-1.5pt}[}
\def\dbr{]\hspace*{-1.5pt}]}

\def\dbar{\;\;\bar{}\!\!\!d}

\def\Rd{\mathbb{R}^d}
\def\S{\mathscr{S}\big(\mathbb{R}^{d}\big)}
\def\S2{\mathscr{S}\big(\mathbb{R}^{2d}\big)}
\def\S'{\mathscr{S}^\prime\big(\mathbb{R}^{d}\big)}
\def\S'2{\mathscr{S}^\prime\big(\mathbb{R}^{2d}\big)}

\def\Ran{\mathfrak{R}an}
\def\bb1{{\rm{1}\hspace{-3pt}\mathbf{l}}}
\def\Sp{\mathcal{S}p}

\begin{document}

\title{The Peierls-Onsager Effective Hamiltonian in a complete gauge covariant setting: Description of the spectrum.}

\date{\today}

\author{Viorel Iftimie\footnote{Institute
of Mathematics Simion Stoilow of the Romanian Academy, P.O.  Box
1-764, Bucharest, RO-70700, Romania.} \ \ and Radu
Purice\footnote{Institute
of Mathematics Simion Stoilow of the Romanian Academy, P.O.  Box
1-764, Bucharest, RO-70700, Romania.}
\footnote{Laboratoire Europ\'een Associ\'e CNRS Franco-Roumain {\it Math-Mode}}}

\maketitle

\begin{abstract}
Using the procedures in \cite{Bu} and \cite{GMS} and the magnetic pseudodifferential calculus we have developped in \cite{MP1,MPR1,IMP1,IMP2} we construct an effective Hamitonian that describes the spectrum in any compact subset of the real axis for a large class of periodic pseudodifferential Hamiltonians in a bounded smooth magnetic field, in a completely gauge covariant setting, without any restrictions on the vector potential and without any adiabaticity hypothesis.

\end{abstract}

\section{Introduction}\label{S.0}

In this paper we consider once again the construction of an effective Hamiltonian for a particle described by a periodic Hamiltonian and subject also to a magnetic field that will be considered bounded and smooth but neither periodic nor slowly varying. Our aim is to use some of the ideas in \cite{Bu,GMS} in conjunction with the magnetic pseudodifferential calculus developed in \cite{MP1,IMP1,IMP2,MPR1} and obtain the following improvements:
\begin{enumerate}
 \item cover also the case of pseudodifferential operators, as for example the relativistic Schr\"{o}dinger operators with principal symbol $<\eta>$;
\item consider magnetic fields that are neither constant nor slowly variable, and thus working in a manifestly covariant form and obtain results that clearly depend only on the magnetic field;
\item give up the adiabatic hypothesis (slowly variable fields) and consider only the intensity of the magnetic field as a small parameter;
\item consider hypothesis formulated only in terms of the magnetic field and not of the vector potential one uses.
\end{enumerate}
Let us point out from the beginning, that as in \cite{GMS} we construct an effective Hamiltonian associated to any compact interval of the energy spectrum but its significance concerns only the description of the real spectrum as a subset of $\mathbb{R}$. In a forthcoming paper our covariant magnetic pseudodifferential calculus will be used in order to construct an effective dynamics associated to any spectral band of the periodic Hamiltonian. Let us mention here that the magnetic pseudodifferential calculus has been used in the Peierls-Onsager problem in \cite{DNL} where some improvements of the results in \cite{PST} are obtained but still in an adiabatic setting. In fact in our following paper we intend to extend the results in \cite{DNL} and construct a more natural framework for the definition of the Peierls-Onsager effective dynamics associated to a spectral band.

Finally let us also point out here that an essential ingredient in the method elaborated in \cite{GMS} is a necessary and sufficient criterion for a tempered distribution to belong to some given Hilbert spaces (Propositions 3.2 and 3.6 in \cite{GMS}). In our 'magnetic' setting some similar criteria have to be proved and this obliges us to some different formulations that allows to avoid a gap in the original proof given in \cite{GMS,DS}.

\subsection{The problem}
\setcounter{equation}{0}
\setcounter{theorem}{0}

We shall constantly use the notation $\X\equiv\Rd$, its dual $\X^*$ being cannonically isomorphic to $\Rd$; let $<.,.>:\X^*\times\X\rightarrow\mathbb{R}$ denote the duality relation. We shall always denote by $\Xi:=\X\times\X^*$ (considered as a symplectic space with the canonical symplectic form $\sigma(X,Y):=\langle\xi,y\rangle-\langle\eta,x\rangle$); we shall denote by $\overline{\Xi}:=\X^*\times\X$.

We shall consider a discrete abelian localy compact subgroup $\Gamma\subset\X$. It is isomorphic to $\mathbb{Z}^d$ and we can view it as a lattice $\Gamma:=\oplus_{j=1}^d\mathbb{Z}e_j$, with $\{e_1,\ldots, e_d\}$ an algebraic basis of $\Rd$ (that we shall call the basis of $\Gamma$).

We consider the quotient group $\Rd/\Gamma$ that is canonically isomorphic to the $d$-dimensional torus $\mathbb{T}\equiv\mathbb{T}_\Gamma\equiv\mathbb{T}$ and let us denote by $\mathfrak{p}:\Rd\ni x\mapsto\hat{x}\in\mathbb{T}$ the canonical projection onto the quotient. Let us consider an \textit{elementary cell}:
$$
E_\Gamma\,=\,\left\{y=\sum\limits_{j=1}^dt_je_j\in\Rd\,\mid\,0\leq t_j<1,\ \forall j\in\{1,\ldots,d\}\right\},
$$
having the interior (as subset of $\Rd$) locally homeomorphic to its projection on $\mathbb{T}$. The dual lattice of $\Gamma$ is then its polar set in $\X^*$ defined as
$$
\Gamma_*\,:=\,\left\{\gamma^*\in\X^*\,\mid\,<\gamma^*,\gamma>\in(2\pi)\mathbb{Z},\ \forall\gamma\in\Gamma\right\}.
$$
Considering the dual basis $\{e^*_1,\ldots,e^*_d\}\subset\X^*$ of $\{e_1,\ldots,e_d\}$, defined by $<e^*_j,e_k>=(2\pi)\delta_{jk}$, we have evidently that $\Gamma_*:=\oplus_{j=1}^d\mathbb{Z}e^*_j$. By definition, we have that $\Gamma_*\subset\X^*$ is the polar of $\Gamma\subset\X$. We define $\mathbb{T}_{\Gamma_*}:=\X^*/\Gamma_*\equiv\mathbb{T}_*\equiv\mathbb{T}_*$ and $E_{\Gamma_*}$ and notice that $\mathbb{T}_{\Gamma_*}$ is isomorphic to the dual group of $\Gamma$ (in the sense of abelian localy compact groups).

We evidently have the following gourp isomorphisms $\X\cong\Gamma\times\mathbb{T}_{\Gamma}$, $\X^*\cong\Gamma_*\times\mathbb{T}_{\Gamma_*}$ that are not topological isomorphisms.

We shall consider the following \textit{free Hamiltonian}:
\begin{equation}\label{0.1}
 H_{0,V}:=-\Delta+V(y),\quad V\in BC^\infty(\mathcal{X},\mathbb{R}),\ \Gamma-\text{periodic},
\end{equation} 
that describes the evolution of an electron in a periodic crystal without external fields. The above operator has a self-adjoint extension in $L^2(\mathcal{X})$ that commutes with the translations $\tau_\gamma$ for any $\gamma\in\Gamma$. We can thus apply the Floquet-Bloch theory. For any $\xi\in\mathcal{X}^*$ we can define the operator 
$$
H_{0,V}(\xi):=\big(D_y+\xi\big)^2+V(y)
$$
that has a self-adjoint extension in $L^2(\mathbb{T})$ that has compact resolvent. Thus its spectrum consists in a growing sequence of finite multiplicity eigenvalues $\lambda_1(\xi)\leq\lambda_2(\xi)\leq...$ that are continuous and $\Gamma^*$-periodic functions of $\xi$. Thus, if we denote by $J_k:=\lambda_k(\mathbb{T}_*)$, we can write that
\begin{equation}
 \sigma\big(H_{0,V}\big)\ =\ \underset{k=1}{\overset{\infty}{\cup}}J_k
\end{equation} 
and it follows that this spectrum is absolutely continuous.

The above analysis implies the following statement that can be considered as \textit{the spectral form of the Onsager-Peierls substitution} in a trivial situation (with 0 magnetic field):
\begin{equation}\label{0.3}
\ \lambda\in\sigma\big(H_{0,V}\big)\ \quad\Longleftrightarrow\quad\ \exists k\geq1\ 0\in\sigma\big(\lambda-\lambda_k(D)\big)\ ,
\end{equation} 
where $\lambda_k(D)$ is the image of the multiplication operator with the function $\lambda_k(\xi)$ on $L^2(\mathcal{X}^*)$ under the conjugation with the Fourier transform (i.e. the Weyl quantization of the symbol $\lambda_k$) and thus defines a bounded self-adjoint operator on $L^2(\mathcal{X})$.

The problem we are interested in, consists in superposing a \textit{magnetic field} $B$ in the above situation; let us first consider a constant magnetic field $B=(B_{jk})_{1\leq j,k\leq d}$ with $B_{jk}=-B_{kj}$. Let us recall that using \textit{the transversal gauge} one can define the following \textit{vector potential} $A=(A_j)_{1\leq j\leq d}$
\begin{equation*}
 A_j(x)\ :=\ -\frac{1}{2}\underset{1\leq k\leq d}{\sum}B_{jk}x_k.
\end{equation*} 
We are considering $A$ as a differential 1-form on $\mathcal{X}$ so that $B$ is the 2-form given by the exterior differential of $A$. Then the associated \textit{magnetic Hamiltonian} is defined as
\begin{equation}
 H_{A,V}:=\big(D+A\big)^2+V(y),
\end{equation} 
that has also a self-adjoint extension in $L^2(\mathcal{X})$. The structure of the spectrum of this operator may be very different of the structure of $\sigma\big(H_{0,V}\big)$ (for example it may be pure point with infinite multiplicity!), but one expects that modulo some small correction (depending on $|B|$), for small $|B|$ the property \ref{0.3} with $D$ replaced by $D+A$ should still be true. More precisely it is conjectured that there exists a symbol $r_k(x,\xi;B,\lambda)$ (in fact a $BC^\infty(\Xi)$ function) such that
\begin{equation}
 \underset{|B|\rightarrow0}{\lim}r_k(x,\xi;B,\lambda)=0\ \text{in}\ BC^\infty(\Xi)
\end{equation} 
and for $\lambda$ in a compact neighborhood of $J_k$ and for small $|B|$ we have that
\begin{equation}\label{0.6}
\ \lambda\in\sigma\big(H_{0A,V}\big)\ \quad\Longleftrightarrow\quad0\in\sigma\big(\lambda-\lambda_k(D+A(x))+r_k(x,D+A(x);B,\lambda)\big),
\end{equation} 
where $r_k(x,D+A(x);B,\lambda)$ is the Weyl quantization of $r_k(x,\xi+A(x);B,\lambda)$.

The first rigorous proof of such a result appeared in \cite{N2} for a simple spectral band (i.e. $\lambda_k(\xi)$ is a non-degenerated eigenvalue of $H_{0,V}(\xi)$ for any $\xi\in\mathcal{X}^*$ and $J_k\cap J_l=\emptyset,\forall l\ne k$). In \cite{HS1} the authors study this case of a simple spectral band but also the general case, by using Wannier functions. In these references the operator appearing on the right hand side of the equivalence \ref{0.6} is considered to act in the Hilbert space $\big[l^2(\Gamma)\big]^N$ (with $N=1$ for the simple spectral band). In fact we shall prove that for a simple spectral band one can replace $l^2(\Gamma)$ with $L^2(\mathcal{X})$. Let us also notice that if one would like to consider also non-constant magnetic fields, then the above Weyl quantization of $A(x)$-dependent symbols gives operators that are not gauge covariant and thus unsuitable for a physical interpretation.

\subsection{The result by Gerard, Martinez and Sj\"{o}strand}

In \cite{GMS} the above three authors consider the evolution of an electron (ignoring the spin) in a periodic crystal under the action of exterior non-constant, slowly varying, magnetic and electric fields. More precisely the magnetic field $B$ is defined as $B=dA$ with a vector potential
\begin{equation}\label{0.7}
 A=(A_1,\ldots,A_d),\quad A_j\in C^\infty(\mathcal{X};\mathbb{R}),\quad \partial^\alpha A_j\in BC^\infty(\mathcal{X})\ \forall |\alpha|\geq1,
\end{equation} 
and the electric potential is described by 
\begin{equation}\label{0.8}
 \phi\in BC^\infty(\mathcal{X};\mathbb{R}).
\end{equation} 
The Hamiltonian is taken to be 
\begin{equation}\label{0.9}
 P_{A,\phi}\ =\ \underset{1\leq j\leq d}{\sum}\big(D_{y_j}+A_j(\epsilon y)\big)^2+V(y)+\phi(\epsilon y),
\end{equation} 
with $|\epsilon|$ small enough; this defines also a self-adjoint operator in $L^2(\mathcal{X})$. In this situation, in order to define an effective Hamiltonian, the authors apply an idea of Buslaev \cite{Bu} (see also \cite{HS1}); this idea consists in ``doubling'' the number of variables and separating the periodic part (that is also ``rapidly varying'') from the non-periodic part (that is also ``slowly varying''). One defines the following operator acting on $\mathcal{X}^2$:
\begin{equation}\label{0.10}
 \widetilde{P}_{A,\phi}\ :=\ \underset{1\leq j\leq d}{\sum}\big(\epsilon D_{x_j}+D_{y_j}+A_j(x)\big)^2+V(y)+\phi(x).
\end{equation} 

Let us point out the very interesting connection between the operators $P_{A,\phi}$ and $\widetilde{P}_{A,\phi}$. If we consider the following change of variables:
$$
\pi_\epsilon:\mathcal{X}^2\rightarrow\mathcal{X}^2,\quad\pi_\epsilon(x,y):=(x-\epsilon y,y),
$$
then for any tempered distribution $F\in\mathscr{S}^\prime(\mathcal{X})$ we have that:
\begin{equation}\label{0.11}
 \big(\widetilde{P}_{A,\phi}\circ\pi^*_\epsilon\big)(\delta_0\otimes F)\ =\ \pi^*_\epsilon\big(\delta_0\otimes(P_{A,\phi}F)\big).
\end{equation} 
Buslaev considers the operator $\widetilde{P}_{A,\phi}$ as a semi-classical operator valued pseudodifferential operator on $\mathcal{X}$ and uses the above remark in order to obtain asymptotic solutions for the equation $P_{A,\phi}u=\lambda u$. Let us develop a little bit this idea in the frame of our previous {\it magnetic pseudodifferential calculus} (\cite{MP1,IMP1,IMP2}), presenting some arguments that will be useful in our proofs.

We recall that given a symbol $a(y,\eta)$ defined on $\Xi$ and a potential vector $A$ defined on $\mathcal{X}$, one can define two 'candidates' for the semi-classical 'magnetic' quantization of the symbol $a$:
\begin{equation}\label{0.12}
 \big(\mathfrak{Op}_{A,h}(a)u\big)(x):=\iint_\Xi e^{i<\eta,x-y>}a\big(\frac{x+y}{2},h\eta+A\big(\frac{x+y}{2}\big)\big)u(y)dy\dbar\eta,\ \forall u\in\mathscr{S}(\mathcal{X}),
\end{equation} 
that is used in \cite{GMS} but is not gauge covariant, and
\begin{equation}\label{0.13}
  \big(\mathfrak{Op}^A_{h}(a)u\big)(x):=\iint_\Xi e^{i<\eta,x-y>}\omega_{h^{-1}A}(x,y)a\big(\frac{x+y}{2},h\eta\big)u(y)dy\dbar\eta,\ \forall u\in\mathscr{S}(\mathcal{X}),
\end{equation} 
with $\omega_{A}(x,y):=\exp\{-i\int_{[x,y]}A\}$, that has been introduced in \cite{MP1} and is gauge covariant. In both the above formulae $h$ is a strictly positive parameter.

For $A=0$ the two quantizations above coincide with the semi-classical Weyl quantization of $a$ denoted by $\mathfrak{Op}_h(a)$. For $h=1$ we use the notations $\mathfrak{Op}_A(a)$, $\mathfrak{Op}^A(a)$ and $\mathfrak{Op}(a)$.

Let us come back now to the operators $P_{A,\phi}$ and $\widetilde{P}_{A,\phi}$ and consider the following notations
$$
A_\epsilon(x):=A(\epsilon x),
$$
\begin{equation}\label{0.14}
 \left\{\begin{array}{lcl}
         p(x,y,\eta)&:=&|\eta|^2+V(y)+\phi(x)\\
	 \widetilde{p}(x,y,\xi,\eta)&:=&|\xi+\eta|^2+V(y)+\phi(x)=p(x,y,\xi+\eta)\\
         \overset{\circ}{p}_\epsilon(y,\eta)&:=&p(\epsilon y,y,\eta).
        \end{array}\right.
\end{equation} 
We evidently have that
\begin{equation}\label{0.15}
 P_{A,\phi}\ =\ \mathfrak{Op}_{A_\epsilon}(\overset{\circ}{p}_\epsilon),
\end{equation} 
while $\widetilde{P}_{A,\phi}$ may be thought as being obtained through the following procedure: one computes the Weyl quantization of $\widetilde{p}$ considered as a symbol in the variables $(y,\eta)\in\Xi$ and obtains an operator valued symbol in the variables $(x,\xi)\in\Xi$, that is then quantized by $\mathfrak{Op}_{A,\epsilon}$:
\begin{equation}\label{0.16}
  \left\{\begin{array}{lcl}
          \mathfrak{q}(x,\xi)&:=&\mathfrak{Op}(\widetilde{p}(x,.,\xi,.)\\
          \widetilde{P}_{A,\phi}&:=&\mathfrak{Op}_{A,\epsilon}(\mathfrak{q}).
         \end{array}\right.
\end{equation} 
The rather strange presence of the parameter $\epsilon$ in \ref{0.9} has to be considered as a reflection of the semi-classical quantization used in \ref{0.16} and of the formula \ref{0.11}. In order to deal with a more natural class of perturbations (thus to eliminate the slow variation hypothesis!) one has to give up the semi-classical quantization in the second step and insert the parameter $\epsilon$ in the symbol (like in \ref{0.15}).

Let us briefly review now the main steps of the argument in \cite{GMS}. As previously remarked, they consider a magnetic field described by a vector potential satisfying \ref{0.7}. They propose to consider the following generalization of $P_{A,\phi}$. 
\begin{itemize}
 \item The starting point is a symbol $p(x,y,\eta)$ that is polynomial in the variable $\eta\in\mathcal{X}^*$ and satisfies the following relations:\\
$i)\quad\ \  p(x,y,\eta)\ =\ \underset{|\alpha|\leq m}{\sum}a_\alpha(x,y)\eta^\alpha,\quad a_\alpha\in BC^\infty(\mathcal{X}\times\mathcal{X};\mathbb{R}),\ m\in\mathbb{N}^*,$\\
$ii)\quad\  a_\alpha(x,y+\gamma)=a_\alpha(x,y),\ \forall|\alpha|\leq m,\ \forall\gamma\in\Gamma, $\\
$iii)\quad \exists c>0\ \text{such that}\ p_m(x,y,\eta):=\underset{|\alpha|=m}{\sum}a_\alpha(x,y)\eta^\alpha\geq c|\eta|^m,\ \forall(x,y)\in\mathcal{X}\times\mathcal{X},\ \forall\eta\in\mathcal{X}^*$ i.e. $p$ is an elliptic symbol (let us notice that this condition implies that $m$ is even).
\item They introduce then the symbols:
$$
\overset{\circ}{p}_\epsilon(y,\eta):=p(\epsilon y,y,\eta),\ \forall\epsilon>0;\qquad\widetilde{p}(x,y,\xi,\eta):=p(x,y,\xi+\eta).
$$
\item The interest is focused on the self-adjoint operator in $L^2(\mathcal{X})$ defined by:
\begin{equation}\label{0.17}
 P_\epsilon\ :=\ \mathfrak{Op}_{A_\epsilon}(\overset{\circ}{p}_\epsilon).
\end{equation} 
\item The auxiliary operator (obtained by doubling the variables) is the self-adjoint operator in $L^2(\mathcal{X}^2)$ defined by:
\begin{equation}\label{0.18}
 \widetilde{P}_\epsilon\ :=\ \mathfrak{Op}_{A,\epsilon}(\mathfrak{q}),\quad\mathfrak{q}(x,\xi):=\mathfrak{Op}\big(\widetilde{p}(x,.,\xi,.)\big).
\end{equation} 
\item One verifies that a relation similar to \ref{0.11} is still verified:
\begin{equation}\label{0.19}
 \big(\widetilde{P}_\epsilon\circ\pi^*_\epsilon\big)\big(\delta_0\otimes F\big)\ =\ \pi^*_\epsilon\big(\delta_0\otimes(P_\epsilon F)\big),\ \forall F\in\mathscr{S}^\prime(\mathcal{X}).
\end{equation} 
\end{itemize}

In order to define an effective Hamiltonian to describe the spectrum of $P_\epsilon$, in \cite{GMS} the authors bring together three important ideas from the literature on the subject.
\begin{enumerate}
 \item First, the idea introduced in \cite{Bu,GRT} of ``doubling the variables" and considering the operator $\widetilde{P}_\epsilon$.
\item Then, the use of an operator valued pseudodifferential calculus, idea introduced in \cite{Bu} and having a rigorous development in \cite{B-K}.
\item The formulation of a Grushin type problem, as proposed in \cite{HS1}.
\end{enumerate}
In the following we shall discuss the use of the Grushin type problem in our spectral problem. The ideas are the following. First one fixes some compact interval $I\subset\mathbb{R}$ and some $\epsilon_0>0$ small enough. Then one has to take into account that $\mathbb{T}$ is compact and thus any elliptic self-adjoint operator in $L^2(\mathbb{T})$ is Fredholm and becomes bijective on specific finite co-dimension subspaces. Thus one can find $N\in\mathbb{N}^*$ and $N$ functions $\phi_j\in C^\infty(\mathcal{X}\times\mathcal{X}\times\mathcal{X}^*)$ (with $1\leq j\leq N$) that are $\Gamma$-periodic in the second variable and such that the following statement is true:
\begin{proposition}
 If we define:
\begin{itemize}
 \item the operator valued symbols
\begin{equation}\label{0.20}
 \left\{\begin{array}{lll}
         R_+:\Xi\rightarrow\mathbb{B}\big(L^2(\mathbb{T});\mathbb{C}^N\big),\ &R_+(x,\xi)u:=\left\{\left<u,\phi_j(x,.,\xi)\right>_{L^2(\mathbb{T})}\right\}_{1\leq j\leq N}\in\mathbb{C}^N,&\forall u\in L^2(\mathbb{T}),\\
&\\
         R_-:\Xi\rightarrow\mathbb{B}\big(\mathbb{C}^N;L^2(\mathbb{T})\big),\ &R_-(x,\xi)c:=\underset{1\leq j\leq N}{\sum}c_j\phi_j(x,.,\xi)\in L^2(\mathbb{T}),&\forall c:=(c_j)_{1\leq j\leq N}\in\mathbb{C}^N,
        \end{array}\right.
\end{equation} 
\item the associated operators obtained by a semi-classical (non-covariant) quantization:
\begin{equation}\label{0.21}
 \boldsymbol{R_\pm}(\epsilon)\ :=\ \mathfrak{Op}_{A,\epsilon}\big(R_\pm\big).
\end{equation} 
\item Then, for $\lambda\in I,\ \epsilon\in(0,\epsilon_0]$, the operator:
\begin{equation}\label{0.22}
 \mathcal{P}_\epsilon\ :=\ \left(\begin{array}{cc}
                                  \widetilde{P}_\epsilon-\lambda&\boldsymbol{R_-}(\epsilon)\\
                                  \boldsymbol{R_+}(\epsilon)&\boldsymbol{0}
                                 \end{array}\right)
\end{equation} 
is self-adjoint in $L^2\big(\mathcal{X}\times\mathbb{T}\big)\oplus L^2\big(\mathcal{X};\mathbb{C}^N\big)$ and has an inverse:
\begin{equation}\label{0.23}
 \mathcal{E}(\epsilon,\lambda)\ :=\ \left(\begin{array}{cc}
                                  \boldsymbol{E}(\epsilon,\lambda)&\boldsymbol{E}_+(\epsilon,\lambda)\\
                                  \boldsymbol{E}_-(\epsilon,\lambda)&\boldsymbol{E}_{-+}(\epsilon,\lambda)
                                 \end{array}\right).
\end{equation} 
\end{itemize}
\end{proposition}
Moreover one can prove then that $\boldsymbol{E}_{-+}(\epsilon,\lambda)=\mathfrak{Op}_{A,\epsilon}\big(E^{-+}_{\epsilon\lambda}\big)$ with $E^{-+}_{\epsilon\lambda}\in BC^\infty\big(\Xi;\mathbb{B}(\mathbb{C}^N)\big)$ uniformly for $(\epsilon,\lambda)\in(0,\epsilon_0]\times I$. Then it is easy to prove that:
\begin{proposition}
 The operator $\boldsymbol{E}_{-+}(\epsilon,\lambda)$ is bounded and self-adjoint in $L^2\big(\mathcal{X};\mathbb{C}^N\big)$ and we have the following equivalence:
\begin{equation}\label{0.24}
\lambda\in\sigma(\widetilde{P}_\epsilon)\quad\Longleftrightarrow\quad 0\in\sigma(\boldsymbol{E}_{-+}(\epsilon,\lambda)).
\end{equation} 
\end{proposition}
Finally, in order to pass from $\widetilde{P}_\epsilon$ to $P_\epsilon$ one makes use of \ref{0.19} and of some unitary transforms of the spaces $l^2(\Gamma)$ and $L^2(\mathcal{X})$. More precisely the following two explicit Hilbert spaces are considered in \cite{GMS}:
\begin{itemize}
 \item 
\begin{equation}\label{0.25}
 \mathfrak{V}_{0,\epsilon}\ :=\ \left\{\ F\in\mathscr{S}^\prime(\mathcal{X})\ \mid\ F=\underset{\gamma\in\Gamma}{\sum}f_\gamma\delta_{\epsilon\gamma},\ \forall (f_\gamma)_{\gamma\in\Gamma}\in l^2(\Gamma)\ \right\},\quad\|F\|_{\mathfrak{V}_{0,\epsilon}}:=\|f\|_{l^2(\Gamma)},
\end{equation} 
that is evidently unitarily equivalent to $l^2(\Gamma)$;
\item 
\begin{equation}\label{0.26}
 \mathfrak{L}_{0,\epsilon}\ :=\ \left\{\ F\in\mathscr{S}^\prime(\mathcal{X}^2)\ \mid\ F(x,y)=\underset{\gamma\in\Gamma}{\sum}v(x)\delta_{0}(x-\epsilon(y-\gamma)),\ \forall v\in L^2(\mathcal{X})\ \right\},\quad\|F\|_{\mathfrak{L}}:=\|v\|_{L^2(\mathcal{X})},
\end{equation} 
that is evidently unitarily equivalent to $L^2(\mathcal{X})$.
\end{itemize}
The following step is to extend the operators $\widetilde{P}_\epsilon$ and $\boldsymbol{R_\pm}(\epsilon)$, considered as pseudodifferential operators, to continuous operators from $\mathscr{S}^\prime(\mathcal{X}^2)$ to $\mathscr{S}^\prime(\mathcal{X}^2)$ and respectively from $\mathscr{S}^\prime(\mathcal{X})$ to $\mathscr{S}^\prime(\mathcal{X})$. Then we can directly restrict them to the subspaces $\mathfrak{L}_{0,\epsilon}$ and respectively $\mathfrak{V}_{0,\epsilon}^N$ and due to the $\Gamma$-periodicity of the initial symbol $p$ the authors prove as in \cite{GMS} that they leave these spaces invariant and that the matrix operator $\mathcal{P}_\epsilon$ still defines an invertible self-adjoint operator in $\mathcal{L}_{0,\epsilon}\oplus\mathfrak{V}_{0,\epsilon}^N$ having as inverse the coresponding restriction of $\mathcal{E}(\epsilon,\lambda)$; moreover the restriction of $\boldsymbol{E}_{-+}(\epsilon,\lambda)$ is a bounded self-adjoint operator in $\mathfrak{V}_{0,\epsilon}^N$ and we still have 
the 
property \ref{0.24}. The remark that allows one to end the analysis is that $P_\epsilon$ acting in $L^2(\mathcal{X})$ is unitarily equivalent with the transformed operator $\widetilde{P}_\epsilon$ acting on $\mathfrak{L}_{0,\epsilon}$ and thus \ref{0.24} implies directly the following equivalence:
\begin{equation}\label{0.27}
\lambda\in\sigma(P_\epsilon)\quad\Longleftrightarrow\quad 0\in\sigma(\boldsymbol{E}_{-+}(\epsilon,\lambda)),
\end{equation} 
with $\boldsymbol{E}_{-+}(\epsilon,\lambda)$ bounded self-adjoint operator on $\mathfrak{V}_{0,\epsilon}^N$ that can be evidently identified with a bounded self-adjoint operator on $[l^2(\Gamma)]^N$.

\subsection{Summary of our results}\label{S.1.3}

Let us briefly comment upon our hypothesis. First the magnetic field $B_\epsilon:=2^{-1}\underset{1\leq j,k\leq d}{\sum}B_{\epsilon,jk}dx_j\wedge dx_k$ is a closed 2-form valued smooth function ($B_{\epsilon,jk}=-B_{\epsilon,kj}$)
\begin{description}
 \item[H.1] For any pair of indices $(j,k)$ between $1$ and $d$ we are given a function $[-\epsilon_0,\epsilon_0]\ni\epsilon\mapsto B_{\epsilon,jk}\in BC^\infty(\mathcal{X};\mathbb{R})$ such that $\underset{\epsilon\rightarrow0}{\lim} B_{\epsilon,jk}=0\ \text{in}\ BC^\infty(\mathcal{X};\mathbb{R})$, for some $\epsilon_0>0$.
\end{description}
Using the transversal gauge we can define a vector potential $A_\epsilon$ (described by a 1-form valued smooth function defined on $[-\epsilon_0,\epsilon_0]$) such that $B_\epsilon=dA_\epsilon$:
\begin{equation}\label{0.28}
 A_{\epsilon,j}(x)\ :=\ -\underset{1\leq k\leq d}{\sum}x_k\int_0^1B_{\epsilon,jk}(sx)sds.
\end{equation} 
We shall not suppose that our vector potential $A_\epsilon$ satisfy \eqref{0.7}, but we shall always suppose the following behavior (that results from our hypothesis (H.1) and \eqref{0.28}):
\begin{equation}\label{0.29}
 \underset{\epsilon\rightarrow0}{\lim} <x>^{-1}A_{\epsilon,j}(x)=0\ \text{in}\ BC^\infty(\mathcal{X};\mathbb{R}).
\end{equation} 

The symbols we are considering are also considered as functions
$$
[-\epsilon_0,\epsilon_0]\ni\epsilon\mapsto p_\epsilon\in C^\infty\big(\mathcal{X}\times\mathcal{X}\times\mathcal{X}^\prime\big)
$$
satisfying conditions of type $S^m_1$ with $m>0$ uniformly in $\epsilon\in[-\epsilon_0,\epsilon_0]$:
\begin{description}
 \item[H.2] $\exists m>0,\ \text{such that}\ \forall(\widetilde{\alpha},\beta)\in\mathbb{N}^{2d}\times\mathbb{N}^d,\ \ \exists C_{\widetilde{\alpha}\beta}>0$ such that
$$
\left|\left(\partial^{\widetilde{\alpha}}_{x,y}\partial^\beta_\eta p_\epsilon\right)(x,y,\eta)\right|\ \leq\ C_{\widetilde{\alpha}\beta}<\eta>^{m-|\beta|},\quad\forall(x,y,\eta)\in\mathcal{X}\times\mathcal{X}\times\mathcal{X}^\prime,\ \forall\epsilon\in[-\epsilon_0,\epsilon_0],
$$
\item[H.3] $\underset{\epsilon\rightarrow0}{\lim}p_\epsilon\ =\ p_0$ in $S^m_1(\mathcal{X})$,
\item[H.4] $\forall\alpha\in\mathbb{N}^d$ with $|\alpha|\geq1$ we have $\underset{\epsilon\rightarrow0}{\lim}\big(\partial^\alpha_xp_\epsilon\big)=0$ in $S^m_1(\mathcal{X})$,
\item[H.5] $p_\epsilon$ is an elliptic symbol uniformly in $\epsilon\in[-\epsilon_0,\epsilon_0]$, i.e. $\exists C>0,\exists R>0$ such that
$$
p_\epsilon(x,y,\eta)\geq C|\eta|^m,\quad\forall(x,y,\eta)\in\mathcal{X}\times\mathcal{X}\times\mathcal{X}^\prime \text{with}\ |\eta|\geq R,\ \forall\epsilon\in[-\epsilon_0,\epsilon_0],
$$
\item[H.6] $p_\epsilon$ is $\Gamma$-periodic with respect to the second variable, i.e.
$$
p_\epsilon(x,y+\gamma,\eta)\ =\ p_\epsilon(x,y,\eta),\quad\forall\gamma\in\Gamma,\ \forall(x,y,\eta)\in\mathcal{X}\times\mathcal{X}\times\mathcal{X}^\prime,\ \forall\epsilon\in[-\epsilon_0,\epsilon_0].
$$
\end{description}
Let us remark here that the hypothesis (H.3) and (H.4) imply that the limit $p_0$ only depends on the second and third variables ($(y,\eta)\in\Xi$) and thus we can write
$$
p_\epsilon(x,y,\eta)\ :=\ p_0(y,\eta)\ +\ r_\epsilon(x,y,\eta),\quad\underset{\epsilon\rightarrow0}{\lim}r_\epsilon(x,y,\eta)=0\ \text{in}\ S^m_1(\mathcal{X}\times\Xi).
$$

Let us also notice that our hypothesis (H.3) is not satisfied if we consider a perturbation of the form (adiabatic electric field) $\phi(\epsilon y)$ but is verified for a perturbation of the form $\epsilon\phi(y)$. One can consider a weaker hypothesis, allowing also for the adiabatic electric field perturbation, without losing the general construction of the effective Hamiltonian, but some consequences that we shall prove would no longer be true.

We associate to our symbols the two types of symbols proposed in \cite{GMS}:
$$
\overset{\circ}{p}_\epsilon(y,\eta):=p_\epsilon(y,y,\eta),\quad\widetilde{p}_\epsilon(x,y,\xi,\eta):=p_\epsilon(x,y,\xi+\eta).
$$

The operator we want to study is 
\begin{equation}\label{0.30}
 P_\epsilon:=\mathfrak{Op}^{A_\epsilon}\big(\overset{\circ}{p}_\epsilon\big).
\end{equation} 
The auxiliary operator is defined as:
\begin{equation}\label{0.31}
 \widetilde{P}_\epsilon:=\mathfrak{Op}^{A_\epsilon}(\mathfrak{q}_\epsilon),\quad\mathfrak{q}_\epsilon(x,\xi):=\mathfrak{Op}\big(\widetilde{p}_\epsilon(x,.,\xi,.)\big).
\end{equation} 
Let us notice that in particular all the above hypothesis are satisfied if we take $B_\epsilon:=\epsilon B$ with $B$ a magnetic field with components of class $BC^\infty(\mathcal{X})$, $A_\epsilon=\epsilon A$ with $A$ a potential vector associated to $B$ by \ref{0.28} and $P_\epsilon$ one of the following possible Schr\"{o}dinger operators:
\begin{equation}\label{0.32}
 P_\epsilon=\underset{1\leq j\leq d}{\sum}\big(D_{y_j}+\epsilon A_j(y)\big)^2+V(y)+\epsilon\phi(y),
\end{equation} 
\begin{equation}\label{0.33}
 P_\epsilon=\mathfrak{Op}^{\epsilon A}(<\eta>)+V(y)+\epsilon\phi(y),
\end{equation} 
\begin{equation}\label{0.34}
 P_\epsilon=\sqrt{\mathfrak{Op}^{\epsilon A}(|\eta|^2)+1}+V(y)+\epsilon\phi(y),
\end{equation} 
where $V$ and $\phi$ satisfy  \ref{0.1} and \ref{0.8}.

In the Proposition \ref{P.A.28} of the Appendix we prove that the difference between \ref{0.33} and \ref{0.34} is of the form $\mathfrak{Op}^{\epsilon A}(q_\epsilon)$ with $\underset{\epsilon\rightarrow0}{\lim}q_\epsilon=0$ in $S^0_1(\Xi)$.

The connection between the operators \ref{0.30} and \ref{0.31} is the following:
\begin{equation}\label{0.35}
\big( \widetilde{P}_\epsilon\circ\pi^*_1\big)\big(\delta_0\otimes F\big)\ =\ \pi^*_1\big(\delta_0\otimes(P_\epsilon F)\big),\ \forall F\in\mathscr{S}^\prime(\mathcal{X}). 
\end{equation} 

In order to define an effective Hamiltonian for $P_\epsilon$ we shall apply the same ideas as in \cite{GMS} with the important remark that the operator valued pseudodifferential calculus we use is not a semi classical calculus but the 'magnetic' calculus so that all our constructions are explicitly gauge covariant. This fact obliges us to a lot of new technical lemmas in order to deal with this new calculus. Our main result is the following:
\begin{theorem}\label{T.0.1}
 We assume the Hypothesis [H.1]-[H.6]. For any compact interval $I\subset\mathbb{R}$ there exists $\epsilon_0>0$ and $N\in\mathbb{N}^*$ such that $\forall\lambda\in I$ and $\forall\epsilon\in[-\epsilon_0,\epsilon_0]$ there exists a bounded self-adjoint operator $\boldsymbol{E}_{-+}(\epsilon,\lambda):=\mathfrak{Op}^{A_\epsilon}\big(E^{-+}_{\epsilon,\lambda}\big)$ acting in $[\mathfrak{V}_{0,1}]^N$, where $E^{-+}_{\epsilon,\lambda}\in BC^\infty\big(\Xi;\mathbb{B}(\mathbb{C}^N)\big)$ uniformly in $(\epsilon,\lambda)\in[-\epsilon_0,\epsilon_0]\times I$ and is $\Gamma^*$-periodic in the variable $\xi\in\mathcal{X}^*$ for which the following equivalence is true:
\begin{equation}\label{0.36}
\lambda\in\sigma(P_\epsilon)\quad\Longleftrightarrow\quad 0\in\sigma\big(\boldsymbol{E}_{-+}(\epsilon,\lambda)\big).
\end{equation} 
\end{theorem}
 
A direct consequence of the above theorem is a stability property for the spectral gaps of the operator $P_\epsilon$ of the same type as that obtained in \cite{AS,N1,AMP} for the Schr\"{o}dinger operator.
\begin{corollary}\label{C.0.2}
 Under the hypothesis [H.1]-[H.6], for any compact interval $K\subset\mathbb{R}$ disjoint from $\sigma(P_0)$, there exists $\epsilon_0>0$ such that $\forall\epsilon\in[-\epsilon_0,\epsilon_0]$ the interval $K$ is disjoint from $\sigma(P_\epsilon)$.
\end{corollary}

In fact we obtain a much stronger result, giving the optimal regularity property but only for $\epsilon=0$ i.e. at vanishing magnetic field.
\begin{proposition}\label{P.6.12}
We denote by $\mathfrak{d}_H(F_1,F_2)$ the Hausdorff distance between the two closed subsets $F_1$ and $F_2$ of $\mathbb{R}$. Then, under the Hypothesis H.1 - H.7 and I.1 - I.3, there exists a strictly positive constant $C$ such that
\begin{equation}\label{6.29a}
\mathfrak{d}_H\left(\sigma\big(P_\epsilon\big)\cap I\,,\,\sigma\big(P_0\big)\cap I\right)\,\leq\, C\epsilon\qquad\forall\epsilon\in[-\epsilon_0,\epsilon_0].
\end{equation}
\end{proposition}

Let us consider now the case of a simple spectral band and generalize the result we discuss previously in this case. By hypothesis we have that $\tau_\gamma P_0=P_0\tau_\gamma$, $\forall\gamma\in\Gamma$ and we can apply the Floquet-Bloch theory. We denote by $\lambda_1(\xi)\leq\lambda_2(\xi)\leq\ldots$ the eigenvalues of the operators $P_{0,\xi}:=\mathfrak{Op}\big(p_0(\cdot,\xi+\cdot)\big)$ that are self-adjoint in $L^2(\mathbb{T})$; they are continuous functions on the torus $\mathbb{T}^{*d}:=\mathcal{X}^*/\Gamma^*$ (and they are even $C^\infty$ in the case of a simple spectral band). Thus $\sigma(P_0)=\underset{j=1}{\overset{d}{\cup}}J_j$ with $J_j:=\lambda_j(\mathbb{T}^{*d})$. Let us consider now the following new Hypothesis:
\begin{description}
 \item[H.7] There exists $k\geq1$ such that $J_k$ {\it is a simple spectral band for} $P_0$, i.e. $\forall\xi\in\mathbb{T}^{*d}$ we have that $\lambda_k(\xi)$ is a non-degenerate eigenvalue of $P_0$ and for any $l\ne k$ we have that $J_l\cap J_k=\emptyset$.
\end{description}
\begin{proposition}\label{P.0.3}
 Assume the hypothesis [H.1]-[H.7] are true and that moreover we have that $p_0(y,-\eta)=p_0(y,\eta)$, $\forall(y,\eta)\in\Xi$. Let $I\subset\mathbb{R}$ be a compact neighborhood of $J_k$ disjoint from $\underset{l\ne k}{\cup}J_l$. Then there exists $\epsilon_0>0$ such that $\forall(\epsilon,\lambda)\in[-\epsilon_0,\epsilon_0]\times I$ in Theorem \ref{T.0.1} we can take $N=1$ and 
\begin{equation}\label{0.37}
 E^{-+}_{\epsilon,\lambda}(x,\xi)\ =\ \lambda-\lambda_k(\xi)+r_{\epsilon,\lambda}(x,\xi),\quad\text{with}\quad\underset{\epsilon\rightarrow0}{\lim}r_{\epsilon,\lambda}=0,\ \text{in}\ BC^\infty(\Xi),\ \text{uniformly in}\ \lambda\in I.
\end{equation} 
\end{proposition}

In the case of a constant magnetic field, under some more assumptions on the symbol $p_\epsilon$ we can have even more information concerning the operator $\boldsymbol{E}_{-+}(\epsilon,\lambda)$.
\begin{proposition}\label{P.0.4}
 Assume that the hypothesis [H.1]-[H.7] are true and that $B_\epsilon$ are constant magnetic fields (for any $\epsilon$) and that the symbols $p_\epsilon$ do not depend on the first variable ($x\in\mathcal{X}$). Then we can complete the conclusion of Theorem \ref{T.0.1} with the following statements:
\begin{enumerate}
 \item $\boldsymbol{E}_{-+}(\epsilon,\lambda)$ is a bounded self-adjoint operator in $[L^2(\mathcal{X})]^N$.
\item The symbol $E^{-+}_{\epsilon,\lambda}$ is independent of the first variable ($y\in\mathcal{X}$) and is $\Gamma^*$-periodic in the second variable ($\xi\in\mathcal{X}^*$).
\end{enumerate}
\end{proposition}

\subsection{Overview of the paper}

The first Section analysis the properties of the auxiliary operator, its self-adjointness and its connection with our main Hamiltonian. The second Section recalls some facts about the Floquet-Bloch theory and the connection between the spectra of the auxiliary operator in $L^2(\mathcal{X}\times\mathcal{X})$ and $L^2(\mathcal{X}\times\mathbb{T})$. In Section 3 we introduce a Grushin type problem and define the principal part of the symbol of the effective Hamiltonian. In Section 4 we use a perturbative method to construct the Peierls effective Hamiltonian and study the connection between its spectrum and the spectrum of the auxiliary operator acting in $L^2(\mathcal{X}\times\mathbb{T})$. Section 5 is devoted to the definition and study of some auxiliary Hilbert spaces of tempered distributions $\mathfrak{V}_0$ and $\mathfrak{L}_0$. In Section 6 we make a rigorous study of the Peierls Hamiltonian reducing the local study of its spectrum to a spectral problem of the effective Hamiltonian acting in $[\mathfrak{V}_0]^N$ (the proof of Theorem \ref{T.0.1}). We also consider an application to the stability of spectral gaps (Corollary \ref{C.0.2}). The Lipschitz regularity of the boundaries of the spectral bands at vanishing magnetic field is also obtained as an application of the main Theorem. The 7-th Section is devoted to the case of a simple spectral band (Proposition \ref{P.0.3}). In Section 8 we consider the particular case of a constant magnetic field (Proposition \ref{P.0.4}). An Appendix is gathering a number of facts concerning operator valued symbols and their associated magnetic pseudodifferential operators and some spaces of periodic distributions. Some of the notations and definitions we use in the paper are introduced in this Appendix.

\section{The auxiliary operator in $L^2(\mathcal{X}\times\mathcal{X})$}\label{S.1}
\setcounter{equation}{0}
\setcounter{theorem}{0}

Let us consider given a family $\{B_\epsilon\}_{\epsilon\in[-\epsilon_0,\epsilon_0]}$ of magnetic fields on $\mathcal{X}$ satisfying Hypothesis H.1 and $\{A_\epsilon\}_{\epsilon\in[-\epsilon_0,\epsilon_0]}$ an associated family of vector potentials (we shall always work with the vector potentials given by formula \eqref{0.23}). Let us also consider a given family of symbols $\{p_\epsilon\}_{\epsilon\in[-\epsilon_0,\epsilon_0]}$ that satisfy the Hypothesis H.2 - H.6.

We shall use the following convention: if $f$ is a function defined on $\mathcal{X}\times\Xi$, we denote by $f^\circ$ the function defined on $\Xi$ by taking the restriction of $f$ at the subset $\Delta_{\mathcal{X}\times\mathcal{X}}\times\mathcal{X}^*$ where $\Delta_{\mathcal{X}\times\mathcal{X}}:=\{(x,x)\in\mathcal{X}\times\mathcal{X}\,\mid\,x\in\mathcal{X}\}$ is the diagonal of the Cartesian product and we denote by $\widetilde{f}$ the function on $\Xi\times\Xi$ defined by the formula $\widetilde{f}(X,Y)\equiv\widetilde{f}(x,\xi,y,\eta):=f(x,y,\xi+\eta)$.

It is evident that with the above hypothesis and notations we have that $p_\epsilon^\circ\in S^m_1(\Xi)$ and is elliptic, both properties being uniform in $\epsilon\in[-\epsilon_0,\epsilon_0]$. Then the operator $P_\epsilon:=\mathfrak{Op}^{A_\epsilon}(p_\epsilon^\circ)$, the main operator we are interested in, is self-adjoint and lower semi-bounded in $L^2(\mathcal{X})$ having the domain $\mathcal{H}^m_{A_\epsilon}(\mathcal{X})$ (the magnetic Sobolev space of order $m$ defined in Definition \ref{D.1.4} below); moreover, with the choice of vector potential that we made, it is essentially self-adjoint on the space of Schwartz test functions $\mathscr{S}(\mathcal{X})$.

Taking into account the example \ref{A.8} it follows that $\{\widetilde{p}_\epsilon\}_{|\epsilon|\leq\epsilon_0}\in S^m_{1,\epsilon}(\mathcal{X}\times\Xi)$ so that by defining $\mathfrak{q}_\epsilon(X):=\mathfrak{Op}\big(\widetilde{p}_\epsilon(X,.)\big)$, we have that for any $s\in\mathbb{R}$ the following is true
$$
\{\mathfrak{q}_\epsilon\}_{|\epsilon|\leq\epsilon_0}\ \in\ S^0_{0,\epsilon}\big(\Xi;\mathbb{B}\big(\mathcal{H}^{s+m}_{\bullet}(\mathcal{X});\mathcal{H}^s_{\bullet}(\mathcal{X})\big)\big)
$$
with the notations introduced at the begining of subsection 9.1 of the Appendix. We can then define the {\it auxiliary operator} $\widetilde{P}_\epsilon:=\mathfrak{Op}^{A_\epsilon}\big(\mathfrak{q}_\epsilon\big)$ that will play a very important role in our arguments. In the following Proposition we collect the properties of the operator $\widetilde{P}_\epsilon$ that result from Proposition \ref{A.7} and exemple \ref{A.8}.

\begin{lemma}\label{L.1.1}
 With the above notations and under the above Hypothesis H.1 - H.6  we have that
\begin{enumerate}
 \item For any $s\in\mathbb{R}$:
$$
\widetilde{P}_\epsilon\ \in\ \mathbb{B}\left(\mathscr{S}\big(\mathcal{X};\mathcal{H}^{s+m}(\mathcal{X})\big);\mathscr{S}\big(\mathcal{X};\mathcal{H}^{s}(\mathcal{X})\big)\right)\ \cap\ \mathbb{B}\left(\mathscr{S}^\prime\big(\mathcal{X};\mathcal{H}^{s+m}(\mathcal{X})\big);\mathscr{S}^\prime\big(\mathcal{X};\mathcal{H}^{s}(\mathcal{X})\big)\right),
$$
uniformly in $\epsilon\in[-\epsilon_0,\epsilon_0]$.
\item 
$
\widetilde{P}_\epsilon\ \in\ \mathbb{B}\left(\mathscr{S}\big(\mathcal{X}^2\big);\mathscr{S}\big(\mathcal{X}^2\big)\right)\ \cap\ \mathbb{B}\left(\mathscr{S}^\prime\big(\mathcal{X}^2\big);\mathscr{S}^\prime\big(\mathcal{X}^2\big)\right)
$,
uniformly in $\epsilon\in[-\epsilon_0,\epsilon_0]$.
\item $\widetilde{P}_\epsilon$ considered as unbounded operator in $L^2(\mathcal{X}^2)$ with domain $\mathscr{S}(\mathcal{X}^2)$ is symetric for any $\epsilon\in[-\epsilon_0,\epsilon_0]$.
\end{enumerate}
\end{lemma}

Let us now discuss some different forms of the operator $\widetilde{P}_\epsilon$ that we shall use further. We consider first the isomorphisms:
\begin{equation}\label{1.1}
 \boldsymbol{\psi}:\mathcal{X}^2\rightarrow\mathcal{X}^2,\quad\boldsymbol{\psi}(x,y)\ :=\ (x,x-y);\qquad\boldsymbol{\psi}^{-1}=\boldsymbol{\psi};
\end{equation}
\begin{equation}\label{1.2}
 \boldsymbol{\chi}:\mathcal{X}^2\rightarrow\mathcal{X}^2,\quad\boldsymbol{\chi}(x,y)\ :=\ (x+y,y);\qquad\boldsymbol{\chi}^{-1}(x,y)\ =\ (x-y,y).
\end{equation}
The operators $\boldsymbol{\psi}^*$ and $\boldsymbol{\chi}^*$ that they induce on $L^2(\mathcal{X})$ ($ \boldsymbol{\psi}^*u:=u\circ\boldsymbol{\psi}$) are evidently unitary.

\begin{lemma}\label{L.1.2}
 For any $u\in\mathscr{S}(\mathcal{X}^2)$, the image $\widetilde{P}_\epsilon u$ may be written in any of the following three equivalent forms:
\begin{equation}\label{1.3}
 \big(\widetilde{P}_\epsilon u\big)(x,y)\ =\ (2\pi)^{-d/2}\int_{\mathcal{X}}\int_{\mathcal{X}^*}e^{i<\eta,y-\tilde{y}>}\omega_{A_\epsilon}(x,x+\tilde{y}-y)\,p_\epsilon\big(x-y+(y+\tilde{y})/2,(y+\tilde{y})/2,\eta\big)\,u(x+\tilde{y}-y,\tilde{y})\,d\tilde{y}\,d\eta,
\end{equation}
\begin{equation}\label{1.4}
 \big(\boldsymbol{\psi}^*\widetilde{P}_\epsilon\boldsymbol{\psi}^* u\big)(x,y)\ =\ \left[\mathfrak{Op}^{A_\epsilon}\left(\big((\id\otimes\tau_y\otimes\id)p_\epsilon\big)^\circ\right)u(.,y)\right](x),
\end{equation}
\begin{equation}\label{1.5}
 \big(\boldsymbol{\chi}^*\widetilde{P}_\epsilon(\boldsymbol{\chi}^*)^{-1} u\big)(x,y)\ =\ \left[\mathfrak{Op}^{(\tau_{-x}A_\epsilon)}\left(\big((\tau_x\otimes\id\otimes\id)p_\epsilon\big)^\circ\right)u(x,.)\right](y).
\end{equation}
\end{lemma}

\begin{proof}
 Let us fix $u\in\mathscr{S}(\mathcal{X}^2)$. Starting from the definitions of $\widetilde{P}_\epsilon$ and $\mathfrak{q}_\epsilon$ and using oscillating integral techniques, we get
$$
\big(\widetilde{P}_\epsilon u\big)(x,y)\ =\ (2\pi)^{-d/2}\int_{\mathcal{X}}\int_{\mathcal{X}^*}e^{i<\xi,x-\tilde{x}>}\omega_{A_\epsilon}(x,\tilde{x})\left[\mathfrak{q}_\epsilon\big((x+\tilde{x})/2,\xi\big)\,u(\tilde{x},.)\right](y)\,d\tilde{x}\,d\xi=
$$
$$
=(2\pi)^{-d}\int_{\mathcal{X}}\int_{\mathcal{X}^*}e^{i<\xi,x-\tilde{x}>}\omega_{A_\epsilon}(x,\tilde{x})\left[\int_{\mathcal{X}}\int_{\mathcal{X}^*}e^{i<\eta,y-\tilde{y}>}\,p_\epsilon\big((x+\tilde{x})/2,(y+\tilde{y})/2,\xi+\eta\big)\,u(\tilde{x},\tilde{y})\,d\tilde{y}\,d\eta\right]d\tilde{x}\,d\xi=
$$
$$
=(2\pi)^{-d}\int_{\mathcal{X}}\int_{\mathcal{X}}\int_{\mathcal{X}^*}e^{i<\eta,y-\tilde{y}>}\,\omega_{A_\epsilon}(x,\tilde{x})\,p_\epsilon\big((x+\tilde{x})/2,(y+\tilde{y})/2,\eta\big)\left[\int_{\mathcal{X}^*}e^{i<\xi,x-\tilde{x}-y+\tilde{y}>}\,d\xi\right]u(\tilde{x},\tilde{y})\,d\tilde{x}\,d\tilde{y}\,d\eta.
$$
By the Fourier inversion theorem the inner oscillating integral is in fact $(2\pi)^{d/2}\delta_0(x-\tilde{x}-y+\tilde{y})=(2\pi)^{d/2}[\tau_{x-y+\tilde{y}}\delta_0](\tilde{x})$ and we can eliminate the integrals over $\xi\in\mathcal{X}^*$ and over $\tilde{x}\in\mathcal{X}$ in order to obtain formula \eqref{1.3}.

In order to prove \eqref{1.4}, we apply \eqref{1.3} to $\boldsymbol{\psi}^*u$, that meaning to replace in \eqref{1.3} $u(x+\tilde{y}-y,\tilde{y})$ with $u(x+\tilde{y}-y,x-y)$ and we finally replace $y$ with $x-y$, obtaining
$$
\big(\boldsymbol{\psi}^*\widetilde{P}_\epsilon\boldsymbol{\psi}^* u\big)(x,y)\ =\ (2\pi)^{-d/2}\int_{\mathcal{X}}\int_{\mathcal{X}^*}e^{i<\eta,x-y-\tilde{y}>}\,\omega_{A_\epsilon}(x,\tilde{y}+y)\,p_\epsilon\big(y+(x-y+\tilde{y})/2,(x-y+\tilde{y})/2,\eta\big)\,u(\tilde{y}+y,y)\,d\tilde{y}\,d\eta.
$$
Changing the integration variable $\tilde{y}$ to $\tilde{y}-y$ we obtain
$$
\big(\boldsymbol{\psi}^*\widetilde{P}_\epsilon\boldsymbol{\psi}^* u\big)(x,y)\ =\ (2\pi)^{-d/2}\int_{\mathcal{X}}\int_{\mathcal{X}^*}e^{i<\eta,x-\tilde{y}>}\,\omega_{A_\epsilon}(x,\tilde{y})\,p_\epsilon\big((x+\tilde{y})/2,(x+\tilde{y})/2-y,\eta\big)\,u(\tilde{y},y)\,d\tilde{y}\,d\eta,
$$
i.e. \eqref{1.4}

The formula \eqref{1.5} can be easily obtained in a similar way, starting with \eqref{1.3} applied to $\big(\boldsymbol{\chi}^*\big)^{-1} u$, i.e. replacing in \eqref{1.3} $u(x+\tilde{y}-y,\tilde{y})$ with $u(x-y,\tilde{y})$ and finally replacing the variable $x\in\mathcal{X}$ with $x+y$; this gives us the result
$$
\big(\boldsymbol{\chi}^*\widetilde{P}_\epsilon(\boldsymbol{\chi}^*)^{-1} u\big)(x,y)\ =\ (2\pi)^{-d/2}\int_{\mathcal{X}}\int_{\mathcal{X}^*}e^{i<\eta,y-\tilde{y}>}\,\omega_{A_\epsilon}(x+y,x+\tilde{y})\,p_\epsilon\big(x+(y+\tilde{y})/2,(y+\tilde{y})/2,\eta\big)\,u(x,\tilde{y})\,d\tilde{y}\,d\eta.
$$
We end the proof of \eqref{1.5} by noticing that the following equalities are true (see also \eqref{A.29}):
$$
\omega_{A_\epsilon}(x+y,x+\tilde{y})=\exp\{-i\int_{[x+y,x+\tilde{y}]}\hspace*{-0.8cm}A_\epsilon\hspace*{0.5cm}\}=\exp\left\{i\left\langle y-\tilde{y},\int_0^1A_\epsilon\big((1-s)(x+y)+s(x+\tilde{y})\big)\,ds\right\rangle\right\}=
$$
$$
=\exp\left\{i\left\langle y-\tilde{y},\int_0^1A_\epsilon\big(x+(1-s)y+s\tilde{y}\big)\,ds\right\rangle\right\}=\exp\left\{i\left\langle y-\tilde{y},\int_0^1\big(\tau_{-x}A_\epsilon\big)\big((1-s)y+s\tilde{y}\big)\,ds\right\rangle\right\}=
$$
$$
=\exp\{-i\int_{[y,\tilde{y}]}\hspace*{-0.5cm}\big(\tau_{-x}A_\epsilon\big)\}=\omega_{(\tau_{-x}A_\epsilon)}(y,\tilde{y}).
$$
\end{proof}
The following Corollary is a direct consequence of Lemma \ref{L.1.2}.
\begin{corollary}\label{C.1.3}
 We have the following two relations between the operators $\widetilde{P}_\epsilon$ and $P_\epsilon$:
\begin{enumerate}
 \item For any $v\in\mathscr{S}^\prime(\mathcal{X})$
\begin{equation}\label{1.8}
 \big(\boldsymbol{\chi}^*\widetilde{P}_\epsilon(\boldsymbol{\chi}^*)^{-1}(\delta_0\otimes v)\big)\ =\ \delta_0\otimes\big(P_\epsilon v\big),
\end{equation}
\begin{equation}\label{1.8'}
 \big(\boldsymbol{\psi}^*\widetilde{P}_\epsilon\boldsymbol{\psi}^*(v\otimes\delta_0)\big)\ =\ \big(P_\epsilon v\big)\otimes\delta_0.
\end{equation}
\item If $p_\epsilon$ does not depend on its second variable $y\in\mathcal{X}$, then the following equality is true:
\begin{equation}\label{1.7}
 \boldsymbol{\psi}^*\widetilde{P}_\epsilon\boldsymbol{\psi}^*\ =\ P_\epsilon\otimes\id.
\end{equation}
\end{enumerate}
\end{corollary}

\begin{proof}
 It is enough to prove \eqref{1.8} for any $v\in\mathscr{S}(\mathcal{X})$. We choose $\varphi\in C^\infty_0(\mathcal{X})$ so that $\{\varphi_n\}_{n\in\mathbb{N}^*}$ with $\varphi_n(x):=n^d\varphi(nx)$ is a $\delta_0$-sequence. We apply now \eqref{1.5} to $u:=\varphi_n\otimes v$ in order to obtain
$$
\big(\boldsymbol{\chi}^*\widetilde{P}_\epsilon(\boldsymbol{\chi}^*)^{-1}(\varphi_n\otimes v)\big)(x,y)=\varphi_n(x)\left[\mathfrak{Op}^{(\tau_{-x}A_\epsilon)}\left(\big((\tau_x\otimes\id\otimes\id)p_\epsilon\big)^\circ\right)v\right](y).
$$
We end the proof of \eqref{1.8} by taking the limit $n\rightarrow\infty$ and noticing that for $x=0$ we obtain in the second factor above: $\big((\tau_0\otimes\id\otimes\id)p_\epsilon\big)^\circ=\big((\id\otimes\id\otimes\id)p_\epsilon\big)^\circ=p_\epsilon$. For \eqref{1.8'} we use exactly the same preocedure starting from \eqref{1.4}.

The equality \eqref{1.7} follows directly from \eqref{1.4} under the given assumption on $p_\epsilon$ that implies that $\big((\id\otimes\tau_y\otimes\id)p_\epsilon\big)^\circ=p_\epsilon^\circ$.
\end{proof}

In order to study the continuity and the self-adjointness of $\widetilde{P}_\epsilon$ we need some more function spaces related to the magnetic Sobolev spaces; in order to define these spaces we shall need the family of operator valued symbols $\{\mathfrak{q}_{s,\epsilon}\}_{(s,\epsilon)\in\mathbb{R}\times[-\epsilon_0,\epsilon_0]}$ introduced in Remark \ref{R.A.25} and the associated operators:
$$
Q_{s,\epsilon}:=\mathfrak{Op}^{A_\epsilon}\big(\mathfrak{q}_{s,\epsilon}\big),\qquad Q^\prime_{s,\epsilon}:=Q_{s,\epsilon}\otimes\id.
$$
Let us still denote by $\widetilde{Q}_{s,\epsilon}:=\boldsymbol{\psi}^*Q^\prime_{s,\epsilon}\boldsymbol{\psi}^*$ with $\boldsymbol{\psi}$ from \eqref{1.1} and let us notice that due to Corollary \ref{C.1.3} (2) the operators $\widetilde{Q}_{s,\epsilon}$ and $Q_{s,\epsilon}$ are in the same relation as the pair $\widetilde{P}_\epsilon$ and $P_\epsilon$.
\begin{definition}\label{D.1.4}
 For magnetic fields $\{B_\epsilon\}_{\epsilon\in[-\epsilon_0,\epsilon_0]}$ verifying Hypothesis H.1 and for choices of vector potentials given by \eqref{0.28} we define the following spaces.
\begin{enumerate}
 \item The magnetic Sobolev space of order $s\in\mathbb{R}$ (as defined in \cite{IMP1}) is 
\begin{equation}\label{1.9}
 \mathcal{H}^s_{A_\epsilon}(\mathcal{X})\ :=\ \left\{\,u\in\mathscr{S}^\prime(\mathcal{X})\,\mid\,Q_{s,\epsilon}u\in L^2(\mathcal{X})\,\right\},
\end{equation}
endowed with the following natural quadratic norm
\begin{equation}\label{1.10}
 \|u\|_{\mathcal{H}^s_{A_\epsilon}(\mathcal{X})}\ :=\ \left\|Q_{s,\epsilon}u\right\|_{L^2(\mathcal{X})},\quad\forall u\in\mathcal{H}^s_{A_\epsilon}(\mathcal{X})
\end{equation}
that makes it a Hilbert space containing $\mathscr{S}(\mathcal{X})$ as a dense subspace.
\item We shall define also $\mathcal{H}^\infty_{A_\epsilon}(\X):=\underset{s\in\mathbb{R}}{\bigcap}\mathcal{H}^s_{A_\epsilon}(\X)$ with the projective limit topology.
\item For $s\in\mathbb{R}$ we consider also the spaces 
\begin{equation}\label{1.11}
 \widetilde{\mathcal{H}}^s_{A_\epsilon}(\mathcal{X}^2)\ :=\ \left\{\,u\in\mathscr{S}^\prime(\mathcal{X}^2)\,\mid\,\widetilde{Q}_{s,\epsilon}u\in L^2(\mathcal{X}^2)\,\right\},
\end{equation}
endowed with the following natural quadratic norm
\begin{equation}\label{1.12}
 \|u\|_{\widetilde{\mathcal{H}}^s_{A_\epsilon}(\mathcal{X}^2)}\ :=\ \left\|\widetilde{Q}_{s,\epsilon}u\right\|_{L^2(\mathcal{X}^2)},\quad\forall u\in\widetilde{\mathcal{H}}^s_{A_\epsilon}(\mathcal{X}^2)
\end{equation}
that makes it a Hilbert space containing $\mathscr{S}(\mathcal{X}^2)$ as a dense subspace.
\end{enumerate}
\end{definition}

\begin{remark}\label{R.1.5}
 $\boldsymbol{\psi}^*$ is a unitary operator from $\widetilde{\mathcal{H}}^s_{A_\epsilon}(\mathcal{X}^2)$ onto $\mathcal{H}^s_{A_\epsilon}(\mathcal{X})\otimes L^2(\mathcal{X})$.
\end{remark}

\begin{proof}
 Let us choose some $u\in\mathscr{S}^\prime(\mathcal{X}^2)$ and notice that
$$
u\in\widetilde{\mathcal{H}}^s_{A_\epsilon}(\mathcal{X}^2)\ \Leftrightarrow\ \boldsymbol{\psi}^*Q^\prime_{s,\epsilon}\boldsymbol{\psi}^*u\in L^2(\mathcal{X}^2)\ \Leftrightarrow\ \big(Q_{s,\epsilon}\otimes\id\big)\boldsymbol{\psi}^*u\in L^2(\mathcal{X}^2)\ \Leftrightarrow\ \boldsymbol{\psi}^*u\in\mathcal{H}^s_{A_\epsilon}(\mathcal{X})\otimes L^2(\mathcal{X})
$$
and evidently we have that
$$
\left\|\boldsymbol{\psi}^*u\right\|_{\mathcal{H}^s_{A_\epsilon}(\mathcal{X})\otimes L^2(\mathcal{X})}\ =\ \|u\|_{\widetilde{\mathcal{H}}^s_{A_\epsilon}(\mathcal{X}^2)}.
$$
\end{proof}

We can prove now a continuity property of the operator $\widetilde{P}_\epsilon$ on the spaces defined by \eqref{1.11}.
\begin{lemma}\label{L.1.16}
 For any $s\in\mathbb{R}$ we have that $\widetilde{P}_\epsilon\in\mathbb{B}\big(\widetilde{\mathcal{H}}^{s+m}_{A_\epsilon}(\mathcal{X}^2);\widetilde{\mathcal{H}}^s_{A_\epsilon}(\mathcal{X}^2)\big)$ uniformly for $\epsilon\in[-\epsilon_0,\epsilon_0]$.
\end{lemma}
\begin{proof}
 Due to \eqref{1.4} and Remark \ref{R.1.5} it is enough to prove that the application
$$
\mathscr{S}(\mathcal{X}^2)\ni u\mapsto\left[\mathfrak{Op}^{A_\epsilon}\left(\big((\id\otimes\tau_y\otimes\id)p_\epsilon\big)^\circ\right)u(.,y)\right](x)\in\mathscr{S}(\mathcal{X}^2)
$$
has a continuous extension from $\mathcal{H}^{s+m}_{A_\epsilon}(\mathcal{X})\otimes L^2(\mathcal{X})$ to $\mathcal{H}^s_{A_\epsilon}(\mathcal{X})\otimes L^2(\mathcal{X})$ uniformly for $\epsilon\in[-\epsilon_0,\epsilon_0]$. But 
$$
\big((\id\otimes\tau_y\otimes\id)p_\epsilon\big)^\circ(x,\xi)\ =\ p_\epsilon(x,x-y,\xi)
$$
and the family $\{p_\epsilon(x,x-y,\xi)\}_{(y,\epsilon)\in\mathcal{X}\times[-\epsilon_0,\epsilon_0]}$ of symbols (in the variables $(x,\xi)\in\Xi$) is a bounded subset of $S^m_1(\Xi)$ (due to our Hypothesis). Applying the continuity properties of magnetic pseudodifferential operators in magnetic Sobolev spaces proved in \cite{IMP1} we conclude that there exists a constant $C>0$ such that 
$$
\left\|\mathfrak{Op}^{A_\epsilon}\left(\big((\id\otimes\tau_y\otimes\id)p_\epsilon\big)^\circ\right)u(.,y)\right\|_{\mathcal{H}^s_{A_\epsilon}(\mathcal{X})}^2\ \leq\ C^2\|u(,.y)\|_{\mathcal{H}^{s+m}_{A_\epsilon}(\mathcal{X})}^2
$$
for any $u\in\mathscr{S}(\mathcal{X}^2)$, $\forall y\in\mathcal{X}$ and $\forall\epsilon\in[-\epsilon_0,\epsilon_0]$. We end then the proof by integrating the above inequality with respect to $y\in\mathcal{X}$.
\end{proof}

In order to prove the self-adjointness of $\widetilde{P}_\epsilon$ in $L^2(\mathcal{X}^2)$ we use the following Remark.
\begin{remark}\label{R.1.17}
 Suppose given $r\in S^m_1(\mathcal{X}\times\Xi)$. Then evidently $r(.,y,.)\in  S^m_1(\Xi)$ uniformly for $y\in\mathcal{X}$. If $B$ is a magnetic field on $\mathcal{X}$ with components of class $BC^\infty(\mathcal{X})$ and $A$ an associated vector potential having components of class $C^\infty_{\text{\sf pol}}(\mathcal{X})$ we define the {\it magnetic pseudodifferential operator with parameter} $y\in\mathcal{X}$
\begin{equation}\label{1.13}
 \big(\mathfrak{R}u\big)(x,y)\ :=\ (2\pi)^{-d/2}\int_{\mathcal{X}}\int_{\mathcal{X}^*}e^{i<\xi,x-\tilde{x}>}\omega_A(x,\tilde{x})\,r\big((x+\tilde{x})/2,y,\xi\big)\,u(\tilde{x},y)\,d\tilde{x}\,d\xi,\quad\forall u\in\mathscr{S}(\mathcal{X}^2),\ \forall(x,y)\in\mathcal{X}^2.
\end{equation}
A straightforward modification of the arguments from \cite{IMP1}, and denoting by $\mathfrak{R}_\epsilon$ the operator defined as above in \eqref{1.13} with a vector potential $A_\epsilon$, allows to prove that
\begin{equation}\label{1.14}
 \mathfrak{R}_\epsilon\ \in\ \mathbb{B}\big(\mathscr{S}(\mathcal{X}^2);\mathscr{S}(\mathcal{X}^2)\big)\ \cap\ \mathbb{B}\big(\mathcal{H}^{s+m}_{A_\epsilon}(\mathcal{X})\otimes L^2(\mathcal{X});\mathcal{H}^{s}_{A_\epsilon}(\mathcal{X})\otimes L^2(\mathcal{X})\big),\quad\forall s\in\mathbb{R}.
\end{equation}
Moreover, if $r$ is elliptic, then for any $u\in L^2(\mathcal{X}^2)$ and any $s\in\mathbb{R}$ we have the equivalence relation:
\begin{equation}\label{1.15}
 u\ \in\ \mathcal{H}^{s+m}_{A_\epsilon}(\mathcal{X})\otimes L^2(\mathcal{X})\ \Longleftrightarrow\ \mathfrak{R}_\epsilon u\ \in\ \mathcal{H}^{s}_{A_\epsilon}(\mathcal{X})\otimes L^2(\mathcal{X}).
\end{equation}
\end{remark}

\begin{proposition}\label{P.1.18}
 $\widetilde{P}_\epsilon$ is a self-adjoint operator in $L^2(\mathcal{X}^2)$ with domain $\widetilde{\mathcal{H}}^m_{A_\epsilon}(\mathcal{X}^2)$. It is essentially self-adjoint on $\mathscr{S}(\mathcal{X}^2)$.
\end{proposition}

\begin{proof}
 Following Lemma \ref{L.1.16} the operator $\widetilde{P}_\epsilon$ with domain $\widetilde{\mathcal{H}}^m_{A_\epsilon}(\mathcal{X}^2)$ is well defined in $L^2(\mathcal{X}^2)$. Moreover we know by Lemma \ref{L.1.1} (3) that $\widetilde{P}_\epsilon$ is symmetric when defined on $\mathscr{S}(\mathcal{X}^2)$ that is dense in $\widetilde{\mathcal{H}}^m_{A_\epsilon}(\mathcal{X}^2)$ for its own norm-topology so that we conclude that $\widetilde{P}_\epsilon$ is symmetric as defined on $\widetilde{\mathcal{H}}^m_{A_\epsilon}(\mathcal{X}^2)$.

 Considering now equation \eqref{1.4} and Remark \ref{R.1.5} it follows that $\widetilde{P}_\epsilon$ as considered in the hypothesis of the Proposition is self-adjoint if and only if the operator $\mathfrak{R}_\epsilon$ defined by  \eqref{1.13} with a symbol $r_\epsilon(x,y,\xi):=p_\epsilon(x,x-y,\xi)$ is self-adjoint in $L^2(\mathcal{X}^2)$ with domain $\mathcal{H}^m_{A_\epsilon}(\mathcal{X})\otimes L^2(\mathcal{X})$. Using the symmetry of $\widetilde{P}_\epsilon$ and its unitary equivalence with $\mathfrak{R}_\epsilon$ we conclude that $\mathfrak{R}_\epsilon$ is also symmetric on its domain.

Let us fix some $v\in\mathcal{D}(\mathfrak{R}_\epsilon^*)$; thus we know that $v\in L^2(\mathcal{X}^2)$ and there exists some $f\in L^2(\mathcal{X}^2)$ such that 
$$
\left(\mathfrak{R}_\epsilon u,v\right)_{L^2(\mathcal{X}^2)}\ =\ \left(u,f\right)_{L^2(\mathcal{X}^2)},\quad\forall u\in\mathscr{S}(\mathcal{X}^2).
$$
We conclude that $\mathfrak{R}_\epsilon v=f$ in $\mathscr{S}^\prime(\mathcal{X}^2)$ and due to \eqref{1.15} we have that $v\in \mathcal{H}^{s}_{A_\epsilon}(\mathcal{X})\otimes L^2(\mathcal{X})$. In conclusion we get that $\mathfrak{R}_\epsilon^*=\mathfrak{R}_\epsilon$.

The last statement of the Proposition follows from the density of $\mathscr{S}(\mathcal{X}^2)$ in $\widetilde{\mathcal{H}}^m_{A_\epsilon}(\mathcal{X}^2)$ and the continuity property proved in Lemma \ref{L.1.16}.
\end{proof}

\section{The auxiliary operator in $L^2\big(\mathcal{X}\times\mathbb{T}\big)$}\label{S.2}
\setcounter{equation}{0}
\setcounter{theorem}{0}

We begin with some elements concerning the Floquet-Bloch theory. We use the notations from the beginning of subsection 9.2 of the Appendix.
\begin{definition}\label{D.2.1}
\hspace*{0.5cm}$\mathscr{S}^\prime_\Gamma\big(\mathcal{X}^2\times\mathcal{X}^*\big):=$
$$
:=\left\{v\in\mathscr{S}^\prime\big(\mathcal{X}^2\times\mathcal{X}^*\big)\,\mid\,v(x,y+\gamma,\theta)=e^{i<\theta,\gamma>}v(x,y,\theta)\,\forall\gamma\in\Gamma,\ v(x,y,\theta+\gamma^*)=v(x,y,\theta)\,\forall\gamma^*\in\Gamma^*\right\},
$$
endowed with the topology induced from $\mathscr{S}^\prime\big(\mathcal{X}^2\times\mathcal{X}^*\big)$.
\end{definition}

\begin{definition}\label{D.2.2}
\hspace*{0.5cm}$\mathscr{F}_0\big(\mathcal{X}^2\times\mathcal{X}^*\big):=\mathscr{S}^\prime_\Gamma\big(\mathcal{X}^2\times\mathcal{X}^*\big)\cap L^2_{\text{\sf loc}}\big(\mathcal{X}^2\times\mathcal{X}^*\big)\cap L^2\big(\mathcal{X}\times E\times E^*\big)$ endowed with the quadratic norm
\begin{equation}\label{2.1}
\|v\|_{\mathscr{F}_0}\ :=\ \sqrt{|E^*|^{-1}\int_{\mathcal{X}}\int_{E}\int_{E^*}|v(x,y,\theta)|^2dx\,dy\,d\theta},\qquad\forall v\in\mathscr{F}_0\big(\mathcal{X}^2\times\mathcal{X}^*\big),
\end{equation}
that makes $\mathscr{F}_0\big(\mathcal{X}^2\times\mathcal{X}^*\big)$ into a Hilbert space.
\end{definition}
We evidently have a continuous embedding of $\mathscr{F}_0\big(\mathcal{X}^2\times\mathcal{X}^*\big)$ into $\mathscr{S}^\prime\big(\mathcal{X}^2\times\mathcal{X}^*\big)$.

\begin{lemma}\label{L.2.3}
 The following map defined on $\mathscr{S}\big(\mathcal{X}^2\big)$:
\begin{equation}\label{2.2}
 \big(\mathcal{U}_\Gamma u\big)(x,y,\theta):=\sum\limits_{\gamma\in\Gamma}e^{i<\theta,\gamma>}u(x,y-\gamma),\qquad \forall(x,y)\in\mathcal{X}^2,\forall\theta\in\mathcal{X}^*,\forall u\in\mathscr{S}\big(\mathcal{X}^2\big),
\end{equation}
extends as a unitary operator $\mathcal{U}_\Gamma:L^2\big(\mathcal{X}^2\big)\rightarrow\mathscr{F}_0\big(\mathcal{X}^2\times\mathcal{X}^*\big)$.

The inverse of the above operator is explicitely given by
\begin{equation}\label{2.3}
 \big(\mathcal{W}_\Gamma v\big)(x,y):=|E^*|^{-1}\int_{E^*}v(x,y,\theta)d\theta,\qquad \forall(x,y)\in\mathcal{X}^2,\,\forall v\in\mathscr{F}_0\big(\mathcal{X}^2\times\mathcal{X}^*\big).
\end{equation}
\end{lemma}
\begin{proof}
 Let us notice that for $u\in\mathscr{S}\big(\mathcal{X}^2\big)$ the series in \eqref{2.2} converges in $\mathscr{E}\big(\mathcal{X}^2\times\mathcal{X}^*\big)$ and its sum, that we denoted by $\mathcal{U}_\Gamma u$ evidently verifies the two relations that characterize the subspace $\mathscr{S}^\prime_\Gamma\big(\mathcal{X}^2\times\mathcal{X}^*\big)$ as subspace of $\mathscr{S}^\prime\big(\mathcal{X}^2\times\mathcal{X}^*\big)$:
\begin{equation}\label{2.4}
 \big(\mathcal{U}_\Gamma u\big)(x,y,\theta+\gamma^*)=\big(\mathcal{U}_\Gamma u\big)(x,y,\theta),\qquad\forall(x,y,\theta)\in\mathcal{X}^2\times\mathcal{X}^*,\ \forall\gamma^*\in\Gamma^*,
\end{equation}
\begin{equation}\label{2.5}
 \big(\mathcal{U}_\Gamma u\big)(x,y+\alpha,\theta)=\sum\limits_{\gamma\in\Gamma}e^{i<\theta,\gamma>}u(x,y+\alpha-\gamma)=e^{i<\theta,\alpha>}\big(\mathcal{U}_\Gamma u\big)(x,y,\theta),\quad\forall(x,y,\theta)\in\mathcal{X}^2\times\mathcal{X}^*,\ \forall\alpha\in\Gamma.
\end{equation}

In particular, considering a fixed pair $(x,y)\in\mathcal{X}^2$, we obtain an element $ \big(\mathcal{U}_\Gamma u\big)(x,y,.)\in\mathscr{S}\big(\mathbb{T}^{*,d}\big)$ and due to Remark \ref{A.10} we can compute its Fourier series that converges in $\mathscr{S}\big(\mathbb{T}^{*,d}\big)$:
\begin{equation}\label{2.6}
 \big(\mathcal{U}_\Gamma u\big)(x,y,\theta)=\sum\limits_{\gamma\in\Gamma}\big(\widehat{\mathcal{U}_\Gamma u}\big)(x,y,\gamma)e^{i<\theta,\gamma>},\quad\big(\widehat{\mathcal{U}_\Gamma u}\big)(x,y,\gamma)=|E^*|^{-1}\int_{E^*}e^{-i<\theta,\gamma>}\big(\mathcal{U}_\Gamma u\big)(x,y,\theta)\,d\theta.
\end{equation}
We also have the Parseval equality
\begin{equation}\label{2.7}
 \sum\limits_{\gamma\in\Gamma}\left|\big(\widehat{\mathcal{U}_\Gamma u}\big)(x,y,\gamma)\right|^2\ =\ |E^*|^{-1}\int_{E^*}\left|\big(\mathcal{U}_\Gamma u\big)(x,y,\theta)\right|^2d\theta.
\end{equation}
Comparing \eqref{2.2} and \eqref{2.6} we conclude that $\big(\widehat{\mathcal{U}_\Gamma u}\big)(x,y,\gamma)=u(x,y-\gamma)$; replacing then in \eqref{2.7} we get the equality
\begin{equation}\label{2.8}
 |E^*|^{-1}\int_{E^*}\left|\big(\mathcal{U}_\Gamma u\big)(x,y,\theta)\right|^2d\theta\ =\ \sum\limits_{\gamma\in\Gamma}\left|u(x,y-\gamma)\right|^2.
\end{equation}
If we integrate the above equality over $\mathcal{X}\times E$ we obtain that
\begin{equation}\label{2.9}
 \sqrt{|E^*|^{-1}\int_{\mathcal{X}}\int_{E}\int_{E^*}\left|\big(\mathcal{U}_\Gamma u\big)(x,y,\theta)\right|^2dx\,dy\,d\theta}\ =\ \sqrt{\int_{\mathcal{X}}\int_{\mathcal{X}}|u(x,y)|^2dx\,dy}.
\end{equation}
We conclude that $\mathcal{U}_\Gamma u\in\mathscr{F}_0\big(\mathcal{X}^2\times\mathcal{X}^*\big)$ and $\mathcal{U}_\Gamma$ extends to an isometry $\mathcal{U}_\Gamma:L^2\big(\mathcal{X}^2\big)\rightarrow\mathscr{F}_0\big(\mathcal{X}^2\times\mathcal{X}^*\big)$. Thus, to end our proof it is enough to prove that the operator $\mathcal{W}_\Gamma$ is an inverse for $\mathcal{U}_\Gamma$. Let us consider some $v\in\mathscr{F}_0\big(\mathcal{X}^2\times\mathcal{X}^*\big)$, then for almost every $x\in\mathcal{X}$ and $y\in E$ we can write that
\begin{equation}\label{2.10}
 v(x,y,\theta)=\sum\limits_{\gamma\in\Gamma}\hat{v}_\gamma(x,y)e^{i<\theta,\gamma>},\quad\text{in }L^2(E^*),
\end{equation}
with
\begin{equation}\label{2.11}
 \hat{v}_\gamma(x,y)=|E^*|^{-1}\int_{E^*}e^{-i<\theta,\gamma>}v(x,y,\theta)d\theta=|E^*|^{-1}\int_{E^*}v(x,y-\gamma,\theta)d\theta=\hat{v}_0(x,y-\gamma).
\end{equation}
Using the above identity \eqref{2.11} and the Parseval equality we notice that we have the following equalities that finally imply that $\mathcal{W}_\Gamma v\in L^2(\mathcal{X}^2)$ and the fact that the map $\mathcal{W}_\Gamma:\mathscr{F}_0\big(\mathcal{X}^2\times\mathcal{X}^*\big)\rightarrow L^2(\mathcal{X}^2)$ is an isometry.
$$
\|\mathcal{W}_\Gamma v\|^2_{L^2(\mathcal{X}^2)}=\int_{\mathcal{X}}\int_{\mathcal{X}}|\hat{v}_0(x,y)|^2dx\,dy=\int_{\mathcal{X}}\left[\sum\limits_{\gamma\in\Gamma}\int_{E}|\hat{v}_0(x,y-\gamma)|^2dy\right]dx=\int_{\mathcal{X}}\left[\sum\limits_{\gamma\in\Gamma}\int_{E}|\hat{v}_\gamma(x,y)|^2dy\right]dx=
$$
$$
=\int_{\mathcal{X}}\left[\int_{E}\sum\limits_{\gamma\in\Gamma}|\hat{v}_\gamma(x,y)|^2dy\right]dx=\int_{\mathcal{X}}\left[\int_{E}|E^*|^{-1}\int_{E^*}|v(x,y,\theta)|^2d\theta\,dy\right]dx=\|v\|_{\mathscr{F}_0\big(\mathcal{X}^2\times\mathcal{X}^*\big)}^2.
$$
Moreover, for any $u\in\mathscr{S}(\mathcal{X}^2)$ we have that
$$
\big(\mathcal{W}_\Gamma\mathcal{U}_\Gamma u\big)(x,y)=|E^*|^{-1}\int_{E^*}\left[\sum\limits_{\gamma\in\Gamma}e^{i<\theta,\gamma>}u(x,y-\gamma)\right]d\theta=|E^*|^{-1}\int_{E^*}u(x,y)d\theta=u(x,y),\qquad\forall(x,y)\in\mathcal{X}^2.
$$
\end{proof}

\begin{lemma}\label{L.2.4}
 With the above definitions for the operators $\mathcal{U}_\Gamma$ and $\mathcal{W}_\Gamma$ we have that
\begin{enumerate}
 \item $\mathcal{U}_\Gamma$ admits a continuous extension to $\mathscr{S}^\prime(\mathcal{X}^2)$ with values in $\mathscr{S}^\prime_\Gamma(\mathcal{X}^2\times\mathcal{X}^*)$.
\item $\mathcal{W}_\Gamma$ admits a continuous extension to $\mathscr{S}^\prime_\Gamma(\mathcal{X}^2\times\mathcal{X}^*)$ with values in $\mathscr{S}^\prime(\mathcal{X}^2)$.
\item We have the equalities: $\mathcal{U}_\Gamma\mathcal{W}_\Gamma=\id_{\mathscr{S}^\prime_\Gamma(\mathcal{X}^2\times\mathcal{X}^*)}$, $\mathcal{W}_\Gamma\mathcal{U}_\Gamma=\id_{\mathscr{S}^\prime(\mathcal{X}^2)}$.
\end{enumerate}
\end{lemma}
\begin{proof}
 Let us consider a tempered distribution $u\in\mathscr{S}^\prime(\mathcal{X}^2)$; in order to prove the convergence of the series \eqref{2.2} in the sense of tempered distributions on $\mathcal{X}^2\times\mathcal{X}^*)$ we choose a test function $\varphi\in\mathscr{S}(\mathcal{X}^2\times\mathcal{X}^*)$ and notice that for any $\nu\in\mathbb{N}$ there exist a seminorm $\|.\|_\nu$ on $\varphi\in\mathscr{S}(\mathcal{X}^2\times\mathcal{X}^*)$ and a seminorm $\|.\|^\prime_\nu$ on $\mathscr{S}^\prime(\mathcal{X}^2)$ such that the following is true:
\begin{equation}\label{2.12}
 \left|\left\langle e^{i<\theta,\gamma>}u(x,y-\gamma),\varphi(x,y,\theta)\right\rangle\right|= \left|\left\langle u(x,y),\int_{\mathcal{X}^*}e^{i<\theta,\gamma>}\varphi(x,y+\gamma,\theta)d\theta\right\rangle\right|\leq\|u\|^\prime_\nu\|\varphi\|_\nu<\gamma>^{-\nu}.
\end{equation}
This last inequality implies the convergence of the series \eqref{2.2} in the sense of tempered distributions on $\mathcal{X}^2\times\mathcal{X}^*)$ and the fact that the map $\mathcal{U}_\Gamma:\mathscr{S}^\prime(\mathcal{X}^2)\rightarrow\mathscr{S}^\prime(\mathcal{X}^2\times\mathcal{X}^*)$ is continuous. The fact that $\mathcal{U}_\Gamma u$ belongs to $\mathscr{S}^\prime_\Gamma(\mathcal{X}^2\times\mathcal{X}^*)$ results either by a direct calculus or by approximating with test functions.

Let us fix now some $v\in\mathscr{S}^\prime_\Gamma(\mathcal{X}^2\times\mathcal{X}^*)$; then, for any test function $\vartheta\in\mathscr{S}(\mathcal{X}^2)$, the application
\begin{equation}\label{2.13}
 v_\vartheta:\mathscr{S}(\mathcal{X})\rightarrow\mathbb{C},\quad\langle v_\vartheta,\varphi\rangle:=\langle v,\vartheta\otimes\varphi\rangle,\quad\forall\varphi\in\mathscr{S}(\mathcal{X}^*)
\end{equation}
is a $\Gamma^*$-periodic tempered distribution that can be canonically identified with an element from $\mathscr{S}^\prime\big(\mathbb{T}\big)$.

We define $\mathcal{W}_\Gamma v$ by
\begin{equation}\label{2.14}
 \left\langle \mathcal{W}_\Gamma v,\vartheta\right\rangle:=|E^*|^{-1}\langle v_\vartheta,1\rangle_{\mathbb{T}^{*,d}},\quad\forall\vartheta\in\mathscr{S}(\mathcal{X}^2).
\end{equation}
It is straightforward to verify that $\mathcal{W}_\Gamma v\in\mathscr{S}^\prime(\mathcal{X}^2)$ and that the application $\mathcal{W}_\Gamma:\mathscr{S}^\prime_\Gamma(\mathcal{X}^2\times\mathcal{X}^*)\rightarrow\mathscr{S}^\prime(\mathcal{X}^2)$ is linear and continuous. Moreover it is evident that for the case $v\in\mathscr{F}_0\big(\mathcal{X}^2\times\mathcal{X}^*\big)$, the definition \eqref{2.14} coincides with \eqref{2.3}.

The equality $\mathcal{W}_\Gamma\mathcal{U}_\Gamma=\id_{\mathscr{S}^\prime(\mathcal{X}^2)}$ results from the one valid on the test functions by the density of $\mathscr{S}(\mathcal{X}^2)$ in $\mathscr{S}^\prime(\mathcal{X}^2)$. In order to prove the other equality we notice that for any $v\in\mathscr{S}^\prime_\Gamma(\mathcal{X}^2\times\mathcal{X}^*)$ and any $\vartheta\in\mathscr{S}(\mathcal{X}^2)$ the tempered distribution $v_\vartheta$ defined in \eqref{2.13} belongs to $\mathscr{S}^\prime\big(\mathbb{T}^{*,d}\big)$ and thus may be written as the sum of a Fourier series converging as tempered distribution in $\mathscr{S}^\prime\big(\mathcal{X}^*\big)$:
\begin{equation}\label{2.15}
 v_\vartheta\ =\ |E^*|^{-1}\sum\limits_{\gamma\in\Gamma}\left\langle v_\vartheta,e^{-i<.,\gamma>}\right\rangle_{\mathbb{T}^{*,d}}e^{i<.,\gamma>}.
\end{equation}
But, from \eqref{2.2} we have that
$$
\big(\mathcal{U}_\Gamma\mathcal{W}_\Gamma v\big)(x,y,\theta)=\sum\limits_{\gamma\in\Gamma}e^{i<\theta,\gamma>}\big(\mathcal{W}_\Gamma v\big)(x,y-\gamma),\quad\text{in }\mathscr{S}^\prime(\mathcal{X}^2\times\mathcal{X}^*).
$$
Let us notice that due to \eqref{2.14} we can write
$$
\left\langle \big(\mathcal{W}_\Gamma v\big)(x,y-\gamma),\vartheta(x,y)\right\rangle=\left\langle \big(\mathcal{W}_\Gamma v\big)(x,y),\vartheta(x,y+\gamma)\right\rangle=|E^*|^{-1}\left\langle v_{(\id\otimes\tau_{-\gamma})\vartheta},1\right\rangle_{\mathbb{T}^{*,d}}.
$$
On the other hand, from \eqref{2.13} we deduce that $\forall\varphi\in\mathscr{S}(\mathcal{X}^*)$ one has that
\begin{equation}\label{2.16}
 \left\langle v_{(\id\otimes\tau_{-\gamma})\vartheta},\varphi\right\rangle=\left\langle v,\left[\big(\id\otimes\tau_{-\gamma}\big)\vartheta\right]\otimes\varphi\right\rangle=\left\langle\big(\id\otimes\tau_\gamma\otimes\id\big)v,\vartheta\otimes\varphi\right\rangle=
\end{equation}
$$
=\left\langle v,\vartheta\otimes\big(e^{-i<.,\gamma>}\varphi\big)\right\rangle=\left\langle v_\vartheta,e^{-i<.,\gamma>}\varphi\right\rangle.
$$
Let us recall the relation between a $\Gamma^*$-periodic tempered distribution and the distribution it induces on the torus (as described in the Remark \ref{A.9}): if $\psi\in C^\infty_0(\mathcal{X}^*)$ is choosen such that $\sum\limits_{\gamma^*\in\Gamma^*}\tau_{\gamma^*}\psi=1$ on $\mathcal{X}^*$ (such a choice is evidently possible), then for any $\rho\in\mathscr{S}(\mathbb{T}^{*,d})$ we have that
$$
\left\langle v_{(\id\otimes\tau_{-\gamma})\vartheta},\rho\right\rangle_{\mathbb{T}^{*,d}}=\left\langle v_{(\id\otimes\tau_{-\gamma})\vartheta},\psi\rho\right\rangle_{\mathcal{X}},\qquad\left\langle v_\vartheta,\rho\right\rangle_{\mathbb{T}^{*,d}}=\left\langle v_\vartheta,\psi\rho\right\rangle_{\mathcal{X}}.
$$
Thus it follows that \eqref{2.16} implies that
$$
\left\langle v_{(\id\otimes\tau_{-\gamma})\vartheta},1\right\rangle_{\mathbb{T}^{*,d}}=\left\langle v_{(\id\otimes\tau_{-\gamma})\vartheta},\psi\right\rangle_{\mathcal{X}}=\left\langle v_\vartheta,e^{-i<.,\gamma>}\psi\right\rangle_{\mathcal{X}}=\left\langle v_\vartheta,e^{-i<.,\gamma>}\right\rangle_{\mathbb{T}^{*,d}}.
$$
We conclude that
$$
\big(\mathcal{U}_\Gamma\mathcal{W}_\Gamma v\big)_\vartheta\ =\ |E^*|^{-1}\sum\limits_{\gamma\in\Gamma}\left\langle v_\vartheta,e^{-i<.,\gamma>}\right\rangle_{\mathbb{T}^{*,d}}e^{i<.,\gamma>}\ =\ v_\vartheta,
$$
so that finally we obtain that $\mathcal{U}_\Gamma\mathcal{W}_\Gamma=\id_{\mathscr{S}^\prime_\Gamma(\mathcal{X}^2\times\mathcal{X}^*)}$.
\end{proof}

\begin{lemma}\label{L.2.5}
 Let $\widetilde{P}_\epsilon$ be the operator defined in Section 1 and $\widetilde{P}_{\epsilon,\Gamma}:=\widetilde{P}_\epsilon\otimes\id$.
\begin{enumerate}
 \item $\widetilde{P}_{\epsilon,\Gamma}$ is a linear continuous operator in $\mathscr{S}^\prime_\Gamma\big(\mathcal{X}^2\times\mathcal{X}^*\big)$.
\item $\mathcal{U}_\Gamma\widetilde{P}_\epsilon=\widetilde{P}_{\epsilon,\Gamma}\mathcal{U}_\Gamma$ on $\mathscr{S}^\prime\big(\mathcal{X}^2\big)$.
\end{enumerate}
\end{lemma}
\begin{proof}
 We evidently have that $\widetilde{P}_{\epsilon,\Gamma}:\mathscr{S}\big(\mathcal{X}^2\times\mathcal{X}^*\big)\rightarrow\mathscr{S}\big(\mathcal{X}^2\times\mathcal{X}^*\big)$ and $\widetilde{P}_{\epsilon,\Gamma}:\mathscr{S}^\prime\big(\mathcal{X}^2\times\mathcal{X}^*\big)\rightarrow\mathscr{S}^\prime\big(\mathcal{X}^2\times\mathcal{X}^*\big)$ are linear and continuous. It is thus sufficient to prove that $\forall v\in\mathscr{S}\big(\mathcal{X}^2\times\mathcal{X}^*\big)$ we have that:
\begin{equation}\label{2.17}
 \left[\Big(\widetilde{P}_\epsilon\otimes\id\Big)v\right](x,y,\theta+\gamma^*)=\left[\Big(\widetilde{P}_\epsilon\otimes\id\Big)(\id\otimes\id\otimes\tau_{-\gamma^*})v\right](x,y,\theta),\quad\forall(x,y,\theta)\in\mathcal{X}^2\times\mathcal{X}^*,\ \forall\gamma^*\in\Gamma^*
\end{equation}
and
\begin{equation}\label{2.18}
\left[\Big(\widetilde{P}_\epsilon\otimes\id\Big)v\right](x,y+\gamma,\theta)=\left[\Big(\widetilde{P}_\epsilon\otimes\id\Big)(\id\otimes\tau_{-\gamma}\otimes\id)v\right](x,y,\theta),\quad\forall(x,y,\theta)\in\mathcal{X}^2\times\mathcal{X}^*,\ \forall\gamma\in\Gamma.
\end{equation}
While the equality \eqref{2.17} is obvious, for the equality \eqref {2.18} we use \eqref{1.3} (with $\tilde{y}$ replaced by $\tilde{y}+\gamma$) and the $\Gamma$-periodicity of $p_\epsilon$ with respect to the second variable.

For the second point of the Lemma we use \eqref{2.2} and \eqref{2.18} and notice that for any $u\in\mathscr{S}(\mathcal{X}^2)$ and for any $(x,y,\theta)\in\mathcal{X}^2\times\mathcal{X}^*$ we have that
$$
\Big(\widetilde{P}_{\epsilon,\Gamma}\mathcal{U}_\Gamma u\Big)(x,y,\theta)=\sum\limits_{\gamma\in\Gamma}e^{i<\theta,\gamma>}\left[\widetilde{P}_\epsilon\big(\id\otimes\tau_{\gamma}\big)u\right](x,y)=\sum\limits_{\gamma\in\Gamma}e^{i<\theta,\gamma>}\big(\widetilde{P}_\epsilon u\big)(x,y-\gamma)=\Big(\mathcal{U}_\Gamma\widetilde{P}_\epsilon u\Big)(x,y,\theta)
$$
\end{proof}

We shall study the self-adjointness of $\widetilde{P}_\epsilon$ acting in some new spaces of functions that are periodic in one argument. Let us first consider the operator $\widetilde{P}_{\epsilon,\Gamma}$.

\begin{definition}\label{D.2.6}
 Recalling the operator $\widetilde{Q}_{s,\epsilon}$ from Definition \ref{D.1.4} (2) we define the operator $\widetilde{Q}_{s,\epsilon,\Gamma}:=\widetilde{Q}_{s,\epsilon}\otimes\id$.
\end{definition}

As we have already noticed about the Definition \ref{D.1.4} (2), the operator $\widetilde{Q}_{s,\epsilon}$ is obtained from $Q_{s,\epsilon}$  by ''doubling the variable`` in the same way as $\widetilde{P}_\epsilon$ is associated to $P_\epsilon$. Then we have a result similar to Lemma \ref{L.2.5} and deduce that $\widetilde{Q}_{s,\epsilon,\Gamma}$ is a linear continuous operator in $\mathscr{S}^\prime_\Gamma\big(\mathcal{X}^2\times\mathcal{X}^*\big)$.

\begin{definition}\label{D.2.6.b}
 For any $s\in\mathbb{R}$ we define 
$$
\mathscr{F}_{s,\epsilon}\big(\mathcal{X}^2\times\mathcal{X}^*\big):=\left\{v\in\mathscr{S}^\prime_\Gamma\big(\mathcal{X}^2\times\mathcal{X}^*\big)\,\mid\,\widetilde{Q}_{s,\epsilon,\Gamma}v\in\mathscr{F}_0\big(\mathcal{X}^2\times\mathcal{X}^*\big)\right\}
$$
that is evidently a Hilbert space for the quadratic norm
\begin{equation}\label{2.19}
 \|v\|_{\mathscr{F}_{s,\epsilon}}\ :=\ \left\|\widetilde{Q}_{s,\epsilon,\Gamma}v\right\|_{\mathscr{F}_0}\qquad\forall v\in\mathscr{F}_{s,\epsilon}\big(\mathcal{X}^2\times\mathcal{X}^*\big).
\end{equation}
\end{definition}

\begin{lemma}\label{L.2.7}
 Let $\widetilde{\mathcal{H}}^s_{A_\epsilon}\big(\mathcal{X}^2\big)$ be the Hilbert space defined in \eqref{1.11}. Then $\mathcal{U}_\Gamma:\widetilde{\mathcal{H}}^s_{A_\epsilon}\big(\mathcal{X}^2\big)\rightarrow\mathscr{F}_{s,\epsilon}\big(\mathcal{X}^2\times\mathcal{X}^*\big)$ is unitary.
\end{lemma}

\begin{proof}
 Let us pick $u\in\widetilde{\mathcal{H}}^s_{A_\epsilon}\big(\mathcal{X}^2\big)$; thus we know that $u\in\mathscr{S}^\prime\big(\mathcal{X}^2\big)$ and $\widetilde{Q}_{s,\epsilon}u\in L^2(\mathcal{X}^2)$. We denote by $v:=\mathcal{U}_\Gamma u\in\mathscr{S}^\prime_\Gamma\big(\mathcal{X}^2\times\mathcal{X}^*\big)$. From Lemma \ref{L.2.5} we deduce that $\mathcal{U}_\Gamma\widetilde{Q}_{s,\epsilon}=\widetilde{Q}_{s,\epsilon,\Gamma}\mathcal{U}_\Gamma$ on $\mathscr{S}^\prime\big(\mathcal{X}^2\big)$, and from Lemma \ref{L.2.3} we have that $ \mathcal{U}_\Gamma\widetilde{Q}_{s,\epsilon}u\in\mathscr{F}_0\big(\mathcal{X}^2\times\mathcal{X}^*\big)$. We conclude that $\widetilde{Q}_{s,\epsilon,\Gamma}v\in\mathscr{F}_0\big(\mathcal{X}^2\times\mathcal{X}^*\big)$ so that $v\in\mathscr{F}_{s,\epsilon}\big(\mathcal{X}^2\times\mathcal{X}^*\big)$. Moreover,
$$
\left\|\mathcal{U}_\Gamma u\right\|_{\mathscr{F}_{s,\epsilon}}=\|v\|_{\mathscr{F}_{s,\epsilon}}=\left\|\widetilde{Q}_{s,\epsilon,\Gamma}v\right\|_{\mathscr{F}_0}=\left\|\mathcal{U}_\Gamma\widetilde{Q}_{s,\epsilon}u\right\|_{\mathscr{F}_0}=\left\|\widetilde{Q}_{s,\epsilon}u\right\|_{L^2(\mathcal{X}^2)}=\|u\|_{\widetilde{\mathcal{H}}^s_{A_\epsilon}},
$$
implying that $\mathcal{U}_\Gamma:\widetilde{\mathcal{H}}^s_{A_\epsilon}\big(\mathcal{X}^2\big)\rightarrow\mathscr{F}_{s,\epsilon}\big(\mathcal{X}^2\times\mathcal{X}^*\big)$ is isometric. In order to prove its surjectivity we consider $v\in\mathscr{F}_{s,\epsilon}\big(\mathcal{X}^2\times\mathcal{X}^*\big)$; then $v\in\mathscr{S}^\prime_\Gamma\big(\mathcal{X}^2\times\mathcal{X}^*\big)$ and also $\widetilde{Q}_{s,\epsilon,\Gamma}v\in\mathscr{F}_0\big(\mathcal{X}^2\times\mathcal{X}^*\big)$. Let us define $u:=\mathcal{W}_\Gamma v\in\mathscr{S}^\prime(\mathcal{X}^2)$ (Lemma \ref{L.2.4}). Then we apply Lemma \ref{L.2.5} and deduce that $\mathcal{W}_\Gamma\widetilde{Q}_{s,\epsilon,\Gamma}=\widetilde{Q}_{s,\epsilon}\mathcal{W}_\Gamma$ on $\mathscr{S}^\prime_\Gamma\big(\mathcal{X}^2\times\mathcal{X}^*\big)$, so that we have $\widetilde{Q}_{s,\epsilon}u=\mathcal{W}_\Gamma\widetilde{Q}_{s,\epsilon,\Gamma}v\in L^2(\mathcal{X}^2)$ and we conclude that $u\in\widetilde{\mathcal{H}}^s_{A_\epsilon}\big(\mathcal{X}^2\big)$ and 
$\mathcal{U}_\Gamma u=v$.
\end{proof}

\begin{lemma}\label{L.2.8}
 The operator $\widetilde{P}_{\epsilon,\Gamma}$ defined on $\mathscr{F}_{m,\epsilon}\big(\mathcal{X}^2\times\mathcal{X}^*\big)$ is self-adjoint as operator acting in the Hilbert space $\mathscr{F}_0\big(\mathcal{X}^2\times\mathcal{X}^*\big)$.
\end{lemma}
\begin{proof}
 By Proposition \ref{P.1.18}, $\widetilde{P}_{\epsilon,\Gamma}$ is self-adjoint as operator acting in $L^2(\mathcal{X}^2)$, with domain $\widetilde{\mathcal{H}}^m_{A_\epsilon}\big(\mathcal{X}^2\big)$. By Lemma \ref{L.2.7} the operators $\mathcal{U}_\Gamma:L^2(\mathcal{X}^2)\rightarrow\mathscr{F}_0\big(\mathcal{X}^2\times\mathcal{X}^*\big)$ and $\mathcal{U}_\Gamma:\widetilde{\mathcal{H}}^m_{A_\epsilon}\big(\mathcal{X}^2\big)\rightarrow\mathscr{F}_{m,\epsilon}\big(\mathcal{X}^2\times\mathcal{X}^*\big)$ are unitary.Finally, by Lemma \ref{L.2.5} we have that $\widetilde{P}_{\epsilon,\Gamma}\mathcal{U}_\Gamma=\mathcal{U}_\Gamma\widetilde{P}_{\epsilon}$ on $\widetilde{\mathcal{H}}^m_{A_\epsilon}\big(\mathcal{X}^2\big)$, so that $\widetilde{P}_{\epsilon,\Gamma}$ is unitarily equivalent with $\widetilde{P}_{\epsilon}$.
\end{proof}

We shall need some more function spaces in order to come back to the operator $\widetilde{P}_{\epsilon}$.
\begin{definition}\label{D.2.9}
 Let $\theta\in\mathcal{X}^*$ and $s\in\mathbb{R}$.
\begin{enumerate}
 \item $\mathscr{S}^\prime_\theta\big(\mathcal{X}^2\big):=\left\{u\in\mathscr{S}^\prime\big(\mathcal{X}^2\big)\,\mid\,\big(\id\otimes\tau_{-\gamma}\big)u=e^{i<\theta,\gamma>}u\ \forall\gamma\in\Gamma\right\}$ endowed with the topology induced from $\mathscr{S}^\prime\big(\mathcal{X}^2\big)$.
\item $\mathcal{H}^s_{\theta,\epsilon}\big(\mathcal{X}^2\big):=\left\{u\in\mathscr{S}^\prime_\theta\big(\mathcal{X}^2\big)\,\mid\,\widetilde{Q}_{s,\epsilon}u\in L^2\big(\mathcal{X}\times E\big)\right\}$ endowed with the following quadratic norm
\begin{equation}\label{2.20}
 \|u\|_{\mathcal{H}^s_{\theta,\epsilon}}:=\left\|\widetilde{Q}_{s,\epsilon}u\right\|_{L^2(\mathcal{X}\times E)},\quad\forall u\in\mathcal{H}^s_{\theta,\epsilon}\big(\mathcal{X}^2\big).
\end{equation}
\item $\mathcal{K}^s_\epsilon\big(\mathcal{X}^2\big):=\mathcal{H}^s_{0,\epsilon}\big(\mathcal{X}^2\big)$.
\end{enumerate}
\end{definition}
As we already noticed in the proof of Lemma \ref{L.2.5}, for any $u\in\mathscr{S}^\prime\big(\mathcal{X}^2\big)$ the following equality holds:
$$
\big(\id\otimes\tau_{-\gamma}\big)\widetilde{P}_{\epsilon}u=\widetilde{P}_{\epsilon}\big(\id\otimes\tau_{-\gamma}\big)u,\quad\forall\gamma\in\Gamma.
$$
It follows that the operators $\widetilde{P}_{\epsilon}$ and $\widetilde{Q}_{s,\epsilon}$ leave the space $\mathscr{S}^\prime_\theta\big(\mathcal{X}^2\big)$ invariant. We shall use the notation $\mathscr{S}^\prime_0\big(\mathcal{X}^2\big)\equiv\mathscr{S}^\prime_\Gamma\big(\mathcal{X}^2\big)$. Let us also notice that for $s=0$ the spaces defined in (2) and (3) above do not depend on $\epsilon$ and will be denoted by $\mathcal{H}_{\theta}\big(\mathcal{X}^2\big)$ and respectively by $\mathcal{K}\big(\mathcal{X}^2\big)$; this last one may be identified with $L^2\big(\mathcal{X}\times\mathbb{T}\big)$.

\begin{lemma}\label{L.2.10}
 Let us consider the map $\boldsymbol{\psi}$ defined by \eqref{1.1}. Then for any $s\in\mathbb{R}$ the adjoint $\boldsymbol{\psi}^*$ is a unitary operator $\mathcal{K}^s_\epsilon\big(\mathcal{X}^2\big)\rightarrow
 \mathcal{H}^s_{A_\epsilon}\big(\mathcal{X}\big)\otimes L^2\big(\mathbb{T}\big)$. In particular $\mathcal{K}^s_\epsilon\big(\mathcal{X}^2\big)$ is a Hilbert space for the norm \eqref{2.20} having $\mathscr{S}\big(\mathcal{X}\times\mathbb{T}\big)$ as a dense subspace.
\end{lemma}
\begin{proof}
 The case $s=0$ is straightforward since the map $\boldsymbol{\psi}^*$ leaves invariant the space $\mathscr{S}^\prime_\Gamma\big(\mathcal{X}^2\big)$ and for any $u\in\mathscr{S}\big(\mathcal{X}\times\mathbb{T}\big)$ we have that
$$
\left\|\boldsymbol{\psi}^*u\right\|_{L^2(\mathcal{X}\times\mathbb{T})}^2=\int_{\mathcal{X}}\left(\int_{E}|u(x,x-y)|^2dy\right)dx=\int_{\mathcal{X}}\left(\int_{-E}|u(x,x+y)|^2dy\right)dx=
$$
$$
=\int_{\mathcal{X}}\left(\int_{x-E}|u(x,y)|^2dy\right)dx=\int_{\mathcal{X}}\left(\int_{E}|u(x,y)|^2dy\right)dx=\|u\|^2_{L^2(\mathcal{X}\times\mathbb{T})}.
$$

For any $s\in\mathbb{R}\setminus\{0\}$ we fix some $u\in\mathscr{S}^\prime_\Gamma\big(\mathcal{X}^2\big)$ and notice that:
$$
u\in\mathcal{K}^s_\epsilon\big(\mathcal{X}^2\big)\Leftrightarrow\widetilde{Q}_{s,\epsilon}u\in L^2(\mathcal{X}\times\mathbb{T})\Leftrightarrow\boldsymbol{\psi}^*Q^\prime_{s,\epsilon}\boldsymbol{\psi}^*u\in L^2(\mathcal{X}\times\mathbb{T})\Leftrightarrow\big(Q_{s,\epsilon}\otimes\id\big)\boldsymbol{\psi}^*u\in L^2(\mathcal{X}\times\mathbb{T})\Leftrightarrow
$$
$$
\Leftrightarrow\boldsymbol{\psi}^*u\in\mathcal{H}^s_{A_\epsilon}\big(\mathcal{X}\big)\otimes L^2\big(\mathbb{T}\big)
$$
and we also have that $\left\|\boldsymbol{\psi}^*u\right\|_{ \mathcal{H}^s_{A_\epsilon}(\mathcal{X})\otimes L^2(\mathbb{T})}=\|u\|_{\mathcal{K}^s_\epsilon}$.

The last statement becomes obvious noticing that $\mathcal{H}^s_{A_\epsilon}\big(\mathcal{X}\big)\otimes L^2\big(\mathbb{T}\big)$ is a Hilbert space with $\mathscr{S}\big(\mathcal{X}\times\mathbb{T}\big)$ a dense subspace in it that is invariant under the map $\boldsymbol{\psi}$.
\end{proof}

\begin{lemma}\label{L.2.11}
 For any $\theta\in\mathcal{X}^*$ and $s\in\mathbb{R}$ we have that the operator $T_\theta:\mathscr{S}\big(\mathcal{X}^2\big)\rightarrow\mathscr{S}\big(\mathcal{X}^2\big)$ defined by 
$$
\big(T_\theta u\big)(x,y):=e^{i<\theta,x-y>}u(x,y),
$$
induces a unitary operator $\mathcal{H}^s_{\theta,\epsilon}\big(\mathcal{X}^2\big)\rightarrow\mathcal{K}^s_\epsilon\big(\mathcal{X}^2\big)$. In particular we have that $\mathcal{H}^s_{\theta,\epsilon}\big(\mathcal{X}^2\big)$ is a Hilbert space containing $\mathscr{S}\big(\mathcal{X}^2\big):=T_\theta^{-1}\big[\mathscr{S}\big(\mathcal{X}\times\mathbb{T}\big)\big]$ as a dense subspace.
\end{lemma}
\begin{proof}
 Let us prove first that for any $\theta\in\mathcal{X}^*$ we have the equality:
\begin{equation}\label{2.21}
 \widetilde{P}_\epsilon T_\theta\ =\ T_\theta\widetilde{P}_\epsilon,
\qquad\text{on }\mathscr{S}^\prime\big(\mathcal{X}^2\big).
\end{equation}
It is clearly enough to prove it on $\mathscr{S}\big(\mathcal{X}^2\big)$; but in this case it results directly from \eqref{1.3} because $(x+\tilde{y}-y)-\tilde{y}=x-y$.

Then the following equality also follows
\begin{equation}\label{2.22}
 \widetilde{Q}_{s,\epsilon}T_\theta\ =\ T_\theta\widetilde{Q}_{s,\epsilon},
\qquad\text{on }\mathscr{S}^\prime\big(\mathcal{X}^2\big).
\end{equation}
We notice further that $T_\theta$ takes the space $\mathscr{S}^\prime_\theta\big(\mathcal{X}^2\big)$ into the space $\mathscr{S}^\prime_\Gamma\big(\mathcal{X}^2\big)$, while the operator $\widetilde{Q}_{s,\epsilon}$ leaves invariant both spaces $\mathscr{S}^\prime_\theta\big(\mathcal{X}^2\big)$ and $\mathscr{S}^\prime_\Gamma\big(\mathcal{X}^2\big)$.

For $u\in\mathscr{S}^\prime_\theta\big(\mathcal{X}^2\big)$ we have the equivalence relations:
$$
u\in\mathcal{H}^s_{\theta,\epsilon}\big(\mathcal{X}^2\big)\Leftrightarrow\widetilde{Q}_{s,\epsilon}u\in L^2(\mathcal{X}\times E)\Leftrightarrow T_\theta\widetilde{Q}_{s,\epsilon}u\in L^2(\mathcal{X}\times\mathbb{T})\Leftrightarrow\widetilde{Q}_{s,\epsilon}T_\theta u\in L^2(\mathcal{X}\times\mathbb{T})\Leftrightarrow
$$
$$
\Leftrightarrow T_\theta u\in\mathcal{K}^s_\epsilon\big(\mathcal{X}^2\big)
$$
and the equality $\left\|T_\theta u\right\|_{\mathcal{K}^s_\epsilon}=\|u\|_{\mathcal{H}^s_{\theta,\epsilon}}$.

The last statement is obvious since Lemma \ref{L.2.10} implies that $\mathcal{K}^s_\epsilon\big(\mathcal{X}^2\big)$ is a Hilbert space having $\mathscr{S}\big(\mathcal{X}\times\mathbb{T}\big)$ as a dense subspace.
\end{proof}

\begin{lemma}\label{L.2.12}
 For any $s\in\mathbb{R}$ we have that $\widetilde{P}_\epsilon\in\mathbb{B}\Big(\mathcal{K}^{s+m}_\epsilon\big(\mathcal{X}^2\big);\mathcal{K}^{s}_\epsilon\big(\mathcal{X}^2\big)\Big)$ uniformly for $\epsilon\in[-\epsilon_0,\epsilon_0]$.
\end{lemma}
\begin{proof}
 We have seen that:
\begin{itemize}
 \item $\mathscr{S}\big(\mathcal{X}\times\mathbb{T}\big)$ is a dense subspace of $\mathcal{K}^{s+m}_\epsilon\big(\mathcal{X}^2\big)$,
\item $\boldsymbol{\psi}^*:\mathcal{K}^{s}_\epsilon\big(\mathcal{X}^2\big)\rightarrow\mathcal{H}^s_{A_\epsilon}\big(\mathcal{X}\big)\otimes L^2(\mathbb{T})$ is a unitary operator that leaves $\mathscr{S}\big(\mathcal{X}\times\mathbb{T}\big)$ invariant.
\end{itemize}
It is thus enough to prove that $\forall s\in\mathbb{R}$, $\exists C_s>0$ such that:
\begin{equation}\label{2.23}
 \left\|\boldsymbol{\psi}^*\widetilde{P}_\epsilon\boldsymbol{\psi}^*u\right\|_{\mathcal{H}^s_{A_\epsilon}\big(\mathcal{X}\big)\otimes L^2(\mathbb{T})}\ \leq\ C_s\|u\|_{\mathcal{H}^{s+m}_{A_\epsilon}\big(\mathcal{X}\big)\otimes L^2(\mathbb{T})},\quad\forall u\in\mathscr{S}\big(\mathcal{X}\times\mathbb{T}\big),\ \forall\epsilon\in[-\epsilon_0,\epsilon_0].
\end{equation}
Formula \eqref{1.4} in Lemma \ref{L.1.2} implies the equality:
\begin{equation}\label{2.24}
 \Big(\boldsymbol{\psi}^*\widetilde{P}_\epsilon\boldsymbol{\psi}^*u\Big)(x,y)=(2\pi)^{-d}\int_{\mathcal{X}}\int_{\mathcal{X}^*}e^{i<\eta,x-\tilde{y}>}\omega_{A_\epsilon}(x,\tilde{y})\,p_\epsilon\Big(\frac{x+\tilde{y}}{2},\frac{x+\tilde{y}}{2}-y,\eta\Big)\,u(\tilde{y},y)\,d\tilde{y}\,d\eta,\quad\forall u\in\mathscr{S}\big(\mathcal{X}^2\big).
\end{equation}
But let us notice that the integral in \eqref{2.24} is well defined for any $u\in\mathscr{S}\big(\mathcal{X}\times\mathbb{T}\big)$ so that we can extend it to such functions (considered as periodic smooth functions in the second variable) either by duality and a computation in $\mathscr{S}^\prime\big(\mathcal{X}^2\big)$ or by approximating with functions from $\mathscr{S}\big(\mathcal{X}^2\big)$ with respect to the topology induced from $\mathscr{S}^\prime\big(\mathcal{X}^2\big)$. By the same time, the properties of the oscillating integral defining the right side of \eqref{2.24} allow to conclude that $\boldsymbol{\psi}^*\widetilde{P}_\epsilon\boldsymbol{\psi}^*\in\mathbb{B}\Big(\mathscr{S}\big(\mathcal{X}\times\mathbb{T}\big);\mathscr{S}\big(\mathcal{X}\times\mathbb{T}\big)\Big)$. Considering now $y\in\mathcal{X}$ in \eqref{2.24} as a parameter, the usual properties of magnetic pseudodifferential operators (see \cite{IMP1}) imply that $\forall\epsilon\in[-\epsilon_0,\epsilon_0]$, $\exists C_s>0$ such that:
\begin{equation}\label{2.25}
 \left\|\boldsymbol{\psi}^*\widetilde{P}_\epsilon\boldsymbol{\psi}^*u(.,y)\right\|_{\mathcal{H}^s_{A_\epsilon}\big(\mathcal{X}\big)}^2\ \leq\ C_s^2\|u(.,y)\|_{\mathcal{H}^{s+m}_{A_\epsilon}\big(\mathcal{X}\big)}^2,\qquad\forall(y,\epsilon)\in\mathcal{X}\times[-\epsilon_0,\epsilon_0],\ \forall u\in\mathscr{S}\big(\mathcal{X}\times\mathbb{T}\big).
\end{equation}
Integrating the above inequality for $y\in\mathbb{T}$ we obtain \eqref{2.23}.
\end{proof}

\begin{proposition}\label{P.2.13}
 $\widetilde{P}_\epsilon$ is a self-adjoint operator in $\mathcal{K}\big(\mathcal{X}^2\big)\equiv L^2(\mathcal{X}\times\mathbb{T})$ with domain $\mathcal{K}^{m}_\epsilon\big(\mathcal{X}^2\big)$; it is essentially self-adjoint on $\mathscr{S}\big(\mathcal{X}\times\mathbb{T}\big)$.
\end{proposition}
\begin{proof}
 Considering Lemma \ref{L.2.10} that implies that for any $s\in\mathbb{R}$ the operator $\boldsymbol{\psi}^*:\mathcal{K}^{s}_\epsilon\big(\mathcal{X}^2\big)\rightarrow\mathcal{H}^s_{A_\epsilon}\big(\mathcal{X}\big)\otimes L^2(\mathbb{T})$ is unitary and leaves invariant the subspace $\mathscr{S}\big(\mathcal{X}\times\mathbb{T}\big)$, it will be enough to prove that $\boldsymbol{\psi}^*\widetilde{P}_\epsilon\boldsymbol{\psi}^*$ is self-adjoint in $L^2(\mathcal{X}\times\mathbb{T})$ with domain $\mathcal{H}^m_{A_\epsilon}\big(\mathcal{X}\big)\otimes L^2(\mathbb{T})$ and essentially self-adjoint on $\mathscr{S}\big(\mathcal{X}\times\mathbb{T}\big)$.

Due to the arguments in the proof of Lemma \ref{L.2.12} we know that $\boldsymbol{\psi}^*\widetilde{P}_\epsilon\boldsymbol{\psi}^*$ is well defined in $L^2(\mathcal{X}\times\mathbb{T})$ with domain $\mathcal{H}^m_{A_\epsilon}\big(\mathcal{X}\big)\otimes L^2(\mathbb{T})$ and on $\mathscr{S}\big(\mathcal{X}\times\mathbb{T}\big)$ is defined by the equality \eqref{2.24}. A straightforward check using \eqref{2.24} shows that the operator $\boldsymbol{\psi}^*\widetilde{P}_\epsilon\boldsymbol{\psi}^*$ is symmetric on $\mathscr{S}\big(\mathcal{X}\times\mathbb{T}\big)$, that is a dense subspace of $\mathcal{H}^m_{A_\epsilon}\big(\mathcal{X}\big)\otimes L^2(\mathbb{T})$. As we know that $\boldsymbol{\psi}^*\widetilde{P}_\epsilon
\boldsymbol{\psi}^*\in\mathbb{B}\big(\mathcal{H}^m_{A_\epsilon}\big(\mathcal{X}\big)\otimes L^2(\mathbb{T});L^2(\mathcal{X}\times\mathbb{T})\big)$, it follows that $\boldsymbol{\psi}^*\widetilde{P}_\epsilon\boldsymbol{\psi}^*$ is symmetric on its domain too. In order to prove its self-adjointness let us fix some $v\in\mathcal{D}\big([\boldsymbol{\psi}^*\widetilde{P}_\epsilon\boldsymbol{\psi}^*]^*\big)$; it follows that $v\in L^2(\mathcal{X}\times\mathbb{T})$ and it exists some $f\in L^2(\mathcal{X}\times\mathbb{T})$ such that we have the equality
$$
\Big(\boldsymbol{\psi}^*\widetilde{P}_\epsilon\boldsymbol{\psi}^*u,v\Big)_{L^2(\mathcal{X}\times\mathbb{T})}\ =\ (u,f)_{L^2(\mathcal{X}\times\mathbb{T})},\qquad\forall u\in\mathscr{S}\big(\mathcal{X}\times\mathbb{T}\big).
$$
Thus $\boldsymbol{\psi}^*\widetilde{P}_\epsilon\boldsymbol{\psi}^*v=f$ as elements of $\mathscr{S}^\prime\big(\mathcal{X}\times\mathbb{T}\big)\equiv\mathscr{S}^\prime_\Gamma\big(\mathcal{X}^2\big)$. We notice that the Remark \ref{R.1.17} remains true if we replace $\mathcal{X}^2$ by $\mathcal{X}\times\mathbb{T}$ and thus we have that $v\in\mathcal{H}^m_{A_\epsilon}\big(\mathcal{X}\big)\otimes L^2(\mathbb{T})$ and thus $v$ belongs to the domain of $\boldsymbol{\psi}^*\widetilde{P}_\epsilon\boldsymbol{\psi}^*$.

The last statement clearly follows from the above results.
\end{proof}

We shall present now a connection between the operators defined in the Propositions \ref{P.1.18} and \ref{P.2.13}.
\begin{proposition}\label{P.2.14}
 Considering $\widetilde{P}_\epsilon$ as operator acting in $\mathscr{S}^\prime\big(\mathcal{X}^2\big)$ we shall denote by $\widetilde{P}_\epsilon^\prime$ the self-adjoint operator that it induces in $L^2(\mathcal{X}^2)$ with domain $\widetilde{H}^m_{A_\epsilon}\big(\mathcal{X}^2\big)$ (as in Proposition \ref{P.1.18}) and by $\widetilde{P}_\epsilon^{\prime\prime}$ the self-adjoint operator that it induces in $L^2(\mathcal{X}\times\mathbb{T})$ with domain $\mathcal{K}^{m}_\epsilon\big(\mathcal{X}^2\big)$ (as in Proposition \ref{P.2.13}). Then we have the equality:
\begin{equation}\label{2.26}
 \sigma\big(\widetilde{P}_\epsilon^\prime\big)\ =\ \sigma\big(\widetilde{P}_\epsilon^{\prime\prime}\big).
\end{equation}
\end{proposition}
\begin{proof}
 From the arguments in the proof of Lemma \ref{L.2.8} we deduce that $\mathcal{U}_\Gamma\widetilde{P}_\epsilon^\prime\mathcal{U}_\Gamma^{-1}=\widetilde{P}_{\epsilon,\Gamma}:=\widetilde{P}_\epsilon\otimes\id$, that is a self-adjoint operator in $\mathscr{F}_0\big(\mathcal{X}^2\times\mathcal{X}^*\big)$ with domain $\mathscr{F}_{m,\epsilon}\big(\mathcal{X}^2\times\mathcal{X}^*\big)$.

On the other side from Lemma \ref{L.2.11} (and the arguments in its proof) we deduce that for any $\theta\in\mathcal{X}^*$ the operator $\widetilde{P}_{\epsilon,\theta}:=T_\theta^{-1}\widetilde{P}_\epsilon^{\prime\prime}T_\theta$ is the self-adjoint operator associated to $ \widetilde{P}_\epsilon$ in $\mathcal{H}_\theta\big(\mathcal{X}^2\big)$, having the domain $\mathcal{H}^m_{\theta,\epsilon}\big(\mathcal{X}^2\big)$.

We shall consider the spaces $\mathscr{F}_0\big(\mathcal{X}^2\times\mathcal{X}^*\big)$ and $\mathscr{F}_{m,\epsilon}\big(\mathcal{X}^2\times\mathcal{X}^*\big)$ as direct integrals of Hilbert spaces over the dual torus; more precisely:
$$
\mathscr{F}_0\big(\mathcal{X}^2\times\mathcal{X}^*\big)\cong\int_{\mathbb{T}^{*,d}}^\oplus\mathcal{H}_\theta\big(\mathcal{X}^2\big)d\theta,\qquad\mathscr{F}_{m,\epsilon}\big(\mathcal{X}^2\times\mathcal{X}^*\big)\cong\int_{\mathbb{T}^{*,d}}^\oplus\mathcal{H}^m_{\theta,\epsilon}\big(\mathcal{X}^2\big)d\theta.
$$

Taking into account that:
$$
\big(\widetilde{P}_{\epsilon,\Gamma}u\big)(x,y,\theta)\ =\ \big((\widetilde{P}_{\epsilon,\theta}u)(.,.,\theta)\big)(x,y),\quad\forall u\in\mathscr{F}_{m,\epsilon}\big(\mathcal{X}^2\times\mathcal{X}^*\big),
$$
and the function: $\mathbb{T}^{*,d}\ni\theta\mapsto\big(\widetilde{P}_{\epsilon,\theta}+i\big)^{-1}\in\mathbb{B}\big(\mathcal{H}_\theta;\mathcal{H}_\theta\big)$ is measurable, we can write:
\begin{equation}\label{2.27}
 \widetilde{P}_{\epsilon,\Gamma}\ =\ \int_{\mathbb{T}^{*,d}}^\oplus\widetilde{P}_{\epsilon,\theta}\,d\theta.
\end{equation}
We can now apply Theorem XIII.85 (d) from \cite{RS-4} in order to conclude that we have the equivalence:
\begin{equation}\label{2.28}
 \lambda\in\sigma\big(\widetilde{P}_{\epsilon,\Gamma}\big)\ \Longleftrightarrow\ \forall\delta>0,\ \left|\left\{\theta\in\mathbb{T}^{*,d}\,\mid\,\sigma\big(\widetilde{P}_{\epsilon,\theta}\big)\cap(\lambda-\delta,\lambda+\delta)\ne\emptyset\right\}\right|\,>0.
\end{equation}
Let us notice that $\sigma\big(\widetilde{P}_{\epsilon,\theta}\big)$ is independent of $\theta\in\mathbb{T}^{*,d}$ and deduce that $\sigma\big(\widetilde{P}_{\epsilon,\Gamma}\big)=\sigma\big(\widetilde{P}_\epsilon^{\prime\prime}\big)$. But the conclusion of the first paragraph in this proof implies that $\sigma\big(\widetilde{P}_{\epsilon,\Gamma}\big)=\sigma\big(\widetilde{P}_\epsilon^{\prime}\big)$ and we finish the proof.
\end{proof}

We shall end up this section with a result giving a connection between the spaces: $\mathcal{K}^{s}_\epsilon\big(\mathcal{X}^2\big)$, $\mathscr{S}\big(\mathcal{X};\mathcal{H}^s(\mathbb{T})\big)$ and $\mathscr{S}^\prime\big(\mathcal{X};\mathcal{H}^s(\mathbb{T})\big)$. We start with a technical Lemma.
\begin{lemma}\label{L.2.15}
Let $B$ be a magnetic field with components of class $BC^\infty(\mathcal{X})$ and $A$ an associated vector potential with components of class $C^\infty_{\text{\sf pol}}$. Let us consider a symbol $q\in S^s_1(\Xi)$ for some $s\in\mathbb{R}$. We denote by $Q:=\mathfrak{Op}^A(q)$, $Q^\prime:=Q\otimes\id$ and $\widetilde{Q}:=\boldsymbol{\psi}^*Q^\prime
\boldsymbol{\psi}^*$, where $\boldsymbol{\psi}$ is defined by \eqref{1.1}. Then we have that $\widetilde{Q}\in\mathbb{B}\big(\mathscr{S}\big(\mathcal{X};\mathcal{H}^s(\mathbb{T})\big);\mathscr{S}\big(\mathcal{X};L^2(\mathbb{T})\big)\big)$ uniformly for $q$ varying in bounded subsets of $S^s_1(\Xi)$ and for $B$ varying in bounded subsets of $BC^\infty(\mathcal{X})$.
\end{lemma}
\begin{proof}
On $\mathscr{S}\big(\mathcal{X};\mathcal{H}^s(\mathbb{T})\big)$ we shall use the following family of seminorms:
\begin{equation}\label{2.29}
|u|_{s,l}\ :=\ \underset{|\alpha|\leq l}{\sup}\left[\int_{\mathcal{X}}<x>^{2l}\left\|\big(\partial^\alpha_x u\big)(x,.)\right\|_{\mathcal{H}^s(\mathbb{T})}^2\ dx\right]^{1/2},\qquad l\in\mathbb{N}, u\in\mathscr{S}\big(\mathcal{X};\mathcal{H}^s(\mathbb{T})\big).
\end{equation}

We have to prove that for any $k\in\mathbb{N}$ there exist $l\in\mathbb{N}$ and $C>0$ such that:
\begin{equation}\label{2.30}
\left|\widetilde{Q} u\right|_{0,k}\ \leq\ C|u|_{s,l},\qquad\forall u\in\mathscr{S}\big(\mathcal{X};\mathcal{H}^s(\mathbb{T})\big),
\end{equation}
 uniformly for $q$ varying in bounded subsets of $S^s_1(\Xi)$ and for $B$ varying in bounded subsets of $BC^\infty(\mathcal{X})$. 
 
 Using \eqref{1.3} and \eqref{1.7}, or a straightforward computation, we obtain that for any $u\in\mathscr{S}\big(\mathcal{X}\times\mathbb{T}\big)$:
 \begin{equation}\label{2.31}
 \big(\widetilde{Q} u\big)(x,y)\ =\ (2\pi)^{-d}\int_{\mathcal{X}}\int_{\mathcal{X}^*}e^{i<\eta,y-\tilde{y}>}\,\omega_A(x,x-y+\tilde{y})\,q\Big(x+\frac{\tilde{y}-y}{2},\eta\Big)\,u(x-y+\tilde{y},\tilde{y})\,d\tilde{y}\,d\eta.
 \end{equation}
 In particular we obtain that $\widetilde{Q}u\in\mathscr{S}\big(\mathcal{X}\times\mathbb{T}\big)$.
 
 For $x,y,\tilde{y}$ and $\eta$ fixed in $\mathcal{X}^*$, we consider the following function of the argument $t\in\mathcal{X}$:
 \begin{equation}
 \Phi(t)\ :=\ \omega_A(x,x-y+t)\,q\Big(x+\frac{t-y}{2},\eta\Big)\,u(x-y+t,\tilde{y}),
 \end{equation}
 and notice that its value for $t=\tilde{y}$ is exactly the factor that multiplies the exponential $e^{i<\eta,y-\tilde{y}>}$ under the integral in \eqref{2.31}; let us consider its Taylor expansion in $t\in\mathcal{X}$ around $t=y$ with integral rest of order $n>d+s$:
 \begin{equation}\label{2.32}
 \Phi(\tilde{y}) \ =\ \underset{|\alpha|<n}{\sum}f_\alpha(x,y,\tilde{y},\eta)\big(\tilde{y}-y\big)^\alpha\ +\ \underset{|\alpha|=n}{\sum}g_\alpha(x,y,\tilde{y},\eta)\big(\tilde{y}-y\big)^\alpha,
 \end{equation}
 where
 \begin{equation}\label{2.33}
 f_\alpha(x,y,\tilde{y},\eta)\ :=\ \underset{\beta\leq\alpha}{\sum}f_{\alpha\beta}(x)\,q_{\alpha\beta}(x,\eta)\big(\partial^\beta_xu\big)(x,\tilde{y}),\quad f_{\alpha\beta}\in C^\infty_{\text{\sf pol}}(\X),\ q_{\alpha\beta}\in S^s_1(\Xi)
 \end{equation}
 and
 \begin{equation}\label{2.34}
 g_\alpha(x,y,\tilde{y},\eta)\ :=\ \underset{\beta\leq\alpha}{\sum}\int_0^1h_{\tau,\alpha,\beta}(x,y-\tilde{y})\,q_{\alpha\beta}\big(x+(1-\tau)\frac{\tilde{y}-y}{2},\eta\big)\big(\partial^\beta_xu\big)(x-(1-\tau)(y-\tilde{y}),\tilde{y})\,d\tau
 \end{equation}
 where $h_{\tau,\alpha,\beta}\in C^\infty_{\text{\sf pol}}(\X\times\X)$ uniformly for $\tau\in[0,1]$ and $q_{\alpha\beta}\in S^s_1(\Xi)$.
 
 We use the relations \eqref{2.32}-\eqref{2.34} in \eqref{2.31} and eliminate the monomials $(\tilde{y}-y)^\alpha$ through partial integrations using the identity 
 $$
 (\tilde{y}-y)^\alpha e^{i<\eta,\tilde{y}-y>}\ =\ \big(-D_\eta\big)^\alpha e^{i<\eta,\tilde{y}-y>}.
 $$
 Finally we obtain:
 \begin{equation}\label{2.35}
 \big(\widetilde{Q} u\big)(x,y)\ =\ \underset{|\alpha|<n}{\sum}\ \underset{\beta\leq\alpha}{\sum}f_{\alpha\beta}(x)\big(T_{\alpha\beta}u\big)(x,y)\ +\ \underset{|\alpha|=n}{\sum}\ \underset{\beta\leq\alpha}{\sum}\int_0^1\big(R_{\alpha\beta}(\tau)u\big)(x,y)\,d\tau,
 \end{equation}
 where
 \begin{equation}\label{2.36}
 \big(T_{\alpha\beta}u\big)(x,y):=(2\pi)^{-d}\int_{\mathcal{X}}\int_{\mathcal{X}^*}e^{i<\eta,y-\tilde{y}>}t_{\alpha\beta}(x,\eta)\big(\partial^\beta_xu\big)(x,\tilde{y})\,d\tilde{y}\,d\eta,\quad t_{\alpha\beta}\in S^{s-|\alpha|}_1(\Xi),
 \end{equation}
 \begin{equation}\label{2.37}
 \big(R_{\alpha\beta}(\tau)u\big)(x,y):=
 \end{equation}
 $$
 (2\pi)^{-d}\int_{\mathcal{X}}\int_{\mathcal{X}^*}e^{i<\eta,y-\tilde{y}>}h_{\tau,\alpha,\beta}(x,y-\tilde{y})\,r_{\alpha\beta}\big(x+(1-\tau)\frac{\tilde{y}-y}{2},\eta\big)\,\big(\partial^\beta_xu\big)\big(x-(1-\tau)(y-\tilde{y}),\tilde{y}\big)\,d\tilde{y}\,d\eta,\quad r_{\alpha\beta}\in S^{s-n}_1(\Xi).
 $$
 
We begin by estimating the term $T_{\alpha\beta}u$, by using Lemma \ref{L.A.19}; Starting from \eqref{2.36} and considering $x\in\X$ as a parameter we conclude that there exists a semi-norm $c_{\alpha\beta}(q)$ of $q\in S^s_1(\Xi)$ such that
\begin{equation}\label{2.38}
\left\| \big(T_{\alpha\beta}u\big)(x,.)\right\|^2_{L^2(\mathbb{T})}\ \leq\ c_{\alpha\beta}(q)^2\left\|\big(\partial^\beta_xu\big)(x,.)\right\|^2_{\mathcal{H}^s(\mathbb{T})},\quad\forall x\in\X,\ \forall u\in\mathscr{S}(\X\times\mathbb{T}).
\end{equation}

Let us consider now the term $R_{\alpha\beta}(\tau)u$. We begin by noticing that due to our hypothesis there exists a constant $C(B)$ (bounded when the components of the magnetic field $B$ take values in bounded subsets of $BC^\infty(\X)$) and there exists an entire number $a\in\mathbb{Z}$ such that
\begin{equation}\label{2.39}
\left|h_{\tau,\alpha,\beta}(x,y-\tilde{y})\right|\ \leq\ C(B)<x>^a<y-\tilde{y}>^a,\quad\forall(x,y,\tilde{y})\in\X^3,\ \forall\tau\in[0,1].
\end{equation}
We integrate by parts in \eqref{2.37}, using the identity
$$
e^{i<\eta,y-\tilde{y}>}\ =\ <y-\tilde{y}>^{-2N}\big(1-\Delta_\eta\big)^Ne^{i<\eta,y-\tilde{y}>}.
$$ 
This allows us to conclude that there exists a seminorm $c^\prime_{\alpha,\beta,N}(p)$ of the symbol $p\in S^s_1(\Xi)$ for which we have the inequality:
\begin{equation}\label{2.40}
\left|\big(R_{\alpha\beta}(\tau)u\big)(x,y)\right|\ \leq\ C(B)c^\prime_{\alpha,\beta,N}(p)<x>^a\int_{\X^*}<\eta>^{s-n}d\eta\int_{\X}<z>^{a-2N}\left|\big(\partial^\beta_xu\big)\big(x-(1-\tau)z,y-z\big)\right|dz
\end{equation}
for any $(x,y)\in\X^2$ and any $\tau\in[0,1]$. We recall our choice $s-n<-d$, we choose further $2N\geq a+2d$ and we estimate the last integral by using the Cauchy-Schwartz inequality. We take the square of the inequality \eqref{2.40} and integrate with respect to $y\in E$ concluding that there exists a constant $C_0>0$ such that
$$
\int_E\left|\big(R_{\alpha\beta}(\tau)u\big)(x,y)\right|^2dy\ \leq
$$
$$
\leq\ C_0C(B)^2c^\prime_{\alpha,\beta,N}(p)^2<x>^{2a}\int_{\X}<z>^{-2d}\left(\int_E\left|\big(\partial^\alpha_xu\big)\big(x-(1-\tau)z,y-z\big)\right|^2dy\right)dz,\quad\forall x\in\X,\ \forall\tau\in[0,1].
$$
 
 For any $\Gamma$-periodic function $v\in L^2_{\text{\sf loc}}(\X)$ and for any $z\in\X$ we have that
 $$
 \int_E|v(y-z)|^2dy\ =\ \int_{\tau_zE}|v(y)|^2dy\ =\ \int_E|v(y)|^2dy
 $$
 so that for any $k\in\mathbb{N}$ there exists $C_k>0$ such that for any $\tau\in[0,1]$ we have that
 \begin{equation}\label{2.41}
 \int_{\X}<x>^{2k}\left\|\big(R_{\alpha\beta}(\tau)u\big)(x,.)\right\|^2_{L^2(\mathbb{T})}\ \leq\ C_kC(B)^2c^\prime_{\alpha,\beta,N}(p)^2\int_{\X}<x>^{2a+2k}\left\|\big(\partial^\alpha_xu\big)(x,.)\right\|^2_{L^2(\mathbb{T})}dx.
 \end{equation}
 
 For the derivatives $\partial^\mu_x\big(T_{\alpha\beta}u\big)(x,.)$ and $\partial^\mu_x\big(R_{\alpha\beta}(\tau)u\big)(x,.)$ (for any $\mu\in\mathbb{N}^d$) we obtain in a similar way estimations of the same form \eqref{2.38} and \eqref{2.41} and using \eqref{2.35} we obtain \eqref{2.30}.
\end{proof}

\begin{lemma}\label{L.2.16}
The following topological embeddings are true (uniformly in $\epsilon\in[-\epsilon_0,\epsilon_0]$):
\begin{equation}\label{2.42}
\mathscr{S}\big(\X;\mathcal{H}^m(\mathbb{T})\big)\ \hookrightarrow\ \mathcal{K}^m_\epsilon(\X\times\X)\ \hookrightarrow\ \mathscr{S}^\prime\big(\X;\mathcal{H}^m(\mathbb{T})\big).
\end{equation}
\end{lemma}
\begin{proof}
In order to prove the first embedding we take into account the density of $\mathscr{S}\big(\X\times\mathbb{T}\big)$ into $\mathscr{S}\big(\X;\mathcal{H}^m(\mathbb{T})\big)$ and the Definition \ref{D.2.9} (c) of the space $\mathcal{K}^m_\epsilon(\X\times\X)$. It is thus enough to prove that there exists a seminorm $|.|_{m,l}$ on $\mathscr{S}\big(\X;\mathcal{H}^m(\mathbb{T})\big)$ such that
\begin{equation}\label{2.43}
\left\|\widetilde{Q}_{m,\epsilon}u\right\|_{L^2(\X\times\mathbb{T})}\ \leq\ C|u|_{m,l},\qquad\forall u\in\mathscr{S}\big(\X\times\mathbb{T}\big).
\end{equation}
But this fact has been proved in Lemma \ref{L.2.15} (inequality \eqref{2.30}).

For the second embedding let us notice that the canonical sesquilinear map on $\mathscr{S}^\prime\big(\X;\mathcal{H}^m(\mathbb{T})\big)\times\mathscr{S}\big(\X;\mathcal{H}^m(\mathbb{T})\big)$ (associated to the duality map) is just a continuous extension of the scalar product
\begin{equation}\label{2.44}
(u,v)_m\ :=\ \int_{\X}\big(u(x,.),v(x,.)\big)_{\mathcal{H}^m(\mathbb{T})}dx,\qquad\forall(u,v)\in\mathscr{S}\big(\X;\mathcal{H}^m(\mathbb{T})\big)\times\mathscr{S}\big(\X;\mathcal{H}^m(\mathbb{T})\big).
\end{equation}
Due to the density of $\mathscr{S}\big(\X\times\mathbb{T}\big)$ into $\mathcal{K}^m_\epsilon(\X\times\X)$, this amounts to prove that it exists a continuous seminorm $|.|_{m,l}$ on $\mathscr{S}\big(\X;\mathcal{H}^m(\mathbb{T})\big)$ such that we have that
\begin{equation}\label{2.45}
|(u,v)_m|\ \leq\ \|u\|_{\mathcal{K}^m_\epsilon}\cdot|v|_{m,l}\qquad\forall (u,v)\in\mathscr{S}\big(\X;\mathcal{H}^m(\mathbb{T})\big)\times\mathscr{S}\big(\X;\mathcal{H}^m(\mathbb{T})\big),
\end{equation}
where $\|u\|_{\mathcal{K}^m_\epsilon}=\left\|\widetilde{Q}_{m,\epsilon}u\right\|_{L^2(\X\times\mathbb{T})}$. Let us notice that
$$
(u,v)_m\ =\ \left(u,\big(1\otimes<D_\Gamma>^{2m}\big)v\right)_{L^2(\X\times\mathbb{T})}
=\left(\widetilde{Q}_{m,\epsilon}u\,,\,\widetilde{Q}_{-m,\epsilon}\big(1\otimes<D_\Gamma>^{2m}\big)v\right)_
{L^2(\X\times\mathbb{T})}.
$$
We denote by $v_\Gamma:=\big(1\otimes<D_\Gamma>^{2m}\big)v\in\mathscr{S}
\big(\X\times\mathbb{T}\big)$ and notice that we have the inequality
\begin{equation}\label{2.46}
|(u,v)_m|\ \leq\ \left\|\widetilde{Q}_{m,\epsilon}u\right\|_{L^2(\X\times\mathbb{T})}\,\left\|\widetilde{Q}_{-m,\epsilon}v_\Gamma\right\|_{L^2(\X\times\mathbb{T})}.
\end{equation}
We conclude thus that the inequality \eqref{2.45} follows if we can prove that there exists a seminorm $|.|_{m,l}$ on $\mathscr{S}\big(\X;\mathcal{H}^m(\mathbb{T})\big)$ such that we have
\begin{equation}\label{2.47}
\left\|\widetilde{Q}_{-m,\epsilon}v_\Gamma\right\|_{L^2(\X\times\mathbb{T})}\ \leq\ C|v|_{m,l},\qquad\forall v\in\mathscr{S}\big(\X\times\mathbb{T}\big).
\end{equation}
From Lemma\ref{L.2.15} (inequality \eqref{2.30}) we know that there exists a seminorm $|.|_{-m,l}$ on $\mathscr{S}\big(\X\times\mathbb{T}\big)$ such that we have
\begin{equation}\label{2.48}
\left\|\widetilde{Q}_{-m,\epsilon}v_\Gamma\right\|_{L^2(\X\times\mathbb{T})}\ \leq\ C|v_\Gamma|_{-m,l},\qquad\forall v\in\mathscr{S}\big(\X\times\mathbb{T}\big).
\end{equation}
Now \eqref{2.47} follows from \eqref{2.48} once we notice that
$$
|v_\Gamma|_{-m,l}\ =\ |v|_{m,l}.
$$
\end{proof}

\section{The Grushin Problem}\label{S.3}
\setcounter{equation}{0}
\setcounter{theorem}{0}

Suppose given a symbol $p$ satisfying the assumptions of Lemma \ref{A.21}, i.e. $p\in S^m_1(\mathbb{T})$ real and elliptic, with $m>0$. The operator $P:=\mathfrak{Op}(p)$ has a self-adjoint realisation in $L^2(\X)$ having domain $\mathcal{H}^m(\X)$ and being lower semibounded and a self-adjoint realisation $P_\Gamma$ in $L^2(\mathbb{T})$ with domain $\mathcal{K}_{m,0}$ being also lower semibounded. 

\begin{lemma}\label{L.3.1}
There exists $N\in\mathbb{N}^*$, $C>0$ and the linear independent family $\{\phi_1,\ldots,\phi_N\}\subset\mathscr{S}(\mathbb{T})$, such that the following inequality is true:
\begin{equation}\label{3.1}
\left(P_\Gamma u,u\right)_{L^2(\mathbb{T})}\ \geq\ C^{-1}\|u\|^2_{\mathcal{K}_{m/2,0}}\ -\ C\sum\limits_{j=1}^N\left|\Big(u,\phi_j\Big)_{L^2(\mathbb{T})}\right|^2,\quad\forall u\in\mathcal{K}_{m,0}.
\end{equation}
\end{lemma}
\begin{proof}
The manifold $\mathbb{T}$ being a compact manifold without border, $\mathcal{K}_{m,0}$ is compactly embedded in $L^2(\mathbb{T})$ and the operator $P_\Gamma$ has compact resolvent. Let us fix some $\lambda\in\mathbb{R}$ and let us denote by $E_\lambda$ the spectral projection of $P_\Gamma$ for the semiaxis $(-\infty,\lambda]$. We choose an orthonormal basis $\{\phi_1,\ldots,\phi_N\}$ for the subspace $\Ran(E_\lambda)$ of $L^2(\mathbb{T})$. Then $\Ran(\bb1-E_\lambda)$ is the orthogonal complement of the space $\Sp\{\phi_1,\ldots,\phi_N\}$ generated by $\{\phi_1,\ldots,\phi_N\}$ in $L^2(\mathbb{T})$; moreover, for any $v\in\mathcal{D}(P_\Gamma)\cap\Ran(\bb1-E_\lambda)$, one has that 
$$
\left(P_\Gamma v,v\right)_{L^2(\mathbb{T})}\ \geq\ \lambda\|v\|^2_{L^2(\mathbb{T})}.
$$ 
In conclusion:
\begin{equation}\label{3.2}
\left(P_\Gamma v,v\right)_{L^2(\mathbb{T})}\ \geq\ \lambda\|v\|^2_{L^2(\mathbb{T})},\quad\forall v\in\mathcal{K}_{m,0}\cap\left[\Sp\{\phi_1,\ldots,\phi_N\}\right]^\bot.
\end{equation}

If $u\in\mathcal{K}_{m,0}$ we have that $v:=u-\sum\limits_{j=1}^N(u,\phi_j)_{L^2(\mathbb{T})}\phi_j$ belongs to the subset of vectors verifying \eqref{3.2} and thus we have that
\begin{equation}\label{3.3}
\left(P_\Gamma v,v\right)_{L^2(\mathbb{T})}\ \geq\ \lambda\left\|u-\sum\limits_{j=1}^N(u,\phi_j)_{L^2(\mathbb{T})}\phi_j\right\|^2_{L^2(\mathbb{T})}\ =\ \lambda\left(\|u\|^2_{L^2(\mathbb{T})}-\sum\limits_{j=1}^N\left|(u,\phi_j)_{L^2(\mathbb{T})}\right|^2\right).
\end{equation}

On the other side, if we know that $P_\Gamma\phi_j=\lambda_j\phi_j$ for any $1\leq j\leq N$, then we have that
$$
\left(P_\Gamma v,v\right)_{L^2(\mathbb{T})}=\left(P_\Gamma u-\sum\limits_{j=1}^N(u,\phi_j)_{L^2(\mathbb{T})}P_\Gamma\phi_j\ ,\ u-\sum\limits_{k=1}^N(u,\phi_k)_{L^2(\mathbb{T})}\phi_k\right)_{L^2(\mathbb{T})}=
$$
$$
=\left(P_\Gamma u,u\right)_{L^2(\mathbb{T})}\ -\sum\limits_{k=1}^N\overline{(u,\phi_k)}_{L^2(\mathbb{T})}(u,P_\Gamma\phi_k)_{L^2(\mathbb{T})}-\sum\limits_{j=1}^N(u,\phi_j)_{L^2(\mathbb{T})}\left(P_\Gamma\phi_j,u\right)_{L^2(\mathbb{T})}+
$$
$$
+\sum\limits_{j,k=1}^N(u,\phi_j)_{L^2(\mathbb{T})}\overline{(u,\phi_k)}_{L^2(\mathbb{T})}\left(P_\Gamma\phi_j,\phi_k\right)_{L^2(\mathbb{T})}=
$$
$$
=\left(P_\Gamma u,u\right)_{L^2(\mathbb{T})}\ -\ \sum\limits_{j=1}^N\lambda_j\left|\left(u,\phi_j\right)_{L^2(\mathbb{T})}\right|^2.
$$
If we compare this inequality with \eqref{3.3} we conclude that
$$
\left(P_\Gamma u,u\right)_{L^2(\mathbb{T})}\ \geq\ \lambda\|u\|^2_{L^2(\mathbb{T})}\ -\ \sum\limits_{j=1}^N(\lambda-\lambda_j)\left|\left(u,\phi_j\right)_{L^2(\mathbb{T})}\right|^2.
$$
Finaly we obtain that
\begin{equation}\label{3.4}
-\|u\|^2_{L^2(\mathbb{T})}\ \geq\ -\frac{1}{\lambda}\left(P_\Gamma u,u\right)_{L^2(\mathbb{T})}\ -\ \sum\limits_{j=1}^N\left(1-\frac{\lambda_j}{\lambda}\right)\left|\left(u,\phi_j\right)_{L^2(\mathbb{T})}\right|^2,\quad\forall u\in\mathcal{K}_{m,0}.
\end{equation}

In order to prove \eqref{3.1} we put together \eqref{3.4} with the G{\aa}rding inequality \eqref{A.36} and conclude that it exists $C_0>0$ such that
$$
\left(P_\Gamma u,u\right)_{L^2(\mathbb{T})}\ \geq\ C_0^{-1}\|u\|^2_{\mathcal{K}_{m/2,0}}\ -\ C_0\|u\|^2_{L^2(\mathbb{T})},\quad\forall u\in\mathcal{K}_{m,0}.
$$
\end{proof}
\begin{remark}\label{R.3.2}
From Remark \ref{R.A.22} we know that for any $\xi\in\X^*$ the operator $P_{\Gamma,\xi}$ is self-adjoint and lower semibounded in $L^2(\mathbb{T})$ on the domain $\mathcal{K}_{m,\xi}$. If we identify $\mathcal{K}_{m,\xi}$ with $\mathcal{H}^m_{\text{\sf loc}}(\X)\cap\mathscr{S}^\prime_\Gamma(\X)$ endowed with the norm $\|<D+\xi>^mu\|_{L^2(E)}$, we deduce that the operator $P_\xi$ is self-adjoint in $L^2_{\text{\sf loc}}(\X)\cap\mathscr{S}^\prime_\Gamma(\X)$ with the domain $\mathcal{K}_{m,\xi}$. Noticing that $P=\sigma_\xi P_\xi\sigma_{-\xi}$ and $\sigma_\xi:\mathcal{K}_{s,\xi}\rightarrow\mathscr{F}_{s,\xi}$ is a unitary operator for any $s\in\mathbb{R}$ and any $\xi\in\X^*$, it follows that $P$ generates in $\mathscr{F}_{0,\xi}$ a self-adjoint lower semibounded operator on the domain $\mathscr{F}_{m,\xi}$.
\end{remark}

\begin{lemma}\label{L.3.3}
Suppose given a compact interval $I\subset\mathbb{R}$; it exists a constant $C>0$, a natural integer $N\in\mathbb{N}$ and the family of functions $\{\psi_1,\ldots,\psi_N\}$ having the following properties:
\begin{enumerate}
\item[a)] $\psi_j\in C^\infty(\Xi)$.
\item[b)] $\psi_j(y,\eta+\gamma^*)\ =\ \psi_j(y,\eta),\quad\forall(y,\eta)\in\Xi,\ \forall\gamma^*\in\Gamma^*,\ 1\leq j\leq N$.
\item[c)] $\{\psi_j(.,\xi)\}_{1\leq j\leq N}$ is an orthonormal system in $\mathscr{F}_{0,\xi}$ for any $\xi\in\X^*$. We denote by $\mathcal{T}_\xi$ the complex linear space generated by the family $\{\psi_j(.,\xi)\}_{1\leq j\leq N}$ in $\mathscr{F}_{0,\xi}$ and by $\mathcal{T}^\bot_\xi$ its orthogonal complement in the same Hilbert space.
\item[d)] The following inequality is true:
\begin{equation}\label{3.5}
\left(\big(P-\lambda\big)u,u\right)_{\mathscr{F}_{0,\xi}}\ \geq\ C\|u\|^2_{\mathscr{F}_{0,\xi}},\quad\forall u\in{\mathscr{F}_{m,\xi}}\cap\mathcal{T}^\bot_\xi,\ \forall\xi\in\X^*,\ \forall\lambda\in I.
\end{equation}
\end{enumerate}
\end{lemma}
\begin{proof}
It is evidently enough to prove \eqref{3.5} for $\lambda=\lambda_0:=\sup I$. We apply Lemma \ref{L.3.1} to the operator $P_{\Gamma,\xi_0}-\lambda_0$ with $\xi_0\in\X^*$ to be considered fixed. We deduce that there exists $C_0>0$, $N_0\in\mathbb{N}^*$ and a family of functions $\{\widetilde{\psi}_1(\cdot,\xi_0),\ldots,\widetilde{\psi}_{N_0}(\cdot,\xi_0)\}$ from $\mathscr{S}(\mathbb{T})$ such that the following inequality is true:
\begin{equation}\label{3.6}
\left(\big(P_{\Gamma,\xi_0}-\lambda_0\big)v,v\right)_{L^2(\mathbb{T})}\ \geq\ C_0^{-1}\|v\|^2_{\mathcal{K}_{m/2,\xi_0}}\ -\ C_0\sum\limits_{j=1}^{N_0}\left|\big(v,\widetilde{\psi}_j(.,\xi_0)\big)_{L^2(\mathbb{T})}\right|^2,\quad\forall v\in\mathcal{K}_{m,\xi_0}.
\end{equation}
Taking into account the result in Example \ref{E.A.20}, we notice that the map $\X^*\ni\xi\mapsto P_{\Gamma,\xi}\in\mathbb{B}\big(\mathcal{K}_{s+m,\xi_0};\mathcal{K}_{s,\xi_0}\big)$ is continuous for any $s\in\mathbb{R}$ (it is even smooth); it follows that
$$
\left|\left(\big(P_{\Gamma,\xi_0}-P_{\Gamma,\xi}\big)v,v\right)_{L^2(\mathbb{T})}\right|\ \leq\ \left\|P_{\Gamma,\xi_0}-P_{\Gamma,\xi}\right\|_{\mathbb{B}(\mathcal{K}_{m/2,\xi_0};\mathcal{K}_{-m/2,\xi_0})}\,\|v\|^2_{\mathcal{K}_{m/2,\xi_0}},\quad\forall v\in\mathcal{K}_{m,\xi_0}.
$$
We conclude that for some smaller constant $C_0$, the inequality \eqref{3.6} is true with $P_{\Gamma,\xi}$ in place of $P_{\Gamma,\xi_0}$ on the left side, for $\xi\in V_0$ some small neighborhood of $\xi_0$ in $\X^*$.

Let us define now the family of functions $\{\psi_1,\ldots,\psi_{N}\}$. Let us first notice that
$$
\psi^\prime_j(.,\xi_0):=e^{i<\xi_0,.>}\widetilde{\psi}_j(.,\xi_0)\in C^\infty(\X)\cap\mathscr{F}_{0,\xi_0}.
$$
Then let us also notice that for any $\delta>0$ we can find functions $\overset{\circ}{\psi}_j\in C^\infty_0(\overset{\circ}{E})$, with $\overset{\circ}{E}$ the interior set of $E$, such that
$$
\left\|\psi^\prime_j(.,\xi_0)-\overset{\circ}{\psi}_j\right\|_{L^2(E)}\ \leq\ \delta,\quad 1\leq j\leq N_0.
$$
Then we define 
$$
\psi_j(x,\xi_0)\ :=\ \sum\limits_{\gamma\in\Gamma}\overset{\circ}{\psi}_j(x-\gamma)e^{i<\xi_0,\gamma>},\quad 1\leq j\leq N_0
$$
and we notice that $\psi_j(.,\xi_0)\in C^\infty(\X)\cap\mathscr{F}_{0,\xi_0}$ and $\psi_j(.,\xi_0)=\overset{\circ}{\psi}_j$ on $\overset{\circ}{E}$. Thus we can finally define
\begin{equation}\label{3.7}
\psi_j(x,\xi)\ :=\ \sum\limits_{\gamma\in\Gamma}\overset{\circ}{\psi}_j(x-\gamma)e^{i<\xi,\gamma>},\quad 1\leq j\leq N_0,\ \forall(x,\xi)\in\Xi.
\end{equation}
These functions evidently verify the properties (a) and (b) in the statement of the Lemma. It is also clear that for any $\xi\in\X^*$ we have that $\psi_j(.,\xi)\in\mathscr{F}_{0,\xi}$. Moreover we have that $\psi_j(.,\xi)=\overset{\circ}{\psi}_j=\psi_j(.,\xi_0)$ on $\overset{\circ}{E}$, so that
$$
\left\|\psi^\prime_j(.,\xi_0)-\psi_j(.,\xi)\right\|_{L^2(E)}\ \leq\ \delta,\quad\forall\xi\in\X^*,\ 1\leq j\leq N_0.
$$
From this estimation we may conclude that $\forall\kappa>0$ we can reduce if necessary the neighborhood $V_0$ fixed above, such that $\forall\xi\in V_0$ and $\forall v\in\mathcal{K}_{m,\xi_0}$ we have that for $1\leq j\leq N_0$:
$$
\left|\big(v,\widetilde{\psi}_j(.,\xi_0)\big)_{L^2(\mathbb{T})}\right|\ =\ \left|\int_Ev(y)\overline{\widetilde{\psi}_j(y,\xi_0)}dy\right|\ =\ \left|\int_Ee^{i<\xi_0,y>}v(y)\overline{\psi^\prime_j(y,\xi_0)}dy\right|\ \leq
$$
$$
\leq\ \left|\int_Ee^{i<\xi_0,y>}v(y)\overline{\psi_j(y,\xi)}dy\right|\,+\,\left|\int_Ee^{i<\xi_0,y>}v(y)\overline{\big[\psi^\prime_j(y,\xi_0)-\psi_j(y,\xi)\big]}dy\right|\ \leq
$$
$$
\leq\ \left|\int_Ee^{i<\xi,y>}v(y)\overline{\psi_j(y,\xi)}dy\right|\,+\,\left|\int_E\big[e^{i<\xi_0,y>}-e^{i<\xi,y>}\big]v(y)\overline{\psi_j(y,\xi)}dy\right|\,+
$$
$$
+\,\left|\int_Ee^{i<\xi_0,y>}v(y)\overline{\big[\psi^\prime_j(y,\xi_0)-\psi_j(y,\xi)\big]}dy\right|\ \leq
$$
$$
\leq\ \left|\int_Ee^{i<\xi,y>}v(y)\overline{\psi_j(y,\xi)}dy\right|\,+\,\kappa\|v\|_{L^2(E)}.
$$
From this estimation we deduce that, by reducing if necessary the constant $C_0>0$ we can replace in the right hand side of \eqref{3.6} the scalar products $\big(v,\widetilde{\psi}_j(.,\xi_0)\big)_{L^2(\mathbb{T})}$ with the scalar products $\big(e^{i<\xi,.>}v,\psi_j(.,\xi)\big)_{\mathscr{F}_{0,\xi}}$ for $\xi\in V_0$.

Let us consider a vector $u\in\mathscr{F}_{m,\xi}$ and associate to it the vector $v:=e^{-i<\xi,.>}u$ that belongs to $\mathcal{K}_{m,\xi}$ and taking into account the equality
$$
e^{i<\xi,.>}P_\xi e^{-i<\xi,.>}u\ =\ Pu
$$
we deduce from \eqref{3.6} and the above arguments that we have
\begin{equation}\label{3.8}
\left(\big(P-\lambda_0\big)u,u\right)_{\mathscr{F}_{0,\xi}}\ \geq\ C_0^{-1}\|u\|^2_{\mathscr{F}_{0,\xi}}\,-\,C_0\sum\limits_{j=1}^{N_0}\left|\left(u,\psi_j(.,\xi)\right)_{\mathscr{F}_{0,\xi}}\right|^2,\quad\forall u\in\mathscr{F}_{m,\xi},\ \forall\xi\in V_0.
\end{equation}

Taking into account that $\mathscr{F}_{s,\xi+\gamma^*}=\mathscr{F}_{s,\xi}$ for any $s\in\mathbb{R}$, any $\xi\in\X^*$ and any $\gamma^*\in\Gamma^*$ and the fact that the functions $\psi_j(x,.)$ are $\Gamma^*$-periodic (for $1\leq j\leq N_0$), we conclude that it is enough to prove \eqref{3.5} for $\xi\in\mathbb{T}_*$. Being compact, $\mathbb{T}_*$ can be covered by a finite number of neighborhoods of type $V_0$ (as defined in the argument above). In this way, repeating the procedure explained above we can find a finite family of functions $\{\psi_1,\ldots,\psi_{\tilde{N}}\}$ (with some quite larger $\tilde{N}$ in principle) that will satisfy the properties (a), (b) and (d) from the statement of the Lemma. We select now out of this family a maximal linearly independent subfamily of $N$ functions $\{\psi_1,\ldots,\psi_N\}$ (it can be characterized by the property that the functions $\{\overset{\circ}{\psi}_1,\ldots,\overset{\circ}{\psi}_N\}$ is a linearly independent system in $C^\infty_0(\overset{\circ}{E})$). Let us notice that this last step (the choice of the maximal linearly independent subfamily) does not change the subspace $\mathcal{T}_\xi$ that they generate. Finally we may use the Gram-Schmidt procedure in order to obtain a family of $N$ orthonormal functions from $\mathscr{F}_{0,\xi}$.
\end{proof}

\begin{lemma}\label{L.3.4}
Under the assumptions of Lemma \ref{L.3.3} we denote by $\Pi_\xi$ the orthogonal projection on $\mathcal{T}_\xi$ in the Hilbert space $\mathscr{F}_{0,\xi}$ and by $S(\xi,\lambda)$ the unbounded operator in $\mathcal{T}^\bot_\xi$ defined on the domain $\mathscr{F}_{m,\xi}\cap\mathcal{T}^\bot_\xi$ by the action of $\big(\bb1-\Pi_\xi\big)\big(P-\lambda\big)$. Then the following statements are true:
\begin{enumerate}
\item[a)] The operator $S(\xi,\lambda)$ is self-adjoint and invertible and $S(\xi,\lambda)^{-1}\in\mathbb{B}\big(\mathcal{T}^\bot_\xi;\mathcal{T}^\bot_\xi\big)$ uniformly with respect to $(\xi,\lambda)\in\mathbb{T}^{d*}\times I$.
\item[b)] The operator $S(\xi,\lambda)^{-1}$ also belongs to $\mathbb{B}\big(\mathcal{T}^\bot_\xi;\mathscr{F}_{m,\xi}\big)$ uniformly with respect to $(\xi,\lambda)\in\mathbb{T}^{d*}\times I$.
\end{enumerate}
\end{lemma}

\begin{proof}
The operator $S(\xi,\lambda)$ is densly defined by definition and is symmetric on its domain because for any couple $(u,v)\in\big[\mathscr{F}_{m,\xi}\cap\mathcal{T}^\bot_\xi\big]^2$ we can write that:
$$
\left(S(\xi,\lambda)u,v\right)_{\mathscr{F}_{0,\xi}}\ =\ \left(\big(\bb1-\Pi_\xi\big)\big(P-\lambda\big)u,v\right)_{\mathscr{F}_{0,\xi}}\ =\ \left(u,\big(P-\lambda\big)\big(\bb1-\Pi_\xi\big)v\right)_{\mathscr{F}_{0,\xi}}\ =\ 
$$
$$
=\ \left(\big(\bb1-\Pi_\xi\big)u,\big(P-\lambda\big)v\right)_{\mathscr{F}_{0,\xi}}\ =\ \left(u,S(\xi,\lambda)v\right)_{\mathscr{F}_{0,\xi}}.
$$

In order to prove now the self-adjointness of the operator $S(\xi,\lambda)$ let us fix some vector $v\in\mathcal{D}\big(S(\xi,\lambda)^*\big)$; thus we deduce first that $v\in\mathcal{T}^\bot_\xi$ and secondly that there exists a vector $f\in\mathcal{T}^\bot_\xi$ such that $\big(S(\xi,\lambda)u,v\big)_{\mathscr{F}_{0,\xi}}=\big(u,f\big)_{\mathscr{F}_{0,\xi}}$ for any $u\in\mathscr{F}_{m,\xi}\cap\mathcal{T}^\bot_\xi$. For any vector $w\in\mathscr{F}_{m,\xi}$ we can write that
$w=\Pi_\xi w\,+\,\big({\bb1}-\Pi_\xi\big)w$ and 
$$
\Pi_\xi w\,=\ \sum\limits_{j=1}^{N}\big(w,\psi_j(.,\xi)\big)_{\mathscr{F}_{0,\xi}}\psi_j(.,\xi)\in\mathscr{F}_{m,\xi}\cap\mathcal{T}_\xi,\qquad\big(\bb1-\Pi_\xi\big)w\in\mathscr{F}_{m,\xi}\cap\mathcal{T}^\bot_\xi.
$$
Thus we have that
$$
\left(\big(P-\lambda\big)\big(\bb1-\Pi_\xi\big)w,v\right)_{\mathscr{F}_{0,\xi}}\,=\,\left(\big(\bb1-\Pi_\xi\big)\big(P-\lambda\big)\big(\bb1-\Pi_\xi\big)w,v\right)_{\mathscr{F}_{0,\xi}}\,=\,\left(S(\xi,\lambda)\big(\bb1-\Pi_\xi\big)w,v\right)_{\mathscr{F}_{0,\xi}}\,=
$$
$$
=\,\left(\big(\bb1-\Pi_\xi\big)w,f\right)_{\mathscr{F}_{0,\xi}}\,=\,\left(w,f\right)_{\mathscr{F}_{0,\xi}},
$$
and
$$
\left(\big(P-\lambda\big)\Pi_\xi w,v\right)_{\mathscr{F}_{0,\xi}}\,=\,\sum\limits_{j=1}^{N}\big(w,\psi_j(.,\xi)\big)_{\mathscr{F}_{0,\xi}}\left(\big(P-\lambda\big)\psi_j(.,\xi),v\right)_{\mathscr{F}_{0,\xi}}\,=\,\left(w,f_0(.,\xi)\right)_{\mathscr{F}_{0,\xi}}
$$
where $f_0(.,\xi):=\sum\limits_{j=1}^{N}\left(\big(P-\lambda\big)\psi_j(.,\xi),v\right)_{\mathscr{F}_{0,\xi}}\psi_j(.,\xi)\in\mathcal{T}_\xi$. In conclusion we have that for anny $w\in\mathscr{F}_{m,\xi}$:
$$
\left(\big(P-\lambda\big)w,v\right)_{\mathscr{F}_{0,\xi}}\ =\ \left(w,f+f_0(.,\xi)\right)_{\mathscr{F}_{0,\xi}},\qquad f+f_0(.,\xi)\in\mathscr{F}_{0,\xi}.
$$
Recalling that $P-\lambda$ is self-adjoint in $\mathscr{F}_{0,\xi}$ with domain $\mathscr{F}_{m,\xi}$ we conclude that $v\in\mathscr{F}_{m,\xi}\cap\mathcal{T}^\bot_\xi=\mathcal{D}\big(S(\xi,\lambda)\big)$.

The invertibility of $S(\xi,\lambda)$ follows from \eqref{3.5} noticing that
\begin{equation}\label{3.9}
{\left(S(\xi,\lambda)u,u\right)_{\mathscr{F}_{0,\xi}}\ \geq\ C\|u\|^2_{\mathscr{F}_{0,\xi}}\qquad\forall u \in\mathscr{F}_{m,\xi}\cal\mathcal{T}^\bot_\xi,\ \forall\xi\in\mathbb{T}}^{*d},\ \forall\lambda\in I.
\end{equation}
From this last inequality follows also that $S(\xi,\lambda)^{-1}\in\mathbb{B}\big(\mathcal{T}^\bot_\xi;\mathscr{F}_{m,\xi}\big)$ uniformly with respect to $(\xi,\lambda)\in\mathbb{T}_*\times I$.

b) For any fixed $\xi_0\in\mathbb{T}_*$ and $\lambda_0\in I$ we know that $P_{\Gamma,\xi_0}-\lambda_0$ is self-adjoint in $\mathcal{K}_{0,\xi_0}$ on the domain $\mathcal{K}_{m,\xi_0}$ and that the Hilbert norm on $\mathcal{K}_{m,\xi_0}$ is equivalent with the graph-norm of $P_{\Gamma,\xi_0}-\lambda_0$. It exists thus a constant $C_0>0$ such that
\begin{equation}\label{3.10}
\|v\|_{\mathcal{K}_{m,\xi_0}}\ \leq\ C_0\left(\|v\|_{\mathcal{K}_{0,\xi_0}}\,+\,\left\|\big(P_{\Gamma,\xi_0}-\lambda_0\big)v\right\|_{\mathcal{K}_{0,\xi_0}}\right),\qquad\forall v\in\mathcal{K}_{m,\xi_0}.
\end{equation}
Taking into account the Example \ref{E.A.20} we know that the application $\X^*\ni\xi\mapsto P_{\Gamma,\xi}\in\mathbb{B}\big(\mathcal{K}_{m,0};\mathcal{K}_{0,0}\big)$ is of class $C^\infty$. Noticing that 
$$
\left\|\big(P_{\Gamma,\xi_0}-\lambda_0\big)v\,-\,\big(P_{\Gamma,\xi}-\lambda\big)v\right\|_{\mathcal{K}_{0,\xi_0}}\ \leq\ C\left\|P_{\Gamma,\xi_0}\,-\,P_{\Gamma,\xi}\right\|_{\mathbb{B}(\mathcal{K}_{m,0};\mathcal{K}_{0,0})}\|v\|_{\mathcal{K}_{m,\xi_0}}\,+\,|\lambda-\lambda_0|\,\|v\|_{\mathcal{K}_{0,\xi_0}},
$$
for any $(\xi,\xi_0)\in\mathbb{T}_*\times\mathbb{T}_*$, any $(\lambda,\lambda_0)\in I\times I$ and any $v\in\mathcal{K}_{m,\xi_0}$, we deduce that there exist a constant $C^\prime_0>0$ and a neighborhood $V_0$ of $\xi_0\in\mathbb{T}_*$ such that
\begin{equation}\label{3.11}
\|v\|_{\mathcal{K}_{m,\xi}}\ \leq\ C^\prime_0\left(\|v\|_{\mathcal{K}_{0,\xi}}\,+\,\left\|\big(P_{\Gamma,\xi}-\lambda\big)v\right\|_{\mathcal{K}_{0,\xi}}\right),\qquad\forall(\xi,\lambda)\in V_0\times I,\ \forall v\in\mathcal{K}_{m,\xi}.
\end{equation}
The manifold $\mathbb{T}_*$ being compact we can find a finite cover with neighborhoods of type $V_0$ as above and we conclude that \eqref{3.11} is true for any $\xi\in\mathbb{T}_*$ with a suitable constant $C^\prime_0$.

Considering now some vector $u\in\mathscr{F}_{m,\xi}$ and denoting by $v:=\sigma_{-\xi}u\in\mathcal{K}_{m,\xi}$ we deduce from \eqref{3.11} that
\begin{equation}\label{3.12}
\|u\|_{\mathscr{F}_{m,\xi}}\ \leq\ C^\prime_0\left(\|u\|_{\mathscr{F}_{0,\xi}}\,+\,\left\|\big(P-\lambda\big)u\right\|_{\mathscr{F}_{0,\xi}}\right),\qquad\forall(\xi,\lambda)\in\mathbb{T}_*\times I,\ \forall u\in\mathscr{F}_{m,\xi}.
\end{equation}

Considering now $u\in\mathscr{F}_{m,\xi}\cap\mathcal{T}^\bot_\xi=\mathcal{D}\big(S(\xi,\lambda)\big)$ we can write that
$$
\big(P-\lambda\big)u\,=\,S(\xi,\lambda)u\,+\,\Pi_\xi\big(P-\lambda\big)u\,=\,S(\xi,\lambda)u\,+\,\sum\limits_{j=1}^{N}\left(u,\big(P-\lambda\big)\psi_j(.,\xi)\right)_{\mathscr{F}_{0,\xi}}\psi_j(.,\xi),
$$
and we know that the norm in $\mathscr{F}_{0,\xi}$ of the second term on the right (the finite sum over $j$) can be bounded by a constant that does not depend on $\xi\in\mathbb{T}^{d*}$ multiplied by $\|u\|_{\mathscr{F}_{0,\xi}}$. Using \eqref{3.12} we deduce that there exists a constant $C>0$ such that
$$
\|u\|_{\mathscr{F}_{m,\xi}}\ \leq\ C\left(\|u\|_{\mathscr{F}_{0,\xi}}\,+\,\left\|S(\xi,\lambda)u\right\|_{\mathscr{F}_{0,\xi}}\right)\qquad\forall(\xi,\lambda)\in\mathbb{T}_*\times I,\ \forall u\in\mathscr{F}_{m,\xi}\cap\mathcal{T}^\bot_\xi.
$$
This inequality clearly implies point (b) of the Lemma.
\end{proof}

Let us define now the following operators associated to the family of functions $\{\psi_j\}_{1\leq j\leq N}$ introduced above. For any $\xi\in\mathbb{T}_*$ we define:
\begin{equation}\label{3.13}
\forall u\in\mathscr{F}_{0,\xi},\qquad\widetilde{R}_+(\xi)u\ :=\ \left\{\big(u,\psi_j(.,\xi)\big)_{\mathscr{F}_{0,\xi}}\right\}_{1\leq j\leq N}\in\mathbb{C}^N;
\end{equation}
\begin{equation}\label{3.14}
\forall\underline{u}:=\{\underline{u}_1,\ldots,\underline{u}_N\}\in\mathbb{C}^N,\qquad\widetilde{R}_-(\xi)\underline{u}:=\sum\limits_{j=1}^{N}\underline{u}_j\psi_j(.,\xi)\in\mathscr{F}_{0,\xi}.
\end{equation}
Evidently we have that $\forall\xi\in\mathbb{T}_*$, $\widetilde{R}_+(\xi)\in\mathbb{B}\big(\mathscr{F}_{0,\xi};\mathbb{C}^N\big)$ and $\widetilde{R}_-(\xi)\in\mathbb{B}\big(\mathbb{C}^N;\mathscr{F}_{0,\xi}\big)$.

We can define now the Grushin type operator associated to our operator $P$:
\begin{equation}\label{3.15}
Q(\xi,\lambda)\ :=\ \left(
\begin{array}{cc}
P-\lambda&\widetilde{R}_-(\xi)\\
\widetilde{R}_+(\xi)&0
\end{array}
\right),
\end{equation}
that due to our previous results belongs to $\mathbb{B}\big(\mathscr{F}_{m,\xi}\times\mathbb{C}^N;\mathscr{F}_{0,\xi}\times\mathbb{C}^N\big)$ uniformly with respect to $(\xi,\lambda)\in\mathbb{T}_*\times I$.

\begin{lemma}\label{L.3.5}
For any values of $(\xi,\lambda)\in\mathbb{T}_*\times I$ the operator $Q(\xi,\lambda)$ acting as an unbounded linear operator in the Hilbert space $\mathscr{F}_{0,\xi}\times\mathbb{C}^N$ is self-adjoint on the domain $\mathscr{F}_{m,\xi}\times\mathbb{C}^N$.
\end{lemma}
\begin{proof}
We know that $P$ is self-adjoint in $\mathscr{F}_{0,\xi}$ with domain $\mathscr{F}_{m,\xi}$ and it is easy to see that $\widetilde{R}_+(\xi)^*=\widetilde{R}_-(\xi)$ so that we conclude that $Q(\xi,\lambda)$ is symmetric on $\mathscr{F}_{m,\xi}\times\mathbb{C}^N$.

Let us consider now a pair $(v,\underline{v})\in\mathcal{D}\big(Q(\xi,\lambda)^*\big)$; this means that $v\in\mathscr{F}_{0,\xi}$, $\underline{v}\in\mathbb{C}^N$ and there exists a pair $(f,\underline{f})\in\mathscr{F}_{0,\xi}\times\mathbb{C}^N$ such that we have
$$
\left(Q(\xi,\lambda)\left(
\begin{array}{c}
u\\
\underline{u}
\end{array}
\right)
\begin{array}{c}
\\
,
\end{array}
\left(
\begin{array}{c}
v\\
\underline{v}
\end{array}
\right)\right)_{\mathscr{F}_{0,\xi}\times\mathbb{C}^N}\ =\ 
\left(\left(
\begin{array}{c}
u\\
\underline{u}
\end{array}
\right)
\begin{array}{c}
\\
,
\end{array}
\left(
\begin{array}{c}
f\\
\underline{f}
\end{array}
\right)\right)_{\mathscr{F}_{0,\xi}\times\mathbb{C}^N}.
$$
Considering the case $\underline{u}=0$ we get that
$$
\left(\big(P-\lambda\big)u,v\right)_{\mathscr{F}_{0,\xi}}\ =\ \big(u,f\big)_{\mathscr{F}_{0,\xi}}\,-\,\big(\widetilde{R}_+(\xi)u,\underline{v}\big)_{\mathbb{C}^N}\ =\ \big(u,f-\widetilde{R}_-(\xi)\underline{v}\big)_{\mathscr{F}_{0,\xi}},\qquad\forall u\in\mathscr{F}_{m,\xi}.
$$
Taking into account the self-adjointness of the operator $P-\lambda$ in $\mathscr{F}_{0,\xi}$ on the domain $\mathscr{F}_{m,\xi}$ we may deduce that in fact $v$ belongs to $\mathscr{F}_{m,\xi}$ and thus $(v,\underline{v})\in\mathcal{D}\big(Q(\xi,\lambda)\big)$.
\end{proof}

\begin{lemma}\label{L.3.6}
The operator $Q(\xi,\lambda)$ defined in \eqref{3.15} is bijective and has an inverse $Q(\xi,\lambda)^{-1}\in\mathbb{B}\big(\mathscr{F}_{0,\xi}\times\mathbb{C}^N;\mathscr{F}_{m,\xi}\times\mathbb{C}^N\big)$ uniformly with respect to $(\xi,\lambda)\in\mathbb{T}_*\times I$.
\end{lemma}
\begin{proof}
Let us first prove the injectivity. Let us choose $u\in\mathscr{F}_{m,\xi}$ and $\underline{u}\in\mathbb{C}^N$ verifying the following equations:
\begin{equation}\label{3.16}
\left\{
\begin{array}{l}
\big(P-\lambda\big)u\,+\,\widetilde{R}_-(\xi)\underline{u}\ =\ 0,\\
\\
\widetilde{R}_+(\xi)u\ =\ 0.
\end{array}\right.
\end{equation}
The second equality in \eqref{3.16} implies that $u\in\mathcal{T}^\bot_\xi$. As by definition we have that $\widetilde{R}_-(\xi)\underline{u}\in\mathcal{T}_\xi$, the first equality in \eqref{3.16} implies that $\big(\bb1-\Pi_\xi\big)\big(P-\lambda\big)u=0$, or equivalently that $S(\xi,\lambda)u=0$. Taking now into account Lemma \ref{L.3.4} it follows that $u=0$. Now, the first equality in \eqref{3.16} implies that we also have $\widetilde{R}_-(\xi)\underline{u}=0$; but the linear independence of the system of functions $\{\psi_j(.,\xi)\}_{1\leq j\leq N}$ implies that the operator $\widetilde{R}_-(\xi)$ is injective and thus we deduce that we also have $\underline{u}=0$.

Let us consider now the surjectivity of the operator $Q(\xi,\lambda)$. Thus let us choose an arbitrary pair $(v,\underline{v})\in\mathscr{F}_{0,\xi}\times\mathbb{C}^N$ and let us search for a pair $(u,\underline{u})\in\mathscr{F}_{m,\xi}\times\mathbb{C}^N$ such that the following equalities are true:
\begin{equation}\label{3.17}
\left\{
\begin{array}{l}
\big(P-\lambda\big)u\,+\,\widetilde{R}_-(\xi)\underline{u}\ =\ v,\\
\\
\widetilde{R}_+(\xi)u\ =\ \underline{v}.
\end{array}
\right.
\end{equation}
Let us denote by $u_1:=\sum\limits_{j=1}^{N}\underline{v}_j\psi_j(.,\xi)$ so that by definition we have that
\begin{equation}\label{3.18}
\widetilde{R}_+(\xi)u_1\ =\ \underline{v},\quad u_1\in\mathscr{F}_{m,\xi}\cap\mathcal{T}_\xi,
\end{equation}
\begin{equation}\label{3.19}
\|u_1\|_{\mathscr{F}_{m,\xi}}\ \leq\ C\|\underline{v}\|_{\mathbb{C}^N},\quad\forall\xi\in\mathbb{T}_*.
\end{equation}
Thus we have to search now for a pair $(u_2,\underline{u})\in\big[\mathscr{F}_{m,\xi}\cap\mathcal{T}^\bot_\xi\big]\times\mathbb{C}^N$ that should verify the equality:
\begin{equation}\label{3.20}
\big(P-\lambda\big)u_2\,+\,\widetilde{R}_-(\xi)\underline{u}\ =\ v\,-\,\big(P-\lambda\big)u_1\ \in\ \mathscr{F}_{0,\xi}.
\end{equation}
In fact, as we have by definition that $\widetilde{R}_+(\xi)u_2=0$, the relations \eqref{3.18} and \eqref{3.20} imply directly that $u=u_1+u_2$ and $\underline{u}$ is the solution we looked for. 

Let us project the equation \eqref{3.20} both on $\mathcal{T}_\xi$ and on its complement $\mathcal{T}^\bot_\xi$ to obtain
\begin{equation}\label{3.21}
S(\xi,\lambda)u_2\ =\ \big(\bb1-\Pi_\xi\big)\big(v\,-\,\big(P-\lambda\big)u_1\big),
\end{equation}
\begin{equation}\label{3.22}
\widetilde{R}_-(\xi)\underline{u}\ =\ \Pi_\xi\big(v-\big(P-\lambda\big)(u_1+u_2)\big).
\end{equation}
Taking into account Lemma \ref{L.3.4}, it follows that equation \eqref{3.21} has a unique solution $u_2\in\mathscr{F}_{m,\xi}\cap\mathcal{T}^\bot_\xi$ and we have the estimation:
\begin{equation}\label{3.23}
\|u_2\|_{\mathscr{F}_{m,\xi}}\ \leq\ C\left\|\big(\bb1-\Pi_\xi\big)\big(v-\big(P-\lambda\big)u_1\big)\right\|_{\mathscr{F}_{0,\xi}}\ \leq\ C\left(\|v\|_{\mathscr{F}_{0,\xi}}\,+\,\|u_1\|_{\mathscr{F}_{m,\xi}}\right)\ \leq C\left(\|v\|_{\mathscr{F}_{0,\xi}}\,+\,\|\underline{v}\|_{\mathbb{C}^N}\right)
\end{equation}
for any $(\xi,\lambda)\in\mathbb{T}_*\times I$.

It is very easy to see that equation \eqref{3.22} always has a unique solution $\underline{u}\in\mathbb{C}^N$ given explicitely by:
$$
\underline{u}_j\ :=\ \big(v-\big(P-\lambda\big)(u_1+u_2)\,,\,\psi_j(.,\xi)\big)_{\mathscr{F}_{0,\xi}}.
$$
From this explicit expression we easily get the following estimation:
\begin{equation}\label{3.24}
\|\underline{u}\|_{\mathbb{C}^N}\ \leq\ C\left(\|v\|_{\mathscr{F}_{0,\xi}}\,+\,\left\|\big(P-\lambda\big)(u_1+u_2)\right\|_{\mathscr{F}_{0,\xi}}\right)\ \leq\ C\left(\|v\|_{\mathscr{F}_{0,\xi}}\,+\,\|u_1\|_{\mathscr{F}_{m,\xi}}\,+\,\|u_2\|_{\mathscr{F}_{m,\xi}}\right)\ \leq
\end{equation}
$$
\leq\ C\left(\|v\|_{\mathscr{F}_{0,\xi}}\,+\,\|\underline{v}\|_{\mathbb{C}^N}\right),\qquad\forall(\xi,\lambda)\in\mathbb{T}^{d*}\times I.
$$

In conclusion we have proved the surjectivity of the operator $Q(\xi,\lambda)$ for any values of $(\xi,\lambda)\in\mathbb{T}_*\times I$ and we finish by noticing that the inequalities \eqref{3.19}, \eqref{3.23} and \eqref{3.24} imply the boundedness of the operator $Q(\xi,\lambda)^{-1}$ uniformly with respect to $(\xi,\lambda)\in\mathbb{T}_*\times I$.
\end{proof}

We define now the following family of $N$ functions:
\begin{equation}\label{3.25}
\phi_j(x,\xi)\ :=\ e^{-i<\xi,x>}\psi_j(x,\xi),\qquad\forall(x,\xi)\in\Xi,\ 1\leq j\leq N,
\end{equation}
with the family $\{\psi_j\}_{1\leq j\leq N}$ defined in Lemma \ref{L.3.3}.

\begin{lemma}\label{L.3.7}
The functions $\{\phi_j\}_{1\leq j\leq N}$ defined in \eqref{3.25} have the following properties:
\begin{enumerate}
\item[a)] $\phi_j\in C^\infty(\Xi)$;
\item[b)] $\phi_j(x+\gamma,\xi)=\phi_j(x,\xi),\qquad\forall(x,\xi)\in\Xi,\ \forall\gamma\in\Gamma$;
\item[c)] $\phi_j(x,\xi+\gamma^*)=e^{-i<\gamma^*,x>}\phi_j(x,\xi),\qquad\forall(x,\xi)\in\Xi,\ \forall\gamma^*\in\Gamma^*$;
\item[d)] For any $\alpha\in\mathbb{N}^d$ and any $s\in\mathbb{R}$ there exists a strictly positive constant $C_{\alpha,s}$ such that:
\begin{equation}\label{3.26}
\left\|\big(\partial^\alpha_\xi\phi_j\big)(.,\xi)\right\|_{\mathcal{K}_{s,\xi}}\ \leq\ C_{\alpha,s},\qquad\forall\xi\in\X^*.
\end{equation}
\end{enumerate}
\end{lemma}
\begin{proof}
The properties (a), (b) and (c) follow easily from Lemma \ref{L.3.3} and the definition \eqref{3.25}. In order to prove property (d) we consider \eqref{A.22} and write:
\begin{equation}\label{3.27}
\left\|\big(\partial^\alpha_\xi\phi_j\big)(.,\xi)\right\|_{\mathcal{K}_{s,\xi}}^2\ \leq\ \frac{1}{|E|}\sum\limits_{\gamma^*\in\Gamma^*}<\xi+\gamma^*>^{2s}\left|\big(\widehat{\partial^\alpha_\xi\phi}_j\big)(\gamma^*,\xi)\right|^2,
\end{equation}
where we have used the notation:
\begin{equation}\label{3.28}
\big(\widehat{\partial^\alpha_\xi\phi}_j\big)(\gamma^*,\xi)\ :=\ \frac{1}{|E|}\int_Ee^{-i<\gamma^*,y>}\big(\partial^\alpha_\xi\phi_j\big)(y,\xi)\,dy.
\end{equation}
We recall once again the identity 
$$
e^{-i<\gamma^*,y>}\ =\ <\gamma^*>^{-2l}\big(\bb1-\Delta_y\big)^le^{-i<\gamma^*,y>},\qquad\forall l\in\mathbb{N},
$$
and taking into account the $\Gamma$-periodicity of the function $\big(\partial^\alpha_\xi\phi_j\big)(y,\xi)$ with respect to the variable $y\in\X$ we integrate by parts in \eqref{3.28} and deduce that $\forall\alpha\in\mathbb{N}^d$ and $\forall l\in\mathbb{N}$ there exists a constant $C_{\alpha,l}>0$ such that we have the estimation:
\begin{equation}\label{3.29}
\left|\big(\widehat{\partial^\alpha_\xi\phi}_j\big)(\gamma^*,\xi)\right|\ \leq\ C_{\alpha,l}<\gamma^*>^{-2l}\left|\frac{1}{|E|}\int_Ee^{-i<\gamma^*,y>}\Big(\big(\bb1-\Delta_y\big)^l\partial^\alpha_\xi\phi_j\Big)(y,\xi)\,dy\right|,\quad\forall\xi\in\X^*,\ \forall\gamma^*\in\Gamma^*.
\end{equation}
Coming back to \eqref{3.27}, taking $l\geq s/2$ and using the estimation \eqref{3.29} and Plancherel identity we obtain the following estimation:
\begin{equation}\label{3.30}
\left\|\big(\partial^\alpha_\xi\phi_j\big)(.,\xi)\right\|_{\mathcal{K}_{s,\xi}}^2\ \leq\ C^2_{\alpha,l}\frac{1}{|E|}\int_E\left|\Big(\big(\bb1-\Delta_y\big)^l\partial^\alpha_\xi\phi_j\Big)(y,\xi)\right|^2dy\ \leq\ C_{\alpha,s}^2,\quad\forall\xi\in E^*.
\end{equation}
In order to extend now to an arbitrary $\xi\in\X^*$ we notice that for any $\xi\in\X^*$ there exist $\eta\in E^*$ and $\gamma^*\in\Gamma^*$ such that $\xi=\eta+\gamma^*$ and using property (c) of the functions $\{\phi_1,\ldots,\phi_N\}$ we see that
$$
\left\|\big(\partial^\alpha_\xi\phi_j\big)(.,\xi)\right\|_{\mathcal{K}_{s,\xi}}\ =\ \left\|\big(\partial^\alpha_\eta\phi_j\big)(.,\eta+\gamma^*)\right\|_{\mathcal{K}_{s,\eta+\gamma^*}}\ =\ \left\|e^{-i<\gamma^*,.>}\big(\partial^\alpha_\eta\phi_j\big)(.,\eta)\right\|_{\mathcal{K}_{s,\eta+\gamma^*}}\ =
$$
$$
=\ \left\|<D+\eta+\gamma^*>^se^{-i<\gamma^*,.>}\big(\partial^\alpha_\eta\phi_j\big)(.,\eta)\right\|_{L^2(E)}\ =\ \left\|<D+\eta>^s\big(\partial^\alpha_\eta\phi_j\big)(.,\eta)\right\|_{L^2(E)}\ =\ \left\|\big(\partial^\alpha_\eta\phi_j\big)(.,\eta)\right\|_{\mathcal{K}_{s,\eta}}\ \leq\ C_{\alpha,s}.
$$
\end{proof}

We denote by $\mathcal{K}_0:=\mathcal{K}_{0,0}\equiv L^2(\mathbb{T})\equiv L^2(E)$ and for any $\xi\in\X^*$ we define the linear operators:
\begin{equation}\label{3.31}
\forall u\in\mathcal{K}_0,\qquad R_+(\xi)u:=\left\{\big(u,\phi_j\big)_{\mathcal{K}_0}\right\}_{1\leq j\leq N},
\end{equation}
\begin{equation}\label{3.32}
\forall\underline{u}\in\mathbb{C}^N,\qquad R_+(\xi)\underline{u}:=\sum\limits_{1\leq j\leq N}\underline{u}_j\phi_j(.,\xi).
\end{equation}
We evidently have that $R_+(\xi)\in\mathbb{B}\big(\mathcal{K}_0;\mathbb{C}^N\big)$, $R_-(\xi)\in\mathbb{B}\big(\mathbb{C}^N;\mathcal{K}_0\big)$ and (using \eqref{3.26}) both are $BC^\infty$ functions of $\xi\in\X^*$. With these operators we can now define the following Grushin type operator:
\begin{equation}\label{3.34}
\mathcal{P}(\xi,\lambda)\ :=\ \left(
\begin{array}{cc}
P_\xi-\lambda&R_-(\xi)\\
R_+(\xi)&0
\end{array}
\right)\ \in\ \mathbb{B}\big(\mathcal{K}_{m,\xi}\times\mathbb{C}^N;\mathcal{K}_0\times\mathbb{C}^N\big),\quad\forall(\xi,\lambda)\in\X^*\times I.
\end{equation}

\begin{proposition}\label{P.3.8}
With the above notations, the following statements are true:
\begin{enumerate}
\item[a)] As a function of $(\xi,\lambda)\in\X^*\times I$, we have that $\mathcal{P}\in C^\infty\big(\X^*\times I;\mathbb{B}\big(\mathcal{K}_{m,0}\times\mathbb{C}^N;\mathcal{K}_0\times\mathbb{C}^N\big)\big)$ and for any $\alpha\in\mathbb{N}^d$ and any $k\in\mathbb{N}$ we have that $\big(\partial^\alpha_\xi\partial^k_\lambda\mathcal{P}\big)(\xi,\lambda)\in\mathbb{B}\big(\mathcal{K}_{m,\xi}\times\mathbb{C}^N;\mathcal{K}_0\times\mathbb{C}^N\big)$ uniformly in $(\xi,\lambda)\in\X^*\times I$.
\item[b)] If we considere $\mathcal{P}(\xi,\lambda)$ as an unbounded operator in $\mathcal{K}_0\times\mathbb{C}^N$ with domain $\mathcal{K}_{m,\xi}\times\mathbb{C}^N$ then for any $(\xi,\lambda)\in\X^*\times I$, it is self-adjoint and unitarily equivalent with the operator $Q(\xi,\lambda)$.
\item[c)] The operator $\mathcal{P}(\xi,\lambda)$ has an inverse:
\begin{equation}\label{3.35}
\mathcal{E}_0(\xi,\lambda)\ :=\ \left(
\begin{array}{cc}
E^0(\xi,\lambda)&E^0_+(\xi,\lambda)\\
E^0_-(\xi,\lambda)&E^0_{-,+}(\xi,\lambda)
\end{array}
\right)\in\mathbb{B}\big(\mathcal{K}_0\times\mathbb{C}^N;\mathcal{K}_{m,\xi}\times\mathbb{C}^N\big),
\end{equation}
uniformly bounded with respect to $(\xi,\lambda)\in\X^*\times I$.
\item[d)] As a function of $(\xi,\lambda)\in\X^*\times I$, we have that $\mathcal{E}_0\in C^\infty\big(\X^*\times I;\mathbb{B}\big(\mathcal{K}_0\times\mathbb{C}^N;\mathcal{K}_{m,0}\times\mathbb{C}^N\big)\big)$ and for any $\alpha\in\mathbb{N}^d$ and any $k\in\mathbb{N}$ we have that $\big(\partial^\alpha_\xi\partial^k_\lambda\mathcal{P}\big)(\xi,\lambda)\in\mathbb{B}\big(\mathcal{K}_0\times\mathbb{C}^N;\mathcal{K}_{m,\xi}\times\mathbb{C}^N\big)$ uniformly in $(\xi,\lambda)\in\X^*\times I$.
\end{enumerate}
\end{proposition}
\begin{proof}
Point (a) follows clearly from the smoothness of the maps $R_-(\xi)$ and $R_+(\xi)$ (proved above) and the arguments in Exemple \ref{E.A.20}.

b) Let us define the operator:
\begin{equation}\label{3.35.b}
U(\xi)\ :=\ \left(
\begin{array}{cc}
\sigma_\xi&0\\
0&\bb1
\end{array}
\right):\mathcal{K}_{s,\xi}\times\mathbb{C}^N\rightarrow\mathscr{F}_{s,\xi}\times\mathbb{C}^N
\end{equation}
and let us notice that it is evidently unitary $\forall(\xi,s)\in\X^*\times\mathbb{R}$. Using Remark \ref{R.A.22} we know that $P_\xi-\lambda=\sigma_{-\xi}\big(P-\lambda\big)\sigma_\xi$ on $\mathcal{K}_{m,\xi}$.  If we use the relations \eqref{3.14}, \eqref{3.25} and \eqref{3.32} we deduce that for any $\underline{u}\in\mathbb{C}^N$:
$$
\sigma_{-\xi}\widetilde{R}_-(\xi)\underline{u}\ =\ \sum\limits_{j=1}^{N}\underline{u}_j\big(\sigma_{-\xi}\psi_j\big)(.,\xi)\ =\ \sum\limits_{j=1}^{N}\underline{u}_j\phi_j(.,\xi)\ =\ R_-(\xi)\underline{u}.
$$
In a similar way, using now \eqref{3.13}, \eqref{3.25} and \eqref{3.31} we obtain that $\forall u\in\mathcal{K}_0$ we have that:
$$
\widetilde{R}_+(\xi)\big(\sigma_\xi u\big)\ =\ \left\{\big(\sigma_\xi u,\psi_j(.,\xi)\big)_{\mathscr{F}_{0,\xi}}\right\}_{1\leq j\leq N}\ =\ \left\{\big(u,\phi_j(.,\xi)\big)_{\mathcal{K}_0}\right\}_{1\leq j\leq N}\ =\ R_+(\xi)u.
$$
We conclude that we have the following equality on $\mathcal{K}_{m,\xi}\times\mathbb{C}^N$:
\begin{equation}\label{3.36}
\mathcal{P}(\xi,\lambda)\ =\ U(\xi)^{-1}Q(\xi,\lambda)U(\xi),\qquad\forall(\xi,\lambda)\in\X^*\times I.
\end{equation}

c) follows easily from point (b).

d) Let us notice that (a) and (c) imply easily that $\mathcal{E}_0\in C^\infty\big(\X^*\times I;\mathbb{B}\big(\mathcal{K}_0\times\mathbb{C}^N;\mathcal{K}_{m,0}\times\mathbb{C}^N\big)\big)$. The last property of $\mathcal{E}_0$ can be proved by recurence, differentiating the equality $\mathcal{P}\mathcal{E}_0=\bb1$ valid on $\mathcal{K}_0\times\mathbb{C}^N$. For example, if $|\alpha|+k=1$ we can write that
$$
\partial^\alpha_\xi\partial^k_\lambda\mathcal{E}_0\ =\ -\mathcal{E}_0\Big(\partial^\alpha_\xi\partial^k_\lambda\mathcal{P}\Big)\mathcal{E}_0
$$
and thus we have the estimation
$$
\left\|\partial^\alpha_\xi\partial^k_\lambda\mathcal{E}_0\right\|_{\mathbb{B}(\mathcal{K}_0\times\mathbb{C}^N;\mathcal{K}_{m,\xi}\times\mathbb{C}^N)}\leq\left\|\mathcal{E}_0\right\|_{\mathbb{B}(\mathcal{K}_0\times\mathbb{C}^N;\mathcal{K}_{m,\xi}\times\mathbb{C}^N)}\left\|\partial^\alpha_\xi\partial^k_\lambda\mathcal{P}\right\|_{\mathbb{B}(\mathcal{K}_{m,\xi}\times\mathbb{C}^N;\mathcal{K}_0\times\mathbb{C}^N)}  \left\|\mathcal{E}_0\right\|_{\mathbb{B}(\mathcal{K}_0\times\mathbb{C}^N;\mathcal{K}_{m,\xi}\times\mathbb{C}^N)},
$$
where the three factors of the right hand side are clearly uniformly bounded with respect to $(\xi,\lambda)\in\X^*\times I$.
\end{proof}

\section{ Construction of the effective Hamiltonian}\label{S.4}
\setcounter{equation}{0}
\setcounter{theorem}{0}

Let us recall our hypothesis (see Section \ref{S.1}): $\{B_\epsilon\}_{|\epsilon|\leq\epsilon_0}$ is a family of magnetic fields satisfying Hypothesis H.1, $\{A_\epsilon\}_{|\epsilon|\leq\epsilon_0}$ is an associated family of vector potentials and $\{p_\epsilon\}_{|\epsilon|\leq\epsilon_0}$ is a family of symbols satisfying Hypothesis H.2 - H.6.

As we have already noticed in Remark \ref{R.A.6}, the symbol at $\epsilon=0$, $p_0(x,y,\eta)$ does not depend on the first variable $x\in\X$; thus if we denote by ${\rm p}_0(y,\eta):=p_0(0,y,\eta)$ and by $r_\epsilon(x,y,\eta):=p_\epsilon(x,y,\eta)-{\rm p}_0(y,\eta)$, we notice that the symbol ${\rm p}_0$ verifies the hypothesis of Section \ref{S.3} and thus ${\rm p}_0\in S^m_1(\mathbb{T})$ being real and elliptic and can write
\begin{equation}\label{4.1}
p_\epsilon\ =\ {\rm p}_0\,+\,r_\epsilon,\qquad\underset{\epsilon\rightarrow0}{\lim}r_\epsilon\,=\,0,\ \text{in}\ S^m_1(\X\times\mathbb{T}).
\end{equation}
We apply the constructions from Section \ref{S.3} to the operator $P_0:=\mathfrak{Op}({\rm p}_0)$. We obtain that for any compact interval $I\subset\mathbb{R}$, for any $\lambda\in I$ and any $\xi\in\X^*$ one can construct the operators $R_{\pm}(\xi)$ as in \eqref{3.31} and \eqref{3.32} in order to use \eqref{3.34} and define the operator
\begin{equation}\label{4.2}
\mathcal{P}_0(\xi,\lambda)\ :=\ \left(
\begin{array}{cc}
P_{0,\xi}-\lambda&R_-(\xi)\\
R_+(\xi)&0
\end{array}
\right)\in\mathbb{B}\big(\mathcal{K}_{m,\xi}\times\mathbb{C}^N;\mathcal{K}_0\times\mathbb{C}^N\big),
\end{equation}
that will verify all the properties listed in Proposition \ref{P.3.8}. In particular
\begin{equation}\label{4.3}
\mathcal{P}_0(.,\lambda)\,\in\,S^0_0\big(\X;\mathbb{B}\big(\mathcal{K}_{m,\xi}\times\mathbb{C}^N;\mathcal{K}_0\times\mathbb{C}^N\big)\big),
\end{equation}
uniformly for $\lambda\in I$. Moreover, it follows that $\mathcal{P}_0(\xi,\lambda)$ is invertible and its inverse denoted by $\mathcal{E}_0(\xi,\lambda)$ is given (as in \eqref{3.35}) by
\begin{equation}\label{4.4}
\mathcal{E}_0(\xi,\lambda)\ :=\ \left(
\begin{array}{cc}
E^0(\xi,\lambda)&E^0_+(\xi,\lambda)\\
E^0_-(\xi,\lambda)&E^0_{-,+}(\xi,\lambda)
\end{array}
\right)\,\in\,\mathbb{B}\big(\mathcal{K}_0\times\mathbb{C}^N;\mathcal{K}_{m,\xi}\times\mathbb{C}^N\big)
\end{equation}
and has the property that
\begin{equation}\label{4.5}
\mathcal{E}_0(.,\lambda)\,\in\,S^0_0\big(\X;\mathbb{B}\big(\mathcal{K}_0\times\mathbb{C}^N;\mathcal{K}_{m,\xi}\times\mathbb{C}^N\big),
\end{equation}
uniformly for $\lambda\in I$.

Let us consider now the following operator:
\begin{equation}\label{4.6}
\mathcal{P}_\epsilon(x,\xi,\lambda)\ :=\ \left(
\begin{array}{cc}
\mathfrak{q}_\epsilon(x,\xi)-\lambda&R_-(\xi)\\
R_+(\xi)&0
\end{array}
\right),\qquad\lambda\in I,\ \epsilon\in[-\epsilon_0,\epsilon_0],\ (x,\xi)\in\Xi
\end{equation}
where we recall that $\mathfrak{q}_\epsilon(x,\xi):=\mathfrak{Op}(\widetilde{p}_\epsilon(x,.,\xi,.))$, $\widetilde{p}_\epsilon(x,y,\xi,\eta):=p_\epsilon(x,y,\xi+\eta)$. Taking into account the example \ref{E.A.20} from the Appendix we notice that $\mathfrak{q}_\epsilon\in S^0_0\big(\X;\mathbb{B}\big(\mathcal{K}_{m,\xi};\mathcal{K}_0\big)\big)$ uniformly in $\epsilon\in[-\epsilon_0,\epsilon_0]$; thus
\begin{equation}\label{4.7}
\mathcal{P}_\epsilon(x,\xi,\lambda)\,\in\,S^0_0\big(\X;\mathbb{B}\big(\mathcal{K}_{m,\xi}\times\mathbb{C}^N;\mathcal{K}_0\times\mathbb{C}^N\big)\big),
\end{equation}uniformly with respect to $(\lambda,\epsilon)\in I\times[-\epsilon_0,\epsilon_0]$.

\begin{lemma}\label{L.4.1}
The operator $\mathcal{P}_{\epsilon,\lambda}:=\mathfrak{Op}^{A_\epsilon}(\mathcal{P}_\epsilon(.,.,\lambda))$ belongs to $\mathbb{B}\big(\mathcal{K}^m_\epsilon(\X^2)\times L^2(\X;\mathbb{C}^N);\mathcal{K}(\X^2)\times L^2(\X;\mathbb{C}^N)\big)$ uniformly with respect to $(\lambda,\epsilon)\in I\times[-\epsilon_0,\epsilon_0]$. Moreover, considering $\mathcal{P}_{\epsilon,\lambda}$ as an unbounded linear operator in the Hilbert space $\mathcal{K}(\X^2)\times L^2(\X;\mathbb{C}^N)$ it defines a self-adjoint operator on the domain $\mathcal{K}^m_\epsilon(\X^2)\times L^2(\X;\mathbb{C}^N)$.
\end{lemma}
\begin{proof}
If we denote by $\mathfrak{R}_{\mp,\epsilon}:=\mathfrak{Op}^{A_\epsilon}(R_\mp(\xi))$ we can write
\begin{equation}\label{4.8}
\mathcal{P}_{\epsilon,\lambda}\ =\ \left(
\begin{array}{cc}
\widetilde{P}_\epsilon-\lambda&\mathfrak{R}_{-,\epsilon}\\
\mathfrak{R}_{+,\epsilon}&0
\end{array}
\right).
\end{equation}
Taking into account Lemma \ref{L.2.12} we may conclude that $\widetilde{P}_\epsilon\in\mathbb{B}\big(\mathcal{K}^m_\epsilon(\X^2);\mathcal{K}(\X^2)\big)$ uniformly with respect to $\epsilon\in[-\epsilon_0,\epsilon_0]$. Noticing that $R_-(\xi)=R_+(\xi)^*$ and belongs to $S^0_0(\X;\mathbb{B}\big(\mathbb{C}^N;\mathcal{K}_0)\big)$, Proposition \ref{P.A.26} implies that 
$$
\mathfrak{R}_{-,\epsilon}\ =\ \mathfrak{R}_{+,\epsilon}^*\,\in\,\mathbb{B}\big(L^2(\X;\mathbb{C}^N);\mathcal{K}(\X^2)\big) 
$$
uniformly with respect to $\epsilon\in[-\epsilon_0,\epsilon_0]$. This gives us the first part of the statement of the Lemma. The self-adjointness follows from the self-adjointness of $\widetilde{P}_\epsilon$ in $\mathcal{K}(\X^2)$ on the domain $\mathcal{K}^m_\epsilon(\X^2)$ and this follows from Proposition \ref{P.2.13}.
\end{proof}

\begin{lemma}\label{L.4.2}
The operator $\mathcal{E}_{0,\epsilon,\lambda}:=\mathfrak{Op}^{A_\epsilon}\big(\mathcal{E}_0(.,\lambda)\big)$ belons to $\mathbb{B}\big(\mathcal{K}(\X^2)\times L^2(\X;\mathbb{C}^N);\mathcal{K}^m_\epsilon(\X^2)\times L^2(\X;\mathbb{C}^N)\big)$ uniformly with respect to $(\lambda,\epsilon)\in I\times[-\epsilon_0,\epsilon_0]$.
\end{lemma}
\begin{proof}
We can write
\begin{equation}\label{4.9}
\mathcal{E}_{0,\epsilon,\lambda}\ =\ \left(
\begin{array}{cc}
\mathfrak{E}^0_{\epsilon,\lambda}&\mathfrak{E}^0_{+,\epsilon,\lambda}\\
\mathfrak{E}^0_{-,\epsilon,\lambda}&\mathfrak{E}^0_{-+,\epsilon,\lambda}
\end{array}
\right),
\end{equation}
with
$$
\mathfrak{E}^0_{\epsilon,\lambda}\ :=\ \mathfrak{Op}^{A_\epsilon}\big(E^0(.,\lambda)\big),\qquad\mathfrak{E}^0_{\pm,\epsilon,\lambda}\ :=\ \mathfrak{Op}^{A_\epsilon}\big(E^0_\pm(.,\lambda)\big),\qquad\mathfrak{E}^0_{-+,\epsilon,\lambda}\ :=\ \mathfrak{Op}^{A_\epsilon}\big(E^0_{-+}(.,\lambda)\big).
$$
From \eqref{4.5} it follows that $\mathcal{E}_0(.,\lambda)\in S^0_0\big(\X;\mathbb{B}\big(\mathcal{K}_0\times\mathbb{C}^N;\mathcal{K}_0\times\mathbb{C}^N\big)\big)$ uniformly with respect to $(\lambda,\epsilon)\in I\times[-\epsilon_0,\epsilon_0]$. In order to prove the boundedness result in the Lemma it is enough to show that
\begin{equation}\label{4.10}
\left(\begin{array}{cc}
\widetilde{Q}_{m,\epsilon}&0\\
0&\bb1
\end{array}
\right)\,\mathcal{E}_{0,\epsilon,\lambda}\,\in\,\mathbb{B}\big(\mathcal{K}(\X^2)\times L^2(\X;\mathbb{C}^N);\mathcal{K}(\X^2)\times L^2(\X;\mathbb{C}^N)\big),
\end{equation}
uniformly with respect to $(\lambda,\epsilon)\in I\times[-\epsilon_0,\epsilon_0]$; here $\widetilde{Q}_{m,\epsilon}$ is defined before Definition \ref{D.1.4}, with some suitable identifications. In that Definition we also argued that the operator $\widetilde{Q}_{m,\epsilon}$ corresponds to the operator $Q_{m,\epsilon}$ from Remark \ref{R.A.25} transformed by {\it doubling the variables} starting from the operator valued symbol $\mathfrak{q}_{m,\epsilon}$. We may thus conclude that $\widetilde{Q}_{m,\epsilon}$ is obtained by the $\mathfrak{Op}^{A_\epsilon}$ quantization of a symbol from $S^0_0\big(\X;\mathbb{B}(\mathcal{K}_0;\mathcal{K}_{m,\xi})\big)$. Taking into account that $E^0(\xi,\lambda)\in S^0_0\big(\X;\mathbb{B}(\mathcal{K}_0;\mathcal{K}_{m,\xi})\big)$ and $E^0_+(\xi,\lambda)\in S^0_0\big(\X;\mathbb{B}(\mathbb{C}^N;\mathcal{K}_{m,\xi})\big)$, the property \eqref{4.10} follows from the {\it Composition Theorem} \ref{T.A.23} a) and from the Proposition \ref{P.A.26}.
\end{proof}

\begin{theorem}\label{T.4.3}
For a sufficiently small $\epsilon_0>0$ we have that for any $(\lambda,\epsilon)\in I\times[-\epsilon_0,\epsilon_0]$ the operator $\mathcal{P}_{\epsilon,\lambda}$ from Lemma \ref{L.4.1} has an inverse denoted by
\begin{equation}\label{4.11}
\mathcal{E}_{\epsilon,\lambda}\ :=\ \left(
\begin{array}{cc}
\mathfrak{E}(\epsilon,\lambda)&\mathfrak{E}_+(\epsilon,\lambda)\\
\mathfrak{E}_-(\epsilon,\lambda)&\mathfrak{E}_{-+}(\epsilon,\lambda)
\end{array}
\right)\,\in\,\mathbb{B}\big(\mathcal{K}(\X^2)\times L^2(\X;\mathbb{C}^N);\mathcal{K}^m_\epsilon(\X^2)\times L^2(\X;\mathbb{C}^N)\big),
\end{equation}
uniformly with respect to $(\epsilon,\lambda)\in[-\epsilon_0,\epsilon_0]\times I$. Moreover we have that 
$$
\mathcal{E}_{\epsilon,\lambda}\ =\ \mathcal{E}_{0,\epsilon,\lambda}\,+\,\mathcal{R}_{\epsilon,\lambda},\qquad\mathcal{R}_{\epsilon,\lambda}=\mathfrak{Op}^{A_\epsilon}\big(\rho_{\epsilon,\lambda}\big),\qquad\underset{\epsilon\rightarrow0}{\lim}\rho_{\epsilon,\lambda}=0\ \text{in}\ S^0_0\big(\X;\mathbb{B}(\mathcal{K}_0\times\mathbb{C}^N;\mathcal{K}_{m,\xi}\times\mathbb{C}^N)\big).
$$
In particular we have that
\begin{equation}\label{4.12}
\mathfrak{E}_{-+}(\epsilon,\lambda)\ =\ \mathfrak{Op}^{A_\epsilon}\big(E^{-+}_{\epsilon,\lambda}\big),\qquad\underset{\epsilon\rightarrow0}{\lim}E^{-+}_{\epsilon,\lambda}=E^0_{-+}(.,\lambda)\ \text{in}\ S^0_0\big(\X;\mathbb{B}(\mathbb{C}^N;\mathbb{C}^N)\big),
\end{equation}
uniformly with respect to $\lambda\in I$.
\end{theorem}
\begin{proof}
We begin the proof with the following remarks.

1. The symbol $\mathcal{E}^0(\xi,\lambda)$ appearing in \eqref{4.4} does not depend on $x\in\X$ and on $\epsilon\in[-\epsilon_0,\epsilon_0]$. We can thus consider that $\mathcal{E}^0(\xi,\lambda)\in S^0_{0,\epsilon}\big(\X;\mathbb{B}(\mathcal{K}_0\times\mathbb{C}^N;\mathcal{K}_{m,\xi}\times\mathbb{C}^N)\big)$ uniformly for $\lambda\in I$.

2. The symbol $\mathcal{P}_0(\xi,\lambda)$ appearing in \eqref{4.2} does not depend on $x\in\X$ and on $\epsilon\in[-\epsilon_0,\epsilon_0]$. We can thus consider that $\mathcal{P}_0(\xi,\lambda)\in S^0_{0,\epsilon}\big(\X;\mathbb{B}(\mathcal{K}_{m,\xi}\times\mathbb{C}^N;\mathcal{K}_0\times\mathbb{C}^N)\big)$ uniformly for $\lambda\in I$.

3. From the relations \eqref{4.1}, \eqref{4.2} and \eqref{4.6} it follows that
$$
\mathcal{P}_\epsilon(x,\xi,\lambda)\,-\,\mathcal{P}_0(\xi,\lambda)\ =\ \left(
\begin{array}{cc}
\mathfrak{q}^\prime_\epsilon(x,\xi)&0\\
0&0
\end{array}
\right),
$$
where $\mathfrak{q}^\prime_\epsilon(x,\xi):=\mathfrak{Op}\big(\widetilde{r}_\epsilon(x,.\xi,.)\big)$ and $\widetilde{r}_\epsilon(x,y,\xi,\eta):=r_\epsilon(x,y,\xi+\eta)$. The fact that $\underset{\epsilon\rightarrow0}{\lim}r_\epsilon=0$ in $S^m_1\big(\X\times\mathbb{T}\big)$ and that the map $S^m_1\big(\X\times\mathbb{T}\big)\ni r_\epsilon\mapsto\mathfrak{q}^\prime_\epsilon\in S^0_0\big(\X;\mathbb{B}(\mathcal{K}_{m,\xi};\mathcal{K}_0)\big)$ is continuous (by an evident generalization of property \eqref{A.32}), we conclude that
\begin{equation}\label{4.13}
\underset{\epsilon\rightarrow0} {\lim}\big[\mathcal{P}_\epsilon(x,\xi,\lambda)-\mathcal{P}_0(\xi,\lambda)\big]=0
\end{equation}
in $S^0_0\big(\X;\mathbb{B}(\mathcal{K}_{m,\xi}\times\mathbb{C}^N;\mathcal{K}_0\times\mathbb{C}^N)\big)$ uniformly with respect to $\lambda\in I$.

Let us come back to the proof of the Theorem and denote by $\mathcal{P}^0_{\epsilon,\lambda}:=\mathfrak{Op}^{A_\epsilon}\big(\mathcal{P}_0(\xi,\lambda)\big)$. We can write that
\begin{equation}\label{4.14}
\mathcal{P}_{\epsilon,\lambda}\,\mathcal{E}_{0,\epsilon,\lambda}\ =\ \mathcal{P}^0_{\epsilon,\lambda}\,\mathcal{E}_{0,\epsilon,\lambda}\,+\,\big(\mathcal{P}_{\epsilon,\lambda}-\mathcal{P}^0_{\epsilon,\lambda})\,\mathcal{E}_{0,\epsilon,\lambda}
\end{equation}
in $\mathbb{B}\big(\mathcal{K}(\X^2)\times L^2(\X:\mathbb{C}^N);\mathcal{K}(\X^2)\times L^2(\X;\mathbb{C}^N)\big)$.

Using the {\it Composition Theorem} \ref{T.A.23} and the above remarks, we conclude that
\begin{equation}\label{4.15}
\mathcal{P}_{\epsilon,\lambda}\,\mathcal{E}_{0,\epsilon,\lambda}\ =\ \bb1\,+\,\mathfrak{Op}^{A_\epsilon}\big(s_{\epsilon,\lambda}\big)
\end{equation}
in $\mathbb{B}\big(\mathcal{K}(\X^2)\times L^2(\X:\mathbb{C}^N);\mathcal{K}(\X^2)\times L^2(\X;\mathbb{C}^N)\big)$, where 
\begin{equation}\label{4.16}
\underset{\epsilon\rightarrow0}{\lim}s_{\epsilon,\lambda}=0,
\end{equation}
in $S^0_0\big(\X;\mathbb{B}(\mathcal{K}_0\times\mathbb{C}^N;\mathcal{K}_0\times\mathbb{C}^N)\big)$ uniformly with respect to $\lambda\in I$.

It follows from Proposition \ref{P.A.27} that for $\epsilon_0>0$ small enough, the operator $\bb1+\mathfrak{Op}^{A_\epsilon}(s_{\epsilon,\lambda})$ is invertible in $\mathbb{B}\big(\mathcal{K}(\X^2)\times L^2(\X;\mathbb{C}^N);\mathcal{K}(\X^2)\times L^2(\X;\mathbb{C}^N)\big)$ for any $(\epsilon,\lambda)\in[-\epsilon_0,\epsilon_0]\times I$ and it exists a symbol $t_{\epsilon,\lambda}$ such that
\begin{equation}\label{4.17}
\underset{\epsilon\rightarrow0}{\lim}t_{\epsilon,\lambda}=0,
\end{equation}
in $S^0_0\big(\X;\mathbb{B}(\mathcal{K}_0\times\mathbb{C}^N;\mathcal{K}_0\times\mathbb{C}^N)\big)$ uniformly with respect to $\lambda\in I$ and
\begin{equation}\label{4.18}
\big[\bb1\,+\,\mathfrak{Op}^{A_\epsilon}(s_{\epsilon,\lambda})\big]^{-1}\ =\ \bb1\,+\,\mathfrak{Op}^{A_\epsilon}(t_{\epsilon,\lambda}).
\end{equation}

Let us define
$$
\mathcal{E}_{\epsilon,\lambda}\ :=\ \mathcal{E}_{0,\epsilon,\lambda}\,\big[\bb1\,+\,\mathfrak{Op}^{A_\epsilon}(t_{\epsilon,\lambda})\big]
$$
and let us notice that it is a right inverse for $\mathcal{P}_{\epsilon,\lambda}$. As the operator $\mathcal{P}_{\epsilon,\lambda}$ is self-adjoint, it follows that $\mathcal{E}_{\epsilon,\lambda}$ defined above is also a left inverse for it.

The other properties in the statement of the Theorem are evident now.
\end{proof}

\begin{remark}\label{R.4.4}
The operator $\mathfrak{E}_{-+}(\epsilon,\lambda)$ defined in \eqref{4.12} will be the effective Hamiltonian associated to the Hamiltonian $P_\epsilon$ and the interval $I$. Its importance will partially be explained in the following Corollary.
\end{remark}

\begin{corollary}\label{C.4.5}
Under the assumptions of Theorem \ref{T.4.3}, for any $\lambda\in I$ and any $\epsilon\in[-\epsilon_0,\epsilon_0]$ the following equivalence is true:
\begin{equation}\label{4.19}
\lambda\,\in\,\sigma\big(\widetilde{P}_\epsilon\big)\quad\Longleftrightarrow\quad0\,\in\,\sigma\big(\mathfrak{E}_{-+}(\epsilon,\lambda)\big).
\end{equation}
\end{corollary}
\begin{proof}
The equality
$$
\mathcal{P}_{\epsilon,\lambda}\,\mathcal{E}_{\epsilon,\lambda}\ =\ \left(
\begin{array}{cc}
\bb1_{\mathcal{K}(\X^2)}&0\\
0&\bb1_{L^2(\X;\mathbb{C}^N)}
\end{array}
\right)
$$
is equivalent with the following system of equations:
\begin{equation}\label{4.20}
\left\{
\begin{array}{rcl}
\big(\widetilde{P}_\epsilon-\lambda\big)\,\mathfrak{E}(\epsilon,\lambda)\,+\,\mathfrak{R}_{-,\epsilon}\,\mathfrak{E}_-(\epsilon,\lambda)&=&\bb1_{\mathcal{K}(\X^2)},\\
\big(\widetilde{P}_\epsilon-\lambda\big)\,\mathfrak{E}_+(\epsilon,\lambda)\,+\,\mathfrak{R}_{-,\epsilon}\,\mathfrak{E}_{-+}(\epsilon,\lambda)&=&0,\\
\mathfrak{R}_{+,\epsilon}\,\mathfrak{E}(\epsilon,\lambda)&=&0,\\
\mathfrak{R}_{+,\epsilon}\,\mathfrak{E}_+(\epsilon,\lambda)&=&\bb1_{L^2(\X;\mathbb{C}^N)}.
\end{array}
\right.
\end{equation}

If $0\notin\sigma\big(\mathfrak{E}_{-+}(\epsilon,\lambda)\big)$ the second equality in \eqref{4.20} implies that 
$$
\mathfrak{R}_{-,\epsilon}\ =\ -\big(\widetilde{P}_\epsilon-\lambda\big)\mathfrak{E}_+(\epsilon,\lambda)\mathfrak{E}_{-+}(\epsilon,\lambda)^{-1}
$$
and by substituting this value in the first equality in \eqref{4.20} we obtain
$$
\big(\widetilde{P}_\epsilon-\lambda\big)\left[\mathfrak{E}(\epsilon,\lambda)-\mathfrak{E}_+(\epsilon,\lambda)\mathfrak{E}_{-+}(\epsilon,\lambda)^{-1}\mathfrak{E}_-(\epsilon,\lambda)\right]\ =\ \bb1_{\mathcal{K}(\X^2)}.
$$
It follows that $\lambda\notin\sigma(\widetilde{P}_\epsilon)$.

Suppose now that $\lambda\notin\sigma(\widetilde{P}_\epsilon)$; then the second equality in \eqref{4.20} implies that
$$
\mathfrak{E}_+(\epsilon,\lambda)\ =\ -\big(\widetilde{P}_\epsilon-\lambda\big)^{-1}\mathfrak{R}_{-,\epsilon}\,\mathfrak{E}_{-+}(\epsilon,\lambda).
$$
After substituting this last expression in the last equality in \eqref{4.20} we get
$$
-\mathfrak{R}_{+,\epsilon}\big(\widetilde{P}_\epsilon-\lambda\big)^{-1}\mathfrak{R}_{-,\epsilon}\,\mathfrak{E}_{-+}(\epsilon,\lambda)\ =\ \bb1_{L^2(\X;\mathbb{C}^N)}.
$$
It follows that $0\notin\sigma\big(\mathfrak{E}_{-+}(\epsilon,\lambda)\big)$ and we obtain the following identity (valid in this case):
$$
\mathfrak{E}_{-+}(\epsilon,\lambda)^{-1}\ =\ -\mathfrak{R}_{+,\epsilon}\big(\widetilde{P}_\epsilon-\lambda\big)^{-1}\mathfrak{R}_{-,\epsilon}.
$$
\end{proof}

We recall that for any $\gamma^*\in\Gamma^*$ we denote by $\sigma_{\gamma^*}$ the operator of multiplication with the character $e^{i<\gamma^*,.>}$ on $\mathscr{S}^\prime(\X;\mathbb{C}^N)$ and we define now $\Upsilon_{\gamma^*}$ as the operator of multiplication with the function $\sigma_{\gamma^*}\otimes\sigma_{-\gamma^*}$ on the space $\mathscr{S}^\prime(\X^2)$. We shall need further the following commutation property.

\begin{lemma}\label{L.4.6}
 For any $\gamma^*\in\Gamma^*$ the following equality is true $\forall(\epsilon,\lambda)\in[-\epsilon_0,\epsilon_0]\times I$:
 \begin{equation}\label{4.21}
 \left(
 \begin{array}{cc}
 \Upsilon_{\gamma^*}&0\\
 0&\sigma_{\gamma^*}
 \end{array}
 \right)\,\mathcal{P}_{\epsilon,\lambda}\ =\ \mathcal{P}_{\epsilon,\lambda}\,
 \left(
 \begin{array}{cc}
 \Upsilon_{\gamma^*}&0\\
 0&\sigma_{\gamma^*}
 \end{array}
 \right)
 \end{equation}
 as operators on $\mathscr{S}(\X\times\mathbb{T})\times\mathscr{S}(\X;\mathbb{C}^N)$ (identifying the test functions on the torus with the associated periodic distributions).
\end{lemma}
\begin{proof}
From equality \eqref{2.21} we deduce that for any $\gamma^*\in\Gamma^*$ we have the following equality on $\mathscr{S}(\X\times\mathbb{T})$:
\begin{equation}\label{4.22}
\Upsilon_{\gamma^*}\,\widetilde{P}_\epsilon\ =\ \widetilde{P}_\epsilon\,\Upsilon_{\gamma^*}.
\end{equation}
Taking now into account Lemma \ref{L.3.7} and the definitions \eqref{3.31} and \eqref{3.32} we obtain that
\begin{equation}\label{4.23}
R_+(\xi+\gamma^*)\ =\ R_+(\xi)\,\sigma_{\gamma^*},\quad R_-(\xi+\gamma^*)\ =\ \sigma_{-\gamma^*}\,R_-(\xi),\qquad\forall\xi\in\X^*,\ \forall\gamma^*\in\Gamma^*.
\end{equation}
Repeating the computations done in Exemple \ref{E.A.4} we obtain for any $\underline{u}\in\mathscr{S}(\X;\mathbb{C}^N)$:
$$
\big(\mathfrak{R}_{-,\epsilon}\sigma_{\gamma^*}\underline{u}\big)(x,y)=\int_\Xi e^{i<\zeta,x-z>}\omega^{A_\epsilon}(x,z)\left[R_-(\zeta)e^{i<\gamma^*,z>}\underline{u}(z)\right](y)dz\,\dbar\zeta=
$$
$$
=e^{i<\gamma^*,x>}\int_\Xi e^{i<\zeta-\gamma^*,x-z>}\omega^{A_\epsilon}(x,z)\left[R_-(\zeta)\underline{u}(z)\right](y)dz\,\dbar\zeta=
$$
$$
=e^{i<\gamma^*,x>}\int_\Xi e^{i<\zeta,x-z>}\omega^{A_\epsilon}(x,z)\left[R_-(\zeta+\gamma^*)\underline{u}(z)\right](y)dz\,\dbar\zeta=
$$
$$
=e^{i<\gamma^*,x>}\int_\Xi e^{i<\zeta,x-z>}\omega^{A_\epsilon}(x,z)\left[\sigma_{-\gamma^*}R_-(\zeta)\underline{u}(z)\right](y)dz\,\dbar\zeta=\big(\Upsilon_{\gamma^*}\mathfrak{R}_{-,\epsilon}\underline{u}\big)(x,y),
$$
concluding that on $\mathscr{S}(\X;\mathbb{C}^N)$ we have the equality
\begin{equation}\label{4.24}
\mathfrak{R}_{-,\epsilon}\sigma_{\gamma^*}\ =\ \Upsilon_{\gamma^*}\mathfrak{R}_{-,\epsilon}.
\end{equation}
In a similar way we obtain that on $\mathscr{S}(\X\times\mathbb{T})$ we have the equality:
\begin{equation}\label{4.25}
\mathfrak{R}_{+,\epsilon}\,\Upsilon_{\gamma^*}\ =\ \sigma_{\gamma^*}\,\mathfrak{R}_{+,\epsilon}.
\end{equation}
\end{proof}
\begin{remark}\label{R.4.7}
Of course the inverse of the operator $\mathcal{P}_{\epsilon,\lambda}$ verifies a commutation equation similar to \eqref{4.21}:
\begin{equation}\label{4.26}
 \left(
 \begin{array}{cc}
 \Upsilon_{\gamma^*}&0\\
 0&\sigma_{\gamma^*}
 \end{array}
 \right)\,\mathcal{E}_{\epsilon,\lambda}\ =\ \mathcal{E}_{\epsilon,\lambda}\,
 \left(
 \begin{array}{cc}
 \Upsilon_{\gamma^*}&0\\
 0&\sigma_{\gamma^*}
 \end{array}
 \right),\qquad\forall\gamma^*\in\Gamma^*,
\end{equation}
on $\mathscr{S}(\X\times\mathbb{T})\times\mathscr{S}(\X;\mathbb{C}^N)$ and for any $(\epsilon,\lambda)\in[-\epsilon_0,\epsilon_0]\times I$.
\end{remark}

\section{The auxiliary Hilbert spaces $\mathfrak{V}_0$ and $\mathfrak{L}_0$}\label{S.5}
\setcounter{equation}{0}
\setcounter{theorem}{0}

The procedure we use after \cite{GMS} for coming back from the operator $\widetilde{P}_\epsilon$ to the basic Hamiltonian $P_\epsilon$ supposes to consider the extension of the pseudodifferential operator $\widetilde{P}_\epsilon$ to the tempered distributions and a restriction of this one to some Hilbert spaces of distributions that we introduce and study in this section.

\begin{definition}\label{D.5.1}
Let us consider the following complex space (denoting by $\delta_\gamma:=\tau_\gamma\delta$ with $\delta$ the Dirac distribution of mass 1 supported in $\{0\}$ and $\gamma\in\Gamma$):
$$
\mathfrak{V}_0\ :=\ \left\{\,w\in\mathscr{S}^\prime(\X)\,\mid\,\exists f\in l^2(\Gamma)\ \text{such that}\ w=\underset{\gamma\in\Gamma}{\sum}f_\gamma\delta_{-\gamma}\,\right\},
$$
endowed with the quadratic norm:
$$
\|w\|_{\mathfrak{V}_0}\ :=\ \sqrt{\underset{\gamma\in\Gamma}{\sum}\left|f_\gamma\right|^2},\qquad\forall w\in\mathfrak{V}_0.
$$
\end{definition}

It is evident that $\mathfrak{V}_0$ is a Hilbert space and is canonically unitarily equivalent with $l^2(\Gamma)$.

The Hilbert space $\mathfrak{V}_0$ has a 'good comparaison property' with respect to the scale of magnetic Sobolev spaces introduced in \cite{IMP1}. In order to study this relation let us choose a family of vector potentials $\{A_\epsilon\}_{|\epsilon|\leq\epsilon_0}$ having components of class $C^\infty_{\text{\sf pol}}(\X)$ and defining the magnetic fields $\{B_\epsilon\}_{|\epsilon|\leq\epsilon_0}$ satisfying Hypothesis H.1.

\begin{lemma}\label{L.5.2}
For any $s>d$ and for any $\epsilon\in[-\epsilon_0,\epsilon_0]$ we have the algebraic and topologic inclusion $\mathfrak{V}_0\hookrightarrow\mathcal{H}^{-s}_{A_\epsilon}(\X)$, uniformly with respect to $\epsilon\in[-\epsilon_0,\epsilon_0]$.
\end{lemma}
\begin{proof}
We use the operator $Q_{s,\epsilon}$ from Remark \ref{R.A.25}. Let $u=\underset{\gamma\in\Gamma}{\sum}f_\gamma\delta_{-\gamma}\in\mathfrak{V}_0$. Then we have (in $\mathscr{S}^\prime(\X)$):
$$
g:=\ Q_{-s,\epsilon}u\ =\ \underset{\gamma\in\Gamma}{\sum}f_\gamma Q_{-s,\epsilon}\delta_{-\gamma}.
$$
A computation made in $\mathscr{S}^\prime(\X)$ shows that (for $s>d$) we have that $Q_{-s,\epsilon}\delta_{-\gamma}$ belongs in fact to $C(\X)$ (as Fourier transform of an integrable function) and moreover:
$$
\big(Q_{-s,\epsilon}\delta_{-\gamma}\big)(x)\ =\ (2\pi)^{-n}\int_\Xi e^{i<\eta,x-y>}\omega_{A_\epsilon}(x,y)\,q_{-s,\epsilon}\big(\frac{x+y}{2},\eta\big)\,\delta_{-\gamma}(y)dy\,d\eta\ =
$$
$$
=(2\pi)^{-n}\int_{\X^*}e^{i<\eta,x+\gamma>}\omega_{A_\epsilon}(x,-\gamma)\,q_{-s,\epsilon}\big(\frac{x-\gamma}{2},\eta\big)\,d\eta.
$$
From this last formula we may deduce that for any $N\in\mathbb{N}$ there exists a strictly positive constant $C_N>0$ such that for any $\epsilon\in[-\epsilon_0,\epsilon_0]$ and for any $x\in\X$ we have the estimation:
$$
\left|\big(Q_{-s,\epsilon}\delta_{-\gamma}\big)(x)\right|\ \leq\ C_N<x+\gamma>^{-N}.
$$
Choosing $N>d$ we notice that for any $x\in\X$:
$$
|g(x)|\ \leq\ C_N\underset{\gamma\in\Gamma}{\sum}\left|f_{\gamma}\right|<x+\gamma>^{-N}\ \leq\ C_N\left(\underset{\gamma\in\Gamma}{\sum}\left|f_{\gamma}\right|^2<x+\gamma>^{-N}\right)^{1/2}\left(\underset{\gamma\in\Gamma}{\sum}<x+\gamma>^{-N}\right)^{1/2}.
$$
We may conclude that $g\in L^2(X)$ and $\|g\|_{L^2(\X)}\leq C_N\|u\|^2_{\mathfrak{V}_0}$. Finally this is equivalent with the fact that $Q_{s,\epsilon}g\in\mathcal{H}^{-s}_{A_\epsilon}(\X)$ and there exists a strictly positive constant $C$ such that
$$
\|u\|_{\mathcal{H}^{-s}_{A_\epsilon}(\X)}\ \leq\ C\|u\|_{\mathfrak{V}_0},\qquad\forall u\in\mathfrak{V}_0,\ \forall\epsilon\in[-\epsilon_0,\epsilon_0].
$$
\end{proof}

We shall need a property characterizing the elements from $\mathfrak{V}_0$ (replacing the property proposed in \cite{GMS} that is not easy to generalise to our situation and moreover its proof given in \cite{GMS} and \cite{DS} has some gaps that make necessary some modifications in the arguments!). 
\begin{lemma}\label{L.5.3}
For any $s>d$ there exists a strictly positive constant $C_s>0$ such that the following inequality is true:
\begin{equation}\label{5.1}
\underset{\gamma\in\Gamma}{\sum}|u(\gamma)|^2\ \leq\ C_s\|u\|^2_{\mathcal{H}^s_{A_\epsilon}(\X)},\qquad\forall u\in\mathscr{S}(\X),\ \forall\epsilon\in[-\epsilon_0,\epsilon_0].
\end{equation}
\end{lemma}
\begin{proof}
For any fixed $u\in\mathscr{S}(\X)$ let us denote by $v:=Q_{s,\epsilon}u\in\mathscr{S}(\X)$. Then $u=Q_{-s,\epsilon}v$ and thus for any $N\in\mathbb{N}$ and for any $x\in\X$ we can write that:
$$
u(x)\ =\ \int_\Xi e^{i<\eta,x-y>}<x-y>^{-2N}\omega_{A_\epsilon}(x,y)\left[\big(\bb1-\Delta_\eta\big)^Nq_{-s,\epsilon}\big(\frac{x+y}{2},\eta\big)\right]v(y)dy\,\dbar\eta.
$$
This equality implies that we can find two strictly positive constants $C$ and $C^\prime$ such that for any $\epsilon\in[-\epsilon_0,\epsilon_0]$ and for any $x\in\X$ one has the estimation:
$$
|u(x)|\ \leq\ C\int_\X<x-y>^{-2N}|v(y)|dy,
$$
$$
|u(x)|^2\ \leq\ C^2\left(\int_\X<x-y>^{-2N}dy\right)\left(\int_\X<x-y>^{-2N}|v(y)|^2dy\right)\ \leq\ C^\prime\int_\X<x-y>^{-2N}|v(y)|^2dy.
$$
We choose now $N\in\mathbb{N}$ large enough and notice that
$$
\underset{\gamma\in\Gamma}{\sum}|u(\gamma)|^2\ \leq\ C^\prime\int_\X\left(\underset{\gamma\in\Gamma}{\sum}<\gamma-y>^{-2N}\right)|v(v)|^2dy\ \leq\ C^{\prime\prime}\|v\|^2_{L^2(\X)}\ \leq\ C_s\|u\|^2_{\mathcal{H}^s_{A_\epsilon}(\X)}.
$$
\end{proof}

For any $\gamma^*\in\Gamma^*$ we use the notation $\sigma_{\gamma^*}$ also for the operator of multiplication with the character $e^{i<\gamma^*,.>}$ on the space of tempered distributions (that it evidently leaves invariant).
\begin{proposition}\label{P.5.4}
We have the following charaterization of the vectors from $\mathfrak{V}_0$:
\begin{enumerate}
\item[a)] Given any vector $u\in\mathfrak{V}_0$ there exists a vector $u_0\in\mathcal{H}^\infty_{A_\epsilon}(\X)$ such that 
\begin{equation}\label{5.2}
u\ =\ \underset{\gamma^*\in\Gamma^*}{\sum}\sigma_{\gamma^*}u_0.
\end{equation}
Moreover the map $\mathfrak{V}_0\ni u\mapsto u_0\in\mathcal{H}^\infty_{A_\epsilon}(\X)$ is continuous uniformly with respect to $\epsilon\in[-\epsilon_0,\epsilon_0]$.
\item[b)] Given any $u_0\in\mathcal{H}^\infty_{A_\epsilon}(\X)$ the series $\underset{\gamma^*\in\Gamma^*}{\sum}\sigma_{\gamma^*}u_0$ converges in $\mathscr{S}^\prime(\X)$ and its sum denoted by $u$ belongs in fact to $\mathfrak{V}_0$. Moreover the map $\mathcal{H}^\infty_{A_\epsilon}(\X)\ni u_0\mapsto u\in\mathfrak{V}_0$ is continuous uniformly with respect to $\epsilon\in[-\epsilon_0,\epsilon_0]$.
\end{enumerate}
\end{proposition}
\begin{proof}
We shall use the notation $u_{\gamma^*}:=\sigma_{\gamma^*}u_0$, for any $\gamma^*\in\Gamma^*$ and for any tempered distribution $u_0\in\mathscr{S}^\prime(\X)$.

a) Lemma \ref{L.5.2} implies that for any $s>d$ and any $\epsilon\in[-\epsilon_0,\epsilon_0]$ we have that $\mathfrak{V}_0\subset\mathcal{H}^{-s}_{A_\epsilon}(\X)$ and there exists a strictly positive constant $C_s>0$, independent of $\epsilon$, such that
\begin{equation}\label{5.3}
\|u\|_{\mathcal{H}^{-s}_{A_\epsilon}(\X)}\ \leq\ C_s\|u\|_{\mathfrak{V}_0},\qquad\forall\epsilon\in[-\epsilon_0,\epsilon_0],\ \forall u\in\mathfrak{V}_0.
\end{equation}

Let us choose a real function $\chi\in C^\infty_0(\X^*)$ such that $\underset{\gamma^*\in\Gamma^*}{\sum}\tau_{\gamma^*}\chi=1$ on $\X^*$. For any distribution $u\in\mathfrak{V}_0$ we define
$$
u_0\ :=\ \mathfrak{Op}^{A_\epsilon}(\chi)u.
$$
Due to the fact that $\chi\in S^{-\infty}_1(\X)$ it follows by the properties of magnetic Sobolev spaces (see \cite{IMP1}) that $u_0\in\mathcal{H}^\infty_{A_\epsilon}(\X)$ and the map $\mathfrak{V}_0\ni u\mapsto u_0\in\mathcal{H}^\infty_{A_\epsilon}(\X)$ is continuous uniformly with respect to $\epsilon\in[-\epsilon_0,\epsilon_0]$.

We shall prove now that for any $s>d$ the series $\underset{\gamma^*\in\Gamma^*}{\sum}u_{\gamma^*}$ is convergent in $\mathcal{H}^{-s}_{A_\epsilon}(\X)$ to an element $v\in\mathcal{H}^{-s}_{A_\epsilon}(\X)$ and that there exists a constant $C>0$ such that
\begin{equation}\label{5.4}
\|v\|_{\mathcal{H}^{-s}_{A_\epsilon}(\X)}\ \leq\ C\|u_0\|_{\mathcal{H}^{s}_{A_\epsilon}(\X)}
\end{equation}
uniformly with respect to $\epsilon\in[-\epsilon_0,\epsilon_0]$.

We denote by 
$$
g_{\gamma^*}:=\ Q_{-s,\epsilon}\sigma_{\gamma^*}u_0\ =\ \sigma_{\gamma^*}\mathfrak{Op}^{A_\epsilon}\big((\id\otimes\tau_{-\gamma^*})q_{-s,\epsilon}\big)u_0
$$
where we have used \eqref{A.4} for the last equality. We notice that the family $\left\{<\gamma^*>^s(\id\otimes\tau_{-\gamma^*})q_{-s,\epsilon}\right\}_{|\epsilon|\leq\epsilon_0}$ is bounded  as subset of $S^s_1(\Xi)$ and thus there exists a constant $C>0$ such that
$$
\|g_{\gamma^*}\|_{L^2(\X)}\ \leq\ C<\gamma^*>^{-s}\|u_0\|_{\mathcal{H}^{s}_{A_\epsilon}(\X)},\qquad\forall\epsilon\in[-\epsilon_0,\epsilon_0],\ \forall\gamma^*\in\Gamma^*.
$$
We conclude that it exists an element $g\in L^2(\X)$ such that $\underset{\gamma^*\in\Gamma^*}{\sum}g_{\gamma^*}=g$ in $L^2(\X)$ and we have the estimation $\|g\|_{L^2(\X)}\leq C^\prime\|u_0\|_{\mathcal{H}^{s}_{A_\epsilon}(\X)}$ for any $\epsilon\in[-\epsilon_0,\epsilon_0]$. Due to the properties of the magnetic pseudodifferential calculus (see \cite{IMP1}) it follows that the series $\underset{\gamma^*\in\Gamma^*}{\sum}u_{\gamma^*}$ converges in $\mathcal{H}^{-s}_{A_\epsilon}(\X)$ to an element $v\in\mathcal{H}^{-s}_{A_\epsilon}(\X)$ and \eqref{5.4} is true.

We still have to show that $v=u$ as tempered distributions. Let us fix a test function $\varphi\in\mathscr{S}(\X)$ and compute:
$$
<v,\varphi>\ =\ \underset{\gamma^*\in\Gamma^*}{\sum}<u_{\gamma^*},\varphi>\ =\ \underset{\gamma^*\in\Gamma^*}{\sum}<\sigma_{\gamma^*}u_0,\varphi>\ =\ \underset{\gamma^*\in\Gamma^*}{\sum}\left\langle\sigma_{\gamma^*}\mathfrak{Op}^{A_\epsilon}(\chi)u,\varphi\right\rangle\ =
$$
$$
=\ \underset{\gamma^*\in\Gamma^*}{\sum}\left\langle\mathfrak{Op}^{A_\epsilon}(\tau_{-\gamma^*}\chi)\sigma_{\gamma^*}u,\varphi\right\rangle\ =\ \underset{\gamma^*\in\Gamma^*}{\sum}\left\langle u,\mathfrak{Op}^{A_\epsilon}(\tau_{-\gamma^*}\chi)\varphi\right\rangle
$$
where we have used the relation $\sigma_{\gamma^*}u=u$ verified by all the elements from $\mathfrak{V}_0$. Let us also notice that for any $s>d$ we have that
$$
\underset{\gamma^*\in\Gamma^*}{\sum}\tau_{-\gamma^*}\chi=1\ \text{in}\ S^s_1(\Xi)
$$
so that we can write that
$$
\varphi\ =\ \underset{\gamma^*\in\Gamma^*}{\sum}\mathfrak{Op}^{A_\epsilon}(\tau_{-\gamma^*}\chi)\varphi,\quad\text{in}\ \mathscr{S}(\X).
$$
We conclude that $<v,\varphi>=<u,\varphi>$ for any $\varphi\in\mathscr{S}(\X)$ and thus $v=u$.

b) During the proof of point (a) we have shown that for any $s>d$ there exists a constant $C_s>0$ such that for any $u_0\in\mathcal{H}^{\infty}_{A_\epsilon}(\X)$ and for any $\epsilon\in[-\epsilon_0,\epsilon_0]$ the series $\underset{\gamma^* \in\Gamma^*}{\sum}u_{\gamma^*}$ converges in $\mathcal{H}^{-s}_{A_\epsilon}(\X)$ to an element $u\in\mathcal{H}^{-s}_{A_\epsilon}(\X)$ and the following estimation is true:
\begin{equation}\label{5.5}
\|u\|_{\mathcal{H}^{-s}_{A_\epsilon}(\X)}\ \leq\ C_s\|u_0\|_{\mathcal{H}^{s}_{A_\epsilon}(\X)},\qquad\forall u_0\in\mathcal{H}^{\infty}_{A_\epsilon}(\X),\ \forall\epsilon\in[-\epsilon_0,\epsilon_0].
\end{equation}

Let us recall the Poisson formula:
\begin{equation}\label{5.6}
\underset{\gamma^*\in\Gamma^*}{\sum}\sigma_{\gamma^*}\ =\ \frac{(2\pi)^d}{|E^*|}\ \underset{\gamma\in\Gamma}{\sum}\delta_{-\gamma},\qquad\text{in }\mathscr{S}^\prime(\X).
\end{equation}

Let us first suppose that $u_0\in\mathscr{S}(\X)$. Multiplying in the equality \eqref{5.6} with $u_0$ we obtain the following equality:
\begin{equation}\label{5.7}
u\ =\ \frac{(2\pi)^d}{|E^*|}\ \underset{\gamma\in\Gamma}{\sum}u_0(-\gamma)\delta_{-\gamma},\qquad\text{in }\mathscr{S}^\prime(\X).
\end{equation}
Lemma \ref{L.5.3} implies that $u\in\mathfrak{V}_0$ and for any $s>d$ we have the estimation:
\begin{equation}\label{5.8}
\|u\|_{\mathfrak{V}_0}\ \leq\ C_s\|u_0\|_{\mathcal{H}^{s}_{A_\epsilon}(\X)}, \qquad\forall u_0\in\mathscr{S}(\X),\ \forall\epsilon\in[-\epsilon_0,\epsilon_0].
\end{equation}

We come now to the general case $u_0\in\mathcal{H}^{\infty}_{A_\epsilon}(\X)$. Let us fix some $\epsilon\in[-\epsilon_0,\epsilon_0]$ and some $s>d$. Using the fact that $\mathscr{S}(\X)$ is dense in $\mathcal{H}^{s}_{A_\epsilon}(\X)$ we can choose a sequence $\{u_0^k\}_{k\in\mathbb{N}^*}\subset\mathscr{S}(\X)$ such that $u_0=\underset{k\nearrow\infty}{\lim}u_0^k$ in $\mathcal{H}^{s}_{A_\epsilon}(\X)$. For each element $u_0^k$ we can associate, as we proved above, an element $u^k\in\mathfrak{V}_0$ such that the following inequalities hold:
$$
\|u^k\|_{\mathfrak{V}_0}\ \leq\ C_s\|u_0^k\|_{\mathcal{H}^{s}_{A_\epsilon}(\X)},\quad\forall k\in\mathbb{N}^*,
$$
$$
\|u^k-u^l\|_{\mathfrak{V}_0}\ \leq\ C_s\|u_0^k-u_0^l\|_{\mathcal{H}^{s}_{A_\epsilon}(\X)},\quad\forall(k,l)\in[\mathbb{N}^*]^2.
$$
It follows that there exists $v\in\mathfrak{V}_0$ such that $v=\underset{k\nearrow\infty}{\lim}u^k$ in $\mathfrak{V}_0$ and the following estimation is valid:
\begin{equation}\label{5.9}
\|v\|_{\mathfrak{V}_0}\ \leq\ C_s\|u_0\|_{\mathcal{H}^{s}_{A_\epsilon}(\X)}.
\end{equation}

We still have to prove that the element $v\in\mathfrak{V}_0$ obtained above is exactly the limit $u=\underset{\gamma^*\in\Gamma^*}{\sum}u_{\gamma^*}$. But we know that $u^k=\underset{\gamma^*\in\Gamma^*}{\sum}\sigma_{\gamma^*}u_0^k$ so that from \eqref{5.5} we deduce that we have the estimations:
$$
\|u^k-u\|_{\mathcal{H}^{-s}_{A_\epsilon}(\X)}\ \leq\ C_s\|u_0^k-u_0\|_{\mathcal{H}^{s}_{A_\epsilon}(\X)}\ \underset{k\nearrow\infty}{\rightarrow}\ 0.
$$
In conclusion $u^k\underset{k\nearrow\infty}{\rightarrow}u$ in $\mathcal{H}^{-s}_{A_\epsilon}(\X)$ and $u^k\underset{k\nearrow\infty}{\rightarrow}v$ in $\mathfrak{V}_0$. But Lemma \ref{L.5.2} implies that $\mathfrak{V}_0$ is continuously embedded in $\mathcal{H}^{-s}_{A_\epsilon}(\X)$ and we conclude that $v=u$.
\end{proof}

\begin{lemma}\label{L.5.6.a}
Let us consider the map $\boldsymbol{\psi}$ defined in \eqref{1.1}. For any vector $v\in L^2(\X)$ the series
$$
w_v\ :=\ \underset{\gamma\in\Gamma}{\sum}\boldsymbol{\psi}^*\big(v\otimes\delta_{-\gamma}\big)
$$
converges in $\mathscr{S}^\prime(\X^2)$ and satisfies the identity $\big(\id\otimes\tau_\alpha\big)w_v=w_v$, $\forall\alpha\in\Gamma$.
\end{lemma}
\begin{proof}
Let us denote by $w_\gamma$ the general term of the series defining $w_v$; then for any $\varphi\in\mathscr{S}(\X^2)$ we have that 
$$
\langle w_\gamma,\varphi\rangle\ =\ \left\langle\boldsymbol{\psi}^*\big(v\otimes\delta_{-\gamma}\big),\varphi\right\rangle\ =\ \left\langle v\otimes\delta_{-\gamma},\boldsymbol{\psi}^*(\varphi)\right\rangle\ =\ \int_\X v(x)\varphi(x,x+\gamma)\,dx.
$$
 But using the definition of test functions it follows that for any $N\in\mathbb{N}$ there exists a defining semi-norm $\nu_N(\varphi)$ on $\mathscr{S}(\X^2)$ such that 
$$
|\varphi(x,x+\gamma)|\ \leq\ \nu_N(\varphi)<x>^{-N}<\gamma>^{-N},\qquad\forall x\in\X,\ \forall\gamma\in\Gamma.
$$
We deduce that for any $N\in\mathbb{N}$ there exists a defining semi-norm $\mu_N(\varphi)$ on $\mathscr{S}(\X^2)$ such that:
$$
|\langle w_\gamma,\varphi\rangle|\ \leq\ \mu_N(\varphi)\|v\|_{L^2(\X)}<\gamma>^{-N},\qquad\forall\varphi\in\mathscr{S}(\X^2).
$$
It follows that the series $\underset{\gamma\in\Gamma}{\sum}w_\gamma$ converges in $\mathscr{S}^\prime(\X^2)$ and its sum $w\in\mathscr{S}^\prime(\X^2)$ satisfies the inequality
\begin{equation}\label{5.10}
|\langle w,\varphi\rangle|\ \leq\ \rho(\varphi)\|v\|_{L^2(\X)},\qquad\forall\varphi\in\mathscr{S}(\X^2),
\end{equation}
for some defining semi-norm $\rho(\varphi)$ on $\mathscr{S}(\X^2)$.

Finally let us notice that for any $\alpha\in\Gamma$ we can write that
$$
\big(\id\otimes\tau_\alpha\big)w_\gamma\ =\ \boldsymbol{\psi}^*\big(v\otimes\delta_{-(\gamma+\alpha)}\big)\ =\ w_{\gamma+\alpha}
$$
and thus $\big(\id\otimes\tau_\alpha\big)w_v=w_v$.
\end{proof}
\begin{definition}\label{D.5.5}
Let us consider the map $\boldsymbol{\psi}$ defined in \eqref{1.1}. We define the following complex space:
$$
\mathfrak{L}_0\ :=\ \left\{\,w\in\mathscr{S}^\prime(\X^2)\,\mid\,\exists v\in L^2(\X)\ \text{such that}\ w=\underset{\gamma\in\Gamma}{\sum}\boldsymbol{\psi}^*\big(v\otimes\delta_{-\gamma}\big)\,\right\}
$$
endowed with the quadratic norm $\|w\|_{\mathfrak{L}_0}:=\|v\|_{L^2(\X)}$.
\end{definition}

\begin{lemma}\label{L.5.6.b}
The complex space $\mathfrak{L}_0$ is a Hilbert space and is embedded continuously into $\mathscr{S}^\prime(\X^2)$.
\end{lemma}
\begin{proof}
The space $\mathfrak{L}_0$ is evidently a Hilbert space canonically unitarily equivalent with $L^2(\X
)$  and the continuity of the embedding into $\mathscr{S}^\prime(\X^2)$ follows easily from \eqref{5.10}.
\end{proof}
\begin{lemma}\label{L.5.7}
For any $s>d$ and any $\epsilon\in[-\epsilon_0,\epsilon_0]$ we have a continuous embedding $\mathfrak{L}_0\hookrightarrow\mathcal{H}^{-s}_{A_\epsilon}(\X)\otimes L^2(\mathbb{T})$ uniformly with respect to $\epsilon\in[-\epsilon_0,\epsilon_0]$.
\end{lemma}
\begin{proof}
Let us fix some $s>d$ and some vector $v\in L^2(\X)$ and define 
$$
u\ :=\ \underset{\gamma\in\Gamma}{\sum}\boldsymbol{\psi}^*\big(v\otimes\delta_{-\gamma}\big)\in\mathfrak{L}_0,
$$
$$
Q^\prime_{-s,\epsilon}\ :=\ Q_{-s,\epsilon}\otimes\id,\qquad g_\gamma\ :=\ \big(Q^\prime_{-s,\epsilon}\circ\boldsymbol{\psi}^*\big)(v\otimes\delta_{-\gamma}).
$$
A straightforward computation in $\mathscr{S}^\prime(\X^2)$ shows that we have the equality:
$$
g_\gamma(x,y)\ =\ v(y-\gamma)\int_\X e^{i<\eta,x-y+\gamma>}\omega_{A_\epsilon}(x,y-\gamma)\,q_{-s,\epsilon}\big(\frac{x+y-\gamma}{2},\eta\big)\,\dbar\eta.
$$
By estimating the integral in the right hand side above we obtain that for any sufficiently large $N\in\mathbb{N}$ there exists $C_N>0$ such that
$$
\left|g_\gamma(x,y)\right|\ \leq\ C_N|v(y-\gamma)|<x-y+\gamma>^{-N},\qquad\forall(x,y)\in\X^2,\ \forall\gamma\in\Gamma.
$$

We conclude that for some suitable constants $C^\prime_N,C^{\prime\prime}_N,\ldots$ we get
$$
\left[\underset{\gamma\in\Gamma}{\sum}|g_\gamma(x,y)|\right]^2\ \leq\ C^2_N\left(\underset{\gamma\in\Gamma}{\sum}|v(y-\gamma)|^2\,<x-y+\gamma>^{-N}\right)\left(\underset{\gamma\in\Gamma}{\sum}<x-y+\gamma>^{-N}\right)\ \leq
$$
$$
\leq\ C^\prime_N\underset{\gamma\in\Gamma}{\sum}|v(y-\gamma)|^2\,<x-y+\gamma>^{-N},
$$
and the integral of the last expression over $\X\times E$ is bounded by $C^{\prime\prime}_N\|v\|^2_{L^2(\X)}$.

If we denote by $g:=\underset{\gamma\in\Gamma}{\sum}g_\gamma=Q^\prime_{-s,\epsilon}u$, we notice that $(\id\otimes\tau_{\alpha})g=g$ for any $\alpha\in\Gamma$ (because by Lemma \ref{L.5.6.a} the vector $u$ has this property). Thus
$$
g\in L^2(\X)\otimes L^2(\mathbb{T}),\qquad\|g\|_{L^2(\X)\otimes L^2(\mathbb{T})}\ \leq\ \sqrt{C^{\prime\prime}}\|v\|_{L^2(\X)}.
$$

It follows that 
$$
u\ =\ Q^\prime_{s,\epsilon}g\,\in\,\mathcal{H}^{-s}_{A_\epsilon}(\X)\otimes L^2(\mathbb{T}),\quad\|u\|_{\mathcal{H}^{-s}_{A_\epsilon}(\X)\otimes L^2(\mathbb{T})}\ \leq\ C^{\prime\prime\prime}_N\|v\|_{L^2(\X)},\quad\forall\epsilon\in[-\epsilon_0,\epsilon_0],\ \forall v\in L^2(\X).
$$
\end{proof}

We shall obtain a characterization of the space $\mathfrak{L}_0$ that is similar to our Proposition \ref{P.5.4}. In order to do that we need some technical results contained in the next Lemma. We shall use the notation:
$$
\mathcal{H}^\infty_{A_\epsilon}(\X)\otimes L^2(\mathbb{T})\ :=\ \underset{s\in\mathbb{R}}{\cap}\left(\mathcal{H}^s_{A_\epsilon}(\X)\otimes L^2(\mathbb{T})\right)
$$
with the natural projective limit topology.

\begin{lemma}\label{L.5.8}
Suppose given some $u_0\in\mathcal{H}^\infty_{A_\epsilon}(\X)\otimes L^2(\mathbb{T})$ and for any $\gamma^*\in\Gamma^*$ let us denote by $u_{\gamma^*}:=\Upsilon_{\gamma^*}u_0$. For any $s>d$ there exists $C_s>0$ such that the series $\underset{\gamma^*\in\Gamma^*}{\sum}u_{\gamma^*}$ converges in $\mathcal{H}^{-s}_{A_\epsilon}(\X)\otimes L^2(\mathbb{T})$ and the sum denoted by $v\in\mathcal{H}^{-s}_{A_\epsilon}(\X)\otimes L^2(\mathbb{T})$ satisfies the estimation:
\begin{equation}\label{5.11}
\|v\|_{\mathcal{H}^{-s}_{A_\epsilon}(\X)\otimes L^2(\mathbb{T})}\ \leq\ C_s\|u_0\|_{\mathcal{H}^{s}_{A_\epsilon}(\X)\otimes L^2(\mathbb{T})},\qquad\forall u_0\in\mathcal{H}^{\infty}_{A_\epsilon}(\X)\otimes L^2(\mathbb{T}),\ \forall\epsilon\in[-\epsilon_0,\epsilon_0].
\end{equation}
\end{lemma}
\begin{proof}
From \eqref{A.4} it follows that on $\mathscr{S}(\X)$ we have the equality
$$
Q_{-s,\epsilon}\sigma_{\gamma^*}\ =\ \sigma_{\gamma^*}\mathfrak{Op}^{A_\epsilon}\big((\id\otimes\tau_{-\gamma^*})\mathfrak{q}_{-s,\epsilon}\big)
$$
so that finally
\begin{equation}\label{5.12}
\big(Q_{-s,\epsilon}\otimes\id\big)u_{\gamma^*}\ =\ \Upsilon_{\gamma^*}\big[\mathfrak{Op}^{A_\epsilon}\big((\id\otimes\tau_{-\gamma^*})\mathfrak{q}_{-s,\epsilon}\big)\otimes\id\big]u_0.
\end{equation}
Taking into account that the family $\left\{<\gamma^*>^s(\id\otimes\tau_{-\gamma^*})\mathfrak{q}_{-s,\epsilon}\right\}_{(\epsilon,\gamma^*)\in[-\epsilon_0,\epsilon_0]\times\Gamma^*}$ is a bounded subset of $S^s(\Xi)$, it follows the existence of a constant $C>0$ such that for any $\epsilon\in[-\epsilon_0,\epsilon_0]$ one has the estimation:
\begin{equation}\label{5.13}
\left\|\big(Q_{-s,\epsilon}\otimes\id\big)u_{\gamma^*}\right\|_{L^2(\X)\otimes L^2(\mathbb{T})}\ \leq\ C<\gamma^*>^{-s}\|u_0\|_{\mathcal{H}^{s}_{A_\epsilon}(\X)\otimes L^2(\mathbb{T})},\qquad\forall u_0\in\mathcal{H}^{\infty}_{A_\epsilon}(\X)\otimes L^2(\mathbb{T}).
\end{equation}
it follows that the series $\underset{\gamma^*\in\Gamma^*}{\sum}\big(Q_{-s,\epsilon}\otimes\id\big)u_{\gamma^*}$ converges in $L^2(\X)\otimes L^2(\mathbb{T})$ uniformly for $\epsilon\in[-\epsilon_0,\epsilon_0]$. The stated inequality follows now by summing up the estimation \eqref{5.13} over all $\Gamma^*$.
\end{proof}

\begin{proposition}\label{P.5.9}
For any $u\in\mathfrak{L}_0$ there exists a vector $u_0\in\mathcal{H}^{\infty}_{A_\epsilon}(\X)\otimes L^2(\mathbb{T})$ such that
\begin{equation}\label{5.14}
u\ =\ \underset{\gamma^*\in\Gamma^*}{\sum}\Upsilon_{\gamma^*}u_0,\ \text{in}\ \mathscr{S}^\prime(\X^2).
\end{equation}
Moreover, the application $\mathfrak{L}_0\ni u\mapsto u_0\in\mathcal{H}^{\infty}_{A_\epsilon}(\X)\otimes L^2(\mathbb{T})$ is continuous uniformly with respect to $\epsilon\in[-\epsilon_0,\epsilon_0]$.
\end{proposition}
\begin{proof}
We recall the notation $u_{\gamma^*}:=\Upsilon{_\gamma^*}u_0$ and, as in the proof of point (a) of Proposition \ref{P.5.4} we fix some real function $\chi\in C^\infty_0(\X)$ satisfying the following identity on $\X$:
$$
\underset{\gamma^*\in\Gamma^*}{\sum}\tau_{\gamma^*}\chi\ =\ 1.
$$
For any $u\in\mathfrak{L}_0$ let us denote by $u_0:=\big(\mathfrak{Op}^{A_\epsilon}(\chi)\otimes\id\big)u$. We notice that $\chi\in S^{-\infty}_1(\Xi)$ and using Lemma \ref{L.5.7} we deduce that $u_0\in\mathcal{H}^{\infty}_{A_\epsilon}(\X)\otimes L^2(\mathbb{T})$ and that the continuity property in the end of the Proposition is clearly true. We still have to verify the equality \eqref{5.14}. 

1. First let us notice that following Lemma \ref{L.5.8}, the series $\underset{\gamma^*\in\Gamma^*}{\sum}u_{\gamma^*}$ converges in $\mathscr{S}^\prime(\X^2)$.

2. An argument similar to that in the proof of Proposition \ref{P.5.4} a) proves that on $\mathscr{S}(\X^2)$ we have the equality
\begin{equation}\label{5.15}
\underset{\gamma^*\in\Gamma^*}{\sum}\mathfrak{Op}^{A_\epsilon}\big(\tau_{-\gamma^*}\chi\big)\otimes\id\ =\ \id.
\end{equation}

3. From \eqref{A.4} we have that $\mathfrak{Op}^{A_\epsilon}\big(\tau_{-\gamma^*}\chi\big)=\sigma_{-\gamma^*}\mathfrak{Op}^{A_\epsilon}(\chi)\sigma_{\gamma^*}$.

4. For any $u\in\mathfrak{L}_0$ there exists $v\in L^2(\X)$ such that $u=\underset{\gamma\in\Gamma}{\sum}\boldsymbol{\psi}^*\big(v\otimes\delta_{-\gamma}\big)$; thus for any $\gamma^*\in\Gamma^*$ we have that:
$$
\Upsilon_{\gamma^*}u\ =\ \underset{\gamma\in\Gamma}{\sum}\Upsilon_{\gamma^*}\boldsymbol{\psi}^*\big(v\otimes\delta_{-\gamma}\big)\ =\ \underset{\gamma\in\Gamma}{\sum}\boldsymbol{\psi}^*\big((\id\otimes\sigma_{\gamma^*})(v\otimes\delta_{-\gamma})\big)\ =\ \underset{\gamma\in\Gamma}{\sum}\boldsymbol{\psi}^*\big(v\otimes\delta_{-\gamma}\big)\ =\ u.
$$

Using the last two remarks above we deduce that for any $u\in\mathfrak{L}_0$ we have the equalities:
$$
\big[\mathfrak{Op}^{A_\epsilon}\big(\tau_{-\gamma^*}\chi\big)\otimes\id\big]u\ =\ \big[\big(\sigma_{-\gamma^*}\mathfrak{Op}^{A_\epsilon}(\chi)\sigma_{\gamma^*}\big)\otimes\id\big]\big(\sigma_{-\gamma^*}\otimes\sigma_{\gamma^*}\big)u\ =
$$
$$
=\ \Upsilon_{-\gamma^*}\big(\mathfrak{Op}^{A_\epsilon}(\chi)\otimes\id\big)u\ =\ \Upsilon_{-\gamma^*}u_0\ =\ u_{-\gamma^*}.
$$
We apply now equality \eqref{5.15} to the vector $u\in\mathfrak{L}_0$ in order to obtain that:
$$
u\ =\ \underset{\gamma^*\in\Gamma^*}{\sum}u_{-\gamma^*}\ =\ \underset{\gamma^*\in\Gamma^*}{\sum}u_{\gamma^*},
$$
as tempered distributions.
\end{proof}

In order to prove the reciprocal statement of Proposition \ref{P.5.9} we need a technical Lemma similar to Lemma\ref{L.5.3}.
\begin{lemma}\label{L.5.10}
For any $s>d$ there exists $C_s>0$ such that the following estimation holds:
\begin{equation}\label{5.16}
\sqrt{\int_\X|u(x,x)|^2\,dx}\ \leq\ C_s\|u\|_{\mathcal{H}^s_{A_\epsilon}(\X)\otimes L^2(\mathbb{T})},\qquad\forall u\in\mathscr{S}(\X\times\mathbb{T}),\ \forall\epsilon\in[-\epsilon_0,\epsilon_0].
\end{equation}
\end{lemma}
\begin{proof}
Let us fix some $u\in\mathscr{S}(\X\times\mathbb{T})$, $\epsilon\in[-\epsilon_0,\epsilon_0]$ and let us define $v:=\big(Q_{s,\epsilon}\otimes\id\big)u\in\mathscr{S}(\X\times\mathbb{T})$. It follows that $u=\big(Q_{-s,\epsilon}\otimes\id\big)v$ and we deduce that for any $N\in\mathbb{N}$ (that we shall choose sufficiently large) we have the identity:
$$
u(x,y)\ =\ \int_\Xi<x-z>^{-2N}e^{i<\zeta,x-z>}\omega_{A_\epsilon}(x,z)\Big[\big((\id-\Delta_{\zeta})^N\mathfrak{q}_{-s,\epsilon}\big)\big(\frac{x+z}{2},\zeta\big)\Big]v(z,y)\,dz\,\dbar\zeta,\qquad\forall(x,y)\in\X^2.
$$
We deduce that there exist the strictly positive constants $C_N,C^\prime_N,\ldots$ such that the following estimations hold:
$$
|u(x,y)|\ \leq\ C_N\int_\X<x-z>^{-2N}|v(z,y)|\,dz,\qquad|u(x,y)|^2\ \leq\ C_N^\prime\int_\X<x-z>^{-2N}|v(z,y)|^2dz.
$$
In conclusion we have that:
$$
\int_\X|u(x,x)|^2dx\ =\ \underset{\gamma\in\Gamma}{\sum}\int_{\tau_{-\gamma}E}|u(x,x)|^2dx\ =\ \underset{\gamma\in\Gamma}{\sum}\int_{E}|u(x+\gamma,x+\gamma)|^2dx\ =\ \underset{\gamma\in\Gamma}{\sum}\int_{E}|u(x+\gamma,x)|^2dx\ \leq
$$
$$
\leq\ C_N^\prime\underset{\gamma\in\Gamma}{\sum}\int_E\int_\X<x-z+\gamma>^{-2N}|v(z,x)|^2dz\,dx\ \leq\ C_N^{\prime\prime}\underset{\gamma\in\Gamma}{\sum}\int_E\int_\X<z-\gamma>^{-2N}|v(z,x)|^2dz\,dx\ \leq
$$
$$
\leq\ C_N^{\prime\prime\prime}\int_E\int_\X|v(z,x)|^2dz\,dx\ =\ C_N^{\prime\prime\prime}\|v\|^2_{L^2(\X\times\mathbb{T})}\ \leq\ C_s^2\|u\|_{\mathcal{H}^s_{A_\epsilon}(\X)\otimes L^2(\mathbb{T})}^2.
$$
\end{proof}

We come now to the reciprocal statement of Proposition \ref{P.5.9}.
\begin{proposition}\label{P.5.11}
Suppose given $u_0\in\mathcal{H}^\infty_{A_\epsilon}(\X)\otimes L^2(\mathbb{T})$ and for any $\gamma^*\in\Gamma^*$ let us consider $u_{\gamma^*}:=\Upsilon_{\gamma^*}u_0$. Then the series $\underset{\gamma^*\in\Gamma^*}{\sum}u_{\gamma^*}$ converges in $\mathscr{S}(\X^2)$ to an element $u\in\mathfrak{L}_0$. Moreover, the application $\mathcal{H}^\infty_{A_\epsilon}(\X)\otimes L^2(\mathbb{T})\ni u_0\mapsto u\in\mathfrak{L}_0$ is continuous uniformly with respect to $\epsilon\in[-\epsilon_0,\epsilon_0]$.
\end{proposition}
\begin{proof}
For any $s>d$ and any $u_0\in\mathcal{H}^\infty_{A_\epsilon}(\X)\otimes L^2(\mathbb{T})$, Lemma \ref{L.5.8} implies that the series $\underset{\gamma^*\in\Gamma^*}{\sum}u_{\gamma^*}$ converges in $\mathcal{H}^{-s}_{A_\epsilon}(\X)\otimes L^2(\mathbb{T})$ to an element $u\in\mathcal{H}^{-s}_{A_\epsilon}(\X)\otimes L^2(\mathbb{T})$ and there exists $C_s>0$ such that the following estimation holds:
\begin{equation}\label{5.17}
\|u\|_{\mathcal{H}^{-s}_{A_\epsilon}(\X)\otimes L^2(\mathbb{T})}\ \leq\ C_s\|u_0\|_{\mathcal{H}^{s}_{A_\epsilon}(\X)\otimes L^2(\mathbb{T})},\qquad\forall u_0\in\mathcal{H}^{s}_{A_\epsilon}(\X)\otimes L^2(\mathbb{T}),\ \forall\epsilon\in[-\epsilon_0,\epsilon_0].
\end{equation}

We still have to prove that $u\in\mathfrak{L}_0$ and that the continuity property stated above is true. As in the proof of Proposition \ref{P.5.4} b) we make use of the Poisson formula \eqref{5.6}. Once we notice that $\boldsymbol{\psi}^*\big(\id\otimes\sigma_{\gamma^*}\big)=\Upsilon_{\gamma^*}$, we conclude that for any $u_0\in\mathscr{S}(\X^2)$ one has the identity:
\begin{equation}\label{5.18}
\underset{\gamma^*}{\sum}u_{\gamma^*}\ =\ \frac{(2\pi)^d}{|E^*|}\left[\underset{\gamma\in\Gamma}{\sum}\boldsymbol{\psi}^*\big(\id\otimes\delta_{-\gamma}\big)\right]u_0.
\end{equation}
But we notice that $\underset{\gamma\in\Gamma}{\sum}\boldsymbol{\psi}^*\big(\id\otimes\delta_{-\gamma}\big)$ belongs to $\mathscr{S}^\prime(\X\times\mathbb{T})$ and we deduce that the identity \eqref{5.18} also holds for $u_0\in\mathscr{S}(\X\times\mathbb{T})$. In this case, $\boldsymbol{\psi}^*\big(\id\otimes\delta_{-\gamma}\big)\cdot u_0$ also belongs to $\mathscr{S}^\prime(\X^2)$ and for any $\varphi\in\mathscr{S}(\X^2)$ we can write that:
$$
\left\langle\boldsymbol{\psi}^*\big(\id\otimes\delta_{-\gamma}\big)\cdot u_0\,,\,\varphi\right\rangle\ =\ \left\langle\boldsymbol{\psi}^*\big(\id\otimes\delta_{-\gamma}\big)\,,\,u_0\varphi\right\rangle\ =\ \left\langle\id\otimes\delta_{-\gamma}\,,\,\boldsymbol{\psi}^*\big(u_0\varphi\big)\right\rangle\ =\ \int_\X\varphi(x,x+\gamma)u_0(x,x)\,dx.
$$
Let us denote by $v_0(x):=u_0(x,x)$ so that we obtain a test function $v_0\in\mathscr{S}(\X)$ and an equality:
\begin{equation}\label{5.19}
\boldsymbol{\psi}^*\big(\id\otimes\delta_{-\gamma}\big)\cdot u_0\ =\ \boldsymbol{\psi}^*\big(v_0\otimes\delta_{-\gamma}\big).
\end{equation}
Let us further denote by $v:=\left((2\pi)^d/|E^*|\right)v_0\in\mathscr{S}(\X)\subset L^2(\X)$. If we insert \eqref{5.19} into \eqref{5.18} we obtain:
\begin{equation}\label{5.20}
u\ :=\ \underset{\gamma^*\in\Gamma^*}{\sum}u_{\gamma^*}\ =\ \underset{\gamma\in\Gamma}{\sum}\boldsymbol{\psi}^*\big(v\otimes\delta_{-\gamma}\big)\,\in\,\mathfrak{L}_0.
\end{equation}

Let us verify now the continuity property. Suppose given $u_0\in\mathcal{H}^\infty_{A_\epsilon}(\X)\otimes L^2(\mathbb{T})$ and suppose given a sequence $\{u_{0,k}\}_{k\in\mathbb{N}^*}\subset\mathscr{S}(\X\times\mathbb{T})$ and some $s>d$ such that $u_0=\underset{k\nearrow\infty}{\lim}u_{0,k}$ in $\mathcal{H}^{s}_{A_\epsilon}(\X)\otimes L^2(\mathbb{T})$. Let us also introduce the notations:
$$
v_k(x):=\frac{(2\pi)^d}{|E^*|}u_{0,k}(x,x),\ \forall x\in\X;\qquad u_k:=\underset{\gamma\in\Gamma}{\sum}\boldsymbol{\psi}^*\big(v_k\otimes\delta_{-\gamma}\big)\,\in\,\mathfrak{L}_0.
$$

From Lemma \ref{L.5.10} we deduce that there exists a strictly positive constant $C_s$ such that for any $\epsilon\in[-\epsilon_0,\epsilon_0]$ and for any pair of indices $(k,l)\in[\mathbb{N}^*]^2$ the following estimations hold:
\begin{equation}\label{5.21}
\|u_k\,-\,u_l\|_{\mathfrak{L}_0}\ :=\ \|v_k\,-\,v_l\|_{L^2(\X)}\ \leq\ C_s\|u_{0,k}\,-\,u_{0,l}\|_{\mathcal{H}^{s}_{A_\epsilon}(\X)\otimes L^2(\mathbb{T})},
\end{equation}
\begin{equation}\label{5.22}
\|u_k\|_{\mathfrak{L}_0}\ \leq\ C_s\|u_{0,k}\|_{\mathcal{H}^{s}_{A_\epsilon}(\X)\otimes L^2(\mathbb{T})}.
\end{equation}
From \eqref{5.21} we deduce that there exists $v\in L^2(\X)$ limit of the sequence $\{v_k\}_{k\in\mathbb{N}^*}$ in $L^2(\X)$ and moreover, using also \eqref{5.22}, that it satisfies the estimation:
\begin{equation}\label{5.23}
\|v\|_{L^2(\X)}\ \leq\ C_s\|u_0\|_{\mathcal{H}^{s}_{A_\epsilon}(\X)\otimes L^2(\mathbb{T})},\qquad\forall\epsilon\in[-\epsilon_0,\epsilon_0].
\end{equation}
Let us denote by $\widetilde{u}:=\underset{\gamma\in\Gamma}{\sum}\boldsymbol{\psi}^*\big(v\otimes\delta_{-\gamma}\big)\,\in\,\mathfrak{L}_0$. From \eqref{5.23} we deduce that
\begin{equation}\label{5.24}
\|\widetilde{u}\|_{\mathfrak{L}_0}\ \leq\ C_s\|u_0\|_{\mathcal{H}^{s}_{A_\epsilon}(\X)\otimes L^2(\mathbb{T})},\qquad\forall\epsilon\in[-\epsilon_0,\epsilon_0].
\end{equation}

In order to end the proof we have to show that $\widetilde{u}=u:=\underset{\gamma^*\in\Gamma^*}{\sum}u_{\gamma^*}$ in $\mathcal{H}^{-s}_{A_\epsilon}(\X)\otimes L^2(\mathbb{T})$. From \eqref{5.20} we know that $u_k:=\underset{\gamma^*\in\Gamma^*}{\sum}\Upsilon_{\gamma^*}u_{0,k}$. If we use now the inequality \eqref{5.17} with $u_0$ replaced by $u_{0,k}\,-\,u_0$, we obtain the estimation:
$$
\|u_k\,-\,u\|_{\mathcal{H}^{-s}_{A_\epsilon}(\X)\otimes L^2(\mathbb{T})}\ \leq\ C_s\|u_{0,k}\,-\,u_0\|_{\mathcal{H}^{s}_{A_\epsilon}(\X)\otimes L^2(\mathbb{T})}.
$$
From this we deduce that $u=\underset{k\nearrow\infty}{\lim}u_k$ in $\mathcal{H}^{-s}_{A_\epsilon}(\X)\otimes L^2(\mathbb{T})$. But from \eqref{5.21} we deduce that $\widetilde{u}=\underset{k\nearrow\infty}{\lim}u_k$ in $\mathfrak{L}_0$ and thus, due to Lemma \ref{L.5.7} also in $\mathcal{H}^{-s}_{A_\epsilon}(\X)\otimes L^2(\mathbb{T})$. In conclusion $\widetilde{u}=u$ and the proof is finished.
\end{proof}

Proceeding as in Lemma \ref{L.5.6.a} we can show that the following definition is meaningful, the series appearing in the definition of the space $\mathfrak{L}_s(\epsilon)$ being convergent as tempered distribution, and the space with the associated norm being a Hilbert space canonically isomorphic with $\mathcal{H}^s_{A_\epsilon}(\X)$ and continuously embedded into $\mathscr{S}^\prime(\X^2)$.
\begin{definition}\label{D.5.12}
For any $s\in\mathbb{R}$ and any $\epsilon\in[-\epsilon_0,\epsilon_0]$, we define the following subspace of tempered distributions:
$$
\mathfrak{L}_s(\epsilon)\ :=\ \left\{\,w\in\mathscr{S}^\prime(\X^2)\,\mid\,\exists v\in\mathcal{H}^s_{A_\epsilon}(\X),\ w\equiv w_v=\underset{\gamma\in\Gamma}{\sum}\boldsymbol{\psi}^*\big(v\otimes\delta_{-\gamma}\big)\,\right\},
$$
endowed with the quadratic norm: 
\begin{equation}\label{5.25}
\|w_v\|_{\mathfrak{L}_s(\epsilon)}\ :=\ \|v\|_{\mathcal{H}^s_{A_\epsilon}(\X)}.
\end{equation}
\end{definition}
\begin{remark}\label{R.5.12}
For any $w\in\mathfrak{L}_s(\epsilon)$ we have the identity: $(\id\otimes\tau_\alpha)w=w, \forall\alpha\in\Gamma$.
\end{remark}

\begin{lemma}\label{L.5.13}
We recall the notation $\widetilde{Q}_{s,\epsilon}$ introduced in Definition \ref{D.1.4} b).
\begin{enumerate}
\item We have the equality:
\begin{equation}\label{5.26}
\mathfrak{L}_s(\epsilon) = \left\{\,w\in\mathscr{S}^\prime(\X^2)\,\mid\,\widetilde{Q}_{s,\epsilon}w\,\in\,\mathfrak{L}_0\,\right\}.
\end{equation}
\item On $\mathfrak{L}_s(\epsilon)$ the definition norm is equivalent with the following norm:
\begin{equation}\label{5.27}
\|w\|^\prime_{\mathfrak{L}_s(\epsilon)}\ :=\ \left\|\widetilde{Q}_{s,\epsilon}w\right\|_{\mathfrak{L}_0}.
\end{equation}
\item If $s\geq0$, then $\mathfrak{L}_s(\epsilon)$ is continuously embedded into $\mathfrak{L}_0$, uniformly with respect to $\epsilon\in[-\epsilon_0,\epsilon_0]$.
\end{enumerate}
\end{lemma}
\begin{proof}
1. For any $w\in\mathfrak{L}_s(\epsilon)$, there exists a vector $v\in\mathcal{H}^s_{A_\epsilon}(\X)$ such that $w\equiv w_v=\underset{\gamma\in\Gamma}{\sum}\boldsymbol{\psi}^*\big(v\otimes\delta_{-\gamma}\big)$. But, we know that by definition we have that $Q_{s,\epsilon}v\in L^2(\X)$, so that we deduce that
$$
\widetilde{Q}_{s,\epsilon}w_v\ =\ \boldsymbol{\psi}^*\big(Q_{s,\epsilon}\otimes\id\big)\boldsymbol{\psi}^*w_v\ =\ \underset{\gamma\in\Gamma}{\sum}\boldsymbol{\psi}^*\big((Q_{s,\epsilon}v)\otimes\delta_{-\gamma}\big)\,\in\,\mathfrak{L}_0.
$$

Reciprocally let $w\in\mathscr{S}^\prime(\X^2)$ be such that $\widetilde{Q}_{s,\epsilon}w$ belongs to $\mathfrak{L}_0$. By the definition of this last space it follows that there exists $f\in L^2(\X)$ such that 
$$
\widetilde{Q}_{s,\epsilon}w\ =\ \underset{\gamma\in\Gamma}{\sum}\boldsymbol{\psi}^*\big(f\otimes\delta_{-\gamma}\big).
$$
It follows that 
$$
w\ =\ \underset{\gamma\in\Gamma}{\sum}\boldsymbol{\psi}^*\big((Q_{-s,\epsilon}f)\otimes\delta_{-\gamma}\big).
$$
But then we have that $v=Q_{-s,\epsilon}f\in\mathcal{H}^s_{A_\epsilon}(\X)$ and in conclusion $w$ belongs to $\mathfrak{L}_s(\epsilon)$.

2. This result follows from the Closed Graph Theorem.

3. This result follows from the continuous embedding of $\mathcal{H}^s_{A_\epsilon}(\X)$ into $L^2(\X)$ for any $s\geq0$ uniformly for $\epsilon\in[-\epsilon_0,\epsilon_0]$.
\end{proof}

\begin{lemma}\label{L.5.14}
For any $m\in\mathbb{R}_+$ and for any $\epsilon\in[-\epsilon_0,\epsilon_0]$ we have the following topological embedding:
\begin{equation}\label{5.29}
\mathfrak{L}_m(\epsilon)\ \hookrightarrow\ \mathscr{S}^\prime\big(\X;\mathcal{H}^m(\mathbb{T})\big),
\end{equation}
uniformly with respect to $\epsilon\in[-\epsilon_0,\epsilon_0]$.
\end{lemma}
\begin{proof}
From the Lemmas \ref{L.5.13} c) and \ref{L.5.7} it follows that we have the following topological embedding: $\mathfrak{L}_m(\epsilon)\hookrightarrow\mathscr{S}^\prime
\big(\X;L^2(\mathbb{T})\big)$ uniformly with respect to $\epsilon\in[-\epsilon_0,\epsilon_0]$. From here on we proceed as in the proof of the secong inclusion in Lemma \ref{L.2.16}. Extending by continuity the sesquilinear form \eqref{2.44} to the canonical sesquilinear form $(.,.)_m$ on $\mathscr{S}^\prime\big(\X;\mathcal{H}^m(\mathbb{T})\big)\times\mathscr{S}\big(\X;\mathcal{H}^m(\mathbb{T})\big)$ we can write it in the following way
\begin{equation}\label{5.30}
(u,v)_m\ =\ \big(u,(\id\otimes<D_\Gamma>^{2m})v\big)_0,\qquad\forall(u,v)\in\mathscr{S}^\prime\big(\X;\mathcal{H}^m(\mathbb{T})\big)\times\mathscr{S}(\X\times\mathbb{T})
\end{equation}
and may extend it to $\mathscr{S}^\prime\big(\X;L^2(\mathbb{T})\big)\times\mathscr{S}(\X\times\mathbb{T})$.

Let us choose now some $u\in\mathfrak{L}_m(\epsilon)$ and let us denote by $f:=\widetilde{Q}_{m,\epsilon}u\in\mathfrak{L}_0$. Because $u=\widetilde{Q}_{-m,\epsilon}\widetilde{Q}_{m,\epsilon}u=\widetilde{Q}_{-m,\epsilon}f$, for any $v\in\mathscr{S}(\X\times\mathbb{T})$ we shall have the equality
\begin{equation}\label{5.31}
(u,v)_m\ =\ \big(f,\widetilde{Q}_{-m,\epsilon}(\id\otimes<D_\Gamma>^{2m})v\big)_0.
\end{equation}

From the fact that $\mathfrak{L}_0$ is continuously embedded into $\mathscr{S}^\prime\big(\X;L^2(\mathbb{T})\big)$, it follows the existence of a defining semi-norm $|.|_l$ on $\mathscr{S}\big(\X;L^2(\mathbb{T})\big)$ such that
\begin{equation}\label{5.32}
\left|(u,v)_m\right|\ \leq\ \|f\|_{\mathfrak{L}_0}\left|\widetilde{Q}_{-m,\epsilon}(\id\otimes<D_\Gamma>^{2m})v\right|_l,\qquad\forall v\in\mathscr{S}(\X\times\mathbb{T}).
\end{equation}
Noticing that $\|f\|_{\mathfrak{L}_0}=\|u\|_{\mathfrak{L}_m(\epsilon)}$ and applying Lemma \ref{L.2.15}, we deduce the existence of a defining semi-norm $|.|_{-m,k}$ on $\mathscr{S}\big(\X;\mathcal{H}^{-m}(\mathbb{T})\big)$ such that
\begin{equation}\label{5.33}
\left|(u,v)_m\right|\ \leq\ C\|u\|_{\mathfrak{L}_m(\epsilon)}\left|(\id\otimes<D_\Gamma>^{2m})v\right|_{-m,k},\qquad\forall v\in\mathscr{S}(\X\times\mathbb{T}).
\end{equation}
But $\left|(\id\otimes<D_\Gamma>^{2m})v\right|_{-m,k}=\|v\|_{m,k}$ so that the statement of the Lemma follows from \eqref{5.33}.
\end{proof}

\section{The proof of Theorem \ref{T.0.1}}
\setcounter{equation}{0}
\setcounter{theorem}{0}

The main technical results discussed in this section concern some continuity properties of the operators $\mathcal{P}_{\epsilon,\lambda}$ and $\mathcal{E}_{\epsilon,\lambda}$ defined in Section \ref{S.4} extended to some spaces of tempered distributions of the type considered in Section \ref{S.5}. We shall suppose that the Hypothesis H.1 - H.6 are satisfied and we shall use the notations introduced in Sections \ref{S.0} and \ref{S.1}. Let us just recall that:
\begin{itemize}
\item The operator $P_{\epsilon}:=\mathfrak{Op}^{A_\epsilon}(\overset{\circ}{p}_\epsilon)$ with $\overset{\circ}{p}_\epsilon(y,\eta):=p(y,y,\eta)$, defines a self-adjoint operator in $L^2(\X)$ on the domain $\mathcal{H}^m_{A_\epsilon}(\X)$.
\item The operator $\widetilde{P}_\epsilon$ defines a self-adjoint operator $\widetilde{P}_\epsilon^\prime$ in the Hilbert space $L^2(\X^2)$ with domain $\widetilde{\mathcal{H}}^m_{A_\epsilon}(\X^2)$ and another self-adjoint operator $\widetilde{P}_\epsilon^{\prime\prime}$ in the Hilbert space $L^2(\X\times\mathbb{T})$ with the domain $\mathcal{K}^m_\epsilon(\X^2)$.
\end{itemize}

\begin{lemma}\label{L.6.1}
For any $\epsilon\in[-\epsilon_0,\epsilon_0]$ we have that:
\begin{enumerate}
\item $\widetilde{P}_\epsilon\in\mathbb{B}\big(\mathfrak{L}_m(\epsilon);\mathfrak{L}_0\big)$ uniformly in $\epsilon\in[-\epsilon_0,\epsilon_0]$.
\item The operator $\widetilde{P}_\epsilon$ considered as an unbounded operator in the Hilbert space $\mathfrak{L}_0$ defines a self-adjoint operator $\widetilde{P}_\epsilon^{\prime\prime\prime}$ having domain $\mathfrak{L}_m(\epsilon)$ and this self-adjoint operator is unitarily equivalent with $P_\epsilon$.
\end{enumerate}
\end{lemma}
\begin{proof}
1. Let us choose two test functions $v$ and $\varphi$ from $\mathscr{S}(\X)$. Using formula \eqref{1.4} we obtain that
$$
\big[\big(\boldsymbol{\psi}^*\widetilde{P}_\epsilon
\boldsymbol{\psi}^*\big)(v\otimes\varphi)\big](x,y)\ =\ \left[\mathfrak{Op}^{A_\epsilon}\big([(\id\otimes\tau_y
\otimes\id)p_\epsilon]^\circ\big)v\right](x)\varphi(y),\qquad\forall(x,y)\in\X^2.
$$
In this equality we insert $\varphi(y)\equiv\varphi_\lambda(y):=\lambda^{-d}\theta\big(\frac{y+\gamma}{\lambda}\big)$ for some $(\lambda,\gamma)\in\mathbb{R}_+^*\times\Gamma$ and for any $y\in\X$, where we denoted by $\theta$ a test function of class $C^\infty_0(\X)$ that satisfies the condition $\int_\X\theta(y)dy=1$. With this choice we consider the limit for $\lambda\searrow0$  as tempered distribution on $\X^2$. Taking into account that for $\lambda\searrow0$ we have that $\varphi_\lambda$ converges in $\mathscr{S}^\prime(\X)$ to $\delta_{-\gamma}$ and using Hypothesis H.6, we conclude that
$$
\big(\boldsymbol{\psi}^*\widetilde{P}_\epsilon
\boldsymbol{\psi}^*\big)(v\otimes\delta_{-\gamma})\ =\ \big(P_\epsilon v\big)\otimes\delta_{-\gamma},\qquad\forall v\in\mathscr{S}(\X),\ \forall\gamma\in\Gamma.
$$
Extending by continuity we can write the equality
\begin{equation}\label{6.1}
\big(\boldsymbol{\psi}^*\widetilde{P}_\epsilon
\boldsymbol{\psi}^*\big)(v\otimes\delta_{-\gamma})\ =\ \big(P_\epsilon v\big)\otimes\delta_{-\gamma},\qquad\forall v\in\mathscr{S}^\prime(\X),\ \forall\gamma\in\Gamma.
\end{equation}

We conclude that for any $u\in\mathfrak{L}_m(\epsilon)$ of the form 
$$
u\ \equiv\ u_v\ :=\underset{\gamma\in\Gamma}{\sum}\boldsymbol{\psi}^*\big(v\otimes\delta_{-\gamma}\big)
$$
for some $v\in\mathcal{H}^m_{A_\epsilon}(\X)$ we can write:
$$
\widetilde{P}_\epsilon u\ =\ \boldsymbol{\psi}^*\left(\underset{\gamma\in\Gamma}{\sum}\big(\boldsymbol{\psi}^*\widetilde{P}_\epsilon
\boldsymbol{\psi}^*\big)(v\otimes\delta_{-\gamma})\right)\ =\ \underset{\gamma\in\Gamma}{\sum}\boldsymbol{\psi}^*\big((P_\epsilon v)\otimes\delta_{-\gamma}\big).
$$
The first statement of the Lemma follows now from the fact that $P_\epsilon\in\mathbb{B}\big(\mathcal{H}^m_{A_\epsilon}(\X);L^2(\X)\big)$ uniformly with respect to $\epsilon\in[-\epsilon_0,\epsilon_0]$.

2. Let us notice that the linear operator defined by
$$
U^s_\epsilon:\mathcal{H}^s_{A_\epsilon}(\X)\rightarrow\mathfrak{L}_s(\epsilon),\qquad U^s_\epsilon v:=\underset{\gamma\in\Gamma}{\sum}\boldsymbol{\psi}^*\big(v\otimes\delta_{-\gamma}\big)
$$
is in fact a unitary operator for any pair $(s,\epsilon)\in\mathbb{R}\times[-\epsilon_0,\epsilon_0]$. Following the arguments from the proof of the first point of the Lemma we have the following equality $\widetilde{P}_\epsilon U^s_\epsilon=U^s_\epsilon P_\epsilon$ valid on $\mathcal{H}^m_{A_\epsilon}(\X)$ (the domain of self-adjointness of $P_\epsilon$).
\end{proof}

We shall study now the {\it effective Hamiltonian} $\mathfrak{E}_{-+}(\epsilon,\lambda)$ defined in Theorem \ref{T.4.3}. The following two technical results will be used in proving the boundedness and self-adjointness of $\mathfrak{E}_{-+}(\epsilon,\lambda)$ in $\mathfrak{V}_0^N$.

\begin{lemma}\label{L.6.2}
Suppose given an operator-valued symbol $\mathfrak{q}\in S^0_0\big(\X;\mathbb{B}(\mathbb{C}^N)\big)$ that is hermitian (i.e. $\mathfrak{q}(x,\xi)^*=\mathfrak{q}(x,\xi),\,\forall(x,\xi)\in\Xi$) and verifies the following invariance property: $(\id\otimes\tau_{\gamma^*})\mathfrak{q}=\mathfrak{q},\,\forall\gamma^*\in\Gamma^*$. Then, for any $\epsilon\in[-\epsilon_0,\epsilon_0]$ the operator $\mathfrak{Op}^{A_\epsilon}(\mathfrak{q})$ belongs to $\mathbb{B}\big(\mathfrak{V}_0^N\big)$ uniformly with respect to $\epsilon\in[-\epsilon_0,\epsilon_0]$ and is self-adjoint. 

Moreover, the application $ S^0_0\big(\X;\mathbb{B}(\mathbb{C}^N)\big)\ni\mathfrak{q}\mapsto\mathfrak{Op}^{A_\epsilon}(\mathfrak{q})\in\mathbb{B}\big(\mathfrak{V}_0^N\big)$ is continuous uniformly with respect to $\epsilon\in[-\epsilon_0,\epsilon_0]$.
\end{lemma}
\begin{proof}
The invariance with respect to translations from $\Gamma^*$ assumed in the statement implies that the operator-valued symbol $\mathfrak{q}$ is in fact a $\Gamma^*$-periodic function with respect to the second variable $\xi\in\X^*$ and thus can be decomposed in a Fourier series (as tempered distributions in $\mathscr{S}^\prime\big(\Xi;\mathbb{B}(\mathbb{C}^N)\big)$):
\begin{equation}\label{6.2}
\mathfrak{q}(x,\xi)\ =\ \underset{\alpha\in\Gamma}{\sum}\hat{\mathfrak{q}}_\alpha(x)e^{i<\xi,\alpha>},\qquad\hat{\mathfrak{q}}_\alpha(x):=|E^*|^{-1}\int_{E^*}e^{-i<\xi,\alpha>}\mathfrak{q}(x,\xi)\,d\xi.
\end{equation}
Due to the regularity of the symbol functions we deduce that for any $\beta\in\mathbb{N}^d$ and for any $k\in\mathbb{N}$ there exists a strictly positive constant $C_{\beta,k}$ such that
\begin{equation}\label{6.3}
\left|\big(\partial^\beta_x\hat{\mathfrak{q}}_\alpha\big)(x)\right|\ \leq\ C_{\beta,k}<\alpha>^{-k},\qquad\forall x\in\X,\ \forall\alpha\in\Gamma
\end{equation}
and we conclude that the series in \eqref{6.2} converges in fact in $BC^\infty\big(\Xi;\mathbb{B}(\mathbb{C}^N)\big)\equiv S^0_0\big(\X;\mathbb{B}(\mathbb{C}^N)\big)$. From \eqref{6.2} we deduce that 
\begin{equation}\label{6.4}
\big(\mathfrak{Op}^{A_\epsilon}(\mathfrak{q})u\big)(x)\ =\ \underset{\alpha\in\Gamma}{\sum}\big(Q_\alpha u\big)(x),\qquad\forall x\in\X,\ \forall u\in\mathscr{S}(\X:\mathbb{C}^N),
\end{equation}
where $Q_\alpha$ is the linear operator defined on $\mathscr{S}(\X;\mathbb{C}^N)$ by the following oscillating integral:
\begin{equation}\label{6.5}
\big(Q_\alpha u\big)(x)\ :=\ \int_\Xi e^{i<\eta,x-y+\alpha>}\omega_{A_\epsilon}(x,y)\,\hat{\mathfrak{q}}_\alpha\big(\frac{x+y}{2}\big)\,u(y)\,dy\,\dbar\eta\ =\ \omega_{A_\epsilon}(x,x+\alpha)\,\hat{\mathfrak{q}}_\alpha(x+\alpha/2)(\tau_{-\alpha}u)(x).
\end{equation}

Both equalities \eqref{6.4} and \eqref{6.5} may be extended by continuity to any $u\in\mathscr{S}^\prime(\X;\mathbb{C}^N)$.
Let us consider then $u\equiv u_{\underline{f}}=\underset{\gamma\in\Gamma}{\sum}\,\underline{f}_\gamma\delta_{-\gamma}\in\mathfrak{V}_0^N$ for some $\underline{f}\in\big[l^2(\Gamma)\big]^N$. Then we can write:
\begin{equation}\label{6.6}
Q_\alpha u\ =\ \underset{\gamma\in\Gamma}{\sum}\omega_{A_\epsilon}(-\gamma-\alpha,-\gamma)\,\hat{\mathfrak{q}}_\alpha(-\gamma-\alpha/2)\,\underline{f}_\gamma\,\delta_{-\alpha-\gamma}\ =\ \underset{\gamma\in\Gamma}{\sum}\omega_{A_\epsilon}(-\gamma,\alpha-\gamma)\,\hat{\mathfrak{q}}_\alpha(-\gamma+\alpha/2)\,\underline{f}_{\gamma-\alpha}\,\delta_{-\gamma}.
\end{equation}
If we use now the formula \eqref{6.6} in \eqref{6.4} we get
\begin{equation}\label{6.7}
\mathfrak{Op}^{A_\epsilon}(\mathfrak{q})u\ =\ \underset{\gamma\in\Gamma}{\sum}\widetilde{f}_\gamma\,\delta_{-\gamma},
\end{equation}
\begin{equation}\label{6.8}
\widetilde{f}_\gamma\ :=\ \underset{\alpha\in\Gamma}{\sum}\omega_{A_\epsilon}(-\gamma,\alpha-\gamma)\,\hat{\mathfrak{q}}_{\alpha}(-\gamma+\alpha/2)\,\underline{f}_{\gamma-\alpha}\ =\ \underset{\alpha\in\Gamma}{\sum}\omega_{A_\epsilon}(-\gamma,-\alpha)\,\hat{\mathfrak{q}}_{\gamma-\alpha}\big(-\frac{\gamma+\alpha}{2}\big)\,\underline{f}_{\alpha}.
\end{equation}
Let us verify that $\widetilde{f}\in\big[l^2(\Gamma)\big]^N$. In fact from \eqref{6.3} and \eqref{6.8} it follows that for any $k\in\mathbb{N}$ (sufficiently large) there exists $C_k>0$ such that 
$$
|\widetilde{f}_{\gamma}|\ \leq\ C_k\underset{\alpha\in\Gamma}{\sum}<\gamma-\alpha>^{-k}|\underline{f}_\alpha|\ \leq\ C_k\sqrt{\underset{\alpha\in\Gamma}{\sum}<\gamma-\alpha>^{-k}}\sqrt{\underset{\alpha\in\Gamma}{\sum}<\gamma-\alpha>^{-k}|\underline{f}_\alpha|^2},
$$
so that we have the estimation:
\begin{equation}\label{6.9}
\|\widetilde{f}\|_{[l^2(\Gamma)]^N}^2\ =\ \underset{\gamma\in\Gamma}{\sum}|\widetilde{f}_\gamma|^2\ \leq\ C^\prime\underset{\alpha\in\Gamma}{\sum}|\underline{f}_\alpha|^2\ =\ C^\prime\|\underline{f}\|_{[l^2(\Gamma)]^N}^2.
\end{equation}
From \eqref{6.7} and \eqref{6.9} we clearly deduce the fact that $\mathfrak{Op}^{A_\epsilon}(\mathfrak{q})\in\mathbb{B}\big(\mathfrak{V}_0^N\big)$ uniformly with respect to $\epsilon\in[-\epsilon_0,\epsilon_0]$ and the continuity of the application $ S^0_0\big(\X;\mathbb{B}(\mathbb{C}^N)\big)\ni\mathfrak{q}\mapsto\mathfrak{Op}^{A_\epsilon}(\mathfrak{q})\in\mathbb{B}\big(\mathfrak{V}_0^N\big)$ uniformly with respect to $\epsilon\in[-\epsilon_0,\epsilon_0]$ clearly follows from \eqref{6.8} and \eqref{6.2}.

In order to prove the self-adjointness of $\mathfrak{Op}^{A_\epsilon}(\mathfrak{q})$ we fix a second element $v\in\mathfrak{V}_0^N$ of the form $v\equiv v_{\underline{g}}=\underset{\gamma\in\Gamma}{\sum}\underline{g}_\gamma\,\delta_{-\gamma}$ for some $\underline{g}\in \big[l^2(\Gamma)\big]^N$. Then we notice that
\begin{equation}\label{6.10}
\left\{
\begin{array}{rcl}
\mathfrak{Op}^{A_\epsilon}(\mathfrak{q})v&=&\underset{\gamma\in\Gamma}{\sum}\widetilde{g}_\gamma\,\delta_{-\gamma}\\
\widetilde{g}_\gamma&=&\underset{\alpha\in\Gamma}{\sum}\omega_{A_\epsilon}(-\gamma,-\alpha)\,\hat{\mathfrak{q}}_{\gamma-\alpha}\big(-\frac{\gamma+\alpha}{2}\big)\,\underline{g}_{\alpha}.
\end{array}
\right.
\end{equation}
Let us point out the following evident equalities:
\begin{equation}\label{6.11}
\big[\hat{\mathfrak{q}}(x)]^*\ =\ \hat{\mathfrak{q}}_{-\alpha}(x);\qquad\overline{\omega_{A_\epsilon}(-\gamma,-\alpha)}\ =\ \omega_{A_\epsilon}(-\alpha,-\gamma)
\end{equation}
in order to deduce that
$$
\left(\mathfrak{Op}^{A_\epsilon}(\mathfrak{q})u\,,\,v\right)_{\mathfrak{V}_0^N}\ =\ \underset{\gamma\in\Gamma}{\sum}\left(\widetilde{f}_\gamma\,,\,\underline{g}_\gamma\right)_{\mathbb{C}^N}\ =\ \underset{(\alpha,\gamma)\in\Gamma^2}{\sum}\left(\omega_{A_\epsilon}(-\gamma,-\alpha)\,\hat{\mathfrak{q}}_{\gamma-\alpha}\big(-\frac{\gamma+\alpha}{2}\big)\,\underline{f}_{\alpha}\,,\,\underline{g}_\gamma\right)_{\mathbb{C}^N}\ =
$$
$$
=\ \underset{(\alpha,\gamma)\in\Gamma^2}{\sum}\left(\underline{f}_{\alpha}\,,\,\omega_{A_\epsilon}(-\alpha,-\gamma)\,\hat{\mathfrak{q}}_{\alpha-\gamma}\big(-\frac{\gamma+\alpha}{2}\big)\,\underline{g}_\gamma\right)_{\mathbb{C}^N}\ =\ \underset{\alpha\in\Gamma}{\sum}\left(\underline{f}_\alpha\,,\,\widetilde{g}_\alpha\right)_{\mathbb{C}^N}\ =\ \left(u\,,\,\mathfrak{Op}^{A_\epsilon}(\mathfrak{q})v\right)_{\mathfrak{V}_0^N}.
$$
\end{proof}

\begin{remark}\label{R.6.3}
Let us point out that a shorter proof of the boundedness of $\mathfrak{Op}^{A_\epsilon}(\mathfrak{q})$ on $\mathfrak{V}_0^N$ may be obtained by using the Proposition \ref{P.5.4} characterizing the distributions from $\mathfrak{V}_0$. The proof we have given has the advantage of giving the explicit form of the operator $\mathfrak{Op}^{A_\epsilon}(\mathfrak{q})$ when we identify $\mathfrak{V}_0^N$ with $\big[l^2(\Gamma)\big]^N$ (see \eqref{6.7} and \eqref{6.8}). Moreover, the self-adjointness is a very easy consequence of these formulae.
\end{remark}

In order to prove that the effective Hamiltonian $\mathfrak{E}_{-+}(\epsilon,\lambda)$ satisfies the hypothesis of the Lemma \ref{L.6.2} we shall need the commutation properties we have proved at the end of Section \ref{S.4}, that we now recall in the following Lemma.
\begin{lemma}\label{L.6.4}
With the notations introduced in Lemma \ref{L.4.6} and Remark \ref{R.4.7}, for any $\gamma^*\in\Gamma^*$ and for any $(\epsilon,\lambda)\in[-\epsilon_0,\epsilon_0]\times I$ the following equalities are true:
\begin{equation}\label{6.12}
\left\{
\begin{array}{rcl}
\mathfrak{R}_{-,\epsilon}\,\sigma_{\gamma^*}&=&\Upsilon_{\gamma^*}\,\mathfrak{R}_{-,\epsilon}\\
\mathfrak{R}_{+,\epsilon}\,\Upsilon_{\gamma^*}&=&\sigma_{\gamma^*}\,\mathfrak{R}_{+,\epsilon},
\end{array}
\right.
\end{equation}
\begin{equation}\label{6.13}
\left\{
\begin{array}{rcl}
\mathfrak{E}(\epsilon,\lambda)\,\Upsilon_{\gamma^*}&=&\Upsilon_{\gamma^*}\,\mathfrak{E}(\epsilon,\lambda)\\
\mathfrak{E}_+(\epsilon,\lambda)\,\sigma_{\gamma^*}&=&\Upsilon_{\gamma^*}\,\mathfrak{E}_+(\epsilon,\lambda)\\
\mathfrak{E}_-(\epsilon,\lambda)\,\Upsilon_{\gamma^*}&=&\sigma_{\gamma^*}\,\mathfrak{E}_-(\epsilon,\lambda)\\
\mathfrak{E}_{-+}(\epsilon,\lambda)\,\sigma_{\gamma^*}&=&\sigma_{\gamma^*}\,\mathfrak{E}_{-+}(\epsilon,\lambda).
\end{array}
\right.
\end{equation}
\end{lemma}
\begin{proof}
The equalities \eqref{6.12} are exactly the equalities \eqref{4.24} and \eqref{4.25} that we have proved in Section \ref{S.4}. The equalities \eqref{6.13} follow from \eqref{4.26} and \eqref{4.11}.
\end{proof}

\begin{lemma}\label{L.6.5}
Under the Hypothesis of Theorem \ref{T.4.3}, we have that $\mathfrak{E}_{-+}(\epsilon,\lambda)\in\mathbb{B}\big(\mathfrak{V}_0^N\big)$ uniformly with respect to $(\epsilon,\lambda)\in[-\epsilon_0,\epsilon_0]\times I$ and is self-adjoint on the Hilbert space $\mathfrak{V}_0^N$.
\end{lemma}
\begin{proof}
We recall that $\mathfrak{E}_{-+}(\epsilon,\lambda):=\mathfrak{Op}^{A_\epsilon}\big(E^{-,+}_{\epsilon,\lambda}\big)$ where $E^{-,+}_{\epsilon,\lambda}\in S^0_0\big(\X;\mathbb{B}(\mathbb{C}^N)\big)$. This Lemma will thus follow directly from Lemma \ref{L.6.2} once we have shown that $E^{-,+}_{\epsilon,\lambda}$ is hermitian and $\Gamma^*$-periodic in the second variable $\xi\in\X^*$.

In order to prove the symmetry we use the fact that the operator $\mathcal{E}_{\epsilon,\lambda}$ is self-adjoint on $\mathcal{K}(\X^2)\times L^2(\X;\mathbb{C}^N)$ and deduce that $\mathfrak{E}_{-+}(\epsilon,\lambda)$ is self-adjoint on the Hilbert space $L^2(\X;\mathbb{C}^N)$. Thus we have the equality $\big[\mathfrak{E}_{-+}(\epsilon,\lambda)\big]^*=\mathfrak{E}_{-+}(\epsilon,\lambda)$ from which we deduce that 
$$
\mathfrak{Op}^{A_\epsilon}\big(\big[E^{-,+}_{\epsilon,\lambda}\big]^*-E^{-,+}_{\epsilon,\lambda}\big)\ =\ 0.
$$
As the application $\mathfrak{Op}^{A_\epsilon}:\mathscr{S}^\prime(\Xi)\rightarrow\mathbb{B}\big(\mathscr{S}(\X);\mathscr{S}^\prime(\X)\big)$ is an isomorphism (see \cite{MP1}) it follows the symmetry relation $\big[E^{-,+}_{\epsilon,\lambda}\big]^*=E^{-,+}_{\epsilon,\lambda}$.

For the $\Gamma^*$-periodicity we use the last equality in \eqref{6.13} that can also be written as 
$$
\sigma_{-\gamma^*}\mathfrak{E}_{-+}(\epsilon,\lambda)\sigma_{\gamma^*}\ =\ \mathfrak{E}_{-+}(\epsilon,\lambda).
$$
Considering now the equality \eqref{A.4}, that evidently remains true also for the $\mathfrak{Op}^{A_\epsilon}$ quantization, we can write
$$
\sigma_{-\gamma^*}\mathfrak{E}_{-+}(\epsilon,\lambda)\sigma_{\gamma^*}\ =\ \mathfrak{Op}^{A_\epsilon}\big((\id\otimes\tau_{-\gamma^*})E^{-+}_{\epsilon,\lambda}\big).
$$
Repeating the above argument based on the injectivity of the quantization map (\cite{MP1}) we conclude that $(\id\otimes\tau_{-\gamma^*})E^{-+}_{\epsilon,\lambda}=E^{-+}_{\epsilon,\lambda}$ for any $\gamma^*\in\Gamma^*$.
\end{proof}

We shall now study the continuity properties of the operators $\mathfrak{R}_{\pm,\epsilon}$, $\mathfrak{E}_{\pm}(\epsilon,\lambda)$ and $\mathfrak{E}(\epsilon,\lambda)$ acting on the distribution spaces $\mathfrak{V}_0$ or $\mathfrak{L}_0$ by using the characterizations of these spaces obtained in Section \ref{S.5}  (Propositions \ref{P.5.4}, \ref{P.5.9} and \ref{P.5.11}) as well as the commutation properties recalled in Lemma \ref{L.6.4}. We shall suppose the hypothesis of Theorem \ref{T.4.3} are satisfied and $(\epsilon,\lambda)\in[-\epsilon_0,\epsilon_0]\times I$.

\begin{lemma}\label{L.6.6}
$\mathfrak{R}_{+,\epsilon}\in\mathbb{B}\big(\mathfrak{L}_0;\mathfrak{V}_0^N\big)$ uniformly with respect to $\epsilon\in[-\epsilon_0,\epsilon_0]$.
\end{lemma}
\begin{proof}
Let us recall that $\mathfrak{R}_{+,\epsilon}=\mathfrak{Op}^{A_\epsilon}\big(R_+\big)$ with $R_+\in S^0_0\big(\X;\mathbb{B}(\mathcal{K}_0;\mathbb{C}^N)\big)$ so that finally we deduce that $\mathfrak{R}_{+,\epsilon}\in\mathbb{B}\big(\mathscr{S}^\prime(\X;\mathcal{K}_0);\mathscr{S}^\prime(\X;\mathbb{C}^N)\big)$. Moreover, from Proposition \ref{P.A.26} we deduce that for any $s\in\mathbb{R}$ we have that $\mathfrak{R}_{+,\epsilon}\in\mathbb{B}\big(\mathcal{H}^s_{A_\epsilon}(\X)\otimes\mathcal{K}_0\,;\,\big[\mathcal{H}^s_{A_\epsilon}(\X)\big]^N\big)$ uniformly with respect to $\epsilon\in[-\epsilon_0,\epsilon_0]$.

Suppose fixed some $u\in\mathfrak{L}_0$; from Proposition \ref{P.5.9} we deduce the existence of a unique $u_0\in\mathcal{H}^\infty_{A_\epsilon}(\X)\otimes\mathcal{K}_0\equiv\mathcal{H}^\infty_{A_\epsilon}(\X)\otimes L^2(\mathbb{T})$ such that $u=\underset{\gamma^*}{\sum}\Upsilon_{\gamma^*}u_0$ with convergence in $\mathscr{S}^\prime(\X^2)$. In fact Lemma \ref{L.5.8} implies the convergence of the above series in $\mathscr{S}^\prime\big(\X;\mathcal{K}_0\big)$. Using now also the second equation in \eqref{6.12} we can write that in $\mathscr{S}^\prime\big(\X;\mathbb{C}^N\big)$ we have the equalities
$$
\mathfrak{R}_{+,\epsilon}u\ =\ \underset{\gamma^*}{\sum}\mathfrak{R}_{+,\epsilon}\Upsilon_{\gamma^*}u_0\ =\ \underset{\gamma^*}{\sum}\sigma_{\gamma^*}\mathfrak{R}_{+,\epsilon}u_0.
$$
But we have seen that $\mathfrak{R}_{+,\epsilon}u_0\in\big[\mathcal{H}^\infty_{A_\epsilon}(\X)\big]^N$ and thus Proposition \ref{P.5.4} b) implies that $\mathfrak{R}_{+,\epsilon}u\in\mathfrak{V}_0^N$. The fact that $\mathfrak{R}_{+,\epsilon}\in\mathbb{B}\big(\mathfrak{L}_0;\mathfrak{V}_0^N\big)$ uniformly with respect to $\epsilon\in[-\epsilon_0,\epsilon_0]$ follows now from this result and the following three remarks:
\begin{enumerate}
\item The above mentioned continuity property of $\mathfrak{R}_{+,\epsilon}$ that follows from Proposition \ref{P.A.26}.
\item The uniform continuity of the application $\mathfrak{L}_0\ni u\mapsto u_0\in\mathcal{H}^\infty_{A_\epsilon}(\X)\otimes\mathcal{K}_0$ with respect to $\epsilon\in[-\epsilon_0,\epsilon_0]$, that follows from Proposition \ref{P.5.9}.
\item The uniform continuity of the application $\mathcal{H}^\infty_{A_\epsilon}(\X)\ni\mathfrak{R}_{+,\epsilon}u_0\mapsto\mathfrak{R}_{+,\epsilon}u\in\mathfrak{V}_0^N$ with respect to $\epsilon\in[-\epsilon_0,\epsilon_0]$, that follows from Proposition \ref{P.5.4} b).
\end{enumerate}
\end{proof}

\begin{lemma}\label{L.6.7}
$\mathfrak{E}_{-}(\epsilon,\lambda)\in\mathbb{B}\big(\mathfrak{L}_0;\mathfrak{V}_0^N\big)$ uniformly with respect to $(\epsilon,\lambda)\in[-\epsilon_0,\epsilon_0]\times I$.
\end{lemma}
\begin{proof}
Let us recall that $\mathfrak{E}_-(\epsilon,\lambda)=\mathfrak{Op}^{A_\epsilon}\big(E^-_{\epsilon,\lambda}\big)$ with $E^-_{\epsilon,\lambda}\in S^0_0\big(\X;\mathbb{B}(\mathcal{K}_0;\mathbb{C}^N)\big)$ uniformly with respect to $(\epsilon,\lambda)\in[-\epsilon_0,\epsilon_0]\times I$. We continue as in the above proof of Lemma \ref{L.6.6}. Considering $\mathfrak{E}_-(\epsilon,\lambda)$ as a magnetic pseudodifferential operator it can be extended to an operator $\mathfrak{E}_-(\epsilon,\lambda)\in\mathfrak{B}\big(\mathscr{S}^\prime(\X;\mathcal{K}_0);\mathscr{S}^\prime(\X;\mathbb{C}^N)\big)$; moreover we notice that $\mathfrak{L}_0$ is continuously embedded into $\mathscr{S}^\prime(\X;\mathcal{K}_0)$ and we conclude that we have indeed that $\mathfrak{E}_-(\epsilon,\lambda)\in\mathbb{B}\big(\mathfrak{L}_0;\mathscr{S}^\prime(\X;\mathbb{C}^N)\big)$. From Proposition \ref{P.5.9} we conclude that for any $u\in\mathfrak{L}_0$ there exists some $u_0\in\mathcal{H}^\infty_{A_\epsilon}(\X)\otimes\mathcal{K}_0$ such that $u=\underset{\gamma^*\in\Gamma^*}{\sum}\Upsilon_{\gamma^*}u_0$, the series converging in $\mathscr{S}^\prime(\X;\mathcal{K}_0)$. Using this identity and the third equality in \eqref{6.13} we obtain that
$$
\mathfrak{E}_-(\epsilon,\lambda)u\ =\ \underset{\gamma^*}{\sum}\mathfrak{E}_-(\epsilon,\lambda)\Upsilon_{\gamma^*}u_0\ =\ \underset{\gamma^*}{\sum}\sigma_{\gamma^*}\mathfrak{E}_-(\epsilon,\lambda)u_0.
$$
Taking into account that for any $s\in\mathbb{R}$ we have that $\mathfrak{E}_-(\epsilon,\lambda)\in\mathbb{B}\big(\mathcal{H}^s_{A_\epsilon}(\X)\otimes\mathcal{K}_0;\big[\mathcal{H}^s_{A_\epsilon}(\X)\big]^N\big)$ uniformly with respect to $(\epsilon,\lambda)\in[-\epsilon_0,\epsilon_0]\times I$, the proof of the Lemma ends similarly to the proof of Lemma \ref{L.6.6} above.
\end{proof}
\begin{lemma}\label{L.6.8}
$\mathfrak{E}_{+}(\epsilon,\lambda)\in\mathbb{B}\big(\mathfrak{V}_0^N;\mathfrak{L}_m(\epsilon)\big)$ uniformly with respect to $(\epsilon,\lambda)\in[-\epsilon_0,\epsilon_0]\times I$.
\end{lemma}
\begin{proof}
Let us recall that $\mathfrak{E}_{+}(\epsilon,\lambda)=\mathfrak{Op}^{A_\epsilon}\big(E^+_{\epsilon,\lambda}\big)$ with $E^+_{\epsilon,\lambda}\in S^0_0\big(\X;\mathbb{B}(\mathbb{C}^N;\mathcal{K}_{m,\xi})\big)$ uniformly with respect to $(\epsilon,\lambda)\in[-\epsilon_0,\epsilon_0]\times I$. We conclude that $\mathfrak{E}_+(\epsilon,\lambda)\in\mathbb{B}\big(\mathscr{S}^\prime(\X;\mathbb{C}^N);\mathscr{S}^\prime(\X;\mathcal{K}_{m,0})\big)$. Noticing that by Lemma \ref{L.5.2} the space $\mathfrak{V}_0^N$ embeds continuously into $\mathscr{S}^\prime(\X;\mathbb{C}^N)$ we conclude that $\mathfrak{E}_+(\epsilon,\lambda)\in\mathbb{B}\big(\mathfrak{V}_0^N;\mathscr{S}^\prime(\X;\mathcal{K}_{m,0})\big)$ uniformly with respect to $(\epsilon,\lambda)\in[-\epsilon_0,\epsilon_0]\times I$.

Suppose now fixed some $u\in\mathfrak{V}_0^N$; we know from Proposition \ref{P.5.4} that there exists an element $u_0\in\big[\mathcal{H}^\infty_{A_\epsilon}(\X)\big]^N$ such that $u=\underset{\gamma^*\in\Gamma^*}{\sum}\sigma_{\gamma^*}u_0$ converging as tempered distribution and such that the application $\mathfrak{V}_0^N\ni u\mapsto u_0\in\big[\mathcal{H}^\infty_{A_\epsilon}(\X)\big]^N$ is continuous uniformly with respect to $\epsilon\in[-\epsilon_0,\epsilon_0]$. Using this result and the second equation in \eqref{6.13} we obtain that
$$
\mathfrak{E}_+(\epsilon,\lambda)u\ =\ \underset{\gamma^*\in\Gamma^*}{\sum}\mathfrak{E}_+(\epsilon,\lambda)\sigma_{\gamma^*}u_0\ =\ \underset{\gamma^*\in\Gamma^*}{\sum}\Upsilon_{\gamma^*}\big(\mathfrak{E}_+(\epsilon,\lambda)u_0\big).
$$
Using now Lemma \ref{L.5.13}, in order to prove that $\mathfrak{E}_+(\epsilon,\lambda)u\in\mathfrak{L}_m(\epsilon)$ all we have to prove is that $\widetilde{Q}_{m,\epsilon}\mathfrak{E}_+(\epsilon,\lambda)u\in\mathfrak{L}_0$. In order to do that we shall need two of the properties of the operator $\widetilde{Q}_{m,\epsilon}$ that we have proved in the previous sections. 

First we know from \eqref{2.22} that
$$
\widetilde{Q}_{m,\epsilon}\Upsilon_{\gamma^*}\ =\ \Upsilon_{\gamma^*}\widetilde{Q}_{m,\epsilon},\qquad\forall\gamma^*\in\Gamma^*.
$$

Secondly, at the end of the proof of Lemma \ref{L.4.2} we have shown that $\widetilde{Q}_{m,\epsilon}=\mathfrak{Op}^{A_\epsilon}\big(\widetilde{\mathfrak{q}}_{m,\epsilon}\big)$ with $\widetilde{\mathfrak{q}}_{m,\epsilon}\in S^0_0(\X;\mathbb{B}(\mathcal{K}_{m,\xi};\mathcal{K}_0)\big)$ uniformly with respect to $\epsilon\in[-\epsilon_0,\epsilon_0]$. If we use The Composition Theorem \ref{T.A.23} we notice that $\widetilde{\mathfrak{q}}_{m,\epsilon}\sharp^{B_\epsilon}E^+_{\epsilon,\lambda}\in S^0_0\big(\X;\mathbb{B}(\mathbb{C}^N;\mathcal{K}_0)\big)$ uniformly with respect to $(\epsilon,\lambda)\in[-\epsilon_0,\epsilon_0]\times I$. Applying then Proposition \ref{P.A.26} gives that $\widetilde{Q}_{m,\epsilon}\mathfrak{E}_+(\epsilon,\lambda)\in\mathbb{B}\big(\big[\mathcal{H}^\infty_{A_\epsilon}(\X)\big]^N;\mathcal{H}^\infty_{A_\epsilon}(\X)\otimes\mathcal{K}_0\big)$ uniformly with respect to $(\epsilon,\lambda)\in[-\epsilon_0,\epsilon_0]\times I$. We conclude that
$$
\widetilde{Q}_{m,\epsilon}\mathfrak{E}_+(\epsilon,\lambda)u\ =\ \underset{\gamma^*\in\Gamma^*}{\sum}\Upsilon_{\gamma^*}\widetilde{Q}_{m,\epsilon}\big(\mathfrak{E}_+(\epsilon,\lambda)u_0\big),
$$
and this last element belongs to $\mathfrak{L}_0$ as implied by Proposition \ref{P.5.11}. The conclusion of the Lemma follows now from the following remarks:
\begin{enumerate}
\item The application $\mathfrak{V}_0^N\ni u\mapsto u_0\in\big[\mathcal{H}^\infty_{A_\epsilon}(\X)\big]^N$ is continuous uniformly with respect to $(\epsilon,\lambda)\in[-\epsilon_0,\epsilon_0]\times I$, as proved in Proposition \ref{P.5.4} a).
\item The application $\mathcal{H}^\infty_{A_\epsilon}(\X)\otimes\mathcal{K}_0\ni \widetilde{Q}_{m,\epsilon}\mathfrak{E}_+(\epsilon,\lambda)u_0\mapsto\widetilde{Q}_{m,\epsilon}\mathfrak{E}_+(\epsilon,\lambda)u\in\mathfrak{L}_0$ is continuous uniformly with respect to $(\epsilon,\lambda)\in[-\epsilon_0,\epsilon_0]\times I$, as proved in Proposition \ref{P.5.11}.
\end{enumerate}
\end{proof}

\begin{lemma}\label{L.6.9}
$\mathfrak{R}_{-,\epsilon}\in\mathbb{B}\big(\mathfrak{V}_0^N;\mathfrak{L}_m(\epsilon)\big)$uniformly with respect to $\epsilon\in[-\epsilon_0,\epsilon_0]$.
\end{lemma}
\begin{proof}
Let us recall that $\mathfrak{R}_{-,\epsilon}=\mathfrak{Op}^{A_\epsilon}\big(R_-\big)$ with $R_-\in S^0_0\big(\X;\mathbb{B}(\mathbb{C}^N;\mathcal{K}_{m,\xi})\big)$ as implied by \eqref{3.22} and \eqref{3.26}. Using now the first equality in \eqref{6.12} we notice that $\mathfrak{R}_{-\epsilon}\sigma_{\gamma^*}=\Upsilon_{\gamma^*}\mathfrak{R}_{-,\epsilon},\ \forall\gamma^*\in\Gamma^*$ and the arguments from the proof of Lemma \ref{L.6.8} may be repeated and one obtains the desired conclusion of the Lemma.
\end{proof}

\begin{lemma}\label{L.6.10}
$\mathfrak{E}(\epsilon,\lambda)\in\mathbb{B}\big(\mathfrak{L}_0;\mathfrak{L}_m(\epsilon)\big)$ uniformly with respect to $(\epsilon,\lambda)\in[-\epsilon_0,\epsilon_0]\times I$.
\end{lemma}
\begin{proof}
Let us recall that $\mathfrak{E}(\epsilon,\lambda)=\mathfrak{Op}^{A_\epsilon}\big(E_{\epsilon,\lambda}\big)$ with $E_{\epsilon,\lambda}\in S^0_0\big(\X;\mathbb{B}(\mathcal{K}_0;\mathcal{K}_{m,\xi})\big)$ uniformly with respect to $(\epsilon,\lambda)\in[-\epsilon_0,\epsilon_0]\times I$. As magnetic pseudodifferential operator we can then extend it to $\mathfrak{E}(\epsilon,\lambda)\in\mathfrak{B}\big(\mathscr{S}^\prime(\X;\mathcal{K}_0);\mathscr{S}^\prime(\X;\mathcal{K}_{m,0})\big)$. Recalling that we have a continuous embedding $\mathfrak{L}_0\hookrightarrow\mathscr{S}^\prime(\X;\mathcal{K}_0)$ we deduce that $\mathfrak{E}(\epsilon,\lambda)\in\mathbb{B}\big(\mathfrak{L}_0;\mathscr{S}^\prime(\X;\mathcal{K}_{m,0})\big)$. We use now Proposition \ref{P.5.9} and the first equality in \eqref{6.13} and write that for any $u\in\mathfrak{L}_0$ there exists $u_0\in\mathcal{H}^\infty_{A_\epsilon}(\X)\otimes\mathcal{K}_0$ such that:
$$
\mathfrak{E}(\epsilon,\lambda)u\ =\ \underset{\gamma^*\in\Gamma^*}{\sum}\mathfrak{E}(\epsilon,\lambda)\Upsilon_{\gamma^*}u_0\ =\ \underset{\gamma^*\in\Gamma^*}{\sum}\Upsilon_{\gamma^*}
\big(\mathfrak{E}(\epsilon,\lambda)u_0\big),
$$
with convergence in the sense of tempered distributions on $\X^2$. From Proposition \ref{P.5.9} we deduce that the application $\mathfrak{L}_0\ni u\mapsto u_0\in\mathcal{H}^\infty_{A_\epsilon}(\X)\otimes\mathcal{K}_0$ is continuous uniformly with respect to $\epsilon\in[-\epsilon_0,\epsilon_0]$ and from The Composition Theorem \ref{T.A.23} we deduce that $\widetilde{\mathfrak{q}}_{m,\epsilon}\sharp^{B_\epsilon}E_{\epsilon,\lambda}\in S^0_0\big(\X;\mathbb{B}(\mathcal{K}_0)\big)$ and the proof of the Lemma ends exactly as the proof of Lemma \ref{L.6.8}.
\end{proof}

Now we shall prove a variant of Theorem \ref{T.4.3} in the frame of the Hilbert spaces $\mathfrak{V}_0$ and $\mathfrak{L}_0$.
\begin{theorem}\label{T.6.11}
We suppose verified the hypothesis of Theorem \ref{T.4.3} and use the same notations; then we have that
\begin{equation}\label{6.14}
\mathcal{P}_{\epsilon,\lambda}\,\in\,\mathbb{B}\big(\mathfrak{L}_m(\epsilon)\times\mathfrak{V}_0^N;\mathfrak{L}_0\times\mathfrak{V}_0^N\big),\quad\mathcal{E}_{\epsilon,\lambda}\,\in\,\mathbb{B}\big(\mathfrak{L}_0\times\mathfrak{V}_0^N;\mathfrak{L}_m(\epsilon)\times\mathfrak{V}_0^N\big),
\end{equation}
uniformly with respect to $(\epsilon,\lambda)\in[-\epsilon_0,\epsilon_0]\times I$. Moreover, for any $(\epsilon,\lambda)\in[-\epsilon_0,\epsilon_0]\times I$ the operator $\mathcal{P}_{\epsilon,\lambda}$ is invertible and its inverse is $\mathcal{E}_{\epsilon,\lambda}$.
\end{theorem}
\begin{proof}
The boundedness properties in \eqref{6.14} follow from Lemmas \ref{L.6.1} (a), \ref{L.6.5}, \ref{L.6.6}, \ref{L.6.7}, \ref{L.6.8}, \ref{L.6.9} and \ref{L.6.10}.

Concerning the invertibility of $\mathcal{P}_{\epsilon,\lambda}$ let us recall that  in Theorem \ref{T.4.3} we have proved that the operator $\mathcal{P}_{\epsilon,\lambda}$ considered as operator in $\mathbb{B}\big(\mathcal{K}^m_\epsilon(\X^2)\times L^2(\X;\mathbb{C}^N);\mathcal{K}(\X^2)\times L^2(\X;\mathbb{C}^N)\big)$ is invertible and its inverse is $\mathcal{E}_{\epsilon,\lambda}\in\mathbb{B}\big(\mathcal{K}(\X^2)\times L^2(\X;\mathbb{C}^N);\mathcal{K}^m_\epsilon(\X^2)\times L^2(\X;\mathbb{C}^N)\big)$.

From \eqref{4.7} we recall that $\mathcal{P}_{\epsilon,\lambda}$ is a magnetic pseudodifferential operator with symbol $\mathcal{P}_\epsilon$ of class $S^0_0\big(\X;\mathbb{B}(\mathcal{K}_{m,\xi}\times\mathbb{C}^N;\mathcal{K}_0\times\mathbb{C}^N)\big)$ uniformly with respect to $(\epsilon,\lambda)\in[-\epsilon_0,\epsilon_0]\times I$. Applying Proposition \ref{P.A.7} we deduce that
\begin{equation}\label{6.15}
\mathcal{P}_{\epsilon,\lambda}\,\in\,\mathbb{B}\big(\mathscr{S}(\X;\mathcal{K}_{m,0})\times\mathscr{S}(\X;\mathbb{C}^N);\mathscr{S}(\X;\mathcal{K}_0)\times\mathscr{S}(\X;\mathbb{C}^N)\big),
\end{equation}
and extending by continuity we also have that
\begin{equation}\label{6.16}
\mathcal{P}_{\epsilon,\lambda}\,\in\,\mathbb{B}\big(\mathscr{S}^\prime(\X;\mathcal{K}_{m,0})\times\mathscr{S}^\prime(\X;\mathbb{C}^N);\mathscr{S}^\prime(\X;\mathcal{K}_0)\times\mathscr{S}^\prime(\X;\mathbb{C}^N)\big).
\end{equation}

Similarly, the operator $\mathcal{E}_{\epsilon,\lambda}$ appearing in Theorem \ref{T.4.3} has a symbol of class $S^0_0\big(\X;\mathbb{B}(\mathcal{K}_0\times\mathbb{C}^N;\mathcal{K}_{m,\xi}\times\mathbb{C}^N)\big)$ and thus defines first an operator of the form
\begin{equation}\label{6.17}
\mathcal{E}_{\epsilon,\lambda}\,\in\,\mathbb{B}\big(\mathscr{S}(\X;\mathcal{K}_0)\times\mathscr{S}(\X;\mathbb{C}^N);\mathscr{S}(\X;\mathcal{K}_{m,0})\times\mathscr{S}(\X;\mathbb{C}^N)\big),
\end{equation}
and extending by continuity we also have that
\begin{equation}\label{6.18}
\mathcal{E}_{\epsilon,\lambda}\,\in\,\mathbb{B}\big(\mathscr{S}^\prime(\X;\mathcal{K}_0)\times\mathscr{S}^\prime(\X;\mathbb{C}^N);\mathscr{S}^\prime(\X;\mathcal{K}_{m,0})\times\mathscr{S}^\prime(\X;\mathbb{C}^N)\big).
\end{equation}

From the first inclusion in Lemma \ref{L.2.16} it follows that $\mathscr{S}(\X;\mathcal{K}_{m,0})\hookrightarrow\mathcal{K}^m_\epsilon(\X^2)$, so that from the invertibility implied by Theorem \ref{T.4.3} (see above in this proof), it also follows that the operator $\mathcal{P}_{\epsilon,\lambda}$ appearing in \eqref{6.15} is invertible and its inverse is the operator $\mathcal{E}_{\epsilon,\lambda}$ appearing in \eqref{6.17}. As both operators $\mathcal{P}_{\epsilon,\lambda}$ and $\mathcal{E}_{\epsilon,\lambda}$ are symmetric, by duality we deduce that also the operators appearing in \eqref{6.16} and \eqref{6.18} are the inverse of one another. This property, together with the embeddings $\mathfrak{L}_m(\epsilon)\hookrightarrow\mathscr{S}^\prime(\X;\mathcal{K}_{m,0})$ given by Lemma \ref{L.5.14}, $\mathfrak{L}_0\hookrightarrow\mathscr{S}^\prime(\X;\mathcal{K}_o)$ given by Lemma \ref{L.5.7} and $\mathfrak{V}_0\hookrightarrow\mathscr{S}^\prime(\X)$ given by Lemma \ref{L.5.2} allow us to end the proof of the Theorem.
\end{proof}

We come now to the proof of the main result of this paper.
\begin{description}
\item[Proof of the Theorem \ref{T.0.1}] We proceed exactly as in the proof of Corollary \ref{C.4.5}. We start from the equality
$$
\mathcal{P}_{\epsilon,\lambda}\,\mathcal{E}_{\epsilon,\lambda}\ =\ \left(
\begin{array}{cc}
\id_{\mathfrak{L}_0}&0\\
0&\id_{\mathfrak{V}_0^N}
\end{array}
\right)
$$
and use the fact that $\widetilde{P}_\epsilon^{\prime\prime\prime}$ is a self-adjoint operator in $\mathfrak{L}_0$ that is unitarily equivalent with $P_{\epsilon}$ (by Lemma \ref{L.6.1}) so that we deduce that $\sigma\big(\widetilde{P}_\epsilon^{\prime\prime\prime}\big)=\sigma\big(P_{\epsilon}\big)$. Then we can write that
\begin{equation}\label{6.19}
0\notin\sigma\big(\mathfrak{E}_{-+}(\epsilon,\lambda)\big)\ \Longrightarrow\ \lambda\notin\sigma\big(\widetilde{P}_\epsilon^{\prime\prime\prime}\big),\ \text{and}\ \big(\widetilde{P}_\epsilon^{\prime\prime\prime}-\lambda\big)^{-1}\ =\ \mathfrak{E}(\epsilon,\lambda)\,-\,\mathfrak{E}_{+,\epsilon}(\epsilon,\lambda)\mathfrak{E}_{-+}(\epsilon,\lambda)^{-1}\mathfrak{E}_{-,\epsilon}(\epsilon,\lambda)
\end{equation}
\begin{equation}\label{6.20}
\lambda\notin\sigma\big(\widetilde{P}_\epsilon^{\prime\prime\prime}\big)\ \Longrightarrow\ 0\notin\sigma\big(\mathfrak{E}_{-+}(\epsilon,\lambda)\big),\ \text{and}\ \mathfrak{E}_{-+}(\epsilon,\lambda)^{-1}\ =\ -\mathfrak{R}_{+,\epsilon}\big(\widetilde{P}_\epsilon^{\prime\prime\prime}-\lambda\big)^{-1}\mathfrak{R}_{-,\epsilon}.
\end{equation}
In conclusion we have obtained that $\lambda\in\sigma\big(\widetilde{P}_\epsilon^{\prime\prime\prime}\big)\ \Leftrightarrow\ 0\in\sigma\big(\mathfrak{E}_{-+}(\epsilon,\lambda)\big)$ and this implies that $\lambda\in\sigma\big(P_\epsilon\big)\ \Leftrightarrow\ 0\in\sigma\big(\mathfrak{E}_{-+}(\epsilon,\lambda)\big)$.\cqfd
\end{description}

An imediate consequence of Theorem \ref{T.0.1} is the following result concerning the stability of spectral gaps for the operator $P_{\epsilon}$:
\begin{description}
\item[Proof of Corollary {C.0.2}] We apply Theorem \ref{T.0.1} and the arguments from its proof above, taking $I=K$ and $\epsilon_0>0$ sufficiently small. Knowing that $\text{\sf dist}\big(K,\sigma\big(P_0\big)\big)>0$, we deduce that we also have $\text{\sf dist}\big(K,\sigma\big(\widetilde{P}_0^{\prime\prime\prime}\big)\big)>0$ and thus we have the estimation:
\begin{equation}\label{6.21}
\underset{\lambda\in K}{\sup}\left\|\big(\widetilde{P}_0^{\prime\prime\prime}-\lambda\big)^{-1}\right\|_{\mathbb{B}(\mathfrak{L}_0)}\ <\ \infty.
\end{equation}
From \eqref{6.20} we deduce that 
$$
\lambda\in K\ \Longrightarrow\ 0\notin\sigma\big(\mathfrak{E}_{-+}(0,\lambda)\big)\ \text{and}\ \mathfrak{E}_{-+}(0,\lambda)^{-1}\ =\ -\mathfrak{R}_{+,0}\big(\widetilde{P}_0^{\prime\prime\prime}-\lambda\big)^{-1}\mathfrak{R}_{-,0},
$$
and thus
\begin{equation}\label{6.22}
\underset{\lambda\in K}{\sup}\left\|\mathfrak{E}_{-+}(0,\lambda)^{-1}\right\|_{\mathbb{B}(\mathfrak{V}_0^N)}\ <\ \infty.
\end{equation}
From Theorem \ref{T.4.3} it follows that for any $(\epsilon,\lambda)\in[-\epsilon_0,\epsilon_0]\times K$:
\begin{equation}\label{6.23}
\mathfrak{E}_{-+}(\epsilon,\lambda)\ =\ \mathfrak{E}_{-+}(0,\lambda)\,+\,\mathfrak{S}_{-+}(\epsilon,\lambda),\qquad\mathfrak{S}_{-+}(\epsilon,\lambda):=\mathfrak{Op}^{A_\epsilon}\big(S^{-+}_{\epsilon,\lambda}\big),
\end{equation}
\begin{equation}\label{6.24}
\underset{\epsilon\rightarrow0}{\lim}S^{-+}_{\epsilon,\lambda}\ =\ 0\ \text{in}\ S^0\big(\X;\mathbb{B}(\mathcal{C}^N)\big),
\end{equation}
uniformly with respect to $\lambda\in K$.

We notice that the symbol $S^{-+}_{\epsilon,\lambda}(x,\xi)$ is $\Gamma^*$-periodic in the second variable $\xi\in\X^*$, so that from Lemma \eqref{L.6.2} we deduce that
\begin{equation}\label{6.25}
\underset{\epsilon\rightarrow0}{\lim}\left\|\mathfrak{S}_{-+}(\epsilon,\lambda)\right\|_{\mathbb{B}(\mathfrak{V}_0^N)}\ =\ 0,
\end{equation}
uniformly with respect to $\lambda\in K$. 

From \eqref{6.22}, \eqref{6.23} and \eqref{6.25} imply that for $\epsilon_0>0$ sufficiently small, the magnetic pseudodifferential operator $\mathfrak{E}_{-+}(\epsilon,\lambda)$ is invertible in $\mathbb{B}(\mathfrak{V}_0)$ for any $(\epsilon,\lambda)\in[-\epsilon_0,\epsilon_0]\times K$; in conclusion $0\notin\sigma\big(\mathfrak{E}_{-+}(\epsilon,\lambda)\big)$ and thus $\lambda\notin\sigma\big(P_\epsilon\big)$ for any $(\epsilon,\lambda)\in[-\epsilon_0,\epsilon_0]\times K$.
\end{description}

The arguments elaborated in the proof of Corollary \ref{C.0.2} allow to obtain an interesting relation between the spectra of the operators $P_\epsilon$ and $P_0$, under some stronger hypothesis. 

\noindent{\bf Hypothesis I.1} Under the conditions of Hypothesis H.1 we suppose further that for any pair $(j,k)$ of indices between $1$ and $d$ the family $\{\epsilon^{-1}B_{\epsilon,jk}\}_{0<|\epsilon|\leq\epsilon_0}$ are bounded subsets of $BC^\infty(\X)$.

\noindent{\bf Hypothesis I.2} We suppose that $p_\epsilon(x,y,\eta)=p_0(y,\eta)+r_\epsilon(x,y,\eta)$ where $p_0$ is a real valued symbol from $S^m_1(\mathbb{T})$ with $m>0$ and the family $\{\epsilon^{-1}r_\epsilon\}_{0<|\epsilon|\leq\epsilon_0}$ is a bounded subset of $S^m_1(\X\times\mathbb{T})$, each symbol $r_\epsilon$ being real valued.

\noindent{\bf Hypothesis I.3} The symbol $p_0$ is elliptic; i.e. there exist $C>0$, $R>0$ such that $p_0(y,\eta)\geq C|\eta|^m$ for any $(y,\eta)\in\Xi$ with $|\eta|\geq R$.

\begin{remark}\label{R.6.12}
If we come back to the proofs of Theorem \ref{T.4.3}, Theorem \ref{T.A.23} (of Composition) and Proposition \ref{P.A.27} and suppose the Hypothesis I.1 - I.3 to be true, we can prove the following fact that extends our property \eqref{4.12}:
\begin{equation}\label{6.26}
\begin{array}{l}
\forall I\subset\mathbb{R}\ \text{compact interval, }\exists\epsilon_0>0,\,\exists N\in\mathbb{N},\ \text{such that:}\\
\left\{
\begin{array}{l}
\forall(\epsilon,\lambda)\in[-\epsilon_0,\epsilon_0]\times I,\qquad\mathfrak{E}_{-+}(\epsilon,\lambda)\ =\ \mathfrak{E}_{-+}(0,\lambda)\,+\,\mathfrak{S}_{-+}(\epsilon,\lambda),\qquad\mathfrak{S}_{-+}(\epsilon,\lambda):=\mathfrak{Op}^{A_\epsilon}\big(S^{-+}_{\epsilon,\lambda}\big),\\
\text{the familly }\left\{\epsilon^{-1}S^{-+}_{\epsilon,\lambda}\right\}_{(|\epsilon|,\lambda)\in(0,\epsilon_0]\times I}\ \text{is a bounded subset of }S^0\big(\X;\mathbb{B}(\mathbb{C}^N)\big).
\end{array}
\right.
\end{array}
\end{equation}
\end{remark}
Once again we notice the $\Gamma^*$-periodicity of the symbol $S^{-+}_{\epsilon,\lambda}(x,\xi)$ with respect to the variable $\xi\in\X^*$ and from Lemma \ref{L.6.2} we deduce that there exists a strictly positive constant $C_1$ such that the following estimation is true:
\begin{equation}\label{6.27}
\left\|\mathfrak{S}_{-+}(\epsilon,\lambda)\right\|_{\mathbb{B}(\mathfrak{V}_0^N)}\ \leq\ C_1\epsilon,\qquad\forall(\epsilon,\lambda)\in[-\epsilon_0,\epsilon_0]\times I.
\end{equation}
Using Lemma \ref{L.6.6} and \ref{L.6.9} we conclude that there exists a strictly positive constant $C_2$ such that the following estimation is true:
\begin{equation}\label{6.28}
\left\|\mathfrak{R}_{+,\epsilon}\right\|_{\mathbb{B}(\mathfrak{L}_0;\mathfrak{V}_0^N)}\ +\ \left\|\mathfrak{R}_{-,\epsilon}\right\|_{\mathbb{B}(\mathfrak{V}_0^N;\mathfrak{L}_m(\epsilon))}\ \leq\ C_2\qquad\forall\epsilon\in[-\epsilon_0,\epsilon_0].
\end{equation}

\begin{description}
\item[Proof of Proposition \ref{P.6.12}]
For $M\subset\mathbb{R}$ and $\delta>0$ we use the notation $M_\delta:=\{t\in\mathbb{R}\,\mid\,{\sf dist}(t,M)\leq\delta\}$. Then we have to prove the following inclusions:
\begin{equation}\label{6.29}
\sigma\big(P_\epsilon\big)\cap I\,\subset\,\sigma\big(P_0\big)_{C\epsilon}\cap I,\qquad\forall\epsilon\in[0,\epsilon_0].
\end{equation}
\begin{equation}\label{6.34}
\sigma\big(P_0\big)\cap I\,\subset\,\sigma\big(P_\epsilon\big)_{C\epsilon}\cap I,\qquad\forall\epsilon\in[0,\epsilon_0].
\end{equation}
Suppose there exists $\lambda\in I$ such that ${\sf dist}\big(\lambda,\sigma\big(P_0\big)\big)>C\epsilon$. From Lemma \ref{L.6.1} we know that $\sigma\big(P_0\big)=\sigma\big(\widetilde{P}_0^{\prime\prime\prime}\big)$ so that we deduce that ${\sf dist}\big(\lambda,\sigma\big(\widetilde{P}_0^{\prime\prime\prime}\big)\big)>C\epsilon$ and conclude that:
\begin{equation}\label{6.30}
\left\|\big(\widetilde{P}_0^{\prime\prime\prime}-\lambda\big)^{-1}\right\|_{\mathbb{B}(\mathfrak{L}_0)}\ \leq\ (C\epsilon)^{-1}.
\end{equation}
From \eqref{6.20} we also deduce that $0\notin\sigma\big(\mathfrak{E}_{-+}(0,\lambda)\big)$ and $\mathfrak{E}_{-+}(0,\lambda)^{-1}=-\mathfrak{R}_{+,0}\big(\widetilde{P}_0^{\prime\prime\prime}-\lambda\big)^{-1}\mathfrak{R}_{-,0}$. Using these facts together with \eqref{6.28} and \eqref{6.30} we obtain the estimation:
\begin{equation}\label{6.31}
\left\|\mathfrak{E}_{-+}(0,\lambda)^{-1}\right\|_{\mathbb{B}(\mathfrak{V}_0^N)}\ \leq\ C_2^2(C\epsilon)^{-1}.
\end{equation}
Using \eqref{6.27} and \eqref{6.31} we also obtain the following estimation:
\begin{equation}\label{6.32}
\left\|\mathfrak{E}_{-+}(0,\lambda)^{-1}\right\|_{\mathbb{B}(\mathfrak{V}_0^N)}\cdot\left\|\mathfrak{S}_{-+}(\epsilon,\lambda)\right\|_{\mathbb{B}(\mathfrak{V}_0^N)}\ \leq\ C_1C_2^2C^{-1},\qquad\forall\epsilon\in[-\epsilon_0,\epsilon_0].
\end{equation}
If we choose now $C>0$ such that $C>C_1C_2^2$, we notice that the operator $\mathfrak{E}_{-+}(\epsilon,\lambda)=\mathfrak{E}_{-+}(0,\lambda)+\mathfrak{S}_{-+}(\epsilon,\lambda)$ is invertible in $\mathbb{B}(\mathfrak{V}_0^N)$ and thus we deduce that $0\notin\sigma\big(\mathfrak{E}_{-+}(\epsilon,\lambda)\big)$. It follows then that $\lambda\notin\sigma\big(P_\epsilon\big)$ for any $\epsilon\in[-\epsilon_0,\epsilon_0]$ and the inclusion \eqref{6.29} follows. 

Let us prove now \eqref{6.34}. Let us suppose that for some $\epsilon$ with $|\epsilon|\in(0,\epsilon_0]$ there exists some $\lambda\in I$ such that ${\sf dist}\big(\lambda,\sigma\big(P_\epsilon\big)\big)>C\epsilon$. Recalling that $\sigma\big(P_\epsilon\big)=\sigma\big(\widetilde{P}_\epsilon^{\prime\prime\prime}\big)$ we deduce that we also have ${\sf dist}\big(\lambda,\sigma\big(\widetilde{P}_\epsilon^{\prime\prime\prime}\big)\big)>C\epsilon$ and thus
\begin{equation}\label{6.35}
\left\|\big(\widetilde{P}_\epsilon^{\prime\prime\prime}-\lambda\big)^{-1}\right\|_{\mathbb{B}(\mathfrak{L}_0)}\ \leq\ (C\epsilon)^{-1}.
\end{equation}
We also deduce that $0\notin\sigma\big(\mathfrak{E}_{-+}(\epsilon,\lambda)\big)$ and $\mathfrak{E}_{-+}(\epsilon,\lambda)^{-1}=-\mathfrak{R}_{+,\epsilon}\big(\widetilde{P}_\epsilon^{\prime\prime\prime}-\lambda\big)^{-1}\mathfrak{R}_{-,\epsilon}$. Using these facts together with \eqref{6.28} and \eqref{6.30} we obtain the estimation:
\begin{equation}\label{6.31.a}
\left\|\mathfrak{E}_{-+}(\epsilon,\lambda)^{-1}\right\|_{\mathbb{B}(\mathfrak{V}_0^N)}\ \leq\ C_2^2(C\epsilon)^{-1}.
\end{equation}
It follows like above that the operator $\mathfrak{E}_{-+}(0,\lambda)=\mathfrak{E}_{-+}(\epsilon,\lambda)-\mathfrak{S}_{-+}(\epsilon,\lambda)$ is invertible in $\mathbb{B}(\mathfrak{V}_0^N)$ and thus we deduce that $0\notin\sigma\big(\mathfrak{E}_{-+}(0,\lambda)\big)$. It follows then that $\lambda\notin\sigma\big(P_0\big)$ and the inclusion \eqref{6.34} follows.  
\end{description}
\begin{remark}\label{R.6.12.a}
The relations \eqref{6.29} and \eqref{6.34} clearly imply that the boundaries of the spectral gaps of the operator $P_\epsilon$ are Lipshitz functions of $\epsilon$ in $\epsilon=0$.
\end{remark}

\section{Some particular situations}
\setcounter{equation}{0}
\setcounter{theorem}{0}

\subsection{The simple spectral band}

In this subsection we shall find some explicit forms for the principal part of the effective Hamiltonian $\mathfrak{E}_{-+}(\epsilon,\lambda)$. We shall suppose the Hypothesis H.1 - H.6 to be satisfied.

For the beginning we concentrate on the operator $P_0=\mathfrak{Op}(p_0)$ with $p_0\in S^m_1(\mathbb{T})$ a real elliptic symbol. We know tht $P_0$ has a self-adjoint realization as operator acting in $L^2(\X)$ with the domain $\mathcal{H}^m(\X)$ (the usual Sobolev space of order $m$). From Lemma \ref{L.A.18} we obtain that $\tau_\gamma P_0=P_0\tau_\gamma$, $\forall\gamma\in\Gamma$ and thus we can use the Floquet theory in order to study the spectrum of the operator $P_0$.

We shall consider the following spaces similar to the one used in Section \ref{S.2} but with one variable less:
\begin{equation}\label{7.1}
\mathscr{S}^\prime_\Gamma(\Xi)\ :=\ \left\{v\in\mathscr{S}^\prime(\Xi)\,\mid\,v(y+\gamma,\eta)=e^{i<\eta,\gamma>}v(y,\eta)\ \forall\gamma\in\Gamma,\ v(y,\eta+\gamma^*)=v(y,\eta)\ \forall\gamma\in\Gamma^*\right\}
\end{equation}
endowed with the topology induced from $\mathscr{S}^\prime(\Xi)$.
\begin{equation}\label{7.2}
\mathscr{F}_0(\Xi)\ :=\ \left\{v\in L^2_{\text{\sf loc}}(\Xi)\cap\mathscr{S}^\prime_\Gamma(\Xi)\,\mid\,v\in L^2(E\times E^*)\right\}
\end{equation}
that is a Hilbert space for the norm $\|v\|_{\mathscr{F}_0(\Xi)}:=\left(|E^*|^{-1}\int_E\int_{E^*}|v(x,\xi)|^2\,dx\,d\xi\right)^{1/2}$.
\begin{equation}\label{7.4}
\forall s\in\mathbb{R},\qquad\mathscr{F}_s(\Xi)\ :=\ \left\{v\in\mathscr{S}^\prime_\Gamma\,\mid\,\big(<D>^s\otimes\id\big)v\in\mathscr{F}_0(\Xi)\right\},
\end{equation}
that is a Hilbert space with the norm $\|v\|_{\mathscr{F}_s(\Xi)}:=\|\big(<D>^s\otimes\id\big)v\|_{\mathscr{F}_0(\Xi)}$.

The following Lemma and its prof are completely similar with the case discussed in Section \ref{S.2}.
\begin{lemma}\label{L.7.1}
With the above notations the following statements are true:
\begin{enumerate}
\item The operator $\mathcal{U}_0:L^2(\X)\rightarrow\mathscr{F}_0(\Xi)$ defined by 
\begin{equation}\label{7.5}
\big(\mathcal{U}_0u\big)(x,\xi)\ :=\ \underset{\gamma\in\Gamma}{\sum}e^{i<\xi,\gamma>}u(x-\gamma),\qquad\forall(x,\xi)\in\Xi,
\end{equation}
is a unitary operator. The inverse of the operator $\mathcal{U}_0$ is the operator $\mathcal{W}_0:\mathscr{F}_0(\Xi)\rightarrow L^2(\X)$ defined as
\begin{equation}\label{7.6}
\big(\mathcal{W}_0v\big)(x)\ :=\ |E^*|^{-1}\int_{E^*}v(x,\xi)\,d\xi,\qquad\forall x\in\X.
\end{equation}
\item For any $s\in\mathbb{R}$ the restriction of the operator $\mathcal{U}_0$ to $\mathcal{H}^s(\X)$ defines a unitary operator $\mathcal{U}_0:\mathcal{H}^s(\X)\rightarrow\mathscr{F}_s(\Xi)$.
\end{enumerate}
\end{lemma}

\begin{lemma}\label{L.7.2}
With the above notations the following statements are true:
\begin{enumerate}
\item The operator $\widehat{P}_0:=P_0\otimes\id$ leaves invariant the subspace $\mathscr{S}^\prime_\Gamma(\Xi)$.
\item Considered as an unbounded operator in the Hilbert space $\mathscr{F}_0(\Xi)$, the operator $\widehat{P}_0$ is self-adjoint and lower semi-bounded on the domain $\mathscr{F}_m(\Xi)$ and is unitarily equivalent to the operator $P_0$.
\end{enumerate}
\end{lemma}
\begin{proof}
The first conclusion follows easily from Lemma \ref{L.A.18}. For the second conclusion let us notice that \eqref{7.5} implies the following equality:
\begin{equation}\label{7.7}
\mathcal{U}_0P_0u\ =\ \widehat{P}_0\mathcal{U}_0u,\qquad\forall u\in\mathcal{H}^m(\X),
\end{equation}
that together withLemma \ref{L.7.1} implies that $\widehat{P}_0$ and $P_0$ are unitarily equivalent.
\end{proof}

Let us notice that for any $v\in\mathscr{F}_0(\Xi)$ and for any $\xi\in\X^*$, the restriction $v(\cdot,\xi)$ defines a function on $\X$ that is of class $\mathscr{F}_{0,\xi}$ and by Remark \ref{R.A.17} this last space is canonically unitarily equivalent with $L^2(E)$; moreover we have the equality
$$
\|v\|_{\mathscr{F}_0(\Xi)}\ =\ |E^*|^{-1/2}\left(\int_{E^*}\|v(\cdot,\xi)\|^2_{\mathscr{F}_{0,\xi}}d\xi\right)^{1/2}.
$$
The following periodicity is evidently true: $\mathscr{F}_{0,\xi}=\mathscr{F}_{0,\xi+\gamma^*}$ for any $\gamma^*\in\Gamma^*$. Moreover, it is easy to see that we can make the following identification:
$$
\mathscr{F}_0(\Xi)\ \equiv\ \int_{\mathbb{T}^{*,d}}^\oplus\mathscr{F}_{0,\xi}\,d\xi.
$$
Similarly, we also have the following relations:
$$
\mathscr{F}_{m,\xi+\gamma^*}=\mathscr{F}_{m,\xi},\ \forall\gamma^*\in\Gamma^*;\qquad\mathscr{F}_m(\Xi)\ \equiv\ \int_{\mathbb{T}^{*,d}}^\oplus\mathscr{F}_{m,\xi}\,d\xi.
$$

Taking into account the Remark \ref{R.A.22} we notice that for any $\xi\in\X^*$ the operator $P_0$ induces on the Hilbert space $\mathscr{F}_{0,\xi}$ a self-adjoint operator with domain $\mathscr{F}_{m,\xi}$ that we shalll denote by $\widehat{P}_0(\xi)$; we evidently have the periodicity $\widehat{P}_0(\xi+\gamma^*)=\widehat{P}_0(\xi)$ for any $\gamma^*\in\Gamma^*$.

If we idntify $\mathcal{K}_0$ with $L^2_{\text{\sf loc}}\cap\mathscr{S}^\prime_\Gamma(\Xi)\equiv L^2(E)$, the same Remark \ref{R.A.22} implies that the operator $\widehat{P}_0(\xi)$ is unitarily equivalent with the operator $\check{P}_0(\xi)$ that is induced by $\mathfrak{Op}\big((\id\otimes\tau_{-\xi})p\big)$ on the space $\mathcal{K}_0$; this is a self-adjoint lower semi-bounded operator on the domain $\mathcal{K}_{m,\xi}$ (identified with $\mathcal{H}^m_{\text{\sf loc}}(\X)\cap\mathscr{S}^\prime_\Gamma(\X)$, with the norm $\|<D+\xi>^m\cdot\|_{L^2(E)}$). More precisely, we can write that
\begin{equation}\label{7.8}
\check{P}_0(\xi)\ =\ \sigma_{-\xi}\widehat{P}_0(\xi)\sigma_\xi,\qquad\forall\xi\in\X^*.
\end{equation}

\begin{lemma}\label{L.7.3}
For any ${\rm z}\in\mathbb{C}\setminus\overline{\underset{\xi\in\X^*}{\cup}\sigma\big(\check{P}_0(\xi)\big)}$ the application
\begin{equation}\label{7.8.a}
\X^*\ni\xi\mapsto\big(\check{P}_0(\xi)\,-\,{\rm z}\big)^{-1}\in\mathbb{B}(\mathcal{K}_0)
\end{equation}
is of class $C^\infty(\X^*)$.
\end{lemma}
\begin{proof}
From the Example \ref{E.A.20} follows that the application
$$
\X^*\ni\xi\mapsto\big(\check{P}_0(\xi)\,-\,{\rm z}\big)^{-1}\in\mathbb{B}(\mathcal{K}_{m,0};\mathcal{K}_0)
$$
is of class $C^\infty(\X^*)$. But from the second rezolvent equality:
$$
\big(\check{P}_0(\xi)\,-\,{\rm z}\big)^{-1}\,-\,\big(\check{P}_0(\xi_0)\,-\,{\rm z}\big)^{-1}\ =\ \big(\check{P}_0(\xi)\,-\,{\rm z}\big)^{-1}\big(\check{P}_0(\xi_0)\,-\,\check{P}_0(\xi)\big)\big(\check{P}_0(\xi_0)\,-\,{\rm z}\big)^{-1}
$$
follows the continuity of the application \eqref{7.8.a} and the existence of the derivatives follows by usual arguments.
\end{proof}

\begin{lemma}\label{L.7.4}
The following equality is true:
\begin{equation}\label{7.9}
\widehat{P}_0\ =\ \int_{\mathbb{T}_*}^\oplus\widehat{P}_0(\xi)\,d\xi.
\end{equation}
\end{lemma}
\begin{proof}
First let us notice the equality
\begin{equation}\label{7.10}
\big(\widehat{P}_0u\big)(\cdot,\xi)\ =\ \widehat{P}_0(\xi)u(\cdot,\xi),\qquad\forall\xi\in\mathbb{T}_*,\ \forall u\in\mathscr{F}_0(\Xi).
\end{equation}
Then, from \eqref{7.8} we deduce the equality:
$$
\big(\widehat{P}_0(\xi)\,+\,i\big)^{-1}\ =\ \sigma_\xi\big(\check{P}_0(\xi)\,+\,i\big)^{-1}\sigma_{-\xi},\qquad\forall\xi\in\X^*
$$
and from Lemma \ref{L.7.3} we obtain the continuity of the application
$$
\mathbb{T}_*\ni\xi\mapsto\big(\widehat{P}_0(\xi)\,+\,i\big)^{-1}\in\mathbb{B}\big(L^2(E)\big),
$$
and that together with \eqref{7.10} imply equality \eqref{7.9}.
\end{proof}

\begin{remark}\label{R.7.5}
Let us notice that $\mathcal{K}_{m,\xi}$ is compactly embedded into $\mathcal{K}_0$ and thus, the operator $\check{P}_0(\xi)$ has compact rezolvent for any $\xi\in\X^*$; it is clearly lower semi-bounded uniformly with respect to $\xi\in\X^*$ (taking into account that $\check{P}_0(\xi+\gamma^*)=\sigma_{-\gamma^*}\check{P}_0(\xi)\sigma_{\gamma^*},\ \forall\gamma^*\in\Gamma^*$). We deduce that $\sigma\big(\widehat{P}_0(\xi)\big)=\sigma\big(\check{P}_0(\xi)\big)=\{\lambda_j(\xi)\}_{j\geq1}$, where for any $\xi\in\X^*$ and any $j\geq1$, $\lambda_j(\xi)$ is a real finitely degenerated eigenvalue and $\underset{j\rightarrow\infty}{\lim}\lambda_j(\xi)=\infty\ \forall\xi\in\X^*$; we can always renumber the eigenvalues and suppose that $\lambda_j(\xi)\leq\lambda_{j+1}(\xi)$ for any $j\geq1$ and for any $\xi\in\X^*$. Due to the $\Gamma^*$-periodicity of $\widehat{P}(\xi)$ we have that $\lambda_j(\xi+\gamma^*)=\lambda_j(\xi)$ for any $j\geq1$, for any $\xi\in\X^*$ and for any $\gamma^*\in\Gamma^*$. These are the {\it Floquet eigenvalues of the operator $\widehat{P}_0$}.
\end{remark}

\begin{lemma}\label{L.7.6}
For each $j\geq1$ the function $\mathbb{T}^{*,d}\ni\xi\mapsto\lambda_j(\xi)\in\mathbb{R}$ is continuous on $\mathbb{T}^{*,d}$ uniformly in $j\geq1$.
\end{lemma}
\begin{proof}
From the uniform lower semi-boundedness it follows the existence of some real number $c\in\mathbb{R}$ such that $\lambda_j(\xi)\geq\,c+1$ for any $j\geq1$ and for any $\xi\in\X^*$. We can thus define the operators $R(\xi):=\big(\check{P}_0(\xi)-c\big)^{-1}$, for any $\xi\in\X^*$ and due to the result in Lemma \ref{L.7.3} we obtain a function of class $C^\infty\big(\X^*;\mathbb{B}(\mathcal{K}_0)\big)$; for any $\xi\in\X^*$ the operator $R(\xi)$ is a bounded self-adjoint operator on $\mathcal{K}_0$ and $\sigma\big(R(\xi)\big)=\{\big(\lambda_j(\xi)-c\big)^{-1}\}_{j\geq1}$. Applying now the {\it Min-Max Principle} (see \cite{RS-4}) we obtain easily that:
$$
\left|\big(\lambda_j(\xi)-c\big)^{-1}\,-\,\big(\lambda_j(\xi_0)-c\big)^{-1}\right|\ \leq\ \left\|R(\xi)\,-\,R(\xi_0)\right\|_{\mathbb{B}(\mathcal{K}_0)},\qquad\forall(\xi,\xi_0)\in[\X^*]^2,\ \forall j\geq1.
$$
\end{proof}

\begin{lemma}\label{L.7.7}
We have the following spectral decomposition: $\sigma\big(P_0\big)=\sigma\big(\widehat{P}_0\big)=\overset{\infty}{\underset{k=1}{\cup}}J_k$ with $J_k:=\lambda_k\big(\mathbb{T}^{*,d}\big)$ is a compact interval in $\mathbb{R}$.
\end{lemma}
\begin{proof}
Lemma \ref{L.7.2} states that the operators $P_0$ and $\widehat{P}_0$ are unitarily equivalent and thus they have the same spectrum. From Theorem XIII.85 (d) in \cite{RS-4} it follows that:
$$
\lambda\,\in\,\sigma\big(\widehat{P}_0\big)\quad\Longleftrightarrow\quad\forall\epsilon>0,\ \left|\left\{\xi\in\mathbb{T}^{*,d}\,\mid\,\sigma\big(\check{P}_0(\xi)\big)\cap(\lambda-\epsilon,\lambda+\epsilon)\ne\emptyset\right\}\right|\,>\,0.
$$
Let us denote by $M:=\overset{\infty}{\underset{k=1}{\cup}}J_k$. If $\lambda\in M$, it exist $\xi_0\in\mathbb{T}^{*,d}$ and $k\geq1$ such that $\lambda=\lambda_k(\xi_0)$. Due to the continuity of $\lambda_k(\xi)$ we know that for any $\epsilon>0$ there exists a neighborhood $V$ of $\xi_0$ in $\mathbb{T}^{*,d}$ such that $|\lambda_k(\xi)-\lambda|\leq\epsilon$ for any $\xi\in V$. It follows that $\lambda\in\sigma\big(\widehat{P}_0\big)$ and thus $M\subset\sigma\big(\widehat{P}_0\big)$.

On the other hand it is evident by definitions that $\sigma\big(\widehat{P}_0\big)\subset\overline{M}$ and thus we need to prove that $M$ is a closed set. Let us fix some $\lambda\in\overline{M}$; it follows that there exists a sequence $\{\mu_l\}_{l\geq1}\subset M$ such that $\mu_l\underset{l\rightarrow\infty}{\rightarrow}\lambda$. We deduce that for any $l\geq1$ there exists a point $\xi^l\in\mathbb{T}^{*,d}$ and an index $k_l\geq1$ such that $\mu_l=\lambda_{k_l}(\xi^l)$. The manifold $\mathbb{T}^{*,d}$ being compact it follows that we may suppose that it exists $\xi\in\mathbb{T}^{*,d}$ such that $\xi^l\underset{l\rightarrow\infty}{\rightarrow}\xi$. Due to the uniform continuity of the functions $\lambda_j$ with respect to $j\geq1$ we deduce that $\lambda_{k_l}(\xi)\underset{l\rightarrow\infty}{\rightarrow}\lambda$. Taking into account that $\lambda_j(\xi)\underset{j\rightarrow\infty}{\rightarrow}\infty$  for any $\xi\in\mathbb{T}^{*,d}$ we deduce that the sequence $\{j_l\}_{l\geq1}$ becomes constant above some rank. We conclude that $\lambda=\lambda_{j_l}(\xi)$ for $l$ sufficiently large and that means that $\lambda\in M$.
\end{proof}

If we suppose that Hypothesis H.7 is satisfied, i.e. there exists $k\geq1$ such that $J_k$ is a simple spectral band for $P_0$, then we have some more regularity for the Floquet eigenvalue $\lambda_k(\xi)$.
\begin{lemma}\label{L.7.8}
Under Hypothesis H.7, if $J_k$ is a simple spectral band for $P_0$, then the function $\lambda_k(\xi)$ is of class $C^\infty(\mathbb{T}^{*,d})$.
\end{lemma}
\begin{proof}
Let us fix a circle $\mathscr{C}$ in the complex plane having its center on the real axis and such that: $J_k$ is contained in the open interior domain delimited by $\mathscr{C}$ and all the other spectral bands $J_l$ with $l\ne k$ are contained in the exterior open domain delimited by $\mathscr{C}$ (that is unbounded). In particular, for such a choice we get that the distance $d\big(\mathscr{C},\sigma(\check{P}_0)\big)>0$.

With the above noations let us define the following operator:
\begin{equation}\label{7.11}
\Pi_k(\xi)\ :=\ -\frac{i}{2\pi}\oint_{\mathscr{C}}\big(\check{P}_0(\xi)-{\rm z}\big)^{-1}d{\rm z},\qquad\forall\xi\in\X^*.
\end{equation}
One verifies easily that it is an orthogonal projection on the subspace $\mathcal{N}_k(\xi):=\ker\big(\check{P}_0(\xi)-{\rm z}\big)^{-1}$, that is a complex vector space of dimension 1. Moreover, it is easy to verify using Lemma \ref{L.7.3}  that the application $\mathbb{T}^{*,d}\ni\xi\mapsto\Pi_k(\xi)\in\mathbb{B}(\mathcal{K}_0)$ is of class $C^\infty\big(\mathbb{T}^{*,d};\mathbb{B}(\mathcal{K}_0)\big)$.

Let us now fix some point $\xi_0\in\mathbb{T}^{*,d}$ and some vector $\phi(\xi_0)\in\mathcal{N}_k(\xi_0)$ having norm $\|\phi(\xi_0)\|_{\mathcal{K}_0}=1$. We can find a sufficiently small open neighborhood $V_0$ of $\xi_0$ in $\mathbb{T}^{*,d}$ such that 
$$
\left\|\Pi_k(\xi)\phi(\xi_0)\right\|_{\mathcal{K}_0}\ \geq\ \frac{1}{2},\qquad\forall\xi\in V_0.
$$
We denote by 
$$
\phi(\xi)\ :=\ \left\|\Pi_k(\xi)\phi(\xi_0)\right\|_{\mathcal{K}_0}^{-1}\Pi_k(\xi)\phi(\xi_0),\qquad\forall\xi\in V_0.
$$

For any $\xi\in V_0$ the vector $\phi(\xi)$ is a norm one vector that generates the subspace $\mathcal{N}_k(\xi)$ and $\phi\in C^\infty\big(V_0;\mathcal{K}_0\big)$.

We choose now $c\in\mathbb{R}$ as in the proof of Lemma \ref{L.7.6}. Then, for any $\xi\in V_0$ we have that
$$
\big(\lambda_k(\xi)\,-\,c\big)^{-1}\phi(\xi)\ =\ \big(\check{P}_0(\xi)\,-\,c\id\big)^{-1}\phi(\xi)
$$
and we conclude that
$$
\big(\lambda_k(\xi)\,-\,c\big)^{-1}\ =\ \left(\big(\check{P}_0(\xi)\,-\,c\id\big)^{-1}\phi(\xi)\,,\,\phi(\xi)\right)_{\mathcal{K}_0}.
$$
Using Lemma \ref{L.7.3} we conclude finally that $\lambda_k\in C^\infty(V_0)$.
\end{proof}

\begin{lemma}\label{L.7.9}
With the above definitions and notations the following statements are true:
\begin{enumerate}
\item For any $(s,\xi)\in\mathbb{R}\times\X^*$ the Hilbert spaces $\mathcal{K}_{s,\xi}$ and $\mathscr{F}_{s,\xi}$ are stable under complex conjugation.
\item $\forall\xi\in\X^*$ and $\forall\gamma^*\in\Gamma^*$ we have that $\check{P}_0(\xi+\gamma^*)=\sigma_{-\gamma^*}\check{P}_0(\xi)\sigma_{\gamma^*}$ and $\lambda_j(\xi+\gamma^*)=\lambda_j(\xi)$ for any $j\geq1$.
\item If the symbol $p_0$ verifies the property
\begin{equation}\label{7.9.a}
p_0(x,-\xi)\ =\ p_0(x,\xi)
\end{equation}
then the following relations hold:
\begin{equation}\label{7.9.e}
\overline{\check{P}_0(\xi)u}\ =\ \check{P}_0(-\xi)\overline{u},\qquad\forall u\in\mathcal{K}_{m,\xi},\ \forall\xi\in\X^*.
\end{equation}
\begin{equation}\label{7.9.f}
\lambda_j(-\xi)\ =\ \lambda_j(\xi),\qquad\forall j\geq1.
\end{equation}
\begin{equation}\label{7.9.g}
\overline{\Pi_k(\xi)u}\ =\ \Pi_k(-\xi)\overline{u},\qquad\forall u\in\mathcal{K}_0,\ \forall\xi\in\X^*,
\end{equation}
for any simple spectral band $J_k$ of $P_0$.
\end{enumerate}
\end{lemma}
\begin{proof} The first statement follows directly from the definitions \eqref{A.21} and \eqref{A.23}, while the second statement follows from Remark \ref{R.7.5}. 

As we know that $\check{P}_0(\xi)$ is induced by $P_{0,\xi}:=\mathfrak{Op}\big((\id\otimes\tau_{-\xi})p\big)$ on the Hilbert space $\mathcal{K}_0$, it is enough to prove that $\overline{P_{0,\xi}u}=P_{0,-\xi}\overline{u}$ for any $u\in\mathscr{S}(\X)$. In fact, for any $x\in\X$ we have that:
$$
\overline{\big(P_{0,\xi}u\big)(x)}\ =\ \int_\Xi e^{i<\eta,y-x>}p_0\big(\frac{x+y}{2},\xi+\eta\big)\overline{u(y)}\,dy\,\dbar\eta\ =
$$
$$
=\ \int_\Xi e^{i<\eta,x-y>}p_0\big(\frac{x+y}{2},\xi-\eta\big)\overline{u(y)}\,dy\,\dbar\eta\ =\ \int_\Xi e^{i<\eta,x-y>}p_0\big(\frac{x+y}{2},-\xi+\eta\big)\overline{u(y)}\,dy\,\dbar\eta\ =\ \big(P_{0,-\xi}\overline{u}\big)(x).
$$

Let us fix now some point $\xi\in\X^*$ and some vector $u\in\mathcal{K}_{m,\xi}$; following statement (1) of this Lemma and \eqref{7.9.e}, the vector $u$ is an eigenvector of $\check{P}_0(\xi)$ for the eigenvalue $\lambda_j(\xi)$ if and only if $\overline{u}$ is eigenvector of $\check{P}_0(-\xi)$ for the eigenvalue $\lambda_j(\xi)$. We deduce that $\{\lambda_j(-\xi)\}_{j\geq1}\,=\,\{\lambda_j(\xi)\}_{j\geq1}$; as both sequences are monotonuous we conclude that $\lambda_j(-\xi)=\lambda_j(\xi)$ for any $j\geq1$ so that we obtain \eqref{7.9.f}. 

Finally \eqref{7.9.g} follows from \eqref{7.9.e} and \eqref{7.11}.
\end{proof}

The next Lemma is very important for the construction in the Grushin problem under Hypothesis H.7; we shall prove it following the ideas in \cite{HS1}.
\begin{lemma}\label{L.7.10}
Supposing that Hypothesis H.7 is also satisfied and supposing that $p(y,-\eta)=p(y,\eta)$ for any $(y,\eta)\in\Xi$, we can construct a function $\phi$ having the following properties:
\begin{enumerate}
\item $\phi\in C^\infty(\Xi;\mathcal{K}_{lm,0})$ for any $l\in\mathbb{N}$.
\item $\phi(y+\gamma,\eta)\ =\ \phi(y,\eta),\quad\forall(y,\eta)\in\Xi,\ \forall\gamma\in\Gamma$.
\item $\phi(y,\eta+\gamma^*)\ =\ e^{-i<\gamma^*,y>}\phi(y,\eta),\quad\forall(y,\eta)\in\Xi,\ \forall\gamma^*\in\Gamma^*$.
\item $\left\|\phi(\cdot,\eta)\right\|_{\mathcal{K}_0}\ =\ 1,\qquad\forall\eta\in\X^*$.
\item $\overline{\phi(y,\eta)}\ =\ \phi(y,-\eta)\quad\forall(y,\eta)\in\Xi$.
\item $\phi(\cdot,\eta)\,\in\,\mathcal{N}_k(\eta)\,=\,\ker\big(\check{P}_0(\eta)\,-\,\lambda_k(\eta)\big),\quad\forall\eta\in\X^*$.
\end{enumerate}
\end{lemma}
\begin{proof}
First we shall construct a function $\phi\in C\big(\X^*;\mathcal{K}_0\big)$ that satisfies properties (2)-(6).

We begin by recalling that $\Pi_k\in C^\infty\big(\X;\mathbb{B}(\mathcal{K}_0)\big)$ and deducing that there exists some $\delta>0$ such that for any pair of points $(\xi^\prime,\xi^{\prime\prime})\in[\overline{E^*}]^2$ with $|\xi^\prime-\xi^{\prime\prime}|<\delta$ the following estimation is true:
\begin{equation}\label{7.12}
\left\|\Pi_k(\xi^\prime)\,-\,\Pi_k(\xi^{\prime\prime})\right\|_{\mathbb{B}(\mathcal{K}_0)}\ <\ \frac{1}{2}.
\end{equation}

We decompose the vectors of $\X^*$ with respect to the dual basis $\{e^*_j\}_{1\leq j\leq d}$ associated to $\Gamma$, and write $\xi=t_1e^*_1+\dots+t_de^*_d$ or $\xi=(t_1,\ldots,t_d)$ for any $\xi\in\X^*$. In particular we have that $\xi\in E^*$ if and only if $-(1/2)\leq t_j<(1/2)$ for any $j=1,\ldots,d$. The construction of the function $\phi$ will be done by induction on the number $d$ of variables $(t_1,\ldots,t_d)$.

We start with the case $d=1$, or equivalently, we consider only momenta of the type $(t,0)$ with $t\in\mathbb{R}$ and $0\in\mathbb{R}^{d-1}$. We choose some $n\in\mathbb{N}^*$ such that $n^{-1}<\delta$ and consider a division of the interval $[0,1/2]\subset\mathbb{R}$:
$$
0=\tau_0<\tau_1<\ldots<\tau_n=1/2,\quad\tau_j=\frac{j}{2n},\ \forall0\leq j\leq n.
$$

Let us fix some vector $\psi(0)\in\mathcal{N}_k((0,0)$ such that $\|\psi(0)\|_{\mathcal{K}_0}=1$ and $\overline{\psi(0)}=\psi(0)$; this last property may be realized noticing that Lemma \ref{L.7.9} implies that if $v\in\mathcal{N}_k((0,0))$ then also $\overline{v}\in\mathcal{N}_k((0,0))$ and if $\|\psi(0)\|_{\mathcal{K}_0}=1$ then it exists $f\in\mathbb{R}$ such that $\overline{\psi(0)}=e^{if}\psi(0)$ and we can replace $\psi(0)$ by $e^{-if/2}\psi(0)$. We use now \eqref{7.12} and deduce that
$$
\left\|\Pi_k((t,0))\psi(0)\right\|_{\mathcal{K}_0}\ \geq\ \frac{1}{2},\quad\forall t\in[0,\tau_1].
$$
Thus we can define:
\begin{equation}\label{7.14}
\psi(t)\ :=\ \left\|\Pi_k((t,0))\psi(0)\right\|_{\mathcal{K}_0}^{-1}\Pi_k((t,0))\psi(0),\quad\forall t\in[0,\tau_1].
\end{equation}
We notice that this function verifies the following properties:
\begin{equation}\label{7.13}
\psi\in C\big([0,\tau_1];\mathcal{K}_0\big);\quad\overline{\psi(0)}=\psi(0);\quad\|\psi(t)\|_{\mathcal{K}_0}=1,\ \forall t\in[0,\tau_1];\quad\psi(t)\in\mathcal{N}_k((t,0)),\ \forall  t\in[0,\tau_1].
\end{equation}

But let us notice that we can use \eqref{7.12} once more starting with $\psi(\tau_1)$, noticing that we also have
$$
\left\|\Pi_k((t,0))\psi(\tau_1)\right\|_{\mathcal{K}_0}\ \geq\ \frac{1}{2},\quad\forall t\in[\tau_1,\tau_2].
$$
and defining
\begin{equation}\label{7.14.b}
\psi(t)\ :=\ \left\|\Pi_k((t,0))\psi(\tau_1)\right\|_{\mathcal{K}_0}^{-1}\Pi_k((t,0))\psi(\tau_1),\quad\forall t\in[\tau_1,\tau_2].
\end{equation}
We notice that $\underset{t\searrow\tau_1}{\lim}\psi(t)=\psi(\tau_1)$ and conclude that the function $\psi:[0,\tau_2]\rightarrow\mathcal{K}_0$, defined by \eqref{7.14} and \eqref{7.14.b} is continuous in $t=\tau_1$ and also verifies
\begin{equation}\label{7.13.2}
\psi\in C\big([0,\tau_2];\mathcal{K}_0\big);\quad\overline{\psi(0)}=\psi(0);\quad\|\psi(t)\|_{\mathcal{K}_0}=1,\ \forall t\in[0,\tau_2];\quad\psi(t)\in\mathcal{N}_k((t,0)),\ \forall  t\in[0,\tau_2].
\end{equation}

Continuing in this way, after a finite number of steps ($n$ steps) we obtain a function $\psi:[0,(1/2)]\rightarrow\mathcal{K}_0$ verifying the properties:
\begin{equation}\label{7.13.n}
\psi\in C\big([0,1/2];\mathcal{K}_0\big);\quad\overline{\psi(0)}=\psi(0);\quad\|\psi(t)\|_{\mathcal{K}_0}=1,\ \forall t\in[0,1/2];\quad\psi(t)\in\mathcal{N}_k((t,0)),\ \forall  t\in[0,1/2].
\end{equation}

We extend now this function to the interval $[-(1/2),(1/2)]$ by defining $\psi(-t):=\overline{\psi(t)}$ for any $t\in[0,1/2]$. It verifies the properties:
\begin{equation}\label{7.13.-}
\psi\in C\big([-1/2,1/2];\mathcal{K}_0\big);\qquad\psi(t)\in\mathcal{N}_k((t,0)),\ \forall  t\in[-1/2,1/2].
\end{equation}
The second property above follows easily from \eqref{7.9.e}.

We take now into account the second statement of the Lemma \ref{L.7.9} and notice that it implies the equality
$$
\check{P}_0((1/2,0))\ =\ \sigma_{-e^*_1}\check{P}_0((-1/2,0))\sigma_{e^*_1}
$$
and from that we deduce that $\sigma_{-e^*_1}\psi(-1/2)$ is an eigenvector of $\check{P}_0((1/2,0))$ for the eigenvalue $\lambda_k((-1/2,0))=\lambda_k((1/2,0))$ (following \eqref{7.9.f}). We deduce that it exists $\kappa\in\mathbb{R}$ such that
$$
\sigma_{-e^*_1}\psi(-1/2)\ =\ e^{i\kappa}\psi(1/2).
$$

Let us define now $\phi(t):=e^{it\kappa}\psi(t)$ for any $t\in[-1/2,1/2]$. It will evidently have the following properties:
\begin{equation}\label{7.phi}
\phi\in C\big([-1/2,1/2];\mathcal{K}_0\big);\quad\phi(t)\in\mathcal{N}_k((t,0)),\ \forall  t\in[-1/2,1/2];\quad\phi(1/2)=\sigma_{-e^*_1}\phi(-1/2).
\end{equation}
We extend this function to $\mathbb{R}$ by the following reccurence relation:
\begin{equation}\label{7.16}
\phi(t)\ :=\ \sigma_{-e^*_1}\phi(t-1),\qquad\forall t\in[j-1/2,j+1/2],\ j\in\mathbb{Z}.
\end{equation}
We obtain a function $\phi:\mathbb{R}\rightarrow\mathcal{K}_0$ verifying the following properties:
\begin{equation}\label{7.17}
\left\{
\begin{array}{l}
\phi\,\in\,C\big(\mathbb{R};\mathcal{K}_0\big).\\
\phi(t+l)\ =\ \sigma_{-le^*_1}\,\phi(t),\qquad\forall t\in\mathbb{R},\ \forall l\in\mathbb{Z}.\\
\|\phi(t)\|_{\mathcal{K}_0}\ =\ 1,\qquad\forall t\in\mathbb{R}.\\
\overline{\phi(t)}\ =\ \phi(-t),\qquad\forall t\in\mathbb{R}.\\
\phi(t)\,\in\,\mathcal{N}_k((t,0)),\qquad\forall t\in\mathbb{R}.
\end{array}
\right.
\end{equation}

Let us denote now by $t^\prime:=(t_1,\ldots,t_{d-1})\in\mathbb{R}^{d-1}$ and suppose by hypothesis that we have constructed a function $\psi:\mathbb{R}^{d-1}\rightarrow\mathcal{K}_0$ satisfying the following properties:
\begin{equation}\label{7.18}
\left\{
\begin{array}{l}
\psi\,\in\,C\big(\mathbb{R}^{d-1};\mathcal{K}_0\big).\\
\psi(t^\prime+l^\prime)\ =\ \sigma_{-<e^{*\prime,l^\prime>}}\,\psi(t^\prime),\qquad\forall t^\prime\in\mathbb{R}^{d-1},\ \forall l^\prime\in\mathbb{Z}^{d-1}.\\
\|\psi(t^\prime)\|_{\mathcal{K}_0}\ =\ 1,\qquad\forall t^\prime\in\mathbb{R}^{d-1}.\\
\overline{\psi(t^\prime)}\ =\ \psi(-t^\prime),\qquad\forall t^\prime\in\mathbb{R}^{d-1}.\\
\psi(t^\prime)\,\in\,\mathcal{N}_k((t^\prime,0)),\qquad\forall t^\prime\in\mathbb{R}^{d-1}.
\end{array}
\right.
\end{equation}
We want to construct now a function $\phi:\mathbb{R}^d\rightarrow\mathcal{K}_0$ satisfying the same properties as $\psi$ above (with $d-1$ replaced by $d$). For doing that let us introduce the following notations $t=(t^\prime,t_d)$ with $t^\prime=(t_1,\ldots,t_{d-1})\in\mathbb{R}^{d-1}$ for any $t=(t_1,\ldots,t_d)\in\mathbb{R}^d$.

We start by defining first $\widetilde{\psi}(t^\prime,0):=\psi(t^\prime)$ for any $t^\prime\in\mathbb{R}^{d-1}$. Repeating the argumet at the beginning of this proof, we extend step by step our function $\widetilde{\psi}$ on subsets of the form $\mathbb{R}^{d-1}\times[\tau_j,\tau_{j+1}]$ with $0\leq j\leq n-1$. Finally we extend it to the subset $\mathbb{R}^{d-1}\times[-1/2,1/2]$ by the definition $\widetilde{\psi}(t^\prime,-t_d):=\overline{\widetilde{\psi}(-t^\prime,t_d)}$ for any $t^\prime\in\mathbb{R}^{d-1}$ and any $t_d\in[0,1/2]$.
We obtain in this way a function $\widetilde{\psi}:\mathbb{R}^{d-1}\times[-1/2,1/2]\rightarrow\mathcal{K}_0$ verifying the properties:
\begin{equation}\label{7.19}
\left\{
\begin{array}{l}
\widetilde{\psi}\,\in\,C\big(\mathbb{R}^{d-1}\times[-1/2,1/2];\mathcal{K}_0\big).\\
\widetilde{\psi}(t^\prime+l^\prime,t_d)\ =\ \sigma_{-<e^{*\prime},l^\prime>}\,\widetilde{\psi}(t^\prime,t_d),\qquad\forall t^\prime\in\mathbb{R}^{d-1},\ \forall l^\prime\in\mathbb{Z}^{d-1},\ \forall t_d\in[-1/2,1/2].\\
\|\widetilde{\psi}(t^\prime,t_d)\|_{\mathcal{K}_0}\ =\ 1,\qquad\forall t^\prime\in\mathbb{R}^{d-1},\ \forall t_d\in[-1/2,1/2].\\
\overline{\widetilde{\psi}(t^\prime,t_d)}\ =\ \widetilde{\psi}(-t^\prime,-t_d),\qquad\forall t^\prime\in\mathbb{R}^{d-1},\ \forall t_d\in[-1/2,1/2].\\
\widetilde{\psi}(t^\prime,t_d)\,\in\,\mathcal{N}_k((t^\prime,0)),\qquad\forall t^\prime\in\mathbb{R}^{d-1},\ \forall t_d\in[-1/2,1/2].
\end{array}
\right.
\end{equation}

Now we come once more to the second statement of the Lemma \ref{L.7.9} and notice that:
$$
\check{P}_0((t^\prime,1/2))\ =\ \sigma_{-e^*_d}\check{P}_0((t^\prime,-1/2))\sigma_{e^*_d},\qquad\forall t^\prime\in\mathbb{R}^{d-1}.
$$
We deduce that $\sigma_{-e^*_d}\widetilde{\psi}(t^\prime,-1/2)$ is a normed eigenvector of $\check{P}_0((t^\prime,1/2))$ for the eigenvalue $\lambda_k((t^\prime,-1/2))=\lambda_k((t^\prime,1/2))$ (we made use of point (3) in Lemma \ref{L.7.9}). We conclude that it exists a function $\kappa^\prime:\mathbb{R}^{d-1}\rightarrow\mathbb{R}$ such that
\begin{equation}\label{7.21}
\sigma_{-e^*_d}\widetilde{\psi}(t^\prime,-1/2)\ =\ e^{i\kappa^\prime(t^\prime)}\widetilde{\psi}(t^\prime,1/2),\qquad\forall t^\prime\in\mathbb{R}^{d-1}.
\end{equation}
From the continuity of the function $\widetilde{\psi}\in C\big(\mathbb{R}^{d-1}\times[-1/2,1/2];\mathcal{K}_0\big)$ we deduce the continuity of the function $e^{i\kappa^\prime}:\mathbb{R}^{d-1}\rightarrow\mathbb{U}(1)$ and also of the function $\kappa^\prime:\mathbb{R}^{d-1}\rightarrow\mathbb{R}$. We take into account now \eqref{7.21} and the second equality in \eqref{7.19} with $t^\prime\in\mathbb{R}^{d-1}$ and $l^\prime\in\mathbb{Z}^{d-1}$ and get
$$
e^{i\kappa^\prime(t^\prime+l^\prime)}\widetilde{\psi}(t^\prime+l^\prime,1/2)\ =\ \sigma_{-e^*_d}\widetilde{\psi}(t^\prime+l^\prime,-1/2)\ =\ \sigma_{-e^*_d}\sigma_{-<e^{*\prime},l^\prime>}\widetilde{\psi}(t^\prime,-1/2)\ =
$$
$$
=\ \sigma_{-<e^{*\prime},l^\prime>}e^{i\kappa^\prime(t^\prime)}\widetilde{\psi}(t^\prime,1/2)\ =\ e^{i\kappa^\prime(t^\prime)}\widetilde{\psi}(t^\prime+l^\prime,1/2).
$$
It follows that $e^{i[\kappa^\prime(t^\prime+l^\prime)-\kappa^\prime(t^\prime)]}=1$ and the function $\kappa^\prime:\mathbb{R}^{d-1}\rightarrow\mathbb{R}$ being continuous we deduce that for any $l^\prime\in\mathbb{Z}^{d-1}$ there exists $n(l^\prime)\in\mathbb{Z}$ such that 
\begin{equation}\label{7.23}
\kappa^\prime(t^\prime+l^\prime)\,-\,\kappa^\prime(t^\prime)\ =\ 2\pi n(l^\prime),\qquad\forall t^\prime\in\mathbb{R}^{d-1}.
\end{equation}
But from \eqref{7.21} we deduce that 
$$
\sigma_{-e^*_d}\widetilde{\psi}(-t^\prime,-1/2)\ =\ e^{i\kappa^\prime(-t^\prime)}\widetilde{\psi}(-t^\prime,1/2),\qquad\forall t^\prime\in\mathbb{R}^{d-1}.
$$
After complex conjugation and making use of \eqref{7.19} we get
$$
\sigma_{e^*_d}\widetilde{\psi}(t^\prime,1/2)\ =\ e^{-i\kappa^\prime(-t^\prime)}\widetilde{\psi}(t^\prime,-1/2),\qquad\forall t^\prime\in\mathbb{R}^{d-1},
$$
or equivalently
$$
\sigma_{-e^*_d}\widetilde{\psi}(t^\prime,-1/2)\ =\ e^{i\kappa^\prime(-t^\prime)}\widetilde{\psi}(t^\prime,1/2),\qquad\forall t^\prime\in\mathbb{R}^{d-1}.
$$
We use once again \eqref{7.21} in order to obtain the equality $e^{i\kappa^\prime(-t^\prime)}=e^{i\kappa^\prime(t^\prime)}$, or equivalently the relation $\kappa^\prime(-t^\prime)-\kappa^\prime(t^\prime)\in2\pi\mathbb{Z}$. As the function $\kappa^\prime$ is continuous, we conclude that the difference $\kappa^\prime(-t^\prime)-\kappa^\prime(t^\prime)$ must be constant; as in $t^\prime=0$ this difference is by definition $0$ we conclude that 
\begin{equation}\label{7.24}
\kappa^\prime(-t^\prime)=\kappa^\prime(t^\prime),\qquad\forall t^\prime\in\mathbb{R}^{d-1}.
\end{equation}

We consider now \eqref{7.23} and notice that for any $l^\prime\in\mathbb{Z}^{d-1}$ we can choose the point $t^\prime:=-(1/2)l^\prime\in\mathbb{R}^{d-1}$ verifying the relation $t^\prime+l^\prime=-t^\prime$ and thus we conclude that $n(l^\prime)=0$ for any $l^\prime\in\mathbb{Z}^{d-1}$ obtaining the equalities
\begin{equation}\label{7.25}
\kappa^\prime(t^\prime+l^\prime)\ =\ \kappa^\prime(t^\prime),\qquad\forall t^\prime\in\mathbb{R}^{d-1},\ \forall l^\prime\in\mathbb{Z}^{d-1}.
\end{equation}

We can now define $\widetilde{\phi}:\mathbb{R}^{d-1}\times[-1/2.1/2]\rightarrow\mathcal{K}_0$ by the following equality:
\begin{equation}\label{7.26}
\widetilde{\phi}(t^\prime,t_d)\ :=\ e^{i\kappa^\prime(t^\prime)t_d}\widetilde{\psi}(t^\prime,t_d),\qquad\forall(t^\prime,t_d)\in\mathbb{R}^{d-1}\times[-1/2,1/2].
\end{equation}
From \eqref{7.24} and \eqref{7.25} the function $\widetilde{\phi}$ has all the properties \eqref{7.19} and it also has the following property:
\begin{equation}\label{7.20}
\widetilde{\phi}(t^\prime+l^\prime,1/2)\ =\ \sigma_{<e^{*\prime},l^\prime>}\sigma_{-e^*_d}\widetilde{\phi}(t^\prime,-1/2),\qquad\forall t^\prime\in\mathbb{R}^{d-1},\ \forall l^\prime\in\mathbb{Z}^{d-1}.
\end{equation}
In fact we see that
$$
\widetilde{\phi}(t^\prime+l^\prime,1/2)\ =\ e^{i\kappa^\prime(t^\prime)/2}\widetilde{\psi}(t^\prime+l^\prime,1/2)\ =\ e^{i\kappa^\prime(t^\prime)/2}\sigma_{-<e^{*\prime},l^\prime>}\widetilde{\psi}(t^\prime,1/2)\ =\ e^{-i\kappa^\prime(t^\prime)/2}\sigma_{-<e^{*\prime},l^\prime>}\sigma_{-e^*_d}\widetilde{\psi}(t^\prime,-1/2)\ =
$$
$$
=\ \sigma_{-<e^{*\prime},l^\prime>}\sigma_{-e^*_d}\widetilde{\phi}(t^\prime,-1/2).
$$

As in the case $d=1$ we extend the function $\widetilde{\phi}$ to the entire $\mathbb{R}^d$ by the following relation:
\begin{equation}\label{7.27}
\widetilde{\phi}(t^\prime,t_d)\ :=\ \sigma_{-e^*_d}\widetilde{\phi}(t^\prime,t_d-1),\qquad\forall t^\prime\in\mathbb{R}^{d-1}\ \forall t_d\in[j-1/2,j+1/2],\ j\in\mathbb{Z}.
\end{equation}
We evidently obtain a function of class $C\big(\X^*;\mathcal{K}_0\big)$ satisfying the properties (2)-(6) in our Lemma. 

We end up our construction by a regularization procedure. Let us choose a real, non-negative, even function $\chi\in C^\infty_0(\X^*)$ and such that $\int_{\X^*}\chi(t)dt=1$. For any $\delta>0$ we denote by $\chi_\delta(\xi):=\delta^{-d}\chi(\xi/\delta)$ and define:
\begin{equation}\label{7.28}
\widetilde{\phi}_\delta(\xi)\ :=\ \int_{\X^*}\chi_\delta(\xi-\eta)\widetilde{\phi}(\eta)\,d\eta,\qquad\forall\xi\in\X^*.
\end{equation}
Evidently $\widetilde{\phi}_\delta\in C^\infty\big(\X^*;\mathcal{K}_0\big)$ for any $\delta>0$. Moreover we also have that:
\begin{equation}\label{7.29}
\widetilde{\phi}_\delta(\xi+\gamma^*)\ =\ \int_{\X^*}\chi_\delta(\xi+\gamma^*-\eta)\widetilde{\phi}(\eta)\,d\eta\ =\ \int_{\X^*}\chi_\delta(\xi-\eta)\widetilde{\phi}(\eta+\gamma^*)\,d\eta\ =\ \int_{\X^*}\chi_\delta(\xi-\eta)\sigma_{-\gamma^*}\widetilde{\phi}(\eta)\,d\eta\ =
\end{equation}
$$
=\ \sigma_{-\gamma^*}\widetilde{\phi}_\delta(\xi),\qquad\forall\xi\in\X^*,\ \forall\gamma^*\in\Gamma^*;
$$
\begin{equation}\label{7.30}
\overline{\widetilde{\phi}_\delta(\xi)}\ =\ \int_{\X^*}\chi_\delta(\xi-\eta)\overline{\widetilde{\phi}(\eta)}\,d\eta\ =\ \int_{\X^*}\chi_\delta(\xi-\eta)\widetilde{\phi}(-\eta)\,d\eta\ =
\end{equation}
$$
=\ \int_{\X^*}\chi_\delta(\xi+\eta)\widetilde{\phi}(\eta)\,d\eta\ =\ \int_{\X^*}\chi_\delta(-\xi-\eta)\widetilde{\phi}(\eta)\,d\eta\ =\ \widetilde{\phi}_\delta(-\xi),\qquad\forall\xi\in\X^*.
$$
We shall prove now the following continuity relation:
\begin{equation}\label{7.31}
\forall\theta>0,\ \exists\delta_0>0:\ \|\widetilde{\phi}_\delta(\xi)\,-\,\widetilde{\phi}(\xi)\|_{\mathcal{K}_0}\leq\,\theta,\ \forall\delta\in[0,\delta_0],\quad\forall\xi\in\X^*.
\end{equation}
In fact, let us notice that:
$$
\|\widetilde{\phi}_\delta(\xi)\,-\,\widetilde{\phi}(\xi)\|_{\mathcal{K}_0}\ \leq\ \int_{\X^*}\|\widetilde{\phi}(\xi-\delta\eta)\,-\,\widetilde{\phi}(\xi)\|_{\mathcal{K}_0}|\chi(\eta)|\,d\eta\ \leq\ C\int_{\text{\sf supp}\chi}\|\widetilde{\phi}(\xi-\delta\eta)\,-\,\widetilde{\phi}(\xi)\|_{\mathcal{K}_0}\,d\eta.
$$
The function $\widetilde{\phi}$ being continuous and satisfying the relation $\widetilde{\phi}(\xi+\gamma^*)=\sigma_{-\gamma^*}\widetilde{\phi}(\xi)$ for all $\gamma^*\in\Gamma^*$ and for all $\xi\in\X^*$ we conclude that it is uniformly continuous on $\X^*$ and the support of $\chi$ being compact, the above estimation implies the stated continuity in $\delta>0$. Thus we conclude that for $\delta\in[0,\delta_0]$ with $\delta_0>0$ small enough we have 
$$
\left\|\Pi_k(\xi)\widetilde{\phi}_\delta(\xi)\right\|_{\mathcal{K}_0}\ \geq\ \frac{1}{2},\qquad\forall\xi\in\X^*
$$
and we can define:
$$
\phi(\xi)\ :=\ \|\Pi_k(\xi)\widetilde{\phi}_{\delta_0}(\xi)\|_{\mathcal{K}_0}^{-1}\Pi_k(\xi)\widetilde{\phi}_{\delta_0}(\xi),\qquad\forall\xi\in\X^*.
$$

It remains to prove that $\phi\in C^\infty\big(\X^*;\mathcal{K}_{lm,0}\big)$ for any $l\in\mathbb{N}$ because all the other properties from the statement of our Lemma are clearly satisfied considering the definition of $\phi$. The case $l=0$ is also clear from the definition.

Let us notice that $\check{P}_0\in C^\infty\big(\X^*;\mathbb{B}(\mathcal{K}_{m,0};\mathcal{K}_0)\big)$ and we know that for any $\xi_0\in\X^*$ there exists a constant $C_0>0$ such that 
$$
\|u\|_{\mathcal{K}_{m,0}}\ \leq\ C_0\left(\left\|\check{P}_0(\xi_0)u\right\|_{\mathcal{K}_0}\,+\,\|u\|_{\mathcal{K}_0}\right),\qquad\forall u\in\mathcal{K}_{m,0}
$$
and evidently
$$
\|u\|_{\mathcal{K}_{m,0}}\ \leq\ C_0\left(\left\|\check{P}_0(\xi)u\right\|_{\mathcal{K}_0}\,+\,\left\|\big(\check{P}_0(\xi_0)\,-\,\check{P}_0(\xi)\big)u\right\|_{\mathcal{K}_0}\,+\,\|u\|_{\mathcal{K}_0}\right),\ \forall\xi\in\X^*,\qquad\forall u\in\mathcal{K}_{m,0}.
$$
Choosing a sufficiently small neighborhood $V_0$ of $\xi_0\in\X^*$ we deduce that it exists $C^\prime_0>0$ such that
$$
\|u\|_{\mathcal{K}_{m,0}}\ \leq\ C^\prime_0\left(\left\|\check{P}_0(\xi)u\right\|_{\mathcal{K}_0}\,+\,\|u\|_{\mathcal{K}_0}\right),\ \forall\xi\in V_0,\qquad\forall u\in\mathcal{K}_{m,0}.
$$
Thus, with $c\in\mathbb{R}$ the constant introduced in the proof of Lemma \ref{L.7.6}, if we take $u:=\big(\check{P}_0(\xi)\,-\,c\big)^{-1}v$ for some $v\in\mathcal{K}_0$ we obtain that
$$
\left\|\big(\check{P}_0(\xi)\,-\,c\big)^{-1}v\right\|_{\mathcal{K}_{m,0}}\ \leq\ C^{\prime\prime}_0\|v\|_{\mathcal{K}_0},\ \forall\xi\in V_0,\qquad\forall v\in\mathcal{K}_0.
$$
We conclude that $\big(\check{P}_0(\xi)\,-\,c\big)^{-1}\in\mathbb{B}(\mathcal{K}_0;\mathcal{K}_{m,0})$ uniformly for $\xi$ in any compact subset of $\X^*$.

Considering now the proof of Lemma \ref{L.7.3} once again we obtain that in fact $\phi\in C^\infty\big(\X^*;\mathcal{K}_{m,0}\big)$ and we get the case $l=1$.

In order to obtain the general situation $l\in\mathbb{N}$ let us notice that for any $\xi\in\X^*$ we know that $\phi(\xi)\in\mathcal{D}\big(\check{P}_0(\xi)^l\big)$ and $\check{P}_0(\xi)^l\phi(\xi)=\lambda_k(\xi)^l\phi(\xi)$ for any $l\in\mathbb{N}$. We also have that $\check{P}_0(\cdot)^l\in C^\infty\big(\X^*;\mathbb{B}(\mathcal{K}_0;\mathcal{K}_{lm,0})\big)$. The fact that $\check{P}_0(\xi)^l$ is induced by the operator $P_{0,\xi}^l$ on the Hilbert space of tempered distributions $\mathcal{K}_0$ and we know that $P_0^l$ is elliptic of order $lm$, allows us to conclude by an argument similar to the one above, that $\phi\in C^\infty\big(\X^*;\mathcal{K}_{lm,0}\big)$ for any $l\in\mathbb{N}$.
\end{proof}

\begin{remark}\label{R.7.11}
The arguments used in the proof of Lemma \ref{L.3.7} allow to deduce from properties (1) - (3) from Lemma \ref{L.7.10} that for any $\alpha\in\mathbb{N}^d$ and for any $s\in\mathbb{R}$ there exists a constant $C_{\alpha,s}>0$ such that:
\begin{equation}\label{7.32}
\left\|\big(\partial^\alpha_\xi\phi\big)(\cdot,\xi)\right\|_{\mathcal{K}_{s,\xi}}\ \leq\ C_{\alpha,s},\quad\forall\xi\in\X^*.
\end{equation}
\end{remark}

\begin{description}
\item[Proof of Proposition \ref{P.0.3}] 
We shall repeat the construction of the Grushin operator \eqref{3.34} from Section \ref{S.3} under the Hypothesis of Proposition \ref{P.0.3}. We shall prove that in this case we can take $N=1$ and $\phi_1(x,\xi)=\phi(x,\xi)$ the function obtained in Lemma \ref{L.7.10}. Due to Lemma \ref{L.7.10} and Remark \ref{R.7.11} this function has all the properties needed in Lemma \ref{L.3.7}. It is thus possible to obtain the operator $\mathcal{P}_0(\xi,\lambda)$ and the essential problem is to prove its invertibility in order to obtain a result similar to Proposition \ref{P.3.8}. Frm that point the proof of Proposition \ref{P.0.3} just repeats the arguments of Section \ref{S.4}.

As in Section \ref{S.3}, for any $\xi\in\X^*$ we construct the operators:
\begin{equation}\label{7.33}
R_+(\xi)\,\in\,\mathbb{B}(\mathcal{K}_0,\mathbb{C}),\quad R_+(\xi)u\ :=\ \left(u,\phi(\cdot,\xi)\right)_{\mathcal{K}_0},\qquad\forall u\in\mathcal{K}_0,
\end{equation}
\begin{equation}\label{7.34}
R_-(\xi)\,\in\,\mathbb{B}(\mathbb{C},\mathcal{K}_0),\quad R_-(\xi)c\ :=\ c\phi(\cdot,\xi),\qquad\forall c\in\mathbb{C}.
\end{equation}
It is evident that $R_+\in S^0_0\big(\X;\mathbb{B}(\mathcal{K}_{m,\xi};\mathbb{C})\big)$ and $R_-\in S^0_0\big(\X;\mathbb{B}(\mathbb{C};\mathcal{K}_{0})\big)$ and from Example \ref{E.A.20} we deduce that $\check{P}_0(\cdot)\,-\,\lambda\id\in S^0_0\big(\X;\mathbb{B}(\mathcal{K}_{m,\xi};\mathcal{K}_0)\big)$ uniformly for $\lambda\in I$.

For any $(\xi,\lambda)\in\X^*\times I$ we define
\begin{equation}\label{7.35}
\mathcal{P}_0(\xi,\lambda)\ :=\ \left(
\begin{array}{cc}
\check{P}_0(\xi)&R_-(\xi)\\
R_+(\xi)&0
\end{array}
\right)\,\in\,\mathbb{B}\big(\mathcal{K}_{m,\xi}\times\mathbb{C};\mathcal{K}_0\times\mathbb{C}\big).
\end{equation}
We denote by $\mathcal{A}_\xi:=\mathcal{K}_{m,\xi}\times\mathbb{C}$ and by $\mathcal{B}_\xi:=\mathcal{K}_0\times\mathbb{C}$ and we notice that
\begin{equation}\label{7.36}
\mathcal{P}_0(\cdot,\lambda)\,\in\,S^0_0\big(\X;\mathbb{B}(\mathcal{A}_\bullet;\mathcal{B}_\bullet\big)\ \text{uniformly in }\lambda\in I.
\end{equation}
In order to construct an inverse for $\mathcal{P}_0(\xi,\lambda)$ we make the following choices:
\begin{equation}\label{7.38}
E^0_+(\xi,\lambda)\,\in\,\mathbb{B}(\mathbb{C};\mathcal{K}_{m,\xi}),\quad E^0_+(\xi,\lambda)c\ :=\ c\phi(\cdot,\xi),\qquad\forall c\in\mathbb{C},
\end{equation}
\begin{equation}\label{7.39}
E^0_-(\xi,\lambda)\,\in\,\mathbb{B}(\mathcal{K}_0;\mathbb{C}),\quad E^0_-(\xi,\lambda)u\ :=\ \left(u,\phi(\cdot,\xi)\right)_{\mathcal{K}_0},\qquad\forall u\in\mathcal{K}_0,
\end{equation}
\begin{equation}\label{7.40}
E^0_{-+}(\xi,\lambda)\,\in\,\mathbb{B}(\mathbb{C}),\quad E^0_{-+}(\xi,\lambda)c\ :=\ \big(\lambda\,-\,\lambda_k(\xi)\big)c,\qquad\forall c\in\mathbb{C},
\end{equation}
\begin{equation}\label{7.41}
E^0(\xi,\lambda)\,\in\,\mathbb{B}(\mathcal{K}_0;\mathcal{K}_{m,\xi}),\quad E^0(\xi,\lambda)\ :=\ \left[\id\,-\,\Pi_k(\xi)\right]\left[\check{P}_0(\xi)\,-\,\lambda\right]\left[\id\,-\,\Pi_k(\xi)\right],
\end{equation}
where $\lambda_k(\xi)$ is the Floquet eigenvalue generating the simple spectral band $J_k$ and $\Pi_k(\xi)$ is the orthogonal projection on $\mathcal{N}_k(\xi)$ as defined in \eqref{7.11}, with the circle $\mathscr{C}$ chosen to contain the interval $I\subset\mathbb{R}$ into its interior domain.

Let us notice that the operator $E^0(\xi,\lambda)$ is well defined; in fact $\Re\text{ange}\left[\id\,-\,\Pi_k(\xi)\right]$ is the orthogonal complement in $\mathcal{K}_0$ of $\mathcal{N}_k(\xi):=\ker\left[\check{P}_0(\xi)\,-\,\lambda\right]$ (linearly generated by $\phi(\cdot,\xi)$) and is a reducing subspace for $\check{P}_0(\xi)$. Thus $\check{P}_0(\xi)$ induces on the space $\Re\text{ange}\left[\id\,-\,\Pi_k(\xi)\right]$ a self-adjoint operator having the spectrum $\{\lambda_j(\xi)\}_{j\ne k}$. If $\lambda\in I$ and $I$ is as in the hypothesis (i.e. $I\cap J_l=\emptyset$ for any $l\ne k$), the distance from $\lambda$ to the spectrum of this induced operator is bounded below by a strictly positive constant $C>0$ that does not depend on $(\xi,\lambda)\in\X^*\times I$; thus the norm of the rezolvent of this induced operator, in point $\lambda\in I\subset\mathbb{R}$ is bounded above by $C^{-1}$. This rezolvent is exactly our operator $E^0(\xi,\lambda)$ defined in \eqref{7.41}. Recalling that $\check{P}_0(\xi)\,-\,\lambda\in S^0_0\big(\X;\mathbb{B}(\mathcal{K}_{m,\xi};\mathcal{K}_0)\big)$ uniformly with respect to $\lambda\in I$ and the formula \eqref{7.11} for the operator $\Pi_k(\xi)$ one proves that $E^0(\cdot,\xi)\in\,S^0_0\big(\X;\mathbb{B}(\mathcal{K}_0;\mathcal{K}_{m,\xi})\big)$ uniformly with respect to $\lambda\in I$.

By definition we have that $E^0_+(\xi,\lambda)\in S^0_0\big(\X;\mathbb{B}(\mathbb{C};\mathcal{K}_{m,\xi})\big)$, $E^0_-(\xi,\lambda)\in S^0_0\big(\X;\mathbb{B}(\mathcal{K}_0;\mathbb{C})\big)$, $E^0_{-+}(\xi,\lambda)\in S^0_0\big(\X;\mathbb{B}(\mathbb{C})\big)$ uniformly for $\lambda\in I$ and thus if we define for any $(\xi,\lambda)\in\X^*\times I$:
\begin{equation}\label{7.37}
\mathcal{E}_0(\xi,\lambda)\ :=\ \left(
\begin{array}{cc}
E^0(\xi,\lambda)&E^0_+(\xi,\lambda)\\
E^0_-(\xi,\lambda)&E^0_{-+}(\xi,\lambda)
\end{array}
\right)\,\in\,\mathbb{B}(\mathcal{B}_\xi;\mathcal{A}_\xi)
\end{equation}
we obtain an operator-valued symbol:
\begin{equation}\label{7.42}
\mathcal{E}_0(\cdots,\lambda)\,\in\,S^0_0\big(\X;\mathbb{B}(\mathcal{B}_\bullet;\mathcal{A}_\bullet)\big),\ \text{uniformly in }\lambda\in I.
\end{equation}

We shall verify now that for each $(\xi,\lambda)\in\X^*\times I$ this defines an inverse for the operator $\mathcal{P}_0(\xi,\lambda)$:
\begin{equation}\label{7.43}
\mathcal{P}_0(\xi,\lambda)\,\mathcal{E}_0(\xi,\lambda)\ =\ \id,\quad\text{on }\mathcal{K}_0\times\mathbb{C}.
\end{equation}
We can write:
$$
\mathcal{P}_0(\xi,\lambda)\,\mathcal{E}_0(\xi,\lambda)\ =\ \left(
\begin{array}{cc}
\big(\check{P}_0(\xi)\,-\,\lambda\big)E^0(\xi,\lambda)\,+\,R_-(\xi)E^0_-(\xi,\lambda)&\big(\check{P}_0(\xi)\,-\,\lambda\big)E^0_+(\xi,\lambda)\,+\,R_-(\xi)E^0_{-+}(\xi,\lambda)\\
R_+(\xi)E^0(\xi,\lambda)&R_+(\xi)E^0_+(\xi,\lambda)
\end{array}
\right).
$$
For any $u\in\mathcal{K}_0$ we can write:
$$
\big(\check{P}_0(\xi)\,-\,\lambda\big)E^0(\xi,\lambda)u\ =\ \left[\check{P}_0(\xi)\,-\,\lambda\right]\left[\id\,-\,\Pi_k(\xi)\right]\left[\check{P}_0(\xi)\,-\,\lambda\right]^{-1}\left[\id\,-\,\Pi_k(\xi)\right]u\ =\ \left[\id\,-\,\Pi_k(\xi)\right]u,
$$
$$
R_-(\xi)E^0_-(\xi,\lambda)u\ =\ \left(u,\phi(\cdot,\xi)\right)_{\mathcal{K}_0}\phi(\cdot,\xi)\ =\ \Pi_k(\xi)u
$$
$$
R_+(\xi)E^0(\xi,\lambda)u\ =\ \left(\left[\id\,-\,\Pi_k(\xi)\right]\left[\check{P}_0(\xi)\,-\,\lambda\right]^{-1}\left[\id\,-\,\Pi_k(\xi)\right]u\,,\,\phi(\cdot,\xi)\right)_{\mathcal{K}_0}\ =\ 0.
$$
For $c\in\mathbb{C}$ we can write that:
$$
\big(\check{P}_0(\xi)\,-\,\lambda\big)E^0_+(\xi,\lambda)c\ =\ c\big(\check{P}_0(\xi)\,-\,\lambda\big)\phi(\cdot,\xi)\ =\ c\big(\lambda_k(\xi)\,-\,\lambda\big)\phi(\cdot,\xi),
$$
$$
R_-(\xi)E^0_{-+}(\xi,\lambda)c\ =\ \big(\lambda\,-\,\lambda_k(\xi)\big)c\phi(\cdot,\xi),
$$
$$
R_+(\xi)E^0_+(\xi,\lambda)c\ =\ c\left(\phi(\cdot,\xi),\phi(\cdot,\xi)\right)_{\mathcal{K}_0}\ =\ c.
$$
These identities imply \eqref{7.43}. The property of being a left inverse is verified by very similar computations or by taking into account the self-adjointness of both $\mathcal{P}_0(\xi,\lambda)$ and $\mathcal{E}_0(\xi,\lambda)$.

From this point one can repeat identically the arguments in the proof of Theorem \ref{T.0.1} noticing that \eqref{4.12} and \eqref{7.40} imply that we can take $E^{-+}_{\lambda,\epsilon}(x,\xi)$ as in \eqref{0.37}.
\end{description}

\subsection{The constant magnetic field}

In this subsection we prove Proposition \ref{P.0.4}. Thus we suppose that the symbols $p_\epsilon$ do not depend on the first argument and the magnetic field has constant components:
\begin{equation}\label{8.1}
\left\{
\begin{array}{l}
B_\epsilon\ =\ \frac{1}{2}\underset{1\leq j,k\leq d}{\sum}B_{jk}(\epsilon)\,dx_j\wedge dx_k,\quad B_{jk}(\epsilon)\ =\ -B_{kj}(\epsilon)\,\in\,\mathbb{R},\\
\underset{\epsilon\rightarrow0}{\lim}B_{jk}(\epsilon)\ =\ 0.
\end{array}
\right.
\end{equation}

Using the transversal gauge \eqref{0.28} we associate to these magnetic fields some vector potentials $A_\epsilon:=(A_{\epsilon,1},\ldots,A_{\epsilon,d})$ satisfying:
\begin{equation}\label{8.2}
A_{\epsilon,j}(x)\ =\ \frac{1}{2}\underset{1\leq k\leq d}{\sum}B_{jk}(\epsilon)x_k.
\end{equation}

\begin{description}
\item[Proof of Proposition \ref{P.0.4} (1)]
We use formula \eqref{1.5} from Lemma \ref{L.1.2} noticing that the linearity of the functions $A_{\epsilon,j}$ and the definition of $\omega_{A}$ imply that
\begin{equation}\label{8.3}
\omega_{\tau_{-x}A_\epsilon}(y,\tilde{y})\ =\ \omega_{A_\epsilon}(y,\tilde{y})\,e^{i<A_\epsilon(x),y-\tilde{y}>}\ =\ \omega_{A_\epsilon+A_\epsilon(x)}(y,\tilde{y}).
\end{equation}
We deduce that for any $u\in\mathscr{S}(\X^2)$ and for any $(x,y)\in\X^2$ we have that:
\begin{equation}\label{8.4}
\left(\boldsymbol{\chi}^*\widetilde{P}_\epsilon(\boldsymbol{\chi}^*)^{-1}u\right)(x,y)\ =\ \left[\big(\id\otimes\sigma_{A(x)}\big)(\id\otimes P_\epsilon)\big(\id\otimes\sigma_{-A(x)}\big)u\right](x,y).
\end{equation}
It follows that the operator $\widetilde{P}_\epsilon$, that is an unbounded self-adjoint operator in $L^2(\X^2)$ denoted in Proposition \ref{P.2.14} by $\widetilde{P}_\epsilon^\prime$, is unitarily equivalent with the oprtator $\id\otimes P_\epsilon$ with $P_\epsilon$ self-adjoint unbounded operator in $L^2(\X)$. It follows that $\sigma\big(\widetilde{P}_\epsilon^\prime\big)=\sigma\big(P_\epsilon\big)$. From Proposition \ref{P.2.14} we deduce that $\sigma\big(\widetilde{P}_\epsilon^\prime\big)=\sigma
\big(\widetilde{P}_\epsilon^{\prime\prime}\big)$ where $\widetilde{P}_\epsilon^{\prime\prime}$ is the self-adjoint realization of $\widetilde{P}_\epsilon$ in the space $L^2\big(\X\times\mathbb{T}\big)$. Finally, from Corollary \ref{C.4.5} we deduce that for any $(\lambda,\epsilon)\in I\times[-\epsilon_0,\epsilon_0]$ we have the equivalence relation:
$$
\lambda\,\in\,\sigma\big(\widetilde{P}_\epsilon^{\prime\prime}\big)\quad\Longleftrightarrow\quad0\,\in\,\sigma\big(\mathfrak{E}_{-+}(\epsilon,\lambda)\big),
$$
where $\mathfrak{E}_{-+}(\epsilon,\lambda)$ is considered as a bounded self-adjoint operator on $\big[L^2(\X)\big]^N$.
\end{description}

In order to prove the second point of Proposition \ref{P.0.4} we shall use the {\it magnetic translations} $T_{\epsilon,a}:=\sigma_{A_\epsilon(a)}\tau_a$ for any $a\in\X$, that define a family of unitary operators in $L^2(\X)$.

\begin{lemma}\label{L.8.1}
For any two families of Hilbert spaces with temperate variation $\{\mathcal{A}_\xi\}_{\xi\in\X^*}$ and $\{\mathcal{B}_\xi\}_{\xi\in\X^*}$ and for any operator-valued symbol $q\in S^0_0\big(\X;\mathbb{B}(\mathcal{A}_\bullet;\mathcal{B}_\bullet)\big)$ the following equality holds:
\begin{equation}\label{8.5}
T_{\epsilon,a}\mathfrak{Op}^{A_\epsilon}(q)\ =\ \mathfrak{Op}^{A_\epsilon}\big((\tau_a\otimes\id)q\big)T_{\epsilon,a},\qquad\forall a\in\X.
\end{equation}
\end{lemma}
\begin{proof}
From Lemma \ref{L.A.18} it follows that 
\begin{equation}\label{8.6}
\tau_a\mathfrak{Op}^{A_\epsilon}(q)\ =\ \mathfrak{Op}^{\tau_aA_\epsilon}\big((\tau_a\otimes\id)q\big)\tau_a.
\end{equation}
We notice that $\tau_aA_\epsilon=A_\epsilon\,-\,A_\epsilon(a)$, so that the equality \eqref{8.6} becomes
\begin{equation}\label{8.7}
\tau_a\mathfrak{Op}^{A_\epsilon}(q)\ =\ \mathfrak{Op}^{(A_\epsilon\,-\,A_\epsilon(a))}\big((\tau_a\otimes\id)q\big)
\tau_a.
\end{equation}
Then, for any $u\in\mathscr{S}(\X;\mathcal{A}_0)$ and for any $x\in\X$ we can write that
$$
\left(\sigma_{-A_\epsilon(a)}\mathfrak{Op}^{A_\epsilon}(q)\sigma_{A_\epsilon(a)}u\right)(x)\ =\ \int_{\Xi}e^{i<\eta,x-y>}e^{-i<A_\epsilon(a),x-y>}\omega_{A_\epsilon}(x,y)\,q\big(\frac{x+y}{2},\eta\big)\,u(y)\,dy\,\dbar\eta.
$$
Noticing that $<A_\epsilon(a),x-y>=-\int_{[x,y]}A_\epsilon(a)$, the above formula implies that 
\begin{equation}\label{8.8}
\mathfrak{Op}^{A_\epsilon}(q)\sigma_{A_\epsilon(a)}\ =\ \sigma_{A_\epsilon(a)}\mathfrak{Op}^{(A_\epsilon-A_\epsilon(a))}(q).
\end{equation}
From \eqref{8.7} and \eqref{8.8} we deduce that
$$
T_{\epsilon,a}\mathfrak{Op}^{A_\epsilon}(q)\ =\ \sigma_{A_\epsilon(a)}\tau_a\mathfrak{Op}^{A_\epsilon}(q)\ =\ \sigma_{A_\epsilon(a)}\mathfrak{Op}^{(A_\epsilon\,-\,A_\epsilon(a))}\big((\tau_a\otimes\id)q\big)
\tau_a\ =\ \mathfrak{Op}^{A_\epsilon}\big((\tau_a\otimes\id)q\big)
\sigma_{A_\epsilon(a)}\tau_a\ =
$$
$$
=\ \mathfrak{Op}^{A_\epsilon}\big((\tau_a\otimes\id)q\big)T_{\epsilon,a}.
$$
\end{proof}

\begin{description}
\item[Proof of Proposition \ref{P.0.4} (2)]
The operator $\mathcal{P}_{\epsilon,\lambda}:=\mathfrak{Op}\big(\mathcal{P}_\epsilon(\cdot,\cdot,\lambda)\big)$ from Theorem \ref{T.4.3} has its symbol defined in \eqref{4.6}:
$$
\mathcal{P}_\epsilon(x,\xi,\lambda)\ :=\ \left(
\begin{array}{cc}
\mathfrak{q}_\epsilon(x,\xi)\,-\,\lambda&R_-(\xi)\\
R_+(\xi)&0
\end{array}
\right),
\qquad\forall(x,\xi)\in\Xi,\ \forall(\lambda,\epsilon)\in I\times[-\epsilon_0,\epsilon_0]
$$
where $\mathfrak{q}_\epsilon(x,\xi):=\mathfrak{Op}\big(\widetilde{p}_\epsilon(x,\cdot,\xi,\cdot)\big)$ with $\widetilde{p}_\epsilon(x.y.\xi,\eta):=p_\epsilon(x,y,\xi+\eta)$. As we have noticed from the beginning, we suppose that our symbol $p_\epsilon$ does not depend on the first variable so that neither the operator-valued symbol $\mathcal{P}_\epsilon$ will not depend on the first variable. From Lemma \ref{L.8.1} the operator 
$$
\mathcal{P}_{\epsilon,\lambda}:\mathscr{S}\big(\X;\mathcal{K}_{m,0}\times\mathbb{C}^N\big)\rightarrow
\mathscr{S}\big(\X;\mathcal{K}_0\times\mathbb{C}^N\big)
$$
commutes with the family $\{T_{\epsilon,a}\otimes\id_{\mathcal{K}_0\times
\mathbb{C}^N}\}_{a\in\X}$. Then its inverse appearing in Theorem \ref{T.4.3}, 
$$
\mathcal{E}_{\epsilon,\lambda}:\mathscr{S}\big(\X;\mathcal{K}_0\times\mathbb{C}^N\big)\rightarrow\mathscr{S}\big(\X;\mathcal{K}_{m,0}\times\mathbb{C}^N\big)
$$
also commutes with the family $\{T_{\epsilon,a}\otimes\id_{\mathcal{K}_0\times
\mathbb{C}^N}\}_{a\in\X}$. From this property we deduce that also the operator $\mathfrak{E}_{-+}(\epsilon,\lambda):L^2(\X;\mathbb{C}^N)\rightarrow L^2(\X;\mathbb{C}^N)$ commutes with the family $\{T_{\epsilon,a}\otimes\id_{\mathbb{C}^N}\}_{a\in\X}$. Using Lemma \ref{L.8.1} once again we deduce that :
$$
\mathfrak{Op}^{A_\epsilon}\big(E^{-+}_{\epsilon,\lambda}\big)\ =\ \mathfrak{E}_{-+}(\epsilon,\lambda)\ =\ \left[T_{\epsilon,a}\otimes\id_{\mathbb{C}^N}\right]
\mathfrak{E}_{-+}(\epsilon,\lambda)\left[T_{\epsilon,a}\otimes\id_{\mathbb{C}^N}
\right]^{-1}\ =\ \mathfrak{Op}^{A_\epsilon}\big((\tau_a\otimes\id)E^{-+}_{\epsilon,\lambda}\big),\ \forall a\in\X.
$$
We conclude that $E^{-+}_{\epsilon,\lambda}(x,\xi)=E^{-+}_{\epsilon,\lambda}(x-a,\xi)$ for any $(x,\xi)\in\Xi$ and for any $a\in\X$. It follows that $E^{-+}_{\epsilon,\lambda}(x,\xi)=E^{-+}_{\epsilon,\lambda}(0,\xi)$ for any $(x,\xi)\in\Xi$. The $\Gamma^*$-periodicity follows as in the general case (see the proof of Lemma \ref{L.6.5}).
\end{description}

\section{Appendices}
\setcounter{equation}{0}
\setcounter{theorem}{0}

\subsection{Magnetic pseudodifferential operators with operator-valued symbols}

\begin{definition}\label{D.A.1}
A family of Hilbert spaces $\{\mathcal{A}_\xi\}_{\xi\in\X^*}$ (indexed by the points in the {\it momentum space}) is said {\it to have temperate variation} when it verifies the following two conditions:
\begin{enumerate}
\item $\mathcal{A}_\xi\ =\ \mathcal{A}_\eta$ as complex vector spaces $\forall(\xi,\eta)\in[\X^*]^2$.
\item There exist $C>0$ and $M\geq0$ such that $\forall u\in\mathcal{A}_0$ we have the estimation:
\begin{equation}\label{A.1}
\|u\|_{\mathcal{A}_\xi}\ \leq\ C<\xi-\eta>^M\|u\|_{\mathcal{A}_\eta},\qquad\forall(\xi,\eta)\in[\X^*]^2.
\end{equation}
\end{enumerate}
\end{definition}
\begin{ex}\label{E.A.2}
We can take $\mathcal{A}_\xi=\mathcal{H}^s(\X)$, with any $s\in\mathbb{R}$ endowed with the $\xi$-dependent norm:
$$
\|u\|_{\mathcal{A}_\xi}\ :=\ \left(\int_\X<\xi+\eta>^{2s}|\hat{u}(\eta)|^2d\eta\right)^{1/2},\qquad\forall u\in\mathcal{H}^s(\X),\ \forall\xi\in\X^*.
$$
The inequality \eqref{A.1} clearly follows from the well known inequality:
\begin{equation}\label{A.2}
<\xi+\eta>^{2s}\ \leq\ C_s<\zeta+\eta>^{2s}<\xi-\zeta>^{2|s|},\qquad\forall(\xi,\eta,\zeta)\in[\X^*]^3,
\end{equation}
where the constant $C_s$ only depends on $s\in\mathbb{R}$. For this specific family we shall use the shorter notation $\mathcal{A}_\xi\equiv\mathcal{H}^s_\xi(\X)$.
\end{ex}

\begin{definition}\label{D.A.3}
Suppose given two families of Hilbert spaces with tempered variation $\{\mathcal{A}_\xi\}_{\xi\in\X^*}$ and $\{\mathcal{B}_\xi\}_{\xi\in\X^*}$; suppose also given $m\in\mathbb{R}$, $\rho\in[0,1]$ and $\mathcal{Y}$ a finite dimensional real vector space. A function $p\in C^\infty\big(\mathcal{Y}\times\X^*;\mathbb{B}(\mathcal{A}_0;\mathcal{B}_0)\big)$ is called {\it an operator-valued symbol of class} $S^m_\rho\big(\mathcal{Y};\mathbb{B}(\mathcal{A}_\bullet;\mathcal{B}_\bullet)\big)$ when it verifies the following property:
\begin{equation}\label{A.3}
\begin{array}{l}
\forall\alpha\in\mathbb{N}^{\dim\mathcal{Y}},\,\forall\beta\in\mathbb{N}^d,\ \exists C_{\alpha,\beta}>0:\\
\left\|\big(\partial^\alpha_y\partial^\beta_\xi p\big)(y,\xi)\right\|_{\mathbb{B}(\mathcal{A}_\xi;\mathcal{B}_\xi)}\ \leq\ C_{\alpha,\beta}<\xi>^{m-\rho|\beta|},\qquad\forall(y,\xi)\in\mathcal{Y}\times\X^*.
\end{array}
\end{equation}
\end{definition}
The space $S^m_\rho\big(\mathcal{Y};\mathbb{B}(\mathcal{A}_\bullet;\mathcal{B}_\bullet)\big)$ endowed with the family of seminorms $\nu_{\alpha,\beta}$ defined as being the smallest constants $C_{\alpha,\beta}$ that satisfy the defining property \eqref{A.3} is a metrizable locally convex linear topological space. In case we have for any $\xi\in\X^*$ that $\mathcal{A}_\xi=\mathcal{A}_0$ and $\mathcal{B}_\xi=\mathcal{B}_0$ as algebraic and topological structures, then we use the notation $S^m_\rho\big(\mathcal{Y};\mathbb{B}(\mathcal{A}_0;\mathcal{B}_0)\big)$. If moreover we have that $\mathcal{A}_0=\mathcal{B}_0=\mathbb{C}$, then we use the simple notation $S^m_\rho(\mathcal{Y})$.

\begin{ex}\label{E.A.4}
If $p\in S^m_1(\X)$ and if for any $\xi\in\X^*$ we denote by $p_\xi:=(\id\otimes\tau_{-\xi})p$, by $P_\xi:=\mathfrak{Op}\big(p_\xi\big)$ and by $\mathfrak{p}$ the application $\Xi\ni(x,\xi)\mapsto P_\xi\in\mathbb{B}(\mathcal{H}^{s+m}_\xi(\X);\mathcal{H}^s_\xi(\X))$, for some $s\in\mathbb{R}$, we can prove that $\mathfrak{p}$ is an operator valued symbol of class $S^0_0\big(\X;\mathbb{B}(\mathcal{H}^{s+m}_\bullet(\X);\mathcal{H}^s_\bullet(\X))\big)$. Moreover the map $S^m_1(\X)\ni p\mapsto\mathfrak{p}\in S^0_0\big(\X;\mathbb{B}(\mathcal{H}^{s+m}_\bullet(\X);\mathcal{H}^s_\bullet(\X))\big)$ is continuous.
\end{ex}
\begin{description}
\item[] \hspace{1cm} In fact, let us recall that for any $\xi\in\X^*$ we have denoted by $\sigma_\xi$ the multiplication operator with the function $e^{i<\xi,\cdot>}$ on the space $\mathscr{S}^\prime(\X)$. Then, for any $u\in\mathscr{S}(\X)$ and for any $\xi\in\X^*$ we have that $u\in\mathcal{H}^{s+m}_\xi(\X)$ and we can write:
$$
\big(\sigma_{-\xi}P_0\sigma_\xi u\big)(x)=\int_\Xi e^{i<\eta-\xi,x-y>}p\big(\frac{x+y}{2},\eta\big)u(y)\,dy\,\dbar\eta=\int_\Xi e^{i<\eta,x-y>}p\big(\frac{x+y}{2},\eta+\xi\big)u(y)\,dy\,\dbar\eta=\big(P_\xi u\big)(x)
$$
and we conclude that
\begin{equation}\label{A.4}
P_\xi\ =\ \sigma_{-\xi}P_0\sigma_\xi,\qquad\forall\xi\in\X^*.
\end{equation}

On the other side, for any $\xi\in\X^*$ we notice that $p_\xi$ is a symbol of class $S^m_1(\X)$ and thus, the usual Weyl calculus implies that $P_\xi\in\mathbb{B}\big(\mathcal{H}^{s+m}(\X);\mathcal{H}^s(\X)\big)$ for any $s\in\mathbb{R}$. We notice easily that for any multi-index $\beta\in\mathbb{N}^d$ we can write:
\begin{equation}\label{A.4a}
\partial^\beta_\xi P_\xi\ =\ \mathfrak{Op}\big(\partial^\beta_\xi p_\xi\big)
\end{equation}
so that we conclude that $P_\xi\in C^\infty\big(\Xi;\mathbb{B}(\mathcal{H}^{s+m}(\X);\mathcal{H}^s(\X))\big)$ (constant with respect to the variable $x\in\X$) for any $s\in\mathbb{R}$.

Let us further notice that for any $u\in\mathscr{S}(\X)$ and any $\xi\in\X^*$ we have the equalities:
\begin{equation}\label{A.5}
\widehat{\sigma_\xi u}\ =\ \tau_\xi\hat{u},
\end{equation}
\begin{equation}\label{A.6}
\left\|\sigma_{-\xi}u\right\|^2_{\mathcal{H}^s_\xi(\X)}\ =\ \int_{\X^*}<\xi+\eta>^{2s}\left|\hat{u}(\xi+\eta)\right|^2d\eta\ =\ \|u\|^2_{\mathcal{H}^s(\X)}.
\end{equation}
Coming back to \eqref{A.4} we deduce that for any $u\in\mathscr{S}(\X)$ and any $\xi\in\X^*$ we have the estimation:
$$
\left\|P_\xi u\right\|^2_{\mathcal{H}^s_\xi(\X)}\ =\ \left\|P_0\sigma_\xi u\right\|^2_{\mathcal{H}^s(\X)}\ \leq\ C_s\left\|\sigma_\xi u\right\|^2_{\mathcal{H}^{s+m}(\X)}\ =\ C_s\|u\|^2_{\mathcal{H}^{s+m}_\xi(\X)}.
$$
Using also the equality \eqref{A.4a} we obtain similar estimations for the derivatives of $P_\xi$ and conclude that $\mathfrak{p}\in S^0_0\big(\X;\mathbb{B}(\mathcal{H}^{s+m}_\bullet(\X);\mathcal{H}^s_\bullet(\X))\big)$ and we have the continuity of the map $S^m_1(\X)\ni p\mapsto\mathfrak{p}\in S^0_0\big(\X;\mathbb{B}(\mathcal{H}^{s+m}_\bullet(\X);\mathcal{H}^s_\bullet(\X))\big)$.
\end{description}

\begin{definition}\label{D.A.5}
We denote by $S^m_{\rho,\epsilon}\big(\X^2;\mathbb{B}(\mathcal{A}_\bullet;\mathcal{B}_\bullet)\big)$ the linear space of families $\{p_\epsilon\}_{|\epsilon|\leq\epsilon_0}$ satisfying the following three conditions:
\begin{enumerate}
\item $\forall\epsilon\in[-\epsilon_0,\epsilon_0],\ p_\epsilon\in S^m_\rho\big(\X^2;\mathbb{B}(\mathcal{A}_\bullet;\mathcal{B}_\bullet)\big)$ uniformly with respect to $\epsilon\in[-\epsilon_0,\epsilon_0]$,
\item $\underset{\epsilon\rightarrow0}{\lim}\,p_\epsilon\ =\ p_0$ in $S^m_\rho\big(\X^2;\mathbb{B}(\mathcal{A}_\bullet;\mathcal{B}_\bullet)\big)$,
\item Denoting the variable in $\X^2$ by $(x,y)$, for any multi-index $\alpha\in\mathbb{N}^d$ with $|\alpha|\geq1$ we have that $\underset{\epsilon\rightarrow0}{\lim}\,\partial^\alpha_xp_\epsilon\ =\ 0$ in $S^m_\rho\big(\X^2;\mathbb{B}(\mathcal{A}_\bullet;\mathcal{B}_\bullet)\big)$,
\end{enumerate}
endowed with the natural locally convex topology of symbols of H\"{o}rmander type.
\end{definition}

As in the case of Definition \ref{D.A.3}, in case we have for any $\xi\in\X^*$ that $\mathcal{A}_\xi=\mathcal{A}_0$ and $\mathcal{B}_\xi=\mathcal{B}_0$ as algebraic and topological structures, then we use the notation $S^m_{\rho,\epsilon}\big(\X^2;\mathbb{B}(\mathcal{A}_0;\mathcal{B}_0)\big)$. If moreover we have that $\mathcal{A}_0=\mathcal{B}_0=\mathbb{C}$, then we use the simple notation $S^m_{\rho,\epsilon}(\X^2)$.

For the families of symbols of type $S^m_{\rho,\epsilon}\big(\X^2;\mathbb{B}(\mathcal{A}_\bullet;\mathcal{B}_\bullet)\big)$ that do not depend on the first variable $x$ in $\X^2$ we shall use the notation $S^m_{\rho,\epsilon}\big(\X;\mathbb{B}(\mathcal{A}_\bullet;\mathcal{B}_\bullet)\big)$.

Let us also notice the following canonical injection:
$$
S^m_{\rho,\epsilon}\big(\X;\mathbb{B}(\mathcal{A}_\bullet;\mathcal{B}_\bullet)\big)\ni p_\epsilon\mapsto\id\otimes p_\epsilon\in S^m_{\rho,\epsilon}\big(\X^2;\mathbb{B}(\mathcal{A}_\bullet;\mathcal{B}_\bullet)\big).
$$

\begin{remark}\label{R.A.6}
We can obtain the following description of the space $S^m_{\rho,\epsilon}\big(\X^2;\mathbb{B}(\mathcal{A}_\bullet;\mathcal{B}_\bullet)\big)$:
\end{remark}
\begin{description}
\item[•] \hspace{1cm} If $\{p_\epsilon\}_{|\epsilon|\leq\epsilon_0}\in S^m_{\rho,\epsilon}\big(\X^2;\mathbb{B}(\mathcal{A}_\bullet;\mathcal{B}_\bullet)\big)$ then we can write the first order Taylor expansion:
$$
p_\epsilon(x,y,\eta)\ =\ p_\epsilon(0,y,\eta)\,+\,\left\langle x\,,\,\int_0^1\big(\nabla_x p_\epsilon\big)(tx,y,\eta)\,dt\right\rangle.
$$
Using conditions (2) and (3) from Definition \ref{D.A.5} we deduce that
$$
p_0(x,y,\eta)\ =\ p_0(0,y,\eta).
$$
We denote by: $p_0(y,\eta)\ :=\ p_0(0,y,\eta)$ as element in $S^m_\rho\big(\X;\mathbb{B}(\mathcal{A}_\bullet;\mathcal{B}_\bullet)\big)$ and by:
$$
r_\epsilon(x,y,\eta)\ :=\ p_\epsilon(x,y,\eta)\,-\,p_0(y,\eta).
$$
This last symbol is of class $S^m_\rho\big(\X^2;\mathbb{B}(\mathcal{A}_\bullet;\mathcal{B}_\bullet)\big)$ uniformly with respect to $\epsilon\in[-\epsilon_0,\epsilon_0]$.

We have: $r_0=0$ and $\underset{\epsilon\rightarrow0}{\lim}\,r_\epsilon\,=\,0$ in $S^m_\rho\big(\X^2;\mathbb{B}(\mathcal{A}_\bullet;\mathcal{B}_\bullet)\big)$ and evidently:
\begin{equation}\label{A.7}
p_\epsilon(x,y,\eta)\ =\ p_0(y,\eta)\,+\,r_\epsilon(x,y,\eta).
\end{equation}
Reciprocally, if $p_0$ and $r_\epsilon$ from the equality \eqref{A.7} have the properties stated above, then $\{p_\epsilon\}_{|\epsilon|\leq\epsilon_0}$ belongs to $S^m_{\rho,\epsilon}\big(\X^2;\mathbb{B}(\mathcal{A}_\bullet;\mathcal{B}_\bullet)\big)$.
\end{description}

We shall consider now {\it magnetic pseudodifferential operators associated to operator-valued symbols}. Let us first consider $p\in S^m_\rho\big(\X;\mathcal{A}_\bullet;\mathcal{B}_\bullet)\big)$ and a magnetic field $B$ with components of class $BC^\infty(\X)$; this magnetic field can always be associated with a vector potential $A$ with components of class $C^\infty_{\text{\sf pol}}(\X)$ (as for example in the transversal gauge). Let us recall the notation $\omega_A(x,y):=\exp\{-i\int_{[x,y]}A\}$. For any $u\in\mathscr{S}(\X;\mathcal{A}_0)$ we can define the oscillating integral (its existence following from the next Proposition):
\begin{equation}\label{A.8}
\big[\mathfrak{Op}^A(p)u\big](x)\ :=\ \int_\Xi e^{i<\eta,x-y>}\omega_A(x,y)\,p\big(\frac{x+y}{2},\eta\big)u(y)\,dy\,\dbar\eta,\qquad\forall x\in\X.
\end{equation}
\begin{proposition}\label{P.A.7}
Under the hypothesis described in the paragraph above \eqref{A.8} the following facts are true:
\begin{enumerate}
\item The integral in \eqref{A.8} exists for any $x\in\X$ as oscillating Bochner integral and defines a function $\mathfrak{Op}^A(p)u\in\mathscr{S}(\X;\mathcal{B}_0)$.
\item The map $\mathfrak{Op}^A(p):\mathscr{S}(\X;\mathcal{A}_0)\rightarrow\mathscr{S}(\X;\mathcal{B}_0)$ defined by \eqref{A.8} and point (1) above is linear and continuous.
\item The formal adjoint $\big[\mathfrak{Op}^A(p)\big]^*:\mathscr{S}(\X;\mathcal{B}_0\big)\rightarrow\mathscr{S}(\X;\mathcal{A}_0)$ of the linear continuous operator defined at point (2) above is equal to $\mathfrak{Op}^A(p^*)$ where $p^*\in S^{m^\prime}_\rho\big(\X;\mathbb{B}(\mathcal{B}_\bullet;\mathcal{A}_\bullet)\big)$ where $m^\prime=m+2(M_\mathcal{A}+M_\mathcal{B})$ and $p^*(x,\xi):=\big[p(x,\xi)\big]^*$ (the adjoint in $\mathbb{B}(\mathcal{A}_0;\mathcal{B}_0)$.
\item The operator $\mathfrak{Op}^A(p)$ extends in a natural way to a linear continuous operator $\mathscr{S}^\prime(\X;\mathcal{A}_0)\rightarrow\mathscr{S}^\prime(\X;\mathcal{B}_0)$ that we denote in the same way.
\end{enumerate}
\end{proposition}
\begin{proof}
Let us first prove the first two points of the Proposition.

Fix some $u\in\mathscr{S}(\X;\mathcal{A}_0)$ and for the begining let us suppose that $p(y,\eta)=0$ for $|\eta|\geq R$, with some $R>0$. Then, for any $x\in\X$, the integral in \eqref{A.8} exists as a Bochner integral of a $\mathcal{B}_0$-valued function. Let us notice that in this case we can integrate by parts in \eqref{A.8} and use the identities:
$$
e^{i<\eta,x-y>}=<x-y>^{-2N_1}\big[(\id-\Delta_\eta)^{N_1}e^{i<\eta,x-y>}\big],\ \forall N_1\in\mathbb{N};\quad
e^{-i<\eta,y>}=<\eta>^{-2N_2}\big[(\id-\Delta_y)^{N_2}e^{-i<\eta,y>}\big],\ \forall N_2\in\mathbb{N}.
$$
In this way we obtain the equality:
\begin{equation}\label{A.9}
\big[\mathfrak{Op}^A(p)u\big](x)\ =
\end{equation}
$$
=\ \int_\Xi e^{i<\eta,x-y>}<x-y>^{-2N_1}(\id-\Delta_\eta)^{N_1}\left[<\eta>^{-2N_2}(\id-\Delta_y)^{N_2}\left(\omega_A(x,y)\big[p\big(\frac{x+y}{2},\eta\big)u(y)\big]\right)\right]dy\dbar\eta.
$$
From this one easily obtains the following estimation: $\exists C(N_1,N_2)>0$, $\exists k(N_2)\in\mathbb{N}$ such that for any $l\in\mathbb{N}$ we have
$$
\left\|\big[\mathfrak{Op}^A(p)u\big](x)\right\|_{\mathcal{B}_0}\ \leq
$$
$$
\leq\ C\left[\int_\Xi<x-y>^{-2N_1}<\eta>^{-2N_2}\big(<x>+<y>\big)^{k(N_2)}
<\eta>^{m+M_{\mathcal{B}}+M_{\mathcal{A}}}<y>^{-l}dy\dbar\eta\right]\times
$$
$$
\times\underset{|\alpha|\leq2N_2}{\sup}\underset{y\in\X}{\sup}\,<y>^l\left\|\big(\partial^\alpha u\big)(y)\right\|_{\mathcal{A}_0},
$$
where $M_\mathcal{A}$ and $M_\mathcal{B}$ are the constants appearing in condition \eqref{A.1} with respect to each of the two families $\{{A}_\xi\}_{\xi\in\X^*}$ and $\{\mathcal{B}_\xi\}_{\xi\in\X^*}$. We make the following choices:
$$
2N_2\geq m+M_\mathcal{A}+M_\mathcal{B}+d,\ l=2N_1+k(N_2)+d+1,\ 2N_1\geq k(N_2)
$$
and we obtain that
\begin{equation}\label{A.10}
\left\|\big[\mathfrak{Op}^A(p)u\big](x)\right\|_{\mathcal{B}_0}\ \leq\ C(N_1)<x>^{-2N_1+k(N_2)}\underset{|\alpha|\leq2N_2}{\sup}\ \underset{y\in\X}{\sup}\,<y>^l\left\|\big(\partial^\alpha u\big)(y)\right\|_{\mathcal{A}_0},\qquad\forall x\in\X.
\end{equation}

Similar estimations may be obtained for the derivatives $\partial^\beta_x\mathfrak{Op}^A(p)u$ and this finishes the proof of the first two points of the Proposition for the "compact support" situation we have considered first.

For the general case we consider a cut-off function $\varphi\in C^\infty_0(\X^*)$ that is equal to 1 for $|\eta|\leq1$. For any $R\geq1$ we define then the 'approximating' symbols $p_R(y,\eta):=\varphi(\eta/R)p(y,\eta)$ that has compact support in the $\eta$-variable and belongs to $S^m_\rho\big(\X;\mathbb{B}(\mathcal{A}_\bullet;\mathcal{B}_\bullet)\big)$ uniformly with respect to $R\in[1,\infty)$. We write the operator $\mathfrak{Op}^A(p_R)$ under a form similar to \eqref{A.9} with the choice of the parameters $N_1$ and $N_2$ as above; then the Dominated Convergence Theorem allows us to take the limit $R\nearrow\infty$ obtaining now for $\mathfrak{Op}^A(p)$ the integral expression \eqref{A.9} that is well defined and verifies the estimation \eqref{A.10}.

The last two points of the statement of the Proposition follow in a standard way from the equality:
\begin{equation}\label{A.11}
\int_\X\left(\big[\mathfrak{Op}^A(p)u\big](x)\,,\,v(x)\right)_{\mathcal{B}_0}dx\ =\ \int_\X\left(u(y)\,,\,\big[\mathfrak{Op}^A(p^*)v\big](y)\right)_{\mathcal{A}_0}dy,\qquad\forall(u,v)\in\mathscr{S}(\X;\mathcal{A}_0)\times\mathscr{S}(\X;\mathcal{B}_0).
\end{equation}
\end{proof}

\begin{ex}\label{E.A.8}
Let us consider a family $\{p_\epsilon\}_{|\epsilon|\leq\epsilon_0}$ of class $S^m_{1,\epsilon}(\X^2)$ and let us define, as in Subsection \ref{S.1.3}:
$$
\widetilde{p}_\epsilon(x,y,\xi,\eta)\ :=\ p_\epsilon(x,y,\xi+\eta);\qquad\mathfrak{q}_\epsilon(x,\xi)\ :=\ \mathfrak{Op}\big(\widetilde{p}_\epsilon(x,\cdot,\xi,\cdot)\big). 
$$
The following two statements are true:
\begin{enumerate}
\item $\{\mathfrak{q}_\epsilon\}_{|\epsilon|\leq\epsilon_0}\,\in\,S^0_{0,\epsilon}\big(\X;\mathbb{B}(\mathcal{H}^{s+m}_\bullet(\X);\mathcal{H}^s_\bullet(\X)\big)$ for any $s\in\mathbb{R}$.
\item If the family of magnetic fields $\{B_\epsilon\}_{|\epsilon|\leq\epsilon_0}$ satisfies Hypothesis H.1 and if the associated vector potentials are choosen as in \eqref{0.28}, then we have that:
\begin{equation}\label{A.12}
\mathfrak{Op}^{A_\epsilon}(\mathfrak{q}_\epsilon)\ \in\ \mathbb{B}\big(\mathscr{S}(\X;\mathcal{H}^{s+m}(\X));\mathscr{S}(\X;\mathcal{H}^{s}(\X))\big)\ \cap\ \mathbb{B}\big(\mathscr{S}^\prime(\X;\mathcal{H}^{s+m}(\X));\mathscr{S}^\prime(\X;\mathcal{H}^{s}(\X))\big),
\end{equation}
for any $s\in\mathbb{R}$. Moreover, 
\begin{equation}\label{A.13}
\mathfrak{Op}^{A_\epsilon}(\mathfrak{q}_\epsilon)\ \in\ \mathbb{B}\big(\mathscr{S}(\X^2);\mathscr{S}(\X^2)\big)\ \cap\ \mathbb{B}\big(\mathscr{S}^\prime(\X^2);\mathscr{S}^\prime(\X^2)\big),
\end{equation}
and all the continuities are uniform with respect to $\epsilon\in[-\epsilon_0,\epsilon_0]$.
\end{enumerate}
\end{ex}
\begin{proof}
We begin by verifying point (1). By similar arguments as in Example \ref{E.A.4} we prove that for any $\epsilon\in[-\epsilon_0,\epsilon_0]$ and $s\in\mathbb{R}$ we have that $\mathfrak{q}_\epsilon\,\in\,S^0_{0}\big(\X;\mathbb{B}(\mathcal{H}^{s+m}_\bullet(\X);\mathcal{H}^s_\bullet(\X))\big)$ uniformly with respect to $\epsilon\in[-\epsilon_0,\epsilon_0]$ and the following application is continuous:
$$
S^m_1(\X^2)\ni p_\epsilon\mapsto\mathfrak{q}_\epsilon\in S^0_{0}\big(\X;\mathbb{B}(\mathcal{H}^{s+m}_\bullet(\X);\mathcal{H}^s_\bullet(\X)\big),\qquad\forall s\in\mathbb{R},
$$
uniformly with respect to $\epsilon\in[-\epsilon_0,\epsilon_0]$. Point (1) follows then clearly.

Concerning the second point of the Proposition, let us notice that \eqref{A.12} and the uniformity with respect to $\epsilon\in[-\epsilon_0,\epsilon_0]$ follow easily from Proposition \ref{P.A.7} and its proof. 

In order to prove \eqref{A.13} let us notice that $\widetilde{p}_\epsilon^\prime(x,\cdot,\xi,\cdot):=<\xi>^{-|m|}\widetilde{p}_\epsilon(x,\cdot,\xi,\cdot)$ defines a symbol of class $S^m_0(\X)$ uniformly with respect to $((x,\xi),\epsilon)\in\Xi\times[-\epsilon_0,\epsilon_0]$ and we can view the element $\widetilde{p}_\epsilon^\prime$ as a function in $BC^\infty\big(\Xi;S^m_0(\X)\big)$. Then, the operator-valued symbol $\mathfrak{q}^\prime_\epsilon(x,\xi):=<\xi>^{-|m|}\mathfrak{q}_\epsilon(x,\xi)$ has the following property:
$$
\forall(\alpha,\beta)\in[\mathbb{N}^d]^2,\ \big(\partial^\alpha_x\partial^\beta_\xi\mathfrak{q}^\prime_\epsilon\big)(x,\xi)\,\in\,\mathbb{B}\big(\mathscr{S}(\X)\big),
$$
uniformly with respect to $((x,\xi),\epsilon)\in\Xi\times[-\epsilon_0,\epsilon_0]$.

Denoting $\mathfrak{s}_{s}(x,\xi):=<\xi>^{s}$ for any $s\in\mathbb{R}$ and writing $\mathfrak{Op}^{A_\epsilon}(\mathfrak{q}_\epsilon)=\mathfrak{Op}^{A_\epsilon}(\mathfrak{s}_{|m|}\mathfrak{q}^\prime)$, the proof of Proposition \ref{P.A.7} implies \eqref{A.13} uniformly with respect to $\epsilon\in[-\epsilon_0,\epsilon_0]$.
\end{proof}

\subsection{Some spaces of periodic distributions}

We shall use the following notations:
\begin{itemize}
\item $\mathscr{S}^\prime_\Gamma(\X)\ :=\ \left\{u\in\mathscr{S}^\prime(\X)\,\mid\,\tau_\gamma u=u,\ \forall\gamma\in\Gamma\right\}$, the space of $\Gamma$-periodic distributions on $\X$.
\item $\mathscr{S}(\mathbb{T}):=C^\infty(\mathbb{T})$ with the usual Fr\'{e}chet topology, We have the evident identification:
$$
\mathscr{S}(\mathbb{T})\ \cong\ \left\{\varphi\in\mathscr{E}(\X)\,\mid\,\tau_\gamma\varphi=\varphi,\ \forall\gamma\in\Gamma\right\}\ =\ \mathscr{S}^\prime_\Gamma(\X)\cap\mathscr{E}(\X).
$$
\item $\mathscr{S}^\prime(\mathbb{T})$ is the dual of $\mathscr{S}(\mathbb{T})$. 
\item We shall denote by $<\cdot,\cdot>_{\mathbb{T}}$ the natural bilinear map defined by the duality relation on $\mathscr{S}^\prime(\mathbb{T})\times\mathscr{S}(\mathbb{T})$ and by $(\cdot,\cdot)_{\mathbb{T}}$ the natural  sesquilinear map on $\mathscr{S}^\prime(\mathbb{T})\times\mathscr{S}(\mathbb{T})$ obtaind by extending the scalar product from $L^2(\mathbb{T})$.
\end{itemize}

\begin{remark}\label{R.A.9}
It is well known that the spaces $\mathscr{S}^\prime_\Gamma(\X)$ and $\mathscr{S}^\prime(\mathbb{T})$ can be identified through the following topological isomorphism:
\begin{equation}\label{A.14}
\mathfrak{i}:\mathscr{S}^\prime(\mathbb{T})\overset{\sim}{\rightarrow}\mathscr{S}^\prime_\Gamma(\X);\qquad<\mathfrak{i}(u),\varphi>:=<u,\underset{\gamma\in\Gamma}{\sum}\tau_\gamma\varphi>_{\mathbb{T}},\qquad\forall(u,\varphi)\in\mathscr{S}^\prime(\mathbb{T})\times\mathscr{S}(\X).
\end{equation}
In order to give an explicit form to the inverse of the isomorphism $\mathfrak{i}$ let us fix some test function $\phi\in C^\infty_0(\X)$ such that $\underset{\gamma\in\Gamma}{\sum}\tau_\gamma\phi=1$ (it is easy to see that there exist enough such functions). Then we can easily verify that:
\begin{equation}\label{A.15}
<\mathfrak{i}^{-1}(v),\theta>_{\mathbb{T}}=<v,\phi\theta>,\qquad\forall(v,\theta)\in\mathscr{S}^\prime_\Gamma(\X)\times\mathscr{S}(\mathbb{T}).
\end{equation}
\end{remark}

\begin{remark}\label{R.A.10}
For any distribution $u\in\mathscr{S}^\prime_\Gamma(\X)\cong\mathscr{S}^\prime(\mathbb{T})$ we have the Fourier series decomposition:
\begin{equation}\label{A.16}
u\ =\ \underset{\gamma^*\in\Gamma^*}{\sum}\hat{u}(\gamma^*)\sigma_{\gamma^*},\qquad\hat{u}(\gamma^*):=|E|^{-1}<u,\sigma_{-\gamma^*}>_{\mathbb{T}},
\end{equation}
where $\sigma_{\gamma^*}(y):=e^{i<\gamma^*,y>},\forall y\in\mathbb{T},\forall\gamma^*\in\Gamma^*$ and the series converges as tempered distribution.

In particular, if $u\in L^2(\mathbb{T})$ we also have the equality:
\begin{equation}\label{A.17}
\|u\|^2_{L^2(\mathbb{T})}\ =\ |E|\underset{\gamma^*\in\Gamma^*}{\sum}|\hat{u}(\gamma^*)|^2.
\end{equation}
\end{remark}

\begin{remark}\label{R.A.11}
A simple computation allows to prove that for any $s\in\mathbb{R}$ and any $\gamma^*\in\Gamma^*$ the following equality is true in $\mathscr{S}^\prime(\X)$:
\begin{equation}\label{A.18}
<D>^s\sigma_{\gamma^*}\ =\ <\gamma^*>^s\sigma_{\gamma^*}.
\end{equation}
Using now \eqref{A.16} we deduce that $<D>^s$ induces on $\mathscr{S}^\prime(\mathbb{T})\cong\mathscr{S}^\prime_\Gamma(\X)$ a well-defined operator, denoted by $<D_\Gamma>^s$, explicitely given by the following formula:
\begin{equation}\label{A.19}
<D_\Gamma>^su\ :=\ \underset{\gamma^*\in\Gamma^*}{\sum}<\gamma^*>^s\hat{u}(\gamma^*)\sigma_{\gamma^*},\qquad\forall u\in\mathscr{S}^\prime(\mathbb{T}).
\end{equation}
\end{remark}

\begin{definition}\label{D.A.12}
Given any $s\in\mathbb{R}$ we define the complex linear space:
$$
\mathcal{H}^s(\mathbb{T})\ :=\ \left\{u\in\mathscr{S}^\prime(\mathbb{T})\,\mid\,<D_\Gamma>^su\in L^2(\mathbb{T})\right\}
$$
endowed with the quadratic norm $\|u\|_{\mathcal{H}^s(\mathbb{T})}:=\left\|<D_\Gamma>^su\right\|_{L^2(\mathbb{T})}$ for which it becomes a Hilbert space.
\end{definition}
Let us notice that for any $u\in\mathscr{S}^\prime(\mathbb{T})$ we have the following equivalence relation:
$$
u\,\in\,\mathcal{H}^s(\mathbb{T})\quad\Longleftrightarrow\quad|E|^{-1}\underset{\gamma^*\in\Gamma^*}{\sum}<\gamma^*>^{2s}|\hat{u}(\gamma^*)|^2\,<\,\infty,
$$
and the formula in the right hand side of the above equivalence relation is equal to $\|u\|^2_{\mathcal{H}^s(\mathbb{T})}$. From these facts it follows easily that $\mathscr{S}(\mathbb{T})$ is dense in $\mathcal{H}^s(\mathbb{T})$.

\begin{lemma}\label{L.A.13}
Let $p\in S^m_1(\X)$ and let us denote by $P:=\mathfrak{Op}(p)$. Then for any $s\in\mathbb{R}$ and for any $u\in\mathcal{H}^{s+m}_{\text{\sf loc}}(\X)\cap\mathscr{S}^\prime(\X)$ we have that $Pu\in\mathcal{H}^s_{\text{\sf loc}}(\X)\cap\mathscr{S}^\prime(\X)$.
\end{lemma}
\begin{proof}
It is clear from the definitions that we have $Pu\in\mathscr{S}^\prime(\X)$. For any relatively compact open subset $\Omega\subset\X$ we can choose two positive test functions $\psi$ and $\chi$ of class $C^\infty_0(\X)$ such that: $\psi=1$ on $\overline{\Omega}$ and $\chi=1$ on a neighborhood $V_\psi$ of the support of $\psi$. Then $\chi u\in\mathcal{H}^{s+m}(\X)$ and we know that $P\chi u\in\mathcal{H}^s(\X)$. Thus the Lemma will follow if we prove that the restriction of $P(1-\chi)u$ to $\Omega$ is of class $C^\infty(\Omega)$. For that, let us choose $\phi\in C^\infty_0(\Omega)$ and notice that
$$
\left\langle P(1-\chi)u,\phi\right\rangle\ =\ \left\langle u,(1-\chi)P^t(\psi\phi)\right\rangle\ =\ \left\langle u,\int_\X K(x,y)\phi(y)dy\right\rangle,
$$
where for any $N\in\mathbb{N}$ we have the following equality (obtained by the usual integration by parts method):
$$
K(x,y)\ :=\ \big(1-\chi(x)\big)\psi(y)\int_\X e^{i<\eta,x-y>}|x-y|^{-2N}\big[\Delta_\eta^N p\big(\frac{x+y}{2},-\eta\big)\big]\dbar\eta.
$$
It is evident that $K\in\mathscr{S}(\X^2)$, so that $v(y):=<u(\cdot),K(\cdot,y)>$ is a function of class $\mathscr{S}(\X)$ and we have the equality:
$$
\left\langle P(1-\chi)u,\phi\right\rangle\ =\ <v,\phi>,
$$
from which we conclude that $P(1\,-\,\chi)u\in C^\infty(\Omega)$.
\end{proof}

\begin{corollary}\label{C.A.14}
The space $\mathcal{H}^s(\mathbb{T})$ can be identified with the usual Sobolev space of order $s$ on  the torus that is defined as $\mathcal{H}^s_{\text{\sf loc}}(\X)\cap\mathscr{S}^\prime_\Gamma(\X)$.
\end{corollary}
\begin{proof}
The Corollary follows from Lemma \ref{L.A.13} using the fact that we have the equality $<D_\Gamma>^s=<D>^s$ on $\mathscr{S}^\prime_\Gamma(\X)$ and thus for any $u\in\mathscr{S}^\prime_\Gamma(\X)$ the following relations hold:
$$
u\in\mathcal{H}^s(\mathbb{T})\ \Leftrightarrow\ <D_\Gamma>^su\in L^2(\mathbb{T})\ \Leftrightarrow\ <D>^su\in L^2_{\text{\sf loc}}(\X)\cap\mathscr{S}^\prime_\Gamma(\X)\ \Leftrightarrow\ u\in\mathcal{H}^s_{\text{\sf loc}}(\X)\cap\mathscr{S}^\prime_\Gamma(\X).
$$
\end{proof}

\begin{remark}\label{R.A.15}
Let us notice the following facts to be used in the paper.
\begin{enumerate}
\item For any pair of real numbers $(s,t)$ the application $<D_\Gamma>^s:\mathcal{H}^{t+s}(\mathbb{T})\rightarrow\mathcal{H}^s(\mathbb{T})$ is an isometric isomorphism.
\item For any pair $(u,v)$ of test functions from $\mathscr{S}(\mathbb{T})$ the following equalities are true:
$$
(u,v)_{\mathcal{H}^s(\mathbb{T})}\ =\ \left(<D_\Gamma>^su,<D_\Gamma>^sv\right)_{\mathbb{T}}\ =\ \left(u,<D_\Gamma>^{2s}v\right)_{\mathbb{T}}.
$$
\item The dual of $\mathcal{H}^s(\mathbb{T})$ can be canonically identified with the space $\mathcal{H}^{-s}(\mathbb{T})$. In fact, using also the above remark, we can identify the dual of $\mathcal{H}^s(\mathbb{T})$ with itself via the operator $<D_\Gamma>^{-2s}$.
\end{enumerate}
\end{remark}

\begin{ex}\label{E.A.16}
For any $s\in\mathbb{R}$ and for any $\xi\in\X^*$ we define the following operator:
\begin{equation}\label{A.20}
<D_\Gamma+\xi>^s:\mathscr{S}^\prime(\mathbb{T})\rightarrow\mathscr{S}^\prime(\mathbb{T}),\qquad<D_\Gamma+\xi>^su:=\underset{\gamma^*\in\Gamma^*}{\sum}<\gamma^*+\xi>^s\hat{u}(\gamma^*)\sigma_{\gamma^*}.
\end{equation}
Considering $u\in\mathscr{S}^\prime_\Gamma(\X)$ we evidently have that $<D_\Gamma+\xi>^su=<D+\xi>^su$.
\end{ex}

We define the following complex linear space:
\begin{equation}\label{A.21}
\mathcal{K}_{s,\xi}\ :=\ \left\{u\in\mathscr{S}^\prime(\mathbb{T})\,\mid\,<D_\Gamma+\xi>^su\in L^2(\mathbb{T})\right\},
\end{equation}
endowed with the quadratic norm:
\begin{equation}\label{A.22}
\|u\|^2_{\mathcal{K}_{s,\xi}}\ :=\ \left\|<D_\Gamma+\xi>^su\right\|^2_{L^2(\mathbb{T})}\ =\ |E|^{-1}\underset{\gamma^*\in\Gamma^*}{\sum}<\gamma^*+\xi>^{2s}|\hat{u}(\gamma^*)|^2
\end{equation}
that defines a structure of Hilbert space on it.

It is clear that $\mathcal{K}_{s,\xi}=\mathcal{H}^s(\mathbb{T})$ as complex vector spaces and for $\xi=0$ even as Hilbert spaces (having the same scalar product). Similar arguments to those in Example \ref{E.A.2} show that the family $\{\mathcal{K}_{s,\xi}\}_{\xi\in\X^*}$ has temperate variation.

Coming back to Corollary \ref{C.A.14} we can consider the elements of $\mathcal{K}_{s,\xi}$ as distributions from $\mathcal{H}^s_{\text{\sf loc}}(\X)\cap\mathscr{S}^\prime_\Gamma(\X)$ and we can define the spaces:
\begin{equation}\label{A.23}
\mathcal{F}_{s,\xi}\ :=\ \left\{u\in\mathscr{S}^\prime(\X)\,\mid\,\sigma_{-\xi}u\in\mathcal{K}_{s,\xi}\right\}.
\end{equation}
It will become a Hilbert space isometrically isomorphic with $\mathcal{K}_{s,\xi}$ (through the operator $\sigma_{-\xi}$) once we endow it with the norm:
\begin{equation}\label{A.24}
\|u\|_{\mathcal{F}_{s,\xi}}\ :=\ \left\|\sigma_{-\xi}u\right\|_{\mathcal{K}_{s,\xi}},\qquad\forall u\in\mathcal{F}_{s,\xi}.
\end{equation}

\begin{remark}\label{R.A.17}
Let us fix some $\xi\in\X^*$. 
\begin{enumerate}
\item Let us denote by:
\begin{equation}\label{A.25}
\mathscr{S}^\prime_\xi(\X)\ :=\ \left\{u\in\mathscr{S}^\prime(\X)\,\mid\,\tau_{-\gamma}u=e^{i<\xi,\gamma>}u,\forall\gamma\in\Gamma\right\}.
\end{equation}
Then $\sigma_{-\xi}:\mathscr{S}^\prime_\xi(\X)\rightarrow\mathscr{S}^\prime_\Gamma(\X)$ is an isomorphism having the inverse $\sigma_\xi$.
\item Let us notice that we can write:
$$
\mathcal{F}_{0,\xi}\ =\ \left\{u\in\mathscr{S}^\prime(\X)\,\mid\,\sigma_{-\xi}u\in L^2_{\text{\sf loc}}(\X)\cap\mathscr{S}^\prime_\Gamma(\X)\right\}\ =\ \mathscr{S}^\prime_\xi(\X)\cap L^2_{\text{\sf loc}}(\X)
$$
and conclude that we can identify $\mathcal{F}_{0,\xi}$ with $L^2(E)$ and notice that we have the equality of the norms $\|u\|_{\mathcal{F}_{0,\xi}}=\|u\|_{L^2(E)}$. In fact the isomorphism $\mathfrak{j}_\xi:\mathcal{F}_{0,\xi}\overset{\sim}{\rightarrow}L^2(E)$ is defined by taking the restriction to $E\subset\X$, i.e. $\mathfrak{j}_\xi u:=u|_E,\forall u\in\mathcal{F}_{0,\xi}$. We can obtain an explicit formula for its inverse $\mathfrak{j}_\xi^{-1}:L^2(E)\overset{\sim}{\rightarrow}\mathcal{F}_{0,\xi}$; for any $v\in L^2(E)$ we define a distribution $\tilde{v}_\xi$ that is equal to $\sigma_{-\xi}v$ on $E$ and is extended to $\X$ by $\Gamma$-periodicity. This clearly gives us a distribution from $L^2_{\text{\sf loc}}(\X)\cap\mathscr{S}^\prime_\Gamma(\X)$. One can easily see that we have $\mathfrak{j}_\xi^{-1}v=\sigma_\xi\tilde{v}_\xi$.
\item From \eqref{A.4} it follows that $<D+\xi>^s=\sigma_{-\xi}<D>^s\sigma_\xi$. Thus:
$$
\mathcal{F}_{s,\xi}\ =\ \left\{u\in\mathscr{S}^\prime(\X)\,\mid\,\sigma_{-\xi}u\in\mathcal{H}^s_{\text{\sf loc}}(\X)\cap\mathscr{S}^\prime_\Gamma(\X)\right\}\ =\ \mathscr{S}^\prime_\xi(\X)\cap\mathcal{H}^s_{\text{\sf loc}}(\X),
$$
and for any $u\in\mathcal{F}_{s,\xi}$ we have that
\begin{equation}\label{A.26}
\|u\|_{\mathcal{F}_{s,\xi}}\ =\ \left\|\sigma_{-\xi}u\right\|_{\mathcal{K}_{s,\xi}}\ =
\end{equation}
$$
=\ \left\|<D_\Gamma+\xi>^s\sigma_{-\xi}u\right\|_{L^2(\mathbb{T})}\ =\ \left\|\sigma_{-\xi}<D>^su\right\|_{L^2(E)}\ =\ \left\|<D>^su\right\|_{L^2(E)}.
$$
In particular we obtain that
\begin{equation}\label{A.27}
\mathcal{F}_{s,\xi}\ =\ \left\{u\in\mathscr{S}^\prime_\xi(\X)\,\mid\,<D>^su\in\mathcal{F}_{0,\xi}\right\},\qquad\|u\|_{\mathcal{F}_{s,\xi}}\ =\ \|<D>^su\|_{\mathcal{F}_{0,\xi}}.
\end{equation}
Moreover, if $s\geq0$ we have that $\mathcal{F}_{s,\xi}=\{u\in\mathcal{F}_{0,\xi}\,\mid\,<D>^su\in\mathcal{F}_{0,\xi}\}$.
\end{enumerate}
\end{remark}

\begin{definition}\label{D.A.18}
We shall denote by $S^m_\rho\big(\X\times\mathbb{T};\mathbb{B}(\mathcal{A}_\bullet;\mathcal{B}_\bullet)\big)$ the space of symbols $p\in S^m_\rho\big(\X^2;\mathbb{B}(\mathcal{A}_\bullet;\mathcal{B}_\bullet)\big)$ that are $\Gamma$-periodic with respect to the second variable (i.e. $p(x,y+\gamma,\xi)=p(x,y,\xi)$, $\forall(x,y)\in\X^2$, $\forall\xi\in\X^*$ and $\forall\gamma\in\Gamma$).
\end{definition}

In a similar way we define the spaces $S^m_\rho\big(\mathbb{T};\mathbb{B}(\mathcal{A}_\bullet;\mathcal{B}_\bullet)\big)$, $S^m_\rho\big(\X\times\mathbb{T};\mathbb{B}(\mathcal{A}_0;\mathcal{B}_0)\big)$, $S^m_\rho\big(\mathbb{T};\mathbb{B}(\mathcal{A}_0;\mathcal{B}_0)\big)$, $S^m_\rho\big(\X\times\mathbb{T}\big)$, $S^m_\rho(\mathbb{T})$, $S^m_{\rho,\epsilon}\big(\X\times\mathbb{T};\mathbb{B}(\mathcal{A}_\bullet;\mathcal{B}_\bullet)\big)$, $S^m_{\rho,\epsilon}\big(\X\times\mathbb{T};\mathbb{B}(\mathcal{A}_0;\mathcal{B}_0)\big)$, $S^m_{\rho,\epsilon}\big(\X\times\mathbb{T}\big)$. Let us notice that we have an evident identification of $S^m_{\rho,\epsilon}\big(\X;\mathbb{B}(\mathcal{A}_\bullet;\mathcal{B}_\bullet)\big)$ with a subspace of $S^m_{\rho,\epsilon}\big(\X\times\mathbb{T};\mathbb{B}(\mathcal{A}_\bullet;\mathcal{B}_\bullet)\big)$.

\begin{lemma}\label{L.A.18}
Under the hypothesis of Proposition \ref{P.A.7}, for any $a\in\X$ we have the equality:
\begin{equation}\label{A.28}
\tau_a\mathfrak{Op}^{A}(p)\ =\ \mathfrak{Op}^{\tau_aA}\big((\tau_a\otimes\id)p\big)\tau_a.
\end{equation}
\end{lemma}
\begin{proof}
We start from equality \eqref{A.8} with $u\in\mathscr{S}\big(\X;\mathcal{A}_0\big)$ and we get:
$$
\left[\tau_a\mathfrak{Op}^{A}(p)u\right](x)\ =\ \int_\Xi e^{i<\eta,x-a-y>}\omega_A(x-a,y)\,p\big(\frac{x-a+y}{2},\eta\big)u(y)\,dy\,\dbar\eta\ =
$$
$$
=\ \int_\Xi e^{i<\eta,x-y>}\omega_A(x-a,y-a)\,p\big(\frac{x+y}{2}-a,\eta\big)u(y-a)\,dy\,\dbar\eta,\qquad\forall x\in\X.
$$
First let us recall that:
\begin{equation}\label{A.29}
\omega_A(x,y)\ =\ e^{-i\int_{[x,y]}A}\ =\ \exp\left\{i\left\langle(x-y),\int_0^1A\big((1-s)x+sy\big)ds\right\rangle\right\}.
\end{equation}
Let us notice that
$$
\int_{[x-a,y-a]}\hspace*{-1cm}A\hspace*{0.8cm}\ =\ -\left\langle(x-y),\int_0^1A\big((1-s)x+sy-a\big)ds\right\rangle
$$
and thus $\omega_A(x-a,y-a)=\omega_{\tau_aA}(x,y)$.
\end{proof}

\begin{lemma}\label{L.A.19}
For any symbol $p\in S^m_1(\mathbb{T})$ the pseudodifferential operator $P:=\mathfrak{Op}(p)$ induces on $\mathbb{T}$ an operator $P_\Gamma\in\mathbb{B}(\mathcal{K}_{s+m,0};\mathcal{K}_{s,0})$ for any $s\in\mathbb{R}$ and the application $S^m_1(\mathbb{T})\ni p\mapsto P_\Gamma\in\mathbb{B}(\mathcal{K}_{s+m,0};\mathcal{K}_{s,0})$ is continuous.
\end{lemma}
\begin{proof}
From equality \eqref{A.28} with $A=0$ and from the fact that $(\tau_\gamma\otimes\id)p=p,\forall\gamma\in\Gamma$ we deduce that $P$ leaves $\mathscr{S}^\prime_\Gamma(\X)$ invariant and thus induces a linear and continuous operator $P_\Gamma:\mathscr{S}^\prime(\mathbb{T})\rightarrow\mathscr{S}^\prime(\mathbb{T})$. If $u\in\mathcal{K}_{s+m,0}=\mathcal{H}^{s+m}(\mathbb{T})$ we can write
$$
\left\|P_\Gamma u\right\|_{\mathcal{K}_{s,0}}\ =\ \left\|<D_\Gamma>^sP_\Gamma u\right\|_{L^2(\mathbb{T})}\ =\ \left\|<D>^sPu\right\|_{L^2(E)}\ =\ \left\|<D>^sP<D>^{-s-m}<D>^{s+m}u\right\|_{L^2(E)}.
$$
From the Weyl calculus we know that $<D>^sP<D>^{-s-m}=\mathfrak{Op}(q)$ for a well defined symbol $q\in S^0_1(\X)$ and the map $S^m_1(\mathbb{T})\ni p\mapsto q\in S^0_1(\X)$ is continuous; by Lemma \ref{L.A.13} we can find a strictly positive constant $C^\prime_0(p)$ (it is one of the defining seminorms for the topology of $S^m_1(\mathbb{T})$ and we can find a number $N\in\mathbb{N}$ (that does not depend on $p$ as seen in the proof of Lemma \ref{L.A.13}) such that
$$
\left\|<D>^sP<D>^{-s-m}v\right\|_{L^2(E)}\ \leq\ C^\prime_0(p)\|v\|_{L^2(F)},\qquad\forall v\in L^2_{\text{\sf loc}}(\X)\cap\mathscr{S}^\prime(\X),
$$
where $F:=\underset{\gamma\in\Gamma_N}{\cup}\tau_\gamma E$, and $\Gamma_N:=\{\gamma\in\Gamma\,\mid\,|\gamma|\leq N\}$. Let us consider now $v=<D>^{s+m}u\in L^2_{\text{\sf loc}}(\X)\cap\mathscr{S}^\prime_\Gamma(\X)$. We deduce that
$$
\|v\|^2_{L^2(F)}\ =\ \underset{|\gamma|\leq N}{\sum}\int_{\tau_\gamma E}|v(x)|^2dx\ \leq\ C_N^2\|v\|^2_{L^2(E)}\ =\ C_N^2\left\|<D>^{s+m}u\right\|_{L^2(E)}\ =\ C_N^2\|u\|^2_{\mathcal{K}_{s+m,0}}.
$$
We conclude that $\|P_\Gamma u\|_{\mathcal{K}_{s,0}}\ \leq\ C_NC^\prime_0(p)\|u\|_{\mathcal{K}_{s+m,0}}$.
\end{proof}

\begin{ex}\label{E.A.20}
For any symbol $p\in S^m_1(\mathbb{T})$ and for any point $\xi\in\X^*$ we know that $(\id\otimes\tau_{-\xi})p\in S^m_1(\mathbb{T})$ and due to Lemma \ref{L.A.19} the operator $P_\xi:=\mathfrak{Op}\big((\id\otimes\tau_{-\xi})p\big)$ induces on $\mathbb{T}$ a well defined operator $P_{\Gamma,\xi}\in\mathbb{B}(\mathcal{K}_{s+m,0};\mathcal{K}_{s,0}\big)$ for any $s\in\mathbb{R}$. From the same Lemma we deduce that the application $\X^*\ni\xi\mapsto P_{\Gamma,\xi}\in\mathbb{B}(\mathcal{K}_{s+m,0};\mathcal{K}_{s,0}\big)$ is continuous; taking into account that $\partial^\alpha_\xi P_\xi=\mathfrak{Op}\big((\id\otimes\tau_{-\gamma})(\id\otimes\partial^\alpha)p\big)$ we deduce that the previous application is in fact of class $C^\infty$.

Let us prove now that for any $s\in\mathbb{R}$ we have that
\begin{equation}\label{A.31}
P_{\Gamma,\xi}\ \in\ S^0_0\big(\mathbb{T};\mathbb{B}(\mathcal{K}_{s+m,\xi};\mathcal{K}_{s,\xi})\big)
\end{equation}
and the application
\begin{equation}\label{A.32}
S^m_1(\mathbb{T})\ni p\mapsto P_{\Gamma,\xi}\,\in\,S^0_0\big(\mathbb{T};\mathbb{B}(\mathcal{K}_{s+m,\xi};\mathcal{K}_{s,\xi})\big)
\end{equation}
is continuous.

These two last statements will follow once we have proved that for any $\alpha\in\mathbb{N}^d$ there exists $c_\alpha(p)$ defining seminorm of the topology of $S^m_1(\mathbb{T})$, such that
\begin{equation}\label{A.33}
\left\|\partial^\alpha_\xi P_{\Gamma,\xi}\right\|_{\mathbb{B}(\mathcal{K}_{s+m,\xi};\mathcal{K}_{s,\xi})}\ \leq\ c_\alpha(p),\ \forall\xi\in\X^*.
\end{equation}
It is clearly enough to prove the case $\alpha=0$. Then, using \eqref{A.4} we deduce that for any $u\in\mathcal{K}_{s+m,\xi}$ we have that:
$$
\left\|P_{\Gamma,\xi}u\right\|_{\mathcal{K}_{s,\xi}}\ =\ \left\|<D_\Gamma+\xi>^sP_{\Gamma,\xi}u\right\|_{L^2(\mathbb{T})}\ =\ \left\|<D+\xi>^sP_\xi u\right\|_{L^2(E)}\ =\ \left\|<D+\xi>^s\sigma_{-\xi}P\sigma_\xi u\right\|_{L^2(E)}\ =
$$
$$
=\ \left\|\sigma_{-\xi}<D>^sP\sigma_\xi u\right\|_{L^2(E)}\ =\ \left\|<D>^sP<D>^{-s-m}<D>^{s+m}\sigma_\xi u\right\|_{L^2(E)}\ =
$$
$$
=\ \left\|<D>^sP<D>^{-s-m}\sigma_\xi<D+\xi>^{s+m}u\right\|_{L^2(E)}.
$$
As in the proof of Lemma \ref{L.A.19} we deduce that 
$$
\left\|<D>^sP<D>^{-s-m}v\right\|_{L^2(E)}\ \leq\ C^\prime_0(p)\|v\|_{L^2(F)},\qquad\forall v\in L^2_{\text{\sf loc}}(\X)\cap\mathscr{S}^\prime(\X).
$$
We consider a vector $w:=<D+\xi>^{s+m}u\in L^2_{\text{\sf loc}}(\X)\cap\mathscr{S}^\prime_\Gamma(\X)$ and $v:=\sigma_\xi w$. Then
$$
\|v\|^2_{L^2(F)}\ =\ \|w\|^2_{L^2(F)}\ \leq\ C_N^2\|w\|^2_{L^2(E)}\ =\ C_N^2\left\|<D+\xi>^{s+m}u\right\|^2_{L^2(E)}\ =\ C_N^2\|u\|^2_{\mathcal{K}_{s+m,\xi}}.
$$
Thus \eqref{A.33} follows for $\alpha=0$ with $C_0(p)=C_NC_0^\prime(p)$.
\end{ex}

\begin{lemma}\label{L.A.21}
Let $p\in S^m_1(\mathbb{T})$ be a real elliptic symbol (i.e. $\exists C>0$, $\exists R>0$ such that $p(y,\eta)\geq C|\eta|^m$ for any $(y,\eta)\in\Xi$ with $|\eta|\geq R$), with $m>0$. Then the operator $P_\Gamma$ defined in Lemma \ref{L.A.19} is self-adjoint on the domain $\mathcal{K}_{m,0}$. Moreover, $P_\Gamma$ is lower semi-bounded and its graph-norm on $\mathcal{K}_{m,0}$ gives a norm equivalent to the defining norm of $\mathcal{K}_{m,0}$.
\end{lemma}
\begin{proof}
Let us first verify the symmetry of $P_\Gamma$ on $\mathcal{K}_{m,0}$. Due to the density of $\mathscr{S}(\mathbb{T})$ in $\mathcal{K}_{m,0}$ and to the fact that $P_\Gamma\in\mathbb{B}\big(\mathcal{K}_{m,0};L^2(\mathbb{T})\big)$, it is enough to verify the symmetry of $P_\Gamma$ on $\mathscr{S}(\mathbb{T})$. Let us choose two vectors $u$ and $v$ from $\mathscr{S}(\mathbb{T})$ so that we have to prove the equality $(P_\Gamma u,v)_{L^2(\mathbb{T})}=(u,P_\Gamma v))_{L^2(\mathbb{T})}$ or equivalently
\begin{equation}\label{A.34}
(P u,v)_{L^2(E)}\ =\ (u,Pv))_{L^2(E)}.
\end{equation}
Identifying $\mathscr{S}(\mathbb{T})$ with $\mathscr{E}(\X)\cap\mathscr{S}^\prime_\Gamma(\X)$ and using the definition of the operator $P$ on the space $\mathscr{S}^\prime(\X)$ one easily verifies that $Pu$ also belongs to $\mathscr{E}(\X)\cap\mathscr{S}^\prime_\Gamma(\X)$ and is explicitely given by the following oscillating integral ($\forall x\in\X$):
\begin{equation}\label{A.35}
\big(Pu\big)(x)\ =\ \int_\Xi e^{i<\eta,x-y>}p\big(\frac{x+y}{2},\eta\big)\,u(y)\,dy\,\dbar\eta\ =
\end{equation}
$$
=\ \underset{\gamma\in\Gamma}{\sum}\int_{\tau_\gamma E}\int_{\X^*}e^{i<\eta,x-y>}p\big(\frac{x+y}{2},\eta\big)\,u(y)\,dy\,\dbar\eta\ =\ \underset{\gamma\in\Gamma}{\sum}\int_{E}\int_{\X^*}e^{i<\eta,x-y+\gamma>}p\big(\frac{x+y-\gamma}{2},\eta\big)\,u(y)\,dy\,\dbar\eta,
$$
the series converging in $\mathscr{E}(\X)$. Using the $\Gamma$-periodicity of $p$ we obtain that:
$$
(Pu,v)_{L^2(E)}\ =\ \int_E\big(Pu\big)(x)\overline{v(x)}\,dx\ =\ \underset{\gamma\in\Gamma}{\sum}\int_{E}\int_{E}\int_{\X^*}e^{i<\eta,x-y+\gamma>}p\big(\frac{x+y-\gamma}{2},\eta\big)\,u(y)\overline{v(x)}\,dx\,dy\,\dbar\eta\ =
$$
$$
=\ \int_{E}u(y)\overline{\left[\underset{\gamma\in\Gamma}{\sum}\int_{E}\int_{\X^*} e^{i<\eta,y-x-\gamma>}p\big(\frac{x+y+\gamma}{2},\eta\big)v(x)dx\,\dbar\eta\right]}dy\ =\ \int_Eu(y)\overline{\big(Pv\big)(y)}\,dy\ =(u,Pv)_{L^2(E)},
$$
and thus we proved the equality \eqref{A.34}.

In order to prove the self-adjointness of $P_\Gamma$ let us choose some vector $u\in\mathcal{D}\big(P_\Gamma^*\big)$; thus it exists $f\in L^2(\mathbb{T})$ such that we have the equality $(P_\Gamma \varphi,u)_{L^2(\mathbb{T})}=(\varphi,f)_{L^2(\mathbb{T})}$, $\forall\varphi\in\mathscr{S}(\mathbb{T})$. Using now the facts that $\mathscr{S}(\mathbb{T})$ is dense in $\mathscr{S}^\prime(\mathbb{T})$ and $P_\Gamma$ is symmetric on $\mathscr{S}(\mathbb{T})$, we deduce that 
$$
(\varphi,f)_{\mathbb{T}}\ =\ \left(P_\Gamma\varphi,u\right)_{\mathbb{T}}\ =\ \left(\varphi,P_\Gamma u\right)_{\mathbb{T}},\qquad\forall\varphi\in\mathscr{S}(\mathbb{T})
$$
and thus we obtain the equality $P_\Gamma u=f$ in $\mathscr{S}^\prime(\mathbb{T})$. By hypothesis $P_\Gamma$ is an elliptic pseudifferential operator of strictly positive order $m$, on the compact manifold $\mathbb{T}$, so that the usual regularity results imply that $u\in\mathcal{K}_{m,0}=\mathcal{D}(P_\Gamma)$. In conclusion $P_\Gamma$ is self-adjoint on the domain $\mathcal{K}_{m,0}$.

The lower semiboundedness property follows from the G\aa rding inequality:
\begin{equation}\label{A.36}
\left(P_\Gamma u,u\right)_{L^2(\mathbb{T})}\ \geq\ C^{-1}\|u\|^2_{\mathcal{K}_{m/2,0}}\,-\,C\|u\|^2_{L^2(\mathbb{T})},\qquad\forall u\in\mathcal{K}_{m,0}.
\end{equation}

The equivalence of the norms stated as the last point of the Lemma follows from the Closed Graph Theorem.
\end{proof}

\begin{remark}\label{R.A.22}
Under the Hypothesis of Lemma \ref{L.A.21}, the same proof also shows that for any $\xi\in\X^*$, the operator $P_{\Gamma,\xi}$ from Example \ref{E.A.20} is self-adjoint and lower semibounded on $L^2(\mathbb{T})$ on the domain $\mathcal{K}_{m,\xi}$. As in Remark \ref{R.A.17} we can identify $\mathcal{K}_{m,\xi}$ with $\mathcal{H}^m_{\text{\sf loc}}(\X)\cap\mathscr{S}^\prime_\Gamma(\X)$ (endowed with the norm $\left\|<D+\xi>^mu\right\|_{L^2(E)}$) and thus we can deduce that the operator $P_\xi$ is a self-adjoint operator in the space $L^2_{\text{\sf loc}}(\X)\cap\mathscr{S}^\prime_\Gamma(\X)$ on the domain $\mathcal{K}_{m,\xi}$. From \eqref{A.4} we know that $P=\sigma_\xi P_\xi\sigma_{-\xi}$ and we also know that $\sigma_\xi:\mathcal{K}_{s,\xi}\rightarrow\mathcal{F}_{s,\xi}$ is a unitary operator for any $s\in\mathbb{R}$ and for any $\xi\in\X^*$ and we conclude that the operator induced by $P$ in $\mathcal{F}_{0,\xi}$ is unitarily equivalent with the operator induced by $P_\xi$ in $\mathcal{K}_{0,\xi}\cong L^2_{\text{\sf loc}}(\X)\cap\mathscr{S}^\prime_\Gamma(\X)$. It follows that the operator $P$ acting in $\mathcal{F}_{0,\xi}$ with domain $\mathcal{F}_{m,\xi}$ is self-adjoint and lower semibounded.
\end{remark}

\subsection{Properties of magnetic pseudodifferential operators with operator-valued symbols}

\begin{theorem}\label{T.A.23}
Let us first consider the composition operation. Suppose chosen three families of Hilbert spaces with temperate variation $\{\mathcal{A}_\xi\}_{\xi\in\X^*}$, $\{\mathcal{B}_\xi\}_{\xi\in\X^*}$ and $\{\mathcal{C}_\xi\}_{\xi\in\X^*}$, two families of symbols $\{p_\epsilon\}_{|\epsilon|\leq\epsilon_0}\in S^m_{\rho,\epsilon}\big(\X;\mathbb{B}(\mathcal{B}_\bullet;\mathcal{C}_\bullet)\big)$ and $\{q_\epsilon\}_{|\epsilon|\leq\epsilon_0}\in S^{m^\prime}_{\rho,\epsilon}\big(\X;\mathbb{B}(\mathcal{A}_\bullet;\mathcal{B}_\bullet)\big)$ and a family of magnetic fields $\{B_\epsilon\}_{|\epsilon|\leq\epsilon_0}$ satisfying Hypothesis H.1 from Section \ref{S.1} with an associated family of vector potentials $\{A_\epsilon\}_{|\epsilon|\leq\epsilon_0}$ given by \eqref{0.28}. Then
\begin{enumerate}
\item There exist a family of symbols 
$$
\{p_\epsilon\sharp^{B_\epsilon}q_\epsilon\}_{|\epsilon|\leq\epsilon_0}\ \in\ S^{m+m^\prime}_{\rho,\epsilon}\big(\X;\mathbb{B}(\mathcal{A}_\bullet;\mathcal{C}_\bullet)\big),\quad\text{such that}\quad\mathfrak{Op}^{A_\epsilon}(p_\epsilon)\mathfrak{Op}^{A_\epsilon}(q_\epsilon)=\mathfrak{Op}^{A_\epsilon}(p_\epsilon\sharp^{B_\epsilon}q_\epsilon).
$$
\item The application 
\begin{equation}\label{A.37}
S^m_{\rho}\big(\X;\mathbb{B}(\mathcal{B}_\bullet;\mathcal{C}_\bullet)\big)\times S^{m^\prime}_{\rho,}\big(\X;\mathbb{B}(\mathcal{A}_\bullet;\mathcal{B}_\bullet)\big)\ni(p_\epsilon,q_\epsilon)\ \mapsto\ p_\epsilon\sharp^{B_\epsilon}q_\epsilon\in S^{m+m^\prime}_{\rho,}\big(\X;\mathbb{B}(\mathcal{A}_\bullet;\mathcal{C}_\bullet)\big)
\end{equation}
is continuous uniformly with respect to $\epsilon\in[-\epsilon_0,\epsilon_0]$.
\item There exists a family of symbols $\{r_\epsilon\}_{|\epsilon|\leq\epsilon_0}\in S^{m+m^\prime-\rho}_{\rho,\epsilon}\big(\X;\mathbb{B}(\mathcal{A}_\bullet;\mathcal{C}_\bullet)\big)$ having the following properties:
\begin{equation}\label{A.38}
\underset{\epsilon\rightarrow0}{\lim}\,r_\epsilon\ =\ 0\quad\text{in}\ S^{m+m^\prime-\rho}_{\rho}\big(\X;\mathbb{B}(\mathcal{A}_\bullet;\mathcal{C}_\bullet)\big)
\end{equation}
\begin{equation}\label{A.39}
p_\epsilon\sharp^{B_\epsilon}q_\epsilon\ =\ p_\epsilon\cdot q_\epsilon\ +\ r_\epsilon,\qquad\forall\epsilon\in[-\epsilon_0,\epsilon_0].
\end{equation}
\end{enumerate}
\end{theorem}
\begin{proof}
By a standard cut-off procedure, as in the proof of Proposition \ref{P.A.7} we may reduce the problem to the case of symbols with compact support in both arguments $(x,\xi)\in\Xi$. 

A direct computation using Stokes formula and the fact that $dB_\epsilon=0$ for any $\epsilon\in[-\epsilon_0,\epsilon_0]$ shows that for point (1) of the Theorem we may take the definition of the composition operation to be the following well defined integral formula:
\begin{equation}\label{A.40}
\big(p_\epsilon\sharp^{B_\epsilon}q_\epsilon\big)(X)\ =\ \int_\Xi\int_\Xi e^{-2i\dbl Y,Z\dbr}\,\omega^{B_\epsilon}(x,y,z)\,p_\epsilon(X-Y)q_\epsilon(X-Z)\,\dbar Y\,\dbar Z,
\end{equation}
where we used the notation $X:=(x,\xi)$, $Y:=(y,\eta)$, $Z:=(z,\zeta)$, $\dbl Y,Z\dbr:=<\eta,z>-<\zeta,y>$ and
$$
\omega^{B_\epsilon}(x,y,z):=e^{-iF_\epsilon(x,y,z)},\qquad F_\epsilon(x,y,z):=\int_{<x-y+z,x-y-z,x+y-z>}\hspace*{-2cm}B_\epsilon
$$
with $<a,b,c>$ the triangle with vertices $a\in\X$, $b\in\X$ and $c\in\X$.

A direct computation (see for example Lemma 1.1 from \cite{IMP1}) shows that all the vectors $\nabla_xF_\epsilon$, $\nabla_yF_\epsilon$ and $\nabla_zF_\epsilon$ have the form $C_\epsilon(x,y,z)y\,+\,D_\epsilon(x,y,z)z$ with $C_\epsilon$ and $D_\epsilon$ functions of class $BC^\infty\big(\X^3;\mathbb{B}(\X)\big)$ satisfying the conditions $\underset{\epsilon\rightarrow0}{\lim}\,C_\epsilon\,=\ \underset{\epsilon\rightarrow0}{\lim}\,D_\epsilon\,=\ 0$ in $BC^\infty\big(\X^3;\mathbb{B}(\X)\big)$. It follows easily then that the derivatives of $\omega^{B_\epsilon}(x,y,z)$ of order at least 1 are finite linear combinations of terms of the form $C_{(\alpha,\beta);\epsilon}y^\alpha z^\beta\omega^{B_\epsilon}(x,y,z)$ with $C_{(\alpha,\beta);\epsilon}\in BC^\infty(\X^3)$ satisfying the property $\underset{\epsilon\rightarrow0}{\lim}\,C_{(\alpha,\beta);\epsilon}=0$ in $BC^\infty(\X^3)$. Applying now the usual integration by parts with respect to the variables $\{y,z,\eta,\zeta\}$ we obtain that for any $N_j\in\mathbb{N}$ ($1\leq j\leq4$) and for any $X\in\Xi$ the following equality is true:
\begin{equation}\label{A.41}
\big(p_\epsilon\sharp^{B_\epsilon}q_\epsilon\big)(X)\ =\ \int_\Xi\int_\Xi e^{-2i\dbl Y,Z\dbr}<\eta>^{-2N_1}<\zeta>^{-2N_2}\big(\id-\frac{1}{4}\Delta_z\big)^{N_1}\big(\id-\frac{1}{4}\Delta_y\big)^{N_2}\times
\end{equation}
$$
\times\left[<y>^{-2N_3}<z>^{-2N_4}\omega^{B_\epsilon}(x,y,z)\left(\big(\id-\frac{1}{4}\Delta_\eta\big)^{N_4}p_\epsilon(X-Y)\right)\left(\big(\id-\frac{1}{4}\Delta_\zeta\big)^{N_3}q_\epsilon(X-Z)\right)\right]\dbar Y\,\dbar Z.
$$
First we apply the differentiation operators on the functions they act on and then we eliminate all the monomials of the form $y^\alpha z^\beta$ that appear from the differentiation of $\omega^{B_\epsilon}$ by integrating by parts using the formulas:
$$
y_je^{-2i\dbl Y,Z\dbr}\ =\ \frac{1}{2i}\partial_{\zeta_j}e^{-2i\dbl Y,Z\dbr},\quad z_je^{-2i\dbl Y,Z\dbr}\ =\ -\frac{1}{2i}\partial_{\eta_j}e^{-2i\dbl Y,Z\dbr}.
$$
These computations allow us to obtain the following estimation (for some $C>0$ and $N\in\mathbb{N}$):
\begin{equation}\label{A.42}
\left\|\big(p_\epsilon\sharp^{B_\epsilon}q_\epsilon\big)(X)\right\|_{\mathbb{B}(\mathcal{A}_\xi;\mathcal{C}_\xi)}\ \leq
\end{equation}
\begin{scriptsize}
$$
\leq\ C\underset{|\alpha|,|\beta|,|\gamma|,|\delta|\leq N}{\max}\ \int_\Xi\int_\Xi<\eta>^{-2N_1}<\zeta>^{-2N_2}<y>^{-2N_3}<z>^{-2N_4}\left\|\partial^\alpha_x\partial^\beta_\xi p_\epsilon(X-Y)\right\|_{\mathbb{B}(\mathcal{B}_\xi;\mathcal{C}_\xi)}\left\|\partial^\gamma_x\partial^\delta_\xi q_\epsilon(X-Z)\right\|_{\mathbb{B}(\mathcal{A}_\xi;\mathcal{B}_\xi)}\,\dbar Y\,\dbar Z,
$$
\end{scriptsize}
for any $\epsilon\in[-\epsilon_0,\epsilon_0]$. We use now \eqref{A.1} and \eqref{A.3} and obtain the following estimations valid for any $\epsilon\in[-\epsilon_0,\epsilon_0]$:
\begin{equation}\label{A.43}
\left\|\partial^\alpha_x\partial^\beta_\xi p_\epsilon(X-Y)\right\|_{\mathbb{B}(\mathcal{B}_\xi;\mathcal{C}_\xi)}\ \leq\ C<\eta>^{2M}\left\|\partial^\alpha_x\partial^\beta_\xi p_\epsilon(X-Y)\right\|_{\mathbb{B}(\mathcal{B}_{\xi-\eta};\mathcal{C}_{\xi-\eta})}\ \leq
\end{equation}
$$
\leq\ C<\eta>^{2M}<\xi-\eta>^{m-\rho|\beta|}\left(\underset{Z\in\Xi}{\sup}<\zeta>^{-m+\rho|\beta|}\left\|\big(\partial^\alpha_z\partial^\beta_\zeta p_\epsilon\big)(Z)\right\|_{\mathbb{B}(\mathcal{B}_{\zeta};\mathcal{C}_{\zeta})}\right).
$$
Repeating the same computations for the derivatives of $q_\epsilon$ and choosing suitable large exponents $N_j$ ($1\leq j\leq4$) in \eqref{A.42} we deduce the existence of two defining seminorms $|\cdot|_{n_1}$ and respectively $|\cdot|_{n_2}$ on the Fr\'{e}chet space $S^m_\rho\big(\X;\mathbb{B}(\mathcal{B}_\bullet;\mathcal{C}_\bullet)\big)$ and respectively on $S^{m^\prime}_\rho\big(\X;\mathcal{A}_\bullet;\mathcal{B}_\bullet)\big)$ such that we have the estimation:
\begin{equation}\label{A.44}
\underset{X\in\Xi}{\sup}<\xi>^{-(m+m^\prime)}\left\|\big(p_\epsilon\sharp^{B_\epsilon}q_\epsilon\big)(X)\right\|_{\mathbb{B}(\mathcal{A}_\xi;\mathcal{C}_\xi)}\ \leq\ |p_\epsilon|_{n_1}\,|q_\epsilon|_{n_2},\qquad\forall\epsilon\in[-\epsilon_0,\epsilon_0].
\end{equation}

The derivatives of $p_\epsilon\sharp^{B_\epsilon}q_\epsilon$ can be estimated in a similar way in order to conclude that $p_\epsilon\sharp^{B_\epsilon}q_\epsilon\in S^{m+m^\prime}_\rho\big(\X;\mathcal{A}_\bullet;\mathcal{C}_\bullet)\big)$ uniformly with respect to $\epsilon\in[-\epsilon_0,\epsilon_0]$ and that property (2) is valid.

Considering now the family of symbols $\{p_\epsilon\sharp^{B_\epsilon}q_\epsilon\}_{|\epsilon|\leq\epsilon_0}$, the hypothesis (2) and (3) from the Definition \ref{D.A.5} follow easily from \eqref{A.38} and \eqref{A.39}. 

In conclusion there is only point (3) that remains to be proved. By the same arguments as above we can once again assume that the symbols $p_\epsilon$ and $q_\epsilon$ have compact support. We begin by using \eqref{A.40} in the equality:
\begin{equation}\label{A.45}
p_\epsilon(X-Y)q_\epsilon(X-Z)\ =
\end{equation}
$$
=\ p_\epsilon(X)q_\epsilon(X)\,-\,\int_0^1\left[\left\langle Y,\nabla_Xp_\epsilon(X-tY)\right\rangle q_\epsilon(X-tZ)\,+\,p_\epsilon(X-tY)\left\langle Z,\nabla_X q_\epsilon(X-tZ)\right\rangle\right]dt.
$$
The first term on the right side of the equality \eqref{A.45} will produce the term $p_\epsilon q_\epsilon$ in the equality \eqref{A.39} (see also the Lemma 2.1 from \cite{IMP1}). Let us study now the term obtained by replacing \eqref{A.40} into \eqref{A.45}. We eliminate $Y$ and $Z$ by integration by parts as in the beginning of this proof taking also into account the following identities:
$$
\eta_je^{-2i\dbl Y,Z\dbr}\ =\ -\frac{1}{2i}\partial_{z_j}e^{-2i\dbl Y,Z\dbr},\quad\zeta_je^{-2i\dbl Y,Z\dbr}\ =\ \frac{1}{2i}\partial_{y_j}e^{-2i\dbl Y,Z\dbr}.
$$
These operations will produce derivatives of $p_\epsilon$ and $q_\epsilon$ with respect to $x\in\X$, that go to 0 for $\epsilon\rightarrow0$ in their symbol spaces topology and derivatives of $F_\epsilon$ with respect to $y$ and $z$; but these derivatives may be once again transformed by integrations by parts into factors of the form $C_\epsilon\in BC^\infty(\X^3)$ having limit 0 for $\epsilon\rightarrow0$ as elements from $BC^\infty(\X^3)$. Thus, the estimations proved in the first part of the proof  imply that the equality \eqref{A.39} holds with $r_\epsilon=\int_0^1s_\epsilon(t)dt$ with $s_\epsilon(t)\in S^{m+m^\prime-\rho}_\rho\big(\X;\mathbb{B}(\mathcal{A}_\bullet;\mathcal{C}_\bullet)\big)$ uniformly with respect to $(\epsilon,t)\in[-\epsilon_0,\epsilon_0]\times[0,1]$ and $\underset{\epsilon\rightarrow0}{\lim}\,s_\epsilon(t)\,=\,0$ in $S^{m+m^\prime-\rho}_\rho\big(\X;\mathbb{B}(\mathcal{A}_\bullet;\mathcal{C}_\bullet)\big)$ uniformly with respect to $t\in[0,1]$. We conclude that $r_\epsilon$ has the properties stated in the Theorem.
\end{proof}

\begin{remark}\label{R.A.24}
The proof of Theorem \ref{T.A.23} also implies the following fact (that we shall use in the paper):
{\it the operation $\sharp^{B_\epsilon}$ is well defined also as operation: $S^m_\rho\big(\X;\mathbb{B}(\mathcal{B}_\bullet;\mathcal{C}_\bullet)\big)\times S^{m^\prime}_\rho\big(\X;\mathbb{B}(\mathcal{A}_\bullet;\mathcal{B}_\bullet)\big)\rightarrow S^{m+m^\prime}_\rho\big(\X;\mathbb{B}(\mathcal{A}_\bullet;\mathcal{C}_\bullet)\big)$ being bilinear and continuous uniformly with respect to $\epsilon\in[-\epsilon_0,\epsilon_0]$.}
\end{remark}

\begin{remark}\label{R.A.25}
As in \cite{IMP2} one can define a family of symbols $\{q_{s,\epsilon}\}_{(s,\epsilon)\in\mathbb{R}\times[-\epsilon_0,\epsilon_0]}$ having the following properties:
\begin{enumerate}
\item $q_{s,\epsilon}\,\in\,S^s_1(\X)$ uniformly with respect to $\epsilon\in[-\epsilon_0,\epsilon_0]$,
\item $q_{s,\epsilon}\sharp^{B_\epsilon}q_{-s,\epsilon}\,=\,1$,
\item $\forall s>0$ we have that $q_{s,\epsilon}(x,\xi)=<\xi>^s+\mu$ with some sufficiently large $\mu>0$ and $q_{0,\epsilon}=1$.
\end{enumerate}
\end{remark}
Evidently that for any Hilbert space $\mathcal{A}$ we can identify the symbol $q_{s,\epsilon}$ with the operator-valued symbol $q_{s,\epsilon}\id_{\mathcal{A}}$ and thus we may consider that $q_{s,\epsilon}\in S^s_1\big(\X;\mathbb{B}(\mathcal{A})\big)$ uniformly with respect to $\epsilon\in[-\epsilon_0,\epsilon_0]$.  We shall use the notation $Q_{s,\epsilon}:=\mathfrak{Op}^{A_\epsilon}(q_{s,\epsilon})$.

\begin{proposition}\label{P.A.26}
Suppose given two Hilbert spaces $\mathcal{A}$ and $\mathcal{B}$ and for any $\epsilon\in[-\epsilon_0,\epsilon_0]$ a symbol $p_\epsilon\in S^m_0\big(\X;\mathbb{B}(\mathcal{A};\mathcal{B})\big)$, uniformly in $\epsilon\in[-\epsilon_0,\epsilon_0]$. Then for any $s\in\mathbb{R}$ the operator $\mathfrak{Op}^{A_\epsilon}(p_\epsilon)$ belongs to the space $\mathbb{B}\big(\mathcal{H}^{s+m}_{A_\epsilon}(\X)\otimes\mathcal{A};\mathcal{H}^s_{A_\epsilon}(\X)\otimes\mathcal{B}\big)$ uniformly with respect to $\epsilon\in[-\epsilon_0,\epsilon_0]$. Moreover, the norm of $\mathfrak{Op}^{A_\epsilon}(p_\epsilon)$ in the above Banach space is bounded from above by a seminorm of $p_\epsilon$ in $S^m_0\big(\X;\mathbb{B}(\mathcal{A};\mathcal{B})\big)$, uniformly with respect to $\epsilon\in[-\epsilon_0,\epsilon_0]$.
\end{proposition}
\begin{proof}
For $m=s=0$ the proposition may be proved by the same arguments as in the scalar case: $\mathcal{A}=\mathcal{B}=\mathbb{C}$ (see for example \cite{IMP1}). Also using the results from \cite{IMP1} we can see that for any $t\in\mathbb{R}$ the operator $Q_{s,\epsilon}$ belongs to the space $\mathbb{B}\big(\mathcal{H}^{t+s}_{A_\epsilon}(\X);\mathcal{H}^t_{A_\epsilon}(\X)\big)$ uniformly with respect to $\epsilon\in[-\epsilon_0,\epsilon_0]$. The proof of the general case folows now from the following identity:
$$
\mathfrak{Op}^{A_\epsilon}(p_\epsilon)\ =\ Q_{-s,\epsilon}\,Q_{s,\epsilon}\,\mathfrak{Op}^{A_\epsilon}(p_\epsilon)\,Q_{-(s+m),\epsilon}\,Q_{s+m,\epsilon}
$$
and the fact that $q_{s,\epsilon}\sharp^{B_\epsilon}p_\epsilon\sharp^{B_\epsilon}q_{-(s+m),\epsilon}$ is a symbol of class $S^0_0\big(\X;\mathbb{B}(\mathcal{A};\mathcal{B})\big)$ uniformly with respect to $\epsilon\in[-\epsilon_0,\epsilon_0]$ (as implied by the Remark \ref{R.A.24}).
\end{proof}

\begin{proposition}\label{P.A.27}
Suppose given a Hilbert space $\mathcal{A}$ and a bounded subset $\{p_\epsilon\}_{|\epsilon|\leq\epsilon_0}\subset S^0_\rho\big(\X;\mathbb{B}(\mathcal{A})\big)$ such that $\underset{\epsilon\rightarrow0}{\lim}\,p_\epsilon=0$ in this space of symbols. Then, for sufficiently small $\epsilon_0>0$ the following statements are true:
\begin{enumerate}
\item $\id+\mathfrak{Op}^{A_\epsilon}(p_\epsilon)$ is invertible in $\mathbb{B}\big(L^2(\X)\otimes\mathcal{A}\big)$ for any $\epsilon\in[-\epsilon_0,\epsilon_0]$.
\item It exists a bounded subset of symbols $\{q_\epsilon\}_{|\epsilon|\leq\epsilon_0}$ from $S^0_\rho\big(\X;\mathbb{B}(\mathcal{A})\big)$ such that $\underset{\epsilon\rightarrow0}{\lim}\,q_\epsilon=0$ in $S^0_\rho\big(\X;\mathbb{B}(\mathcal{A})\big)$ and the following equality holds:
\begin{equation}\label{A.46}
\big[\id\,+\,\mathfrak{Op}^{A_\epsilon}(p_\epsilon)\big]^{-1}\ =\ \id\,+\,\mathfrak{Op}^{A_\epsilon}(q_\epsilon).
\end{equation}
\end{enumerate}
\end{proposition}
\begin{proof}
The first statement above is quite evident once we notice that following Proposition \ref{P.A.26} we can choose some small enough $\epsilon_0>0$ such that $\left\|\mathfrak{Op}^{A_\epsilon}(p_\epsilon)\right\|_{\mathbb{B}(L^2(\X)\otimes\mathcal{A})}\leq(1/2)$ for any $\epsilon\in[-\epsilon_0,\epsilon_0]$. By a straightforward modification of the arguments given in \textsection 6.1 from \cite{IMP2} in order to deal with operator-valued symbols, we deduce that there exists a bounded subset $\{r_\epsilon\}_{|\epsilon|\leq\epsilon_0}$ in $S^0_\rho\big(\X;\mathbb{B}(\mathcal{A})\big)$ such that:
\begin{equation}\label{A.47}
\big[\id\,+\,\mathfrak{Op}^{A_\epsilon}(p_\epsilon)\big]^{-1}\ =\ \mathfrak{Op}^{A_\epsilon}(r_\epsilon).
\end{equation}
The equality \eqref{A.46} follows if we notice that
$$
\big[\id\,+\,\mathfrak{Op}^{A_\epsilon}(p_\epsilon)\big]^{-1}\ =\ \id\,-\,\mathfrak{Op}^{A_\epsilon}(p_\epsilon)\big[\id\,+\,\mathfrak{Op}^{A_\epsilon}(p_\epsilon)\big]^{-1}\ =\ \id\,-\,\mathfrak{Op}^{A_\epsilon}(p_\epsilon)\mathfrak{Op}^{A_\epsilon}(r_\epsilon)
$$
and also that Remark \ref{R.A.24} implies that $q_\epsilon:=-p_\epsilon\sharp^{B_\epsilon}r_\epsilon$ has all the stated properties.
\end{proof}

\subsection{Relativistic Hamiltonians}

We shall close this subsection with the study of a property that connects the two relativistic Schr\"{o}dinger Hamiltonians $\mathfrak{Op}^{A_\epsilon}(h_R)$ and $\big[\mathfrak{Op}^{A_\epsilon}(h_{NR})\big]^{1/2}$ with $h_R(x,\xi):=<\xi>\equiv\sqrt{1+|\xi|^2}$ and $h_{NR}(x,\xi):=1+|\xi|^2\equiv<\xi>^2$. We shall use some arguments presented in \textsection 6.3 of \cite{IMP2}.
\begin{proposition}\label{P.A.28}
There exists a bounded subset $\{q_\epsilon\}_{|\epsilon|\leq\epsilon_0}$ of symbols from $S^0_1(\X)$ such that $\underset{\epsilon\rightarrow0}{\lim}\,q_\epsilon=0$ in $S^0_1(\X)$ and 
\begin{equation}\label{A.48}
\big[\mathfrak{Op}^{A_\epsilon}(h_{NR})\big]^{1/2}\ =\ \mathfrak{Op}^{A_\epsilon}(h_R)\,+\,\mathfrak{Op}^{A_\epsilon}(q_\epsilon).
\end{equation}
\end{proposition}
\begin{proof}
Following \cite{IMP2}, if we denote by $p^-$ the inverse of the symbol $p$ with respect to the composition $\sharp^{B_\epsilon}$,
\begin{equation}\label{A.49}
\big[\mathfrak{Op}^{A_\epsilon}(h_{NR})\big]^{1/2}\ =\ \mathfrak{Op}^{A_\epsilon}(h_{NR})\mathfrak{Op}^{A_\epsilon}\left(-\frac{1}{2\pi i}\int_{-i\infty}^{i\infty}{\rm z}^{-1/2}\big(<\xi>^2-{\rm z}\big)^-\,d{\rm z}\right).
\end{equation}
Recalling the proof of point (3) in Theorem \ref{T.A.23} we can easily prove that:
\begin{equation}\label{A.50}
\big(<\xi>^2-{\rm z}\big)\,\sharp^{B_\epsilon}\,\big(<\xi>^2-{\rm z}\big)^{-1}\ =\ 1\,+\,r_{\epsilon,{\rm z}}
\end{equation}
where $<{\rm z}>r_{\epsilon,{\rm z}}\in S^0_1(\X)$ uniformly for $(\epsilon,{\rm z})\in[-\epsilon_0,\epsilon_0]\times i\mathbb{R}$ and $\underset{\epsilon\rightarrow0}{\lim}\,<{\rm z}>r_{\epsilon,{\rm z}}=0$ in $S^0_1(\X)$ uniformly with respect to ${\rm z}\in i\mathbb{R}$. Following the proof of Proposition \ref{P.A.27}, for $\epsilon_0>0$ sufficiently small there exists a symbol $f_{\epsilon,{\rm z}}$ such that $<{\rm z}>f_{\epsilon,{\rm z}}\in S^0_1(\X)$ uniformly with respect to $(\epsilon,{\rm z})\in[-\epsilon_0,\epsilon_0]\times i\mathbb{R}$, $\underset{\epsilon\rightarrow0}{\lim}<{\rm z}>f_{\epsilon,{\rm z}}=0$ in $S^0_1(\X)$ uniformly with respect to ${\rm z}\in i\mathbb{R}$ and we also have
\begin{equation}\label{A.52}
\big(1\,+\,r_{\epsilon,{\rm z}}\big)^-\ =\ 1\,+\,f_{\epsilon,{\rm z}}.
\end{equation}

From \eqref{A.49} and from the properties of the symbol $r_{\epsilon,{\rm z}}$ it follows that we can define:
\begin{equation}\label{A.53}
\big(<\xi>^2-{\rm z}\big)^-\ :=\ \big(<\xi>^2-{\rm z}\big)^{-1}\,\sharp^{B_\epsilon}\,\big(1\,+\,f_{\epsilon,{\rm z}}\big)\ =\ \big(<\xi>^2-{\rm z}\big)^{-1}\,+\,\big(<\xi>^2-{\rm z}\big)^{-1}\,\sharp^{B_\epsilon}\,f_{\epsilon,{\rm z}}.
\end{equation}
Using \eqref{A.53} in \eqref{A.49} we notice that the term $\big(<\xi>^2-{\rm z}\big)^{-1}$ produces by magnetic quantization a term of the form $\mathfrak{Op}^{A_\epsilon}(h_{NR})\mathfrak{Op}^{A_\epsilon}(h_{R}^{-1})$ and using Theorem \ref{T.A.23} this operator may be put in the form $\mathfrak{Op}^{A_\epsilon}(h_{R})+\mathfrak{Op}^{A_\epsilon}(q_\epsilon^\prime)$ where $q_\epsilon^\prime\in S^0_1(\X)$ uniformly with respect to $\epsilon\in[-\epsilon_0,\epsilon_0]$ with $\underset{\epsilon\rightarrow0}{\lim}\,q_\epsilon^\prime=0$ in $S^0_1(\X)$. If we notice that $h_{NR}\sharp^{B_\epsilon}\big(h_{NR}-{\rm z}\big)^{-1}\in S^0_1(\X)$ uniformly with respect to $(\epsilon,{\rm z})\in[-\epsilon_0,\epsilon_0]\times i\mathbb{R}$, then we can see that the second term of \eqref{A.53} gives in \eqref{A.49} by magnetic quantization an expression of the form $\mathfrak{Op}^{A_\epsilon}(q_\epsilon^{\prime\prime})$ with $q_\epsilon^{\prime\prime}\in S^0_1(\X)$ uniformly with respect to $\epsilon\in[-\epsilon_0,\epsilon_0]$ and such that $\underset{\epsilon\rightarrow0}{\lim}\,q_\epsilon^{\prime\prime}=0$ in $S^0_1(\X)$.
\end{proof}

 {\bf Acknowledgements:}
R.Purice aknowledges the CNCSIS support
under the Ideas Programme, PCCE project no. 55/2008 {\it Sisteme
diferen\c{t}iale \^{\i}n analiza neliniar\u{a} \c{s}i aplica\c{t}ii}.
\newpage

E-mail:Viorel.Iftimie@imar.ro, Radu.Purice@imar.ro

\end{document}